\documentclass[sigconf, nonacm]{acmart}

\usepackage{soul,color}
\usepackage{xcolor,colortbl}
\usepackage{makecell}
\usepackage{cuted}
\usepackage{flushend}
\usepackage{subfigure}
\usepackage{enumerate}
\usepackage[linesnumbered,ruled,vlined]{algorithm2e}

\SetCommentSty{mycommfont}
\SetKwInput{KwInput}{Input}               
\SetKwInput{KwOutput}{Output}  
\usepackage{booktabs}

\usepackage{newtxmath}
\usepackage{bm}
\usepackage{graphicx}
\usepackage{hyperref}
\usepackage{url}

\usepackage{tikz}
\usetikzlibrary{positioning,chains,fit,shapes,calc}
\usetikzlibrary{snakes}
\usepackage{subfigure} 
\usepackage{lipsum}

\usepackage{float}
\usepackage{bm}

\begin{document}
\title{sGrapp: Butterfly Approximation in Streaming  Graphs}

\author{Aida Sheshbolouki}
\affiliation{%
  \institution{University of Waterloo}}
\email{aida.sheshbolouki@uwaterloo.ca}

\author{M. Tamer {\"O}zsu}
\affiliation{%
  \institution{University of Waterloo}}
\email{tamer.ozsu@uwaterloo.ca}

\begin{abstract}
We study the fundamental problem of butterfly (i.e. (2,2)-bicliques) counting in bipartite streaming graphs. Similar to triangles in unipartite graphs, enumerating butterflies is crucial in understanding the structure of bipartite graphs. This benefits many applications where studying the cohesion in a graph shaped data is of particular interest. Examples include investigating the structure of computational graphs or input graphs to the algorithms,  as well as dynamic phenomena and analytic tasks over complex real graphs. Butterfly counting is computationally expensive, and known techniques do not scale to large graphs; the problem is even harder in streaming graphs. In this paper, following a data-driven methodology, we first conduct an empirical analysis to uncover temporal organizing principles of butterflies in real streaming graphs and then we introduce an approximate adaptive window-based algorithm, sGrapp, for counting butterflies as well as its optimized version sGrapp-x. sGrapp is designed to operate efficiently and effectively over any graph stream with any temporal behavior. Experimental studies of sGrapp and sGrapp-x show superior performance in terms of both accuracy and efficiency.
\end{abstract}

\maketitle

\section{Introduction}\label{sec:intro}
In this paper we address the problem of counting butterfly patterns in large, bipartite streaming graphs. 
A butterfly  (also called (2,2)-biclique or rectangle) is a complete bipartite subgraph with two vertices of one type and two vertices of another type (rightmost in Figure \ref{fig:patterns}). Similar to the triangles in unipartite graphs, butterflies are the simplest and most local form of a cycle in bipartite graphs. Enumerating butterflies is important in measuring graph cohesion and clustering or community structure~\cite{aksoy2017measuring}. Clustering or community structure is measured by the transitivity/clustering coefficient that is computed as the fraction of three-paths (called caterpillars-- left four in Figure \ref{fig:patterns})  which form a butterfly ~\cite{PhysRevE.72.056127, zhang2008clustering, aksoy2017measuring}. Graph cohesion can be measured by the number of butterflies-per-vertex and by the local clustering coefficient. Study of such local structural measures unveils hidden ordering and hierarchies in graphs displaying structural deviations from uncorrelated random connections~\cite{caldarelli2004structure, ravasz2003hierarchical, newman2003structure}. A recent study investigates the predictive performance of deep neural networks by means of clustering coefficient~\cite{you2020graph}. Other applications are realistic graph models~\cite{aksoy2017measuring, kim2012multiplicative} and representative graph sampling~\cite{zhang2017clustering}. The study of different phenomena in complex graphs such as social collective behaviours~\cite{david2020herding}, synchronization ~\cite{sheshbolouki2015feedback, ziaeemehr2020emergence}, information propagation ~\cite{PhysRevE.72.066116}, and epidemic spreading ~\cite{PhysRevE.69.066116} rely on clustering coefficient. Moreover, clustering coefficient plays an important role in graph analytics tasks such as link prediction ~\cite{huang2010link} and community detection ~\cite{zhang2008clustering}, and in general any graph processing algorithm relying on counting the mutual neighbors or Jaccard similarity. The distribution of local clustering coefficient is used as a feature to uncover statistical differences between normal and fraudulent data in applications such as spam detection~\cite{becchetti2008efficient}.

We study the problem in the context of streaming graphs, because the graphs that are used in many modern applications are not static and not available to  algorithms in their entirety; rather the graph vertices and edges are streamed and the graph ``emerges'' over time. These are called \emph{streaming graphs} and they differ from dynamic graphs that are fully available but undergo changes over time. A driving example is the stream of user-product interactions in e-commerce services. Alibaba has reported that customer purchase activities during a heavy period in 2017 resulted in generation of 320 PB of log data in a six hour period, and it had to deal with a high velocity stream of data that incurred a processing rate of 470 million event logs per second. Other e-commerce sites have similar activity albeit at somewhat lower levels. Other applications such as web recommenders, fraud detection, and social network analysis rely on butterfly counting over streaming graphs. 

\begin{figure}[t]
    \centering
    \subfigure{\includegraphics[width=0.08\textwidth]{ 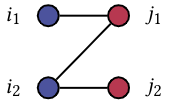}}
    \subfigure{\includegraphics[width=0.08\textwidth]{ 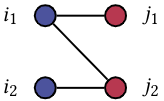}} 
    \subfigure{\includegraphics[width=0.08\textwidth]{ 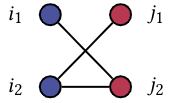}}
    \subfigure{\includegraphics[width=0.08\textwidth]{ 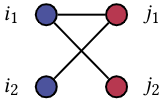}}
    \subfigure{\includegraphics[width=0.08\textwidth]{ 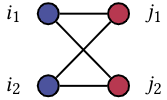}}
    \caption{Caterpillar and butterfly (rightmost) patterns.} 
    \label{fig:patterns}
\end{figure}

Bipartite graphs that model  networks with two disjoint sets of vertices are prevalent in real applications: interaction graphs that model the interactions (e.g. comments, reviews, purchases, ratings, etc) between users and items, affiliation graphs that model the membership of actors/people in groups, authorship graphs that model the links between authors and their works, text graphs that model the occurrence of words in documents, and feature graphs that model the assignment of features to entities. In particular, user-product graphs are currently recognized as the most common graphs in industry that require attention. 
It is important to study the underlying patterns and structures of bipartite graphs, and in this paper we focus on butterfly patterns. A natural question that arises is why the bipartite graph cannot be projected into a unipartite graph on which the existing approaches to count the triangles are used? The answer is that the projected graph is misleading and counting on it is inefficient. First, the projected unipartite graph loses fine-grained pattern information~\cite{sariyuce2018peeling,latapy2006basic}, since the one-to-many relationship information are projected to pairwise relationships and the projection is not bijective. Second, the projected unipartite graph will have significantly more edges than the bipartite graph since each $i-(j-)$vertex $v$ with degree $d_v$  produces $d_v(d_v-1)/2$ homogeneous edges. That is, the number of edges in the original bipartite graph is $\Sigma_v d_v$ while in the projected graph it is $\Sigma_v$ $d_v \choose 2$. It has been shown that projection can lead to an edge inflation of $200\times$ ~\cite{latapy2006basic}. In the case of streaming bipartite graphs that already have a high number of edges, the projection will exacerbate the computational footprint. Finally, the patterns that emerge in the projected unipartite graph are not reliable signals of the original bipartite graph since the edge inflation artificially changes the patterns. For instance, it has been shown that the clustering coefficient is high in the projected mode~\cite{newman2001random, guillaume2004bipartite} and unipartite projection misleads the community detection analysis~\cite{guimera2007module, barber2007modularity}. Due to these issues, it is important to devise techniques to directly study bipartite graphs. 

Exact butterfly counting  is feasible only when the entire graph is available to the processing algorithm. As noted earlier, this is not possible in streaming graphs (and even in massive static graphs~\cite{lyu2020maximum}). The alternative is approximating. One such approach is to use random sampling/sparsification~\cite{buriol2006counting, sanei2018butterfly}, which requires determining the sampling probability, reservoir size, and scaling factor. The sampling process is done several times and can be a potential overhead lowering the processing throughput. Another approach in streaming graphs is to batch the incoming graph vertices and edges into a window and process them when the window moves; this is what we follow. Most existing streaming proposals ~\cite{buriol2006counting, bar2002reductions, buriol2007estimating} assume that (a) all the edges incident to a vertex arrive together (i.e. incidence streams) and (b)  vertex degrees are bounded. Neither of these are likely to hold in real-life streaming graphs. We propose a butterfly counting algorithm that can efficiently return an accurate answer over any graph stream without these unrealistic assumptions. It has been shown that the space lower bound for an approximate butterfly count that bounds the relative error to $0<\delta<0.01$ is $O(n^2)$ where $n$ is the number of vertices~\cite{sanei2019fleet}. This is not feasible in streaming systems. We analyze the computational and error bounds of our proposed algorithm. We also validate our algorithm's accuracy and efficiency empirically.

We follow a data-driven approach to algorithm design: we conduct a deep empirical analysis of a number of real graphs with varying temporal/structural characteristics to determine the temporal occurrence of connectivities. We formulate this as a power law (Section \ref{sec:butterflyemergence}) that grounds our algorithm,  \textit{sGrapp}, to exploit these patterns. Data-driven approach has previously been used to design a graph generator/model preserving the mined patterns in a set of unipartite real graphs~\cite{leskovec2005graphs}. However, to the best of our knowledge, this is the first time  this approach is followed for designing a graph processing algorithm.  sGrapp is a \textbf{s}treaming \textbf{gr}aph \textbf{app}roximation algorithm for butterfly counting in bipartite graphs (Section \ref{sec:approx}) and is based on (a) our novel stream processing framework, which uses time-based windows that can adapt to the  temporal distribution of the stream (Section \ref{subsec:adaptivewin}) and (b) our algorithm for exact butterfly counting in streaming graph snapshots (Section \ref{subsec:butterflydensification}). Our experimental analysis (Section \ref{sec:experiments}) shows that sGrapp achieves $160\times$ higher throughput and $0.02\times$ lower estimation error than baselines and can process $1.5\times10^6$ edges-per-second. It can achieve an average window error of less than $0.05$ in graph streams with almost uniform temporal distribution. We introduce optimizations that lower the average window error to less than $0.14$ in graph streams with non-uniform temporal distribution without affecting the throughput. sGrapp handles graph streams with both high number of edges and high average degree with a sublinear memory footprint, which is lower than that of the baselines. Empirical analysis shows that the performance of sGrapp is independent of its input data, hence can be applied to any real graph stream.
\section{Background}\label{sec:background}
\subsection{Preliminaries}\label{sec:overview}
We define a graph $G$ as a pair of vertex and edge sets $G=(V , E)$. Since $G$ is a bipartite graph, $V = V_i \cup V_j$ and $V_i \cap V_j=\emptyset$. We use user-item bipartite graphs in which $V_i$  (called i-vertices) represents users and $V_j$ (called j-vertices) represents items.

\begin{definition}[Streaming Graph Record]\label{def:sgr}
A streaming graph record (sgr) $r=(\tau, p)$ is a pair  where $\tau$ is the event (application) timestamp of the record assigned by the data source, and payload $p=\langle e/v , op\rangle$ indicates an edge $e \in E$ or a vertex $v \in V $ of the [property] graph $G$, and an operation $op \in \{ insert, delete, update\}$ that defines the type of the record.
\end{definition}

In this paper, the operations are limited to edge insertion. If there are duplicate edge arrivals, the algorithm ignores the duplicates. 

\begin{definition}[Streaming Graph]\label{def:insertstream}
 
\sloppy A streaming graph $S$ is an unbounded sequence of streaming graph records $S= \langle r^1, r^2, \cdots \rangle$ in which each record $r^m$ arrives at a particular time $t^m$ ($t^m \leq t^n$ for $m<n$). 
\end{definition}

\begin{definition}[Time-based Window]\label{def:twin}
A time-based window $W$  over a streaming graph $S$ is denoted by time interval $[W^b,W^e)$ where $W^b$ and $W^e$ are the beginning and end times of window $W$ and $W_e - W_b = |W|$. The window contents 
is the multiset of sgrs where the timestamp $\tau_i$ of each record $r^i$ is in the window interval.
\end{definition}

\begin{definition}[Time-based Sliding Window]\label{def:window}
A time-based sliding window $W$ with a slide interval $\beta$ is a time-based window that progresses every $\beta$ time units.
At any time point $\tau$, a time-based sliding window $W$ with a slide interval $\beta$ defines a time interval $(W^b,W^e]$ where $W^e = \lfloor \tau / \beta \rfloor \cdot \beta$ and $W^b = W^e - |W|$.
\end{definition}

\begin{definition}[Time-based Tumbling Window]\label{def:tumblingwindow}
A tumbling window is a time-based window where, for two subsequent windows $W_{i}$ and $W_{i+1}$, $W_{i+1}^{b} = W_{i}^{e}$ and $W_{i+1}^{e} = W_{i+1}^{b}+|W_{i+1}|$. Simply, when subsequent sliding windows are disjoint, they are called tumbling windows.
\end{definition}

\begin{definition}[Time-based Landmark Window]\label{def:landmarkwindow}
A landmark window is a constantly expanding time-based window denoted by a pair $\langle W^b,|W| \rangle$ where, $ W^b$ is the fixed beginning time and $|W|$ is the expansion size.  For two subsequent windows $W_i$ and $W_{i+1}$, $W^{b}_{i+1} = W^{b}_i$ and $W^{e}_{i+1} = W^{e}_i+|W_{i+1}|$. Simply, when the beginning border is fixed and the end border moves forward, the window is called landmark.
\end{definition}

\begin{definition}[Streaming Graph Snapshot]\label{def:streaminggraphsnapshot}
 A \textit{streaming graph snapshot} $G_{W,t}$ is the graph formed by the records in the time-based window $W$ at time $t$.
\end{definition}

Table \ref{tab:notations} lists the notations used in the paper.

\begin{table*}[h]
\caption{Frequent notations. Similar notations stand for j-vertices where applicable.}

    \begin{tabular}{|p{2.75cm}|p{10.5cm}|}\hline
        
        Notation & Description \\ \hline\hline
       
        
        $r^m=(\tau, p)$ &  A streaming graph record (sgr) with timestamp $\tau$, payload $p$, and arrival time $t_m$  \\ \hline
        
        $\tau$ & sgr timestamp (real time-label) \\ \hline
        
        $t$ & Computational time point or time of sgr arrival at the computational system  \\ \hline
        
        $\mathcal{R}$ & Average stream rate  \\ \hline
        
        $p=\langle e/v , op\rangle$ &  An edge $e \in E$ or a vertex $v \in V$ , and an operation $op \in \{ insert, delete, update\}$\\ \hline
       
        $W_i := [W_i^b,W_i^e)$ &   $i$th time-based window $W$ as an interval of width $|W|$\\ \hline
        
        $\beta$ &  Slide size for a sliding window  \\ \hline
        
        $G_{W,t}=( V(t),E(t))$ &   A graph snapshot formed by window $W$ at time $t$ \\ \hline
        
        $deg(i)$ & Degree of vertex $i$   \\ \hline
        
        $N_i$ & Neighborhood of vertex $i$   \\ \hline
        
        $P$/$\gamma$/$M$ & FLEET's sampling probability/subsampling probability/reservoir capacity   \\ \hline
      
        $K_i$ & Average degree of i-vertices   \\ \hline
        
        $\eta$/$\alpha$ & Butterfly densification power law exponent for all/inter-window butterflies   \\ \hline
        
        
        $N_{hub}(t)$ & Number of hubs at time $t$   \\ \hline
       
        $N_t$ & Number of unique timestamps in data stream   \\ \hline
       
        $B(t)$ & The number of butterflies since the initial time point until $t$    \\ \hline
        
        $B_i$ & Butterfly support of vertex $i$   \\ \hline
        
        $B^{W_k}$ &   Number of butterflies introduced by at least one vertex in the window $W_k$\\ \hline
        
        $\hat{B}(t=W_k^e)=\hat{B_k}$ &  Estimation of number of butterflies at time $t=W_k^e$   \\ \hline
        
        
        $B_G^{W_k}$ &  Number of butterflies in graph  corresponding to window $W_k$  \\ \hline
        
        $B^{interW} \&~\hat{B}^{interW}$ &  Number of inter-window butterflies \& its estimate   \\ \hline
        
        
        $N_t^w$ & Number of unique timestamps per window \\ \hline
        
         
        $K_{i,W_k}$ & the lower bound of degree of i(j)-vertices in window $W_k$\\ \hline
        
        $V_{i,W_k}$/$E_{W_k}$ & Set of i-vertices/edges in the interval $[W_k^b,W_k^e)$ \\ \hline

        $E_k$ & Set of edges in the interval $[W_0^b,W_k^e)$ \\ \hline

        $Pr(N_{iHub}^t\ge 1)$ & Probability of having at least one i-hub in the butterflies at time $t$ \\ \hline
        
\end{tabular} \label{tab:notations}
\normalsize
\end{table*}
\subsection{Related Work}\label{sec:relatedworks}
he existing works in butterfly counting can be classified  along three dimensions: graph characteristic (bipartite/unipartite), data location (disk-resident/in-memory) and graph availability (static/ dynamic/streaming). Detailed coverage of each design point is beyond the scope of this paper; we focus on two particular design points that are most relevant to our work: static bipartite graphs and streaming bipartite graphs.

\subsubsection{Counting in Static Bipartite Graphs}\label{subsec:relatedworksexact}
The literature on counting (bi)cliques in static bipartite graphs~\cite{sariyuce2018peeling, wang2014rectangle, wang2019vertex, sanei2018butterfly} and static unipartite graphs~\cite{wang2010triangulation, hellings2012efficient} is quite rich. A major challenge in this context is the massive size of these graphs. Some studies have focused on disk-resident data and optimized I/O access patterns for counting the exact number of cliques~\cite{hellings2012efficient, chu2011triangle, becchetti2008efficient, hu2014efficient, hu2013massive, pagh2014input}. 
Other studies consider in-memory algorithms and use random sampling so that the induced graph can fit in main memory for estimating the number of (bi)cliques~\cite{buriol2006counting, sanei2018butterfly}. There are studies that propose scaling out computation by parallelization~\cite{kim2014opt, arifuzzaman2013patric}.  

Butterfly counting algorithms in bipartite graphs follow either vertex-centric or edge-centric processing. One straightforward edge-centric approach is to take each pair of disjoint edges $(e_{i_1,j_1},e_{i_2,j_2})$ in the graph (Figure \ref{fig:exactmethods}a)
and check for the existence of the two other edges that complete the butterfly pattern. The complexity of this approach is $\mathcal{O}(|E|^2)$ which is too expensive for graphs  with a high number of edges. Another edge-centric approach~\cite{chiba1985arboricity} takes an edge $e_{i_1,j_1}$ and examines the existence of the three complementary edges. That is, the algorithm checks the connections between neighbors of $i_1$ and neighbors of $j_1$  denoted as $j_2$ and $i_2$, respectively to see whether they are connected by an edge $e_{i_2,j_2}$ (Figure \ref{fig:exactmethods}b).
This approach can be implemented with an algorithm that has complexity  $\mathcal{O}(\sum_{\langle i_1,j_1 \rangle\in E} Min(deg(i_1),deg(j_1)))$, which is not appropriate for dense graphs with high number of edges and high average degrees. The state-of-the-art approach~\cite{wang2014rectangle, wang2019vertex, sanei2018butterfly} is vertex-centric that takes a vertex $v_i$ and traverses all two-hop neighbors to identify triples $\langle i_1, j_1, i_2 \rangle$ and $\langle i_1, j_2, i_2 \rangle$. That is, it finds all triples (i.e. two-paths) with common end vertices (i.e. the same two-hop neighbor) and then combines them to get the number of all butterflies (Figure \ref{fig:exactmethods}c). 
The complexity of this approach is $\mathcal{O}(\sum_{i_1\in V_i}\sum_{j_1\in N_{i_1}}deg(j_1))$, which is challenging for graphs with high average i- and j-degrees as a result of traversing two hop neighbors~\cite{wang2019vertex}.

\begin{figure}[t]
   
    \subfigure[]{\includegraphics[width=0.1\textwidth]{  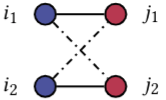}}
    \subfigure[]{\includegraphics[width=0.1\textwidth]{  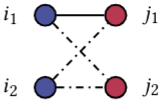}}
    \subfigure[]{\includegraphics[width=0.1\textwidth]{  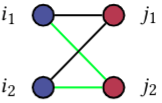}}
    \subfigure[]{\includegraphics[width=0.1\textwidth]{  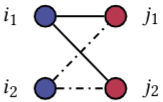}}
    \caption{Butterfly counting methods.}
    \label{fig:exactmethods}
\end{figure}

\subsubsection{Counting in Streaming Bipartite Graphs}\label{subsec:relatedworksapprox}

In the streaming graph context, the literature is also rich for counting in unipartite graphs ~\cite{wang2017approximately, wang2010triangulation, buriol2006counting, bar2002reductions, becchetti2008efficient, buriol2007estimating, bera2017towards, braverman2013hard}. However, to the best of our knowledge, the only butterfly counting study over bipartite streaming graphs
is FLEET ~\cite{sanei2019fleet}, which introduces a suite of algorithms. FLEET1 samples the edges of a window with probability $P$ into a reservoir with fixed capacity $M$ to bound the memory consumption and increments the butterfly count by the number of incident butterflies for each sampled edge. When the size of reservoir exceeds $M$, the edges are sub-sampled with probability $\gamma$ and the butterfly count is set to the exact number of butterflies in the reservoir. The sampling probability is then multiplied by $\gamma$ for the following edges. FLEET2 avoids re-computing the exact number of butterflies in the reservoir during the sub-sampling iterations. FLEET3 avoids re-computation and also updates the estimate before sampling the edges into the reservoir. FLEETSSW uses count-based sliding windows with limited graph size in each window, and FLEETTSW uses time-based sliding windows with fixed window length across windows. To overcome the variable number of edges inside each window, FLEETTSW assumes an upper-bound for the number of edges in a window on top of a FIFO-based sampling scheme. As we  discuss in Section \ref{sec:approx}, there exist a number of inter-window butterflies in the stream that are missed by the FLEET algorithms. Moreover, FLEET requires determining a sub-sampling probability and a normalization factor to scale-up the estimation computed over the sampled edges, and the specification of a time when the result is ready to be returned. FLEET requires a sufficiently large amount of memory to guarantee a desired level of accuracy. 
\section{Analysis of Graph Characteristics}\label{sec:butterflyemergence}

In this section, we  present our investigations into the emergence  of butterfly patterns in graph streams and on the underlying contributors to these patterns. We use the insights provided by this analysis to introduce an approximation algorithm for butterfly counting in streaming graphs in Section \ref{sec:approx}. The analysis results themselves are also important as they expose how butterfly patterns exist in real world graphs.


\subsection{Graph stream data}\label{subsec:graphstreamdataset}
We study a set of real world graphs and make use of a set of synthetic graphs to explore additional features. Table \ref{tab:graphs} provides the statistics about the graphs we study; these graphs are also used in the experiments discussed in Section \ref{sec:experiments}.

\begin{table*}[t]\caption{Bipartite and temporal graph datasets used. $\langle k_i\rangle$ and $\langle k_j\rangle$ denote the average degree of i-vertices and j-vertices, respectively. $N$ and $m=m_0$ are parameters of BA graphs and refer to the total number of vertices and average degree in the unipartite BA graph, respectively. $N_t$ denotes the number of unique timestamps. $B_G$ denotes the number of butterflies in the graph.}
\small
    \begin{tabular}{|p{3cm}|p{1.5cm}|p{1.5cm}|p{1.5cm}|p{0.6cm}|p{0.6cm}|p{1cm}|p{1cm}|p{1.5cm}|p{2.4cm}|}\hline
          \makecell{Graph dataset} & \makecell{$|V_i|$} & \makecell{$|V_j|$} & \makecell{$|E|$} & \makecell{$\langle k_i\rangle$} & \makecell{$\langle k_j\rangle$}& \makecell{$N$} & \makecell{$m=m_0$}& \makecell{$N_t$} & \makecell{$B_G$} \\  \hline\hline
       
         \makecell{Epinions\\BA+Epinions stamps\\BA+random stamps} & 
         \makecell{$22,164$\\$22,514$\\$22,514$} & \makecell{$296,277$\\$21,455$\\$21,455$} & \makecell{$922,267$\\$922,254$\\$922,254$} & \makecell{$41$\\$41$\\$41$} & \makecell{$3$\\$43$\\$43$} & \makecell{\\$22,515$\\$22,515$} & \makecell{\\$41$\\$41$}  & \makecell{$4,318$\\$4,318$\\$921,159$}& \makecell{$170,303,771,005$\\$$\\$$}  \\ \hline
        
         \makecell{MovieLens1m\\BA+ML1m stamps\\BA+random stamps} &
         \makecell{$6,040$\\$6,106$\\$6,106$} & \makecell{$3,706$\\$6,022$\\$6,022$} & \makecell{$1,000,210$\\$999,901$\\$999,901$} & \makecell{$166$\\$164$\\$164$}& \makecell{$270$\\$166$\\$166$} & \makecell{\\$6,107$\\$6,107$} & \makecell{\\$166$\\$166$} & \makecell{$458,455$\\$458,312$\\$994,467$}& \makecell{$16,671,201,295$\\$$\\$$}  \\ \hline
        
        \makecell{MovieLens100k\\BA+ML100k stamps\\BA+random stamps} & 
        \makecell{$943$\\$995$\\$995$} & \makecell{$1,682$\\$982$\\$982$} & \makecell{$100,000$\\$99,905$\\$99,905$} & \makecell{$106$\\$100$\\$100$}& \makecell{$59$\\$100$\\$100$} & \makecell{\\$966$\\$966$} & \makecell{\\$106$\\$106$} & \makecell{$49,282$\\$49,254$\\$996,555$}& \makecell{$220,548,028$\\$$\\$$}  \\ \hline

        \makecell{MovieLens10m} &\makecell{69,878} & \makecell{10,677} & \makecell{10,000,054} & \makecell{143} & \makecell{937} & \makecell{} & \makecell{} & \makecell{7,096,905} & \makecell{1,197,019,065,804} \\ \hline 
        
        \makecell{edit-frwiki} &\makecell{288,275} & \makecell{3,992,426} & \makecell{46,168,355} & \makecell{160} & \makecell{$11$}& \makecell{}& \makecell{} & \makecell{39,190,059} & \makecell{$601.2\times 10^{9}$} \\ \hline
        
        \makecell{edit-enwiki} &\makecell{262,373,039} & \makecell{266,665,865} & \makecell{266,769,613} & \makecell{70} & \makecell{12} & \makecell{}& \makecell{}  & \makecell{134,075,025} & \makecell{$2\times 10^{12}$} \\ \hline
        
    \end{tabular}\label{tab:graphs}
\end{table*}

\begin{table*}[ht]\caption{$R^2$ and RMSE of ten fitting functions for the temporal evolution of butterfly frequency in three real-world graph streams. Filled cells decode increasing function and best fits are highlighted in gray cells.}
\small
    \begin{tabular}{|p{1.5cm}|p{1cm}|p{1.3cm}|p{1cm}|p{1.5cm}|p{1cm}|p{1.4cm}|p{1.4cm}|p{1.4cm}|p{1.4cm}|p{1.4cm}|}\hline
          \makecell{$R^2$\\RMSE} & \makecell{Linear} & \makecell{Quadratic} & \makecell{Cubic} & \makecell{4th degree\\ polynomial} & \makecell{Quintic} &
          \makecell{6th degree\\ polynomial}  & \makecell{7th degree\\ polynomial}  & \makecell{8th degree\\ polynomial} 
          & \makecell{9th degree\\ polynomial}  & \makecell{10th degree\\ polynomial} \\ \hline\hline
         \makecell{Epinions} & \cellcolor{pink!40}\makecell{$0.9947$\\$1.481e^4$} & \cellcolor{pink!40}\makecell{$0.9951$\\$1.435e^4$} & \cellcolor{pink!40}\makecell{$0.9951$\\$1.432e^4$} & \makecell{$0.9975$\\$1.028e^4$} & \cellcolor{pink!40}\makecell{$0.9977$\\$9751$} &
         \cellcolor{pink!40}\makecell{$0.9977$\\$9716$} & \makecell{$0.9978$\\$9598$}& 
         \cellcolor{pink!40}\makecell{$0.9984$\\$8130$}& \cellcolor{gray!40}\makecell{$\boldsymbol{0.9987}$\\$\boldsymbol{7409}$} & \makecell{$0.9987$\\ $7386$}  \\ \hline
        
         \makecell{ML100k} & \cellcolor{pink!40}\makecell{$0.931$\\$2.31e^6$} &\cellcolor{pink!40} \makecell{$0.9977$\\$4.18e^5$} & \cellcolor{pink!40}\makecell{$0.9978$\\$4.167e^5$} & \makecell{$0.9978$\\$4.126e^5$} & \makecell{$0.9983$\\$3.673e^5$}  & \cellcolor{pink!40}\makecell{$0.9983$\\$3.584e^5$} & \cellcolor{gray!40}\makecell{$\boldsymbol{0.9993}$ \\ $\boldsymbol{2.286e^5}$ } & \cellcolor{gray!40}\makecell{$\boldsymbol{0.9993}$\\$\boldsymbol{2.286e^5}$} & \makecell{$0.9997$\\$1.552e^5$} & \makecell{$0.9997$\\$1.552e^5$} \\ \hline
        
        \makecell{ML1m} & \cellcolor{pink!40}\makecell{$0.8751$\\$2.119e^{6}$} & \cellcolor{pink!40}\makecell{$0.9951$\\$4.196e^{5}$} & \cellcolor{pink!40}\makecell{$0.9953$\\$4.111e^5$} & \makecell{$0.9977$\\$2.895e^5$} & \makecell{$0.9989$\\$1.976e^5$} &
        \makecell{$0.9989$\\$1.961e^5$} & \makecell{$0.999$\\$1.94e^5$} & \makecell{$0.999$\\$1.937e^5$} & \cellcolor{pink!40}\makecell{$0.999$\\$1.933e^5$} & \cellcolor{gray!40}\makecell{$\boldsymbol{0.999}$\\ $\boldsymbol{1.933e^5}$}\\ \hline
        
        \makecell{ML10m} & \cellcolor{pink!40}\makecell{$0.8943$\\$3.223e^6$} & \cellcolor{pink!40}\makecell{$0.9983$\\$4.034e^5$}& \cellcolor{pink!40}\makecell{$0.999$\\$3.149e^5$}& \makecell{$0.9992$\\$2.841e^5$}& \makecell{$0.9993$\\$2.701e^5$}& \makecell{$0.9993$\\$2.699e^5$}& \cellcolor{gray!40}\makecell{$\boldsymbol{0.9993}$\\$\boldsymbol{2.605e^5}$}& \makecell{$0.9994$\\$2.493e^5$}& \makecell{$0.9996$\\$1.868e^5$}& \makecell{$0.9997$\\$1.781e^5$}\\ \hline
        
         \makecell{Edit-FrWiki} & \makecell{$0.9228$\\$8.09e^4$} & \cellcolor{pink!40}\makecell{$0.9932$\\$2.408e^4$} &  \cellcolor{pink!40}\makecell{$0.9932$\\$2.397e^4$} &\cellcolor{pink!40}\makecell{$0.9953$\\$1.998e^4$} & \cellcolor{gray!40}\makecell{$\boldsymbol{0.9966}$\\$\boldsymbol{1.693e^4}$} & \makecell{$0.9968$\\$1.653e^4$} & \makecell{$0.9979$\\$1.319e^4$} & \makecell{$0.9988$\\$1.01e^4$} &  \makecell{$0.9988$\\$9928$} & \makecell{$0.9989$\\$9725$} \\ \hline
         
         \makecell{Edit-EnWiki} & \cellcolor{pink!40}\makecell{$0.971$\\$1990$}  & \cellcolor{pink!40}\makecell{$.9879$\\$1288$} & \cellcolor{pink!40}\makecell{$0.9879$\\$1285$} & \cellcolor{pink!40}\makecell{$0.9903$\\$1150$} & \cellcolor{pink!40}\makecell{$0.9918$\\$1060$} & \makecell{$0.9928$\\$990$} & \makecell{$0.9951$\\$821.3$} & \cellcolor{pink!40}\makecell{$0.9957$\\$769.9$} & \cellcolor{gray!40}\makecell{$\boldsymbol{0.9964}$\\$\boldsymbol{696.5}$} & \makecell{$0.9967$\\$671.7$} \\ \hline

    \end{tabular}\label{tab:rmser2reals}
\end{table*}


\textbf{Real-world graphs } -- In this study, we use six real world graphs: four rating graphs including Epinions
, MovieLens100k
, MovieLens1m
, MovieLens10m
, and two Wikipedia edit networks in English
 and French
obtained from the KONECT repository~\cite{kunegis2013konect}. All of these networks include information generated from interaction of a set of users  with a set of items (products, movies, or wikipedia pages). These datasets cover graphs with different edge density levels and are suitable for deep analysis and evaluations.

\textbf{Synthetic graphs } -- In addition to the real world graphs, we use six synthetic random graphs in this study to bolster the analysis of real world graphs. In fact synthetic graphs are configurable and have known structural properties that ease the understanding of their patterns. We use these synthetic graphs to better understand and explain what is happening in real world graphs through the comparisons and contradictory case investigations. These synthetic graphs are generated with respect to the three real world graphs (Epinions, MovieLens100k, and MovieLense1m) in that the the synthetic graphs and the corresponding real world graphs have (roughly) same structural statistics. We use the Barabasi-Albert (BA) model ~\cite{barabasi1999emergence} to generate the structure of random graphs as the baseline for analyzing real world graphs. We chose this model because it is a popular and widely adopted model for generating scale free graphs ~\cite{huang2016leopard, hadian2016roll, ma2019linc, talukdar2010automatically, liu2008towards, mondal2012managing, yang2012towards, jin2010computing,   chang2012exact, lee2020measurements, bernaschi2019spiders}. Given the total number of vertices $N$, the initial number of vertices $m_0$ and the number of connections of new vertices $m$ ($m\leq m_0$) as inputs, the BA graph model applies the rich-get-richer preferential attachment rule to generate a unipartite scale-free random graph. Precisely, this graph model creates an initial complete graph with $m_0$ vertices and keeps adding $N-m_0$ new vertices to this initial graph. The new vertices are connected to $m$ existing vertices with higher probability of attachment dictated by the attachment rule. The BA preferential attachment rule states that the probability is determined based on the degree of the vertex, therefore the higher the degree (i.e. the older the vertex), the higher the probability of attachment. The original BA model produces growing unipartite graphs with no timestamps. Therefore, we extended the model to generate bipartite and temporal graphs with respect to a given real graph such that the structure is dynamic but the timestamps are static. We introduce a three-step procedure to create a bipartite and temporal scale-free BA graph as a baseline for a given real-world graph:

 \begin{figure}[ht]
   \centering
   \includegraphics[width=\linewidth]{ 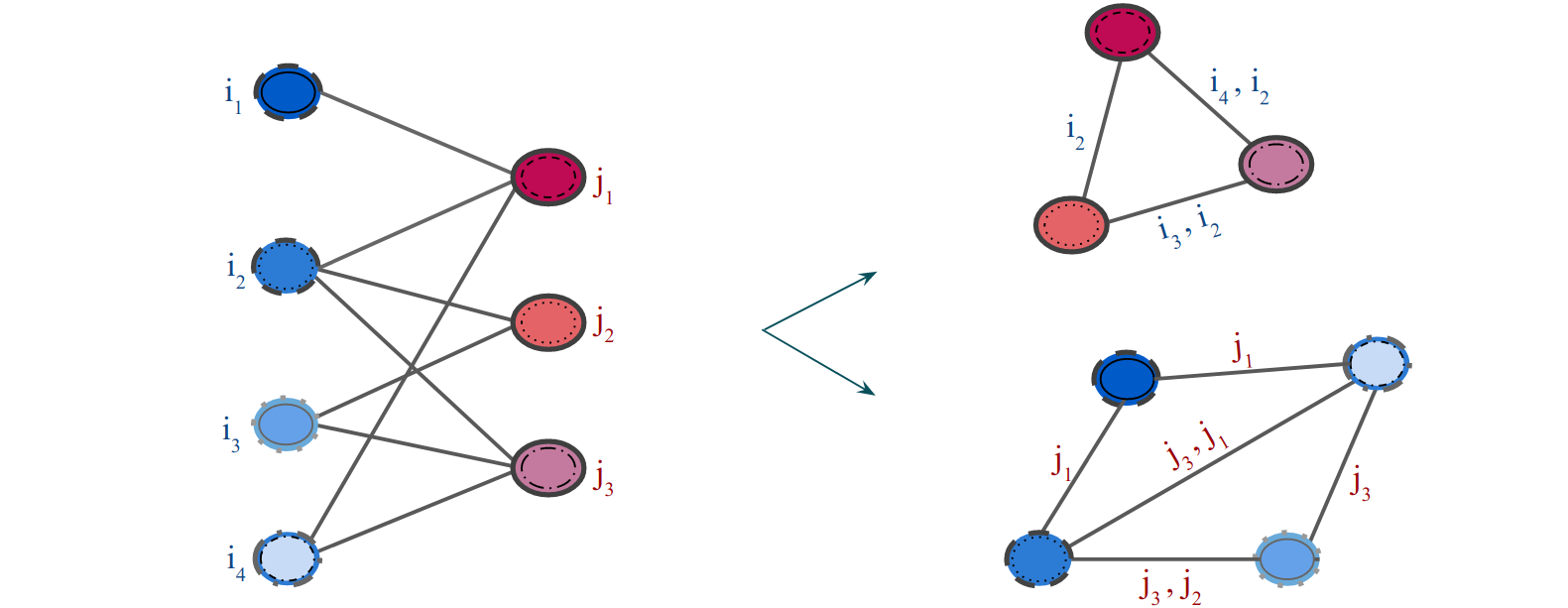}
   \caption{Projecting a bipartite graph to two unipartite graphs. There is a link between two vertices in unipartite mode if they have any common neighbors in the bipartite mode. Edge labels in the unipartite graph reflect the common neighbors.}
  \Description{Projecting a bipartite graph to two unipartite graphs. There is a link between two vertices in unipartite mode if they have any common neighbors in the bipartite mode. Edge labels in the unipartite graph reflect the common neighbors.}\label{fig:projection}
 \end{figure}
 \begin{enumerate}
    \item \textbf{Create Unipartite BA graph} -- The input parameters to the BA model (i.e. $N$, $m$, and $m_0$) should be set such that the average degree of i-vertices and the number of total edges ($|E|$) in real-world and synthetic graphs are (roughly) the same. That is because of the edge-centric nature of our intended analysis. Therefore, we set the parameters $m=m_0$ equal to the average degree of i-vertices (i.e. users) in the real-world graph and determine the value of $N$ in a way that it satisfies the equation for the number of edges in BA graph, that is $m_0(m_0-1)/2 + (N-m_0)m=|E|$ . Given the input parameters, the edge list of the scale-free unipartite directed graph is generated.
     \item \textbf{Project the graph to bipartite mode} -- A common approach to project a bipartite graph $BG=(V, E_{ij}, \Sigma, \psi, \phi)$ to unipartite modes $G_i=(V_i, E_i, \Sigma, \psi, \phi)$ and $G_j=(V_j, E_j, \Sigma, \psi, \phi)$ is to connect a pair of vertices if they have a common neighbor (Figure \ref{fig:projection}). That is,  $(i_m,i_n)\in E_i$ if $\exists j\in V_j$ : $(i_m,j)\in E_{ij} \And (i_n,j)\in E_{ij}$ and the same connection rule for j-vertices. Accordingly, we propose a \textbf{reverse-engineering technique for projecting the unipartite graphs to bipartite mode}. Precisely, given a unipartite BA graph $G_i$ with $N_i$ vertices (assuming the vertices as i-vertices), the bipartite mode $BG$ is generated by the procedure below:
     \begin{enumerate}[\indent(a)]
         \item Assign $N_j$ labels $\{L_k|1\leq k\leq N_j\}$ to arbitrary edges in $G_i$.
         \item Create a set of $N_j$ j-vertices.
         \item Project each edge $(i_m, i_n)\in  E_i$ with label $L_k$ into two edges $(i_m,j_k)$ and $(i_n,j_k)$.
     \end{enumerate}
    Clearly, this procedure can yield a bipartite BA graph with a pre-specified number of i- and j-vertices. Therefore, it can mimic the number of vertices in the real-world graph exactly. However, the number of edges in the output bipartite BA graph does not match that of the unipartite BA graph and if we create a unipartite BA graph with specific number of edges, then the number of i-vertices would be affected accordingly. Therefore, this projection method can not yield bipartite BA graphs that have specific number of edges and vertices at the same time and solely adjusting the number of edges will affect the number of vertices. On the other hand, the intended analysis in this work is edge-centric, therefore it is important to create synthetic bipartite graphs with the same number of edges as the real-world graphs. 
\begin{figure*}[t]
\centering
\includegraphics[width=\linewidth]{ 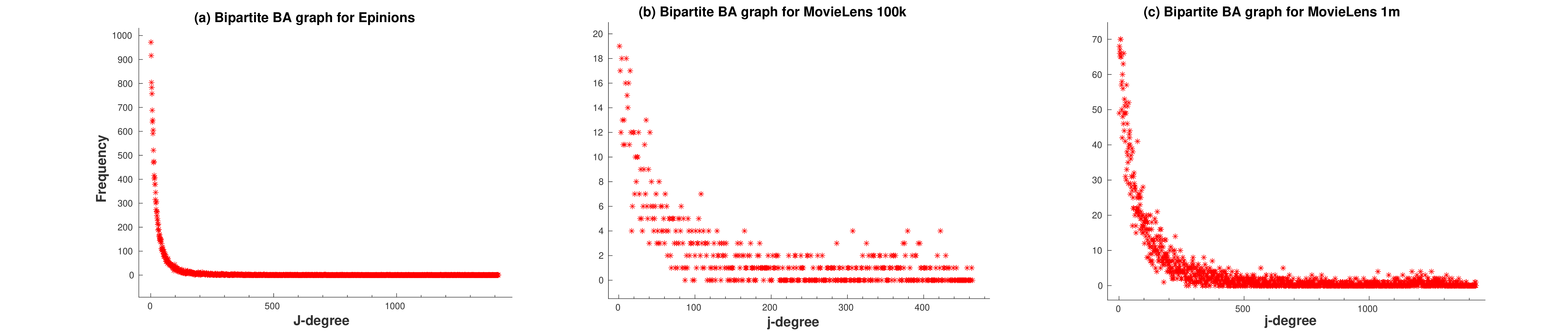}
\caption{The j-degree distribution of Projected Bipartite BA graphs for three real-world graphs}
\Description{}\label{fig:jdegdist}
\end{figure*}

     To address this problem, we follow a simple projection method. Given the list of directed edges in the unipartite BA graph, the sources of edges are treated as i-vertices and the destinations as the j-vertices. Hence, the BA graph is projected to bipartite mode with same number of edges as that of the unipartite and the corresponding real-world graph. The number of i-vertices in the projected bipartite BA graph (equal to the $N$ of unipartite BA graph) is very close to that of the real-world graph. In spite of different number of j-vertices in the projected and real-world graphs, this projection method is preferable as it solves the aforementioned issue. Moreover, this method preserves the scale-free characteristic of the uni-partite graph since the j-degree (i-degree) distribution in bipartite graph is equivalent to the in-degree (out-degree) distribution of vertices in the unipartite graph and the j-degree distribution is scale-free (see Figure \ref{fig:jdegdist}).
   
    \item \textbf{Assign timestamps to the synthetic edge}s -- Given the timestamps of the a real-world graph and the bipartite structure of the corresponding random graph, timestamps are assigned to the edges in two ways: 
     \begin{enumerate}[\indent(a)]
         \item Each BA edge is randomly assigned a timestamp within the range of timestamps of the corresponding real-world graph and the resulting graph is called \textit{BA+random stamps}.
         \item The un-ordered timestamps of the corresponding real-world graph are assigned to arbitrary BA edges and the resulting graph is called \textit{BA+real stamps}. This method guarantees same temporal distribution for the edges of BA and real-world graphs and supports fair comparisons.
     \end{enumerate}
 \end{enumerate}
 
\begin{figure*}[t]
  \centering
  \includegraphics[width=\linewidth]{ 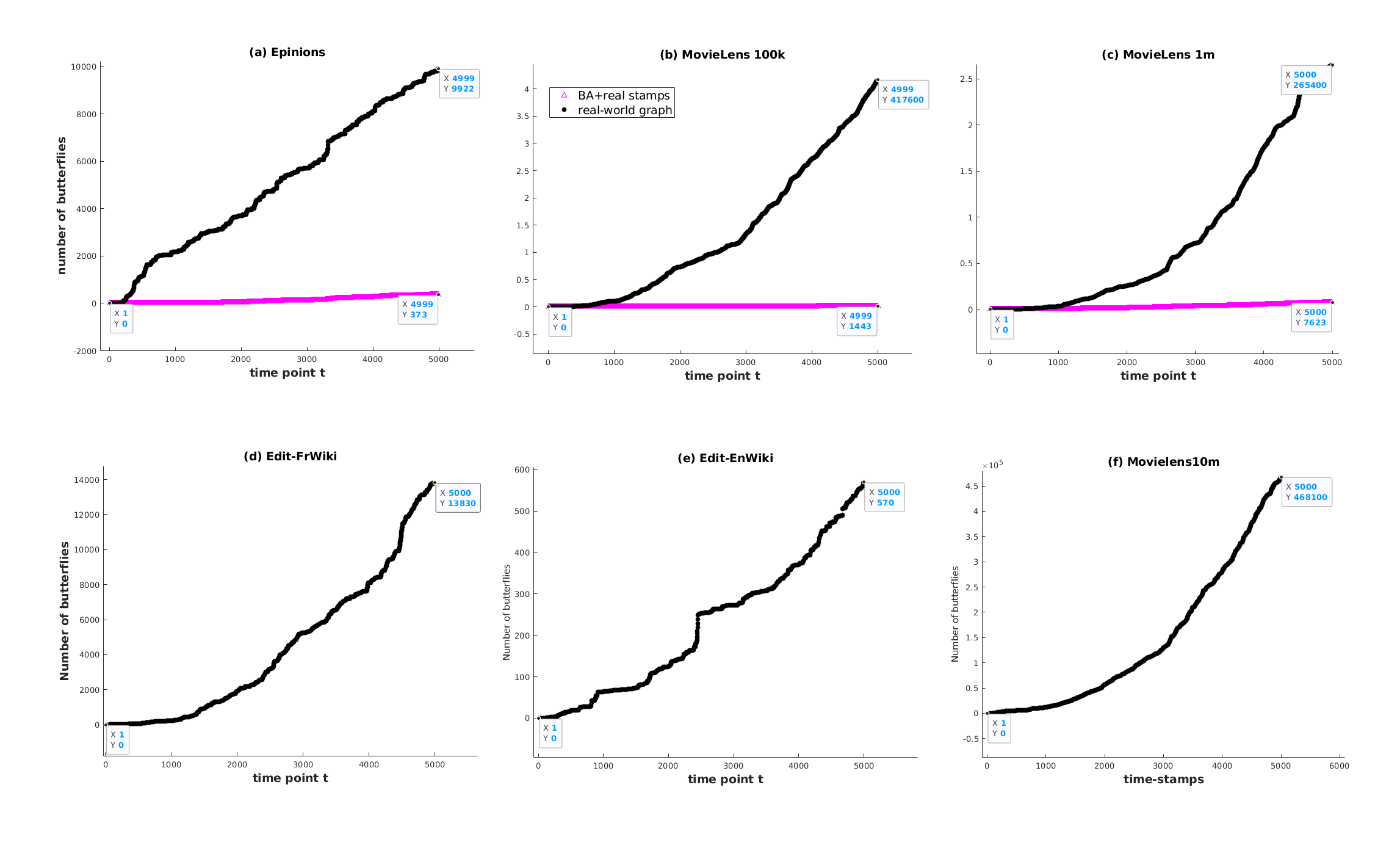}
  \caption{ [Best viewed in colored.] Temporal evolution of butterfly frequency .}
  \Description{Temporal evolution of butterfly frequency.}\label{fig:cliquesnum}
\end{figure*}

\begin{figure*}[t]
    \centering
    \subfigure{\includegraphics[width=0.15\textwidth]{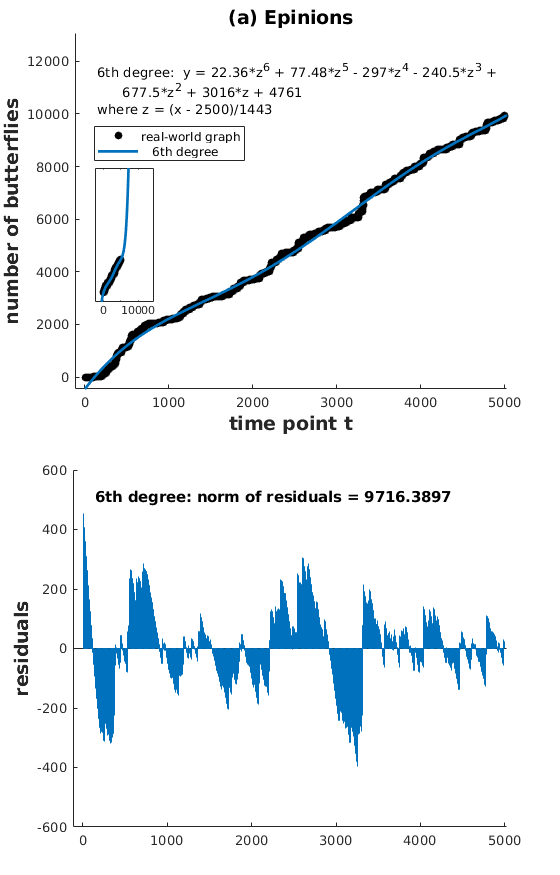}} 
    \subfigure{\includegraphics[width=0.15\textwidth]{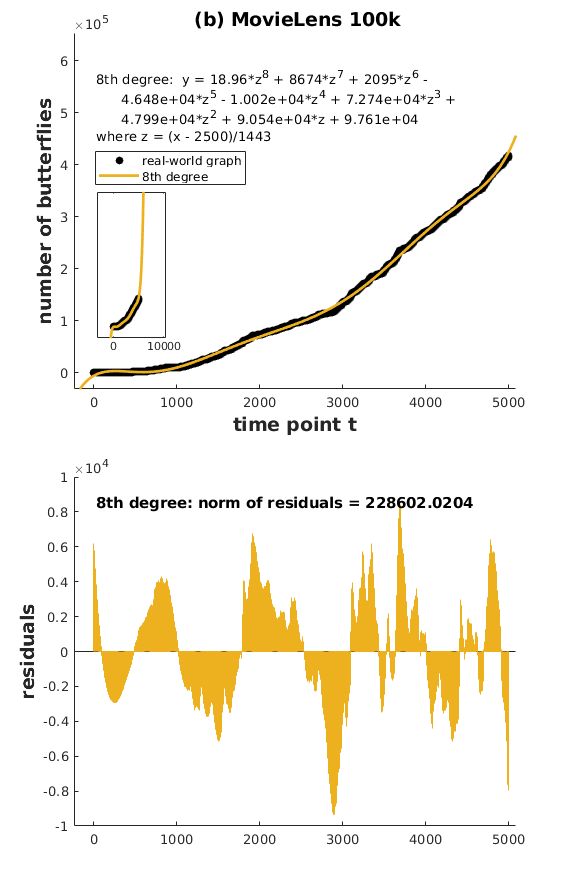}} 
    \subfigure{\includegraphics[width=0.15\textwidth]{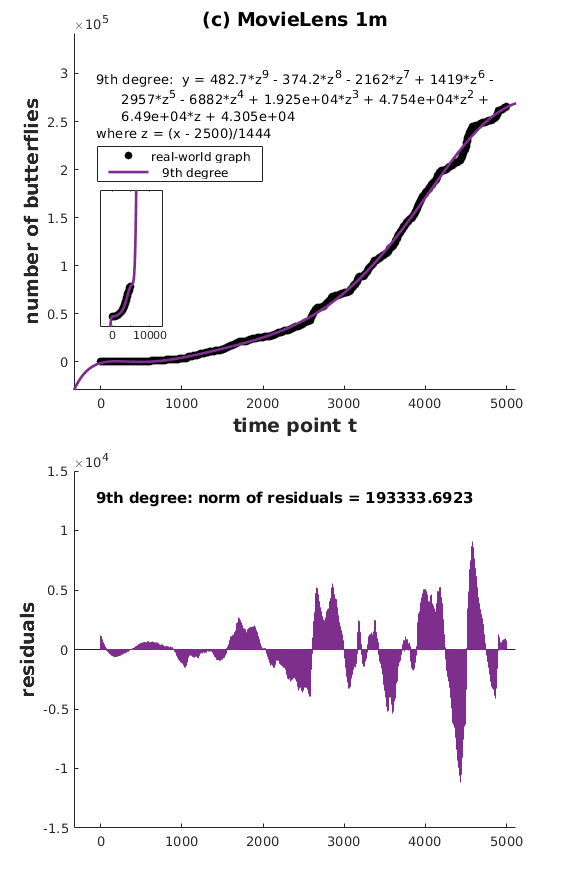}}
    \subfigure{\includegraphics[width=0.15\textwidth]{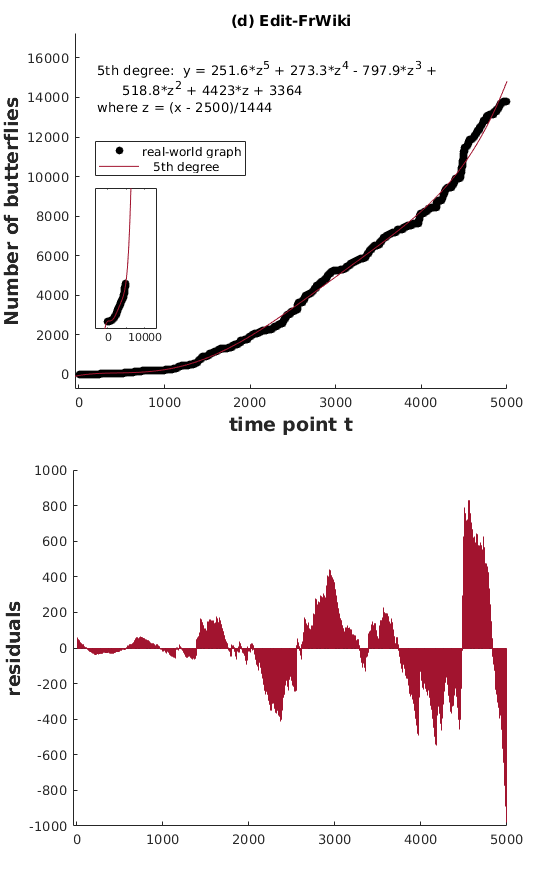}}
    \subfigure{\includegraphics[width=0.15\textwidth]{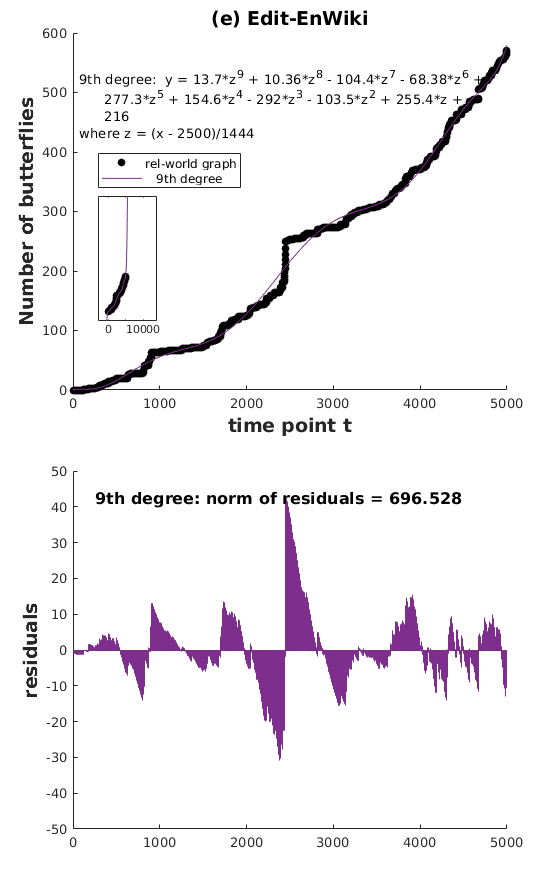}}
    \subfigure{\includegraphics[width=0.15\textwidth]{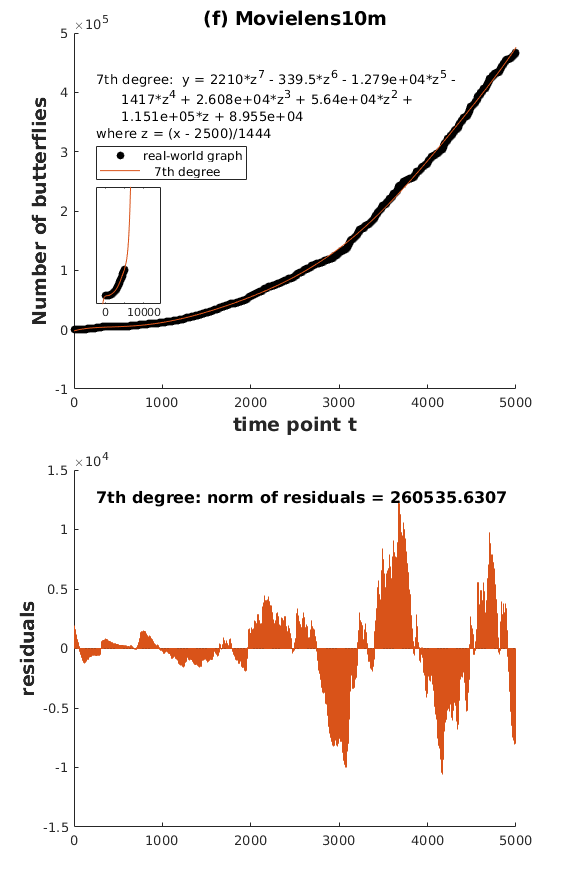}}
    
    \caption{[Best viewed in colored.] Best Fitting functions for the temporal evolution of butterfly frequency (top) and the residual errors of the estimated fitting function (bottom).}
    \label{fig:cliquesfitting}
\end{figure*} 
 
All the edge lists (real and synthetic) are sorted based on the timestamps to simulate the streaming graph records in the analysis.

\subsection{Butterfly 
Emergence Patterns}\label{subsec:butterflydensification}
Network motifs are ``patterns of interconnections occurring in complex networks at numbers that are significantly higher than those in randomized networks'' ~\cite{motifs}. Identifying the motifs helps characterize the graph and also benefits graph querying systems that are based on subgraph-centric programming model (i.e. operates on subgraphs rather than vertices or edges) and can be optimized by indexing the network motifs. That is, network motifs represent the regularities in the graph data and are helpful in building indexes over frequent and regular graph structures (structural indexing) ~\cite{yan2004graph, zhao2007graph, sasaki2020structural}.  On the other hand, the frequent butterflies in a graph is a sign of high clustering coefficient. While butterflies are known to be motifs in static graphs, their temporal emergence patterns is not well studied. Therefore, we study the number of butterflies emerging in the real-world graphs over time. Also, we compare these with the occurrence patterns in randomized graphs to see if the occurrence frequency is higher in real-world graphs. This is required for a sound and complete recognition of butterflies as temporal motifs, since motif definition requires such comparison.

For this analysis we use an exact butterfly counting algorithm for graph \emph{snapshots} (called \emph{countButterflies(G)} -- Algorithm \ref{alg:countB}). Given a bipartite graph snapshot $G_{W,t}=( V(t),E(t))$ at a time point $t$, the goal is to compute $B(t)$ as  the number of all quadruples $\langle i_1, i_2, j_1, j_2 \rangle$ in $G_{W,t}$ such that they form a butterfly via four edges $\{e_{i_1,j_1} , e_{i_1,j_2} , e_{i_2,j_1} , e_{i_2,j_2}\}$ (Figure \ref{fig:patterns}--rightmost). 

Algorithm \ref{alg:countB} follows a vertex-centric approach that does not require accessing  two-hop neighbors (i.e. it is not triple-based) and can be  computed by looping over either i-vertices or j-vertices depending on their average degree (denoted by $K_i$ and $K_j$). The algorithm takes a vertex $i_1$ (provided that $K_i\leq K_j$) and considering each pair of j-neighbors $j_1$ and $j_2$, identifies the common i-neighbors of $j_1$ and $j_2$, i.e. vertices such as $i_2$ (Figure \ref{fig:exactmethods}.d). We use sublists to avoid iterating over repeated j-neighbors (lines 6-8 in Algorithm \ref{alg:countB}) and we identify the common neighbors by iterating over the lower degree j-vertex (line 10 in Algorithm \ref{alg:countB}).  

\begin{algorithm}\caption{countButterflies(G)}\label{alg:countB}
\DontPrintSemicolon
 \KwInput{ 
    $G=\langle V_i\cup V_j , E_{ij} \rangle$,  static graph \\ 
    }
    \KwOutput{$B_G$, The number of butterflies in G}
    
    $Butterflies\gets \emptyset$ \tcp*{An empty hashSet of quadruples}
    $jNeighbors \gets \emptyset$   \tcp*{An empty Set}
     $vi_2s \gets \emptyset$ \tcp*{An empty Set}
    \tcc{loop over $i_1 \in V_i$ if $K_i< K_j$, otherwise loop over $j_1 \in V_j$}
    \For{$i_1 \in V_i$}{   
        $jNeighbors \gets N_{i_1}$ \tcp*{j-neighbors of vertex $i_1$}
        \For{$index1\in [1, size(jNeighbors)]$}{
            $j_1\gets jNeighbors[index1]$\\
            \For{$index2\in [index1+1, size(jNeighbors)]$}{
                $j_2\gets jNeighbors[index2]$\\
                $vi_2s \gets  N_{j_1} \cap  N_{j_2}$  \tcp*{common i-neighbors}
                $Butterflies.add([i_1,j_1,i_2,j_2])$
            }
        }
    }
    $B_G\gets size(Butterflies)$
    
\end{algorithm}

It is important to calculate the exact number of butterflies to make sure that analysis are correct and the identified patterns are reliable. Hence, we adopt an eager computation model where the exact number of butterflies is computed after each edge is added (connecting new/existing vertices) (Algorithm \ref{alg:countB}). We do this in the time period $0$ to $5000$ due to the computational expenses of the computation model. Note that the frequency distribution of edge insertions occurring in time-intervals of variant sizes follows the same shape. This means that the distribution with respect to scaling across time scales is invariant (i.e. self-similar~\cite{wang2002data}). Therefore, we can rely on the analysis on a fraction of the subsequent streaming edges. 

To compare the numbers with that of a random graph (see the definition of network motifs), we just use the corresponding BA graph with the same real timestamp. This enables fair comparison of structural evolution of real-world and synthetic random graphs.

As shown in Figure \ref{fig:cliquesnum},  real-world graphs display rapid temporal evolution of the number of butterflies. To further investigate the growth pattern of butterfly frequency in these graphs, we examine ten polynomial functions of degree one to ten to fit the data points of temporal butterfly frequency evolution  (black lines in Figure \ref{fig:cliquesnum}) and picked the best fitting function (Table \ref{tab:rmser2reals}). The best fitting function satisfies three conditions: (i) it has the lowest RMSE; (ii) it has the highest coefficient of determination ($R^2$); and (iii) it is a non-decreasing function. $RMSE$ quantifies the estimation error, while  $R^2$ quantifies the linear correlation between the estimated fitting function and the data points. Figure \ref{fig:cliquesfitting} illustrates the best fitting function and its estimation errors (residuals) used in calculation of the RMSE. Note that high RMSE values are due to the increasing function giving rise to high residuals. We do not compare the RMSE of different graphs, instead we compare the RMSE of different fitting functions for each graph. Therefore, the absolute value of RMSE is not as important as its relative value for different functions. As shown in Figure \ref{fig:cliquesfitting} all the plots are properly fitted to polynomial functions of degree above 5 (best fitted to 5th, 7th, 9th and 10th degrees). We term this the \textbf{butterfly densification power-law} (following the power-law terminology ~\cite{leskovec2005graphs}): the number of butterflies at time point $t$ (i.e. $B(t)$) follows a power law function of the number of edges at $t$ (i.e. $B(t)\propto f(|E(t)|^\eta$), $\eta>1 $). Moreover, the outstanding frequency of butterflies in the real-world graphs compared to that of random graphs suggests that butterflies are network motifs \emph{across the time line}.

\subsection{Bursty Butterfly Formation} \label{subsec:burstybutterflyformation}
In the previous subsection we observed the densification of butterflies as network motifs. Now, we study how these motifs are formed over time. To this end, we check the distribution of inter-arrival time of a pair of edges forming a butterfly. That is, for any pair of edges $\langle e_1,e_2 \rangle$ with time stamps $\tau_1$ and $\tau_2$ that co-exist in a butterfly, the inter-arrival time is $|\tau_1-\tau_2|$. We adopt a lazy computation model to compute the inter-arrival distribution once after adding $5000$ edges to the graph (i.e. at the time point $t=5000$).

\begin{figure*}[h]
    \centering
    \includegraphics[width=0.15\textwidth]{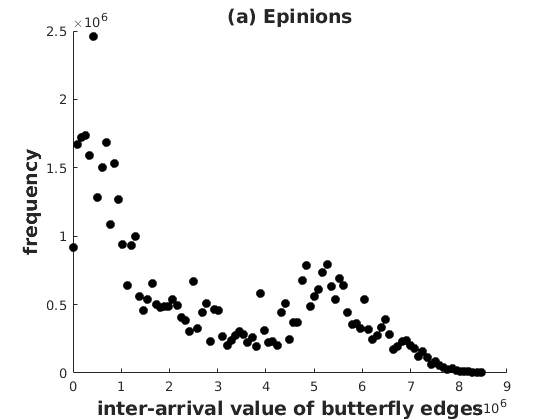} 
    \subfigure{\includegraphics[width=0.15\textwidth]{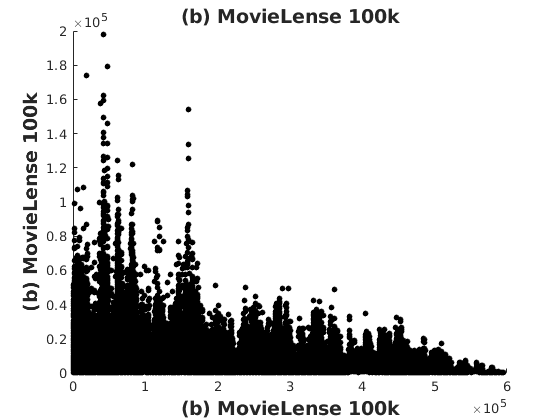}} 
    \subfigure{\includegraphics[width=0.15\textwidth]{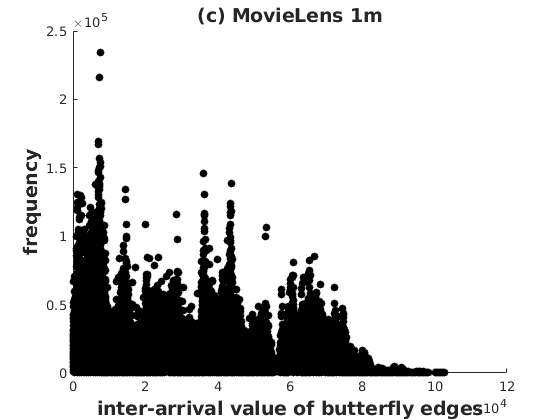}}
    \subfigure{\includegraphics[width=0.15\textwidth]{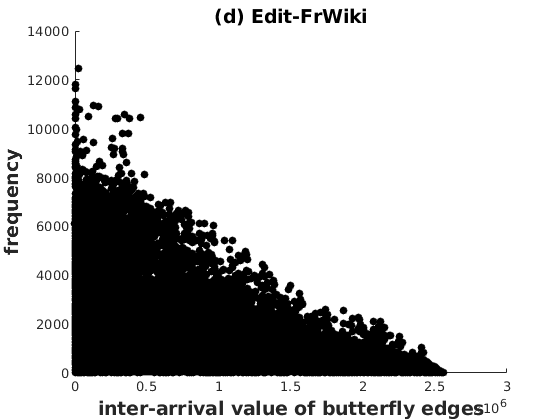}}
    \subfigure{\includegraphics[width=0.15\textwidth]{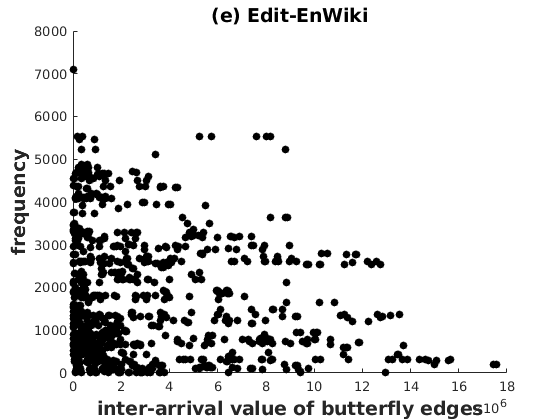}}
    \subfigure{\includegraphics[width=0.15\textwidth]{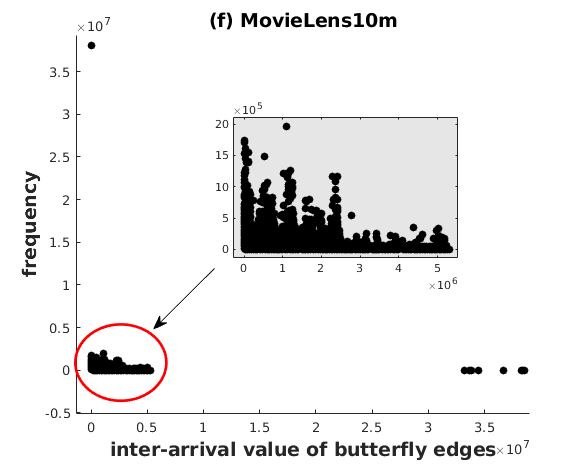}}
    
    \caption{Distribution of inter-arrival time of edges forming butterflies in real-world graphs.}
    \label{fig:interarrivalreal}
\end{figure*}

\begin{figure*}[h]
    \centering
    \subfigure{\includegraphics[width=0.15\textwidth]{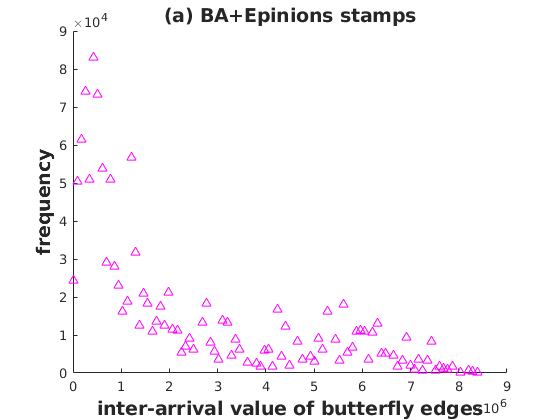}} 
    \subfigure{\includegraphics[width=0.15\textwidth]{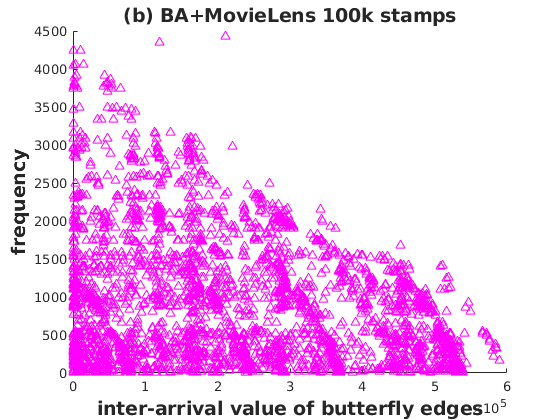}} 
    \subfigure{\includegraphics[width=0.15\textwidth]{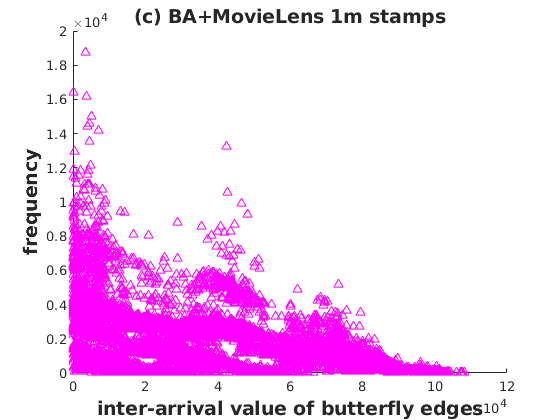}}
    
    \caption{Distribution of inter-arrival time of edges forming butterflies in BA+real stamps graphs.}
    \label{fig:interarrivalBA}
\end{figure*}

As shown in the Figures \ref{fig:interarrivalreal} and \ref{fig:interarrivalBA}, the distribution of inter-arrival values is skewed to the right. The left peaks and the heavy tail of the distribution reveal different patterns. The leftmost peaks highlight that many butterflies are formed by edges with close timestamps. On the other hand, according to  Figure \ref{fig:cliquesnum}, the number of butterflies increase significantly over time. It can be inferred that \textit{butterflies are formed in a bursty fashion}. 

Next, we investigate the vertices that form the butterflies to see (a) whether the bursty butterfly generation is contributed by hubs (i.e. vertices with degree above the average of unique vertex degrees) or normal vertices (Subsection \ref{subsubsec:hubcontribution}), and (b) if hubs are the main contributors, are they young, old, or both? (Subsection \ref{subsubsec:hubagecontribution}).

\subsubsection{Hubs contribution to butterfly emergence}\label{subsubsec:hubcontribution}

We hypothesize that butterflies are contributed by hubs and to test this, we study following items:
\begin{itemize}
    \item The probability of forming butterflies by hubs
    \item The correlation between degree and support of vertices
    \item The connection patterns of hubs
\end{itemize}

\textbf{The probability of forming butterflies by hubs --}
We enumerate butterflies formed at time $t=0$ to $t=5000$ and check the fraction of butterflies formed by zero to four hubs  (Table \ref{tab:hubs}) and the fraction of butterflies formed by zero, one, or two i-/j-hubs (Table \ref{tab:ijhubs}). It is evident that, butterflies mostly include one or, with higher probability, two hubs which are usually i-hubs. 

\begin{table*}[h]\caption{Fraction of butterflies including zero, one, two, three, or four hub(s) at after applying 5000 edge-insertions sgrs.}
    \begin{tabular}{|p{2.9cm}|p{1cm}|p{1cm}|p{1cm}|p{1cm}|p{1cm}|}\hline
        
        \makecell{Fraction} & \makecell{0 hub} & \makecell{1 hub} & \makecell{2 hubs} & \makecell{3 hubs} & \makecell{4 hubs} \\ \hline\hline
        
        \makecell{Epinions\\BA+Epinions stamps} &\makecell{$0.09$\\$0.11$} & \makecell{$0.29$\\$0.44$} &\makecell{$0.55$\\$0.39$}&\makecell{$0.07$\\$0.06$} & \makecell{$0$\\$0$}\\ \hline
        
        \makecell{ML100k\\BA+ML100k stamps} & \makecell{$0.07$\\$0.24$} & \makecell{$0.35$\\$0.28$} &\makecell{$0.48$\\$0.28$}&\makecell{$0.09$\\$0.15$} & \makecell{$0.01$\\$0.05$} \\ \hline
         
        \makecell{ML1m\\BA+ML1m stamps} &\makecell{$0.07$\\$0.01$} & \makecell{$0.38$\\$0.33$} &\makecell{$0.48$\\$0.6$}&\makecell{$0.07$\\$0.06$} & \makecell{$0$\\$0$} \\ \hline
        
        \makecell{ML10m} &\makecell{$0.09$} & \makecell{$0.34$} &\makecell{$0.37$}&\makecell{$0.17$} & \makecell{$0.03$} \\ \hline
        
        \makecell{Edit-Frwiki} &\makecell{$0.08$} & \makecell{$0.29$} &\makecell{$0.53$}&\makecell{$0.1$} & \makecell{$0$} \\ \hline
        
        \makecell{Edit-Enwiki} &\makecell{$0.1$} & \makecell{$0.48$} &\makecell{$0.41$}&\makecell{$0.01$} & \makecell{$0$} \\ \hline
        
    \end{tabular}\label{tab:hubs}
\end{table*}

\begin{table*}[h]\caption{Fraction of butterflies including zero, one, or two i-hub(s) or j-hub(s) at after applying 5000 edge-insertions sgrs.}
    \begin{tabular}{|p{2.9cm}|p{1cm}|p{1cm}|p{1.2cm}|p{1cm}|p{1cm}|p{1.2cm}|}\hline
        
        \makecell{Fraction} & \makecell{0 i-hub} & \makecell{1 i-hub} & \makecell{2 i-hubs} & \makecell{0 j-hub} & \makecell{1 j-hub} & \makecell{2 j-hubs} \\ \hline\hline
        
        \makecell{Epinions\\BA+Epinions stamps} &\makecell{$0.11$\\$0.19$} & \makecell{$0.35$\\$0.56$} & \makecell{$0.54$\\$0.25$} &
                             \makecell{$0.85$\\$0.7$} & \makecell{$0.13$\\$0.25$} &\makecell{$0.02$\\$0.05$} \\ \hline
        
        \makecell{ML100k\\BA+ML100k stamps} &\makecell{$0.10$\\$0.48$} & \makecell{$0.46$\\$0.39$} & \makecell{$0.44$\\$0.13$} &
                           \makecell{$0.75$\\$0.37$} & \makecell{$0.21$\\$0.41$} &\makecell{$0.04$\\$0.23$} \\ \hline
         
        \makecell{ML1m\\BA+ML1m stamps} & \makecell{$0.1$\\$0.01$} & \makecell{$0.43$\\$0.36$} & \makecell{$0.47$\\$0.63$} &
                          \makecell{$0.84$\\$0.9$} & \makecell{$0.15$\\$0.1$} &\makecell{$0.01$\\$0$} \\ \hline
                          
        \makecell{ML10m} & \makecell{$0.25$} & \makecell{$0.54$} & \makecell{$0.21$} &
                          \makecell{$0.47$} & \makecell{$0.33$} &\makecell{$0.2$} \\ \hline
                          
        \makecell{Edit-Frwiki} & \makecell{$0.11$} & \makecell{$0.35$} & \makecell{$0.54$} &
                          \makecell{$0.81$} & \makecell{$0.18$} &\makecell{$0.01$} \\ \hline
                         
        \makecell{Edit-Enwiki} & \makecell{$0.1$} & \makecell{$0.5$} & \makecell{$0.4$} &
                          \makecell{$0.97$} & \makecell{$0.03$} &\makecell{$0$} \\ \hline
        
    \end{tabular}\label{tab:ijhubs}
\end{table*}

\textbf{The correlation between degree and support of vertices --} We study the correlation between degree $deg(i)$ and butterfly support $B_i$, where $B_i$ is defined as the number of butterflies incident to each vertex. For computing the $B_i$, we extend \emph{countButterflies(G)} (Algorithm \ref{alg:countB}) to obtain \emph{ButterflySupport(G)} (Algorithm \ref{alg:bSupport}).

\begin{algorithm}\caption{ButterflySupport(G)}\label{alg:bSupport}
\DontPrintSemicolon
 \KwInput{ 
    $G=\langle V_i\cup V_j , E_{ij} \rangle$,  static graph \\ 
    }
    \KwOutput{$vSupport$, butterfly support of vertices}
    $vSupport\gets \emptyset$ \tcp*{An empty hashMap}
    $Butterflies\gets \emptyset$ \tcp*{An empty hashSet of quadruples}
    $jNeighbors \gets \emptyset$   \tcp*{An empty Set}
     $vi2s \gets \emptyset$ \tcp*{An empty Set}
    \tcc{loop over $i_1 \in V_i$ if $K_i<K_j$, otherwise loop over $j_1 \in V_j$}
    \For{$v_{i1} \in V_i$}{   
        $jNeighbors \gets N_{i_1}$ \tcp*{j-neighbors of vertex $i_1$}
        \For{$index1\in [1, size(jNeighbors)]$}{
            $j_1\gets jNeighbors[index1]$\\
            \For{$index2\in [index1+1, size(jNeighbors)]$}{
                $j_2\gets jNeighbors[index2]$\\
                $vi_2s \gets  N_{j_1} \cap  N_{j_2}$  \tcp*{common i-neighbors}
                \For{$i_2\in vi_2s$}{
                    \If{$[i_1,j_1,i_2,j_2] \not \in Butterflies$}{
                        $Butterflies.add([i_1, j_1,i_2,j_2])$\\
                         $vSupport.put(i_1, vSupport.get(i_1)+1)$\\
                         $vSupport.put(j_1, vSupport.get(j_1)+1)$\\
                         $vSupport.put(i_2, vSupport.get(i_2)+1)$\\
                         $vSupport.put(j_2, vSupport.get(j_2)+1)$
                    }
                }
                
            }
        }
    }
    
\end{algorithm}

\begin{figure*}[t]
  \centering
  \includegraphics[width=\linewidth]{ 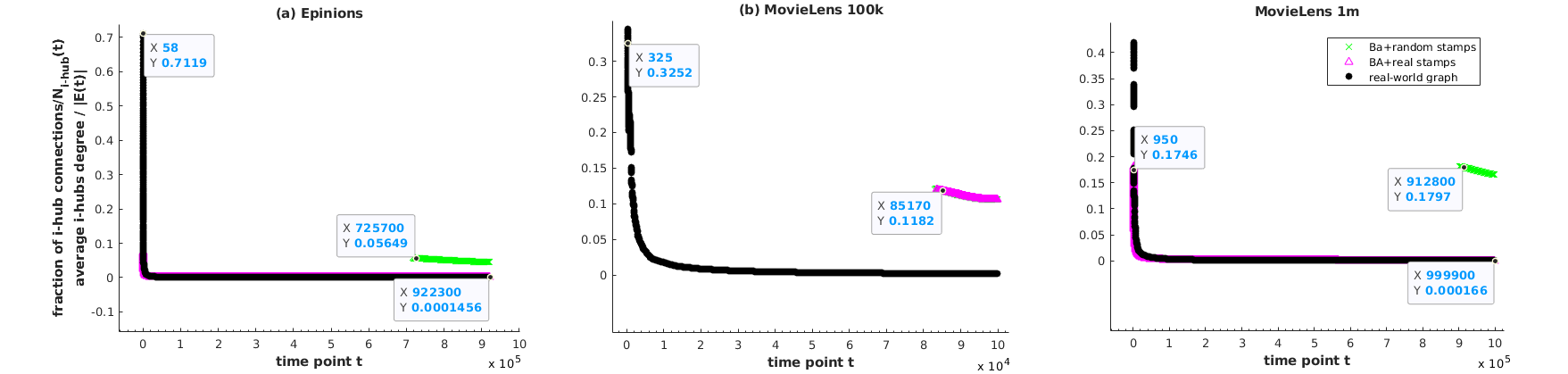}
  \caption{[Best viewed in colored.] Temporal evolution of the normalized fraction of i-hub connection (average i-hub degree).}
  \Description{Temporal evolution of the normalized fraction of i-hub connection (average i-hub degree).}\label{fig:nihubdegree}
\end{figure*}

\begin{figure*}[t]
  \centering
  \includegraphics[width=\linewidth]{ 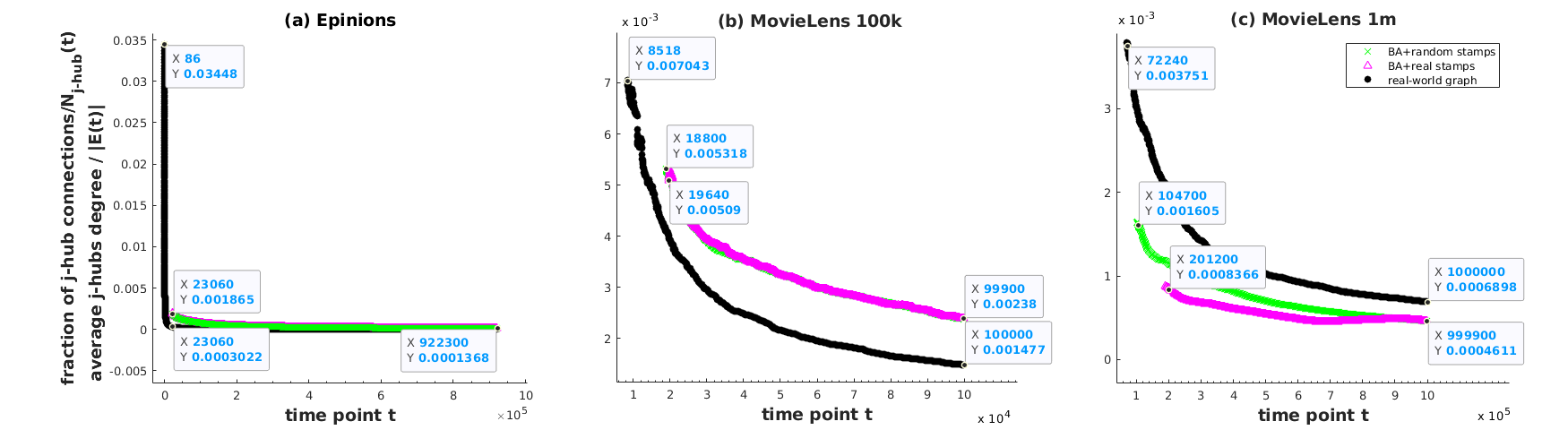}
  \caption{[Best viewed in colored.] Temporal evolution of the normalized fraction of j-hub connection (average j-hub degree).}
  \Description{Temporal evolution of the normalized fraction of j-hub connection (average j-hub degree).}\label{fig:njhubdegree}
\end{figure*}

We refer to the correlation computed over the i-vertices and j-vertices as i-correlation (equation \ref{eq:cor}) and j-correlation (similarly computed), respectively. We use the Pearson correlation coefficient at time point $t=5000$ for all the $N=|V_i|$ or $|V_j|$ seen i-(j-)vertices in the graph snapshot. It should be noted that a positive correlation coefficient means $deg(i)$ and $B_i$ increase or decrease together, while a negative correlation means increasing one quantity implies decreasing the other one. Values close to $1$ demonstrate strong correlation.

{\footnotesize
\begin{equation}\label{eq:cor}
   \frac{N\sum_{i\in V_i} deg(i)B_i-\sum_{i\in V_i} deg(i)\sum_{i\in V_i} B_i}{\sqrt{[N\sum_{i\in V_i}deg(i)^2-(\sum_{i\in V_i} deg(i))^2][N\sum_{i\in V_i} B_i^2-(\sum_{i\in V_i} B_i)^2]}}
\end{equation}
}

As provided in Table \ref{tab:cor}, there is a strong positive correlation between the degree and the support of vertices in real-world graphs. i.e. the higher the degree, the higher the butterfly support and vice versa. This highlights the impact of hubs in the emergence of enormous number of  butterflies in the real-world graphs.

\textbf{The connection patterns of hubs --} We quantify the extent to which i-(j-)hubs  dominate the edges over time by means of two equivalent measures: (i) the fraction of i-(j-)hub connections (denoted by $\frac{\sum_{i=1}^{N_{hub}(t)} (deg(hub_i))}{E(t)}$) normalized over the number of hubs  at time point $t$ (denoted by $N_{hub}(t)$), and (ii) the average degree of i-(j-)hubs (denoted by $\frac{\sum_{i=1}^{N_{hub}(t)} (deg(hub_i))}{N_{hub}(t)}$) normalized over the total number of edges at time point $t$ (denoted by $|E(t)|$).  Both quantities are calculated by $\frac{\sum_{i=1}^{N_{hub}(t)} (deg(hub_i))}{E(t)*N_{hub}(t)}$ at any given time point $t$. We adopt an eager computation model to compute this value when a new edge is added. The time point $t$ can be interpreted as the number of edges added to the graph since the initial time point $t=0$. 

As shown in the Figures \ref{fig:nihubdegree} and \ref{fig:njhubdegree}, while the number of edges added to the graph increases, the normalized fraction of i-(j-)hub connections (average degree of i-(j-)hubs) decreases over time in both real-world and BA graphs. 


Figures \ref{fig:nihubdegree} and \ref{fig:njhubdegree} also reveal that (a) unlike real-world graphs, i- and j-hubs emerge later in the BA graphs (originated by the BA's preferential attachment rule), and (b) the average degree of hubs in early time points is higher in real-world graphs than that of BA graphs. This is due to the bursty characteristic of graph stream (i.e. arrival of a bunch of edges with same time-stamp and same i- or j- vertex). In summary,  early in the stream, the BA graphs have lower number of hubs with lower degrees compared to the real-world graphs. Figure \ref{fig:cliquesnum} also illustrates the low number of butterflies in BA graphs earlier in the stream when there are no hubs in these graphs or the average hub degree is low. On the other hand, real world graphs have high number of hub connections and high number of butterflies. These observations again verify the contribution of hubs to the emergence of butterflies; When the number of hubs is low and the average degree of hubs is also low, the number of butterflies is also low (as seen in BA graphs). Also, when the number of hubs and their average degree is high, the number of butterflies is high (as seen in real-world graphs). 

\begin{table}[h]\caption{Correlation between the butterfly support and the degree of i-vertices (i-correlation) and j-vertices (j-correlation).}
    \begin{tabular}{|p{3.6cm}|p{2cm}|p{2cm}|}\hline
          \makecell{} & \makecell{i-correlation} & \makecell{j-correlation} \\ \hline\hline
        \makecell{Epinions\\BA+Epinions stamps} & \makecell{$0.86$\\$0.56$} & \makecell{$0.73$\\$0.72$} \\ \hline
        
        \makecell{MovieLens1m\\BA+MovieLens1m stamps} & \makecell{$0.98$\\$0.92$} & \makecell{$0.92$\\$0.89$} \\ \hline
        
        \makecell{MovieLens100k\\BA+MovieLens100k stamps} & \makecell{$0.95$\\$0.63$} & \makecell{$0.93$\\$0.88$} \\ \hline
        
        \makecell{MovieLens10m} & \makecell{$0.83$} & \makecell{$0.93$} \\ \hline
        
        \makecell{Edit-Frwiki} & \makecell{$0.91$} & \makecell{$0.85$} \\ \hline
        
        \makecell{Edit-Enwiki} & \makecell{$0.89$} & \makecell{$0.62$} \\ \hline
        
    \end{tabular}\label{tab:cor}
\end{table}

\subsubsection{Contribution of hub age  to butterfly emergence}\label{subsubsec:hubagecontribution}
We hypothesize that butterflies are contributed by old hubs and to test this, we study following items:
\begin{itemize}
    \item The evolution of young and old hubs
    \item The inter-arrival of butterfly edges
\end{itemize}

\textbf{The evolution of young and old hubs --} To further investigate how the age of hubs contribute to the emergence of butterflies, we first check the evolution of young and old hubs. 
As mentioned before, we define the i-(j-)hub as any i-(j-)vertex whose degree is above the average of unique i-(j-)degrees in the graph. Accordingly, young(old) hubs are defined as any hub whose timestamp is in the last(first) $25\%$ of ordered set of already seen timestamps. The vertex timestamps are determined as the timestamp of the sgr by which the vertex has been added to the graph for the first time. For instance, if a vertex $i$ arrives via the inserting edges $e_1=\langle i,j_1 \rangle$ and  $e_2=\langle i,j_2 \rangle$, the time stamp of vertex $i$ is set to the timestamp of $e_1$, which has  arrived before $e_2$ (assuming subscript identify order of arrival). We adopt a lazy  computation model to compute the number of young/old i-(j-)hubs using a time-based landmark window where the computation is done over a growing graph generated by the edges in the append-only window following each expansion. Window expansion lengths are set to cover $0.1*N_t$ unique timestamps in each window in Epinions, ML100k, ML1m, and ML10m. In the larger graph streams Edit-EnWiki and Edit-FrWiki, this value is equal to $0.01*N_t$. $N_t$ is the number of unique timestamps in data stream.

As shown in the Figure \ref{fig:younghubs}, young i-hubs and/or j-hubs are formed in the real-world graphs over time, while in BA graphs with random timestamps the number of young i-(j-)hubs is always zero. In BA graphs with real timestamps, the timestamp of hubs are shuffled, therefore the old hubs are identified as young hubs that should be ignored. Figure \ref{fig:oldhubs} demonstrates that old hubs increase over time in BA graphs, which is not always the case for real-world graphs. Moreover the number of old hubs in real world graphs is less than that of BA graphs. 

\begin{figure*}[h]
    \centering
    \subfigure{\includegraphics[width=0.14\textwidth]{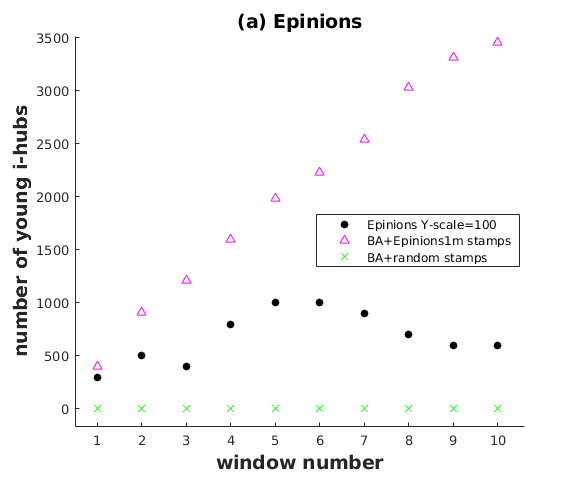}} 
    \subfigure{\includegraphics[width=0.14\textwidth]{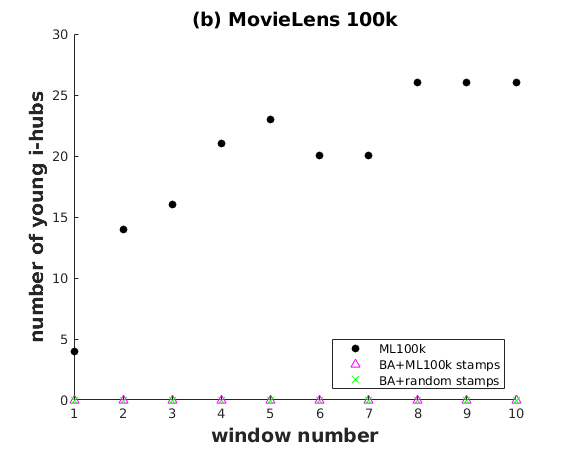}} 
    \subfigure{\includegraphics[width=0.14\textwidth]{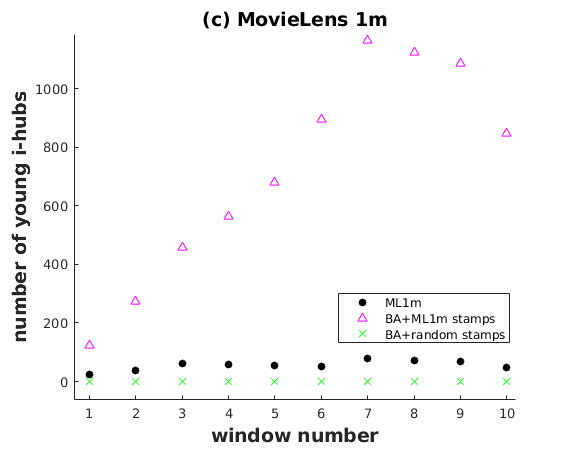}}
    \subfigure{\includegraphics[width=0.14\textwidth]{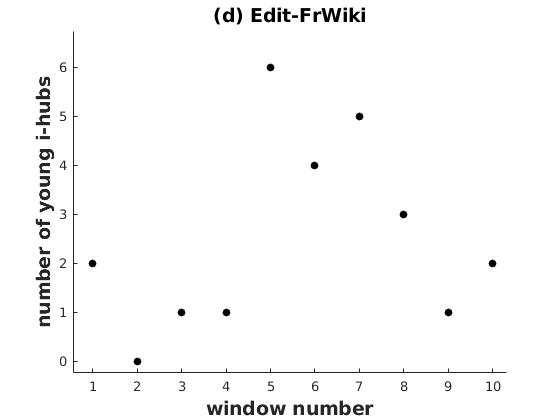}} 
    \subfigure{\includegraphics[width=0.14\textwidth]{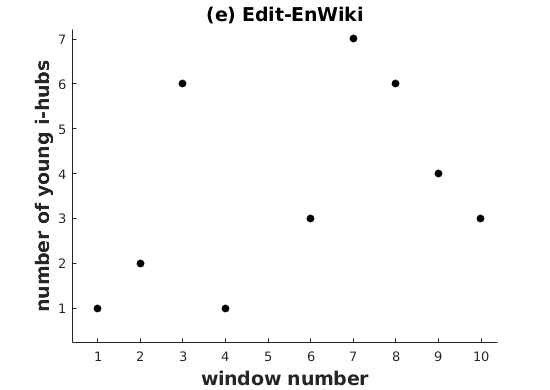}} 
    \includegraphics[width=0.14\textwidth]{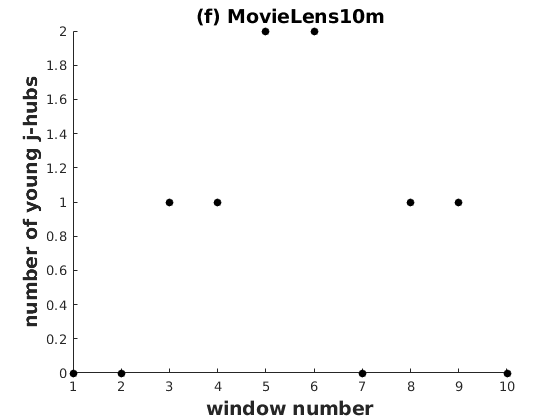}
    
    \subfigure{\includegraphics[width=0.14\textwidth]{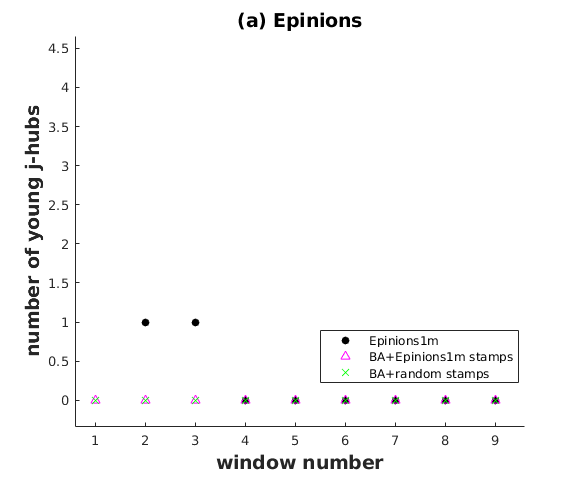}} 
    \subfigure{\includegraphics[width=0.14\textwidth]{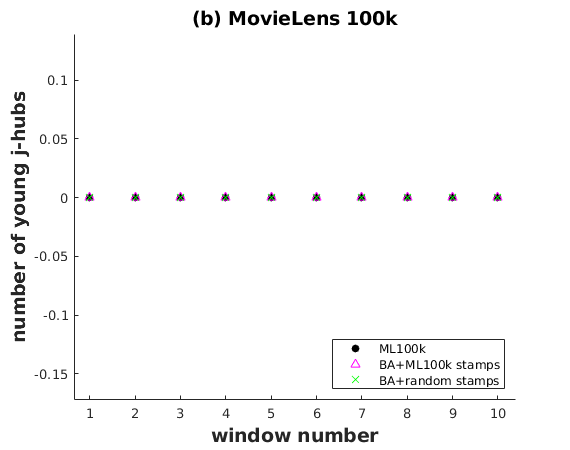}} 
    \subfigure{\includegraphics[width=0.14\textwidth]{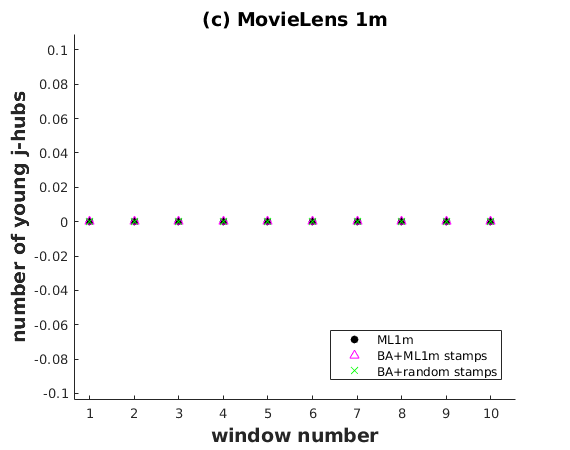}}
    \subfigure{\includegraphics[width=0.14\textwidth]{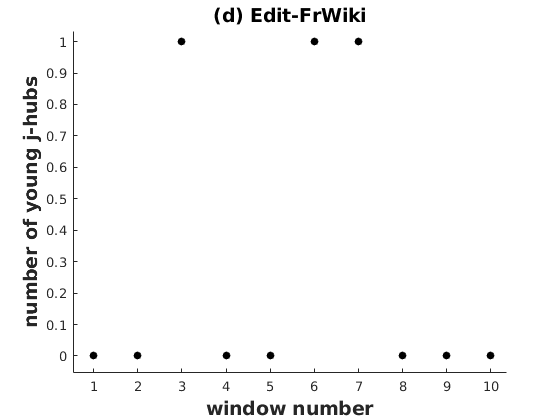}} 
    \subfigure{\includegraphics[width=0.14\textwidth]{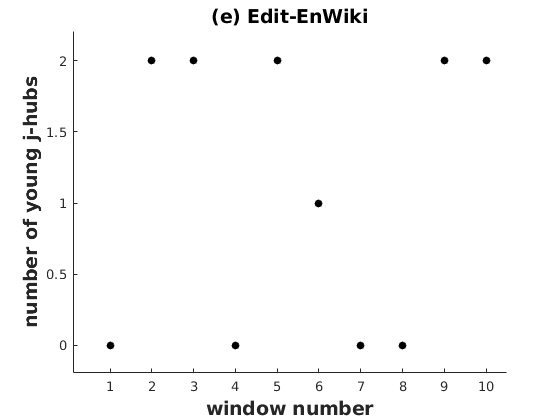}} 
    \subfigure{\includegraphics[width=0.14\textwidth]{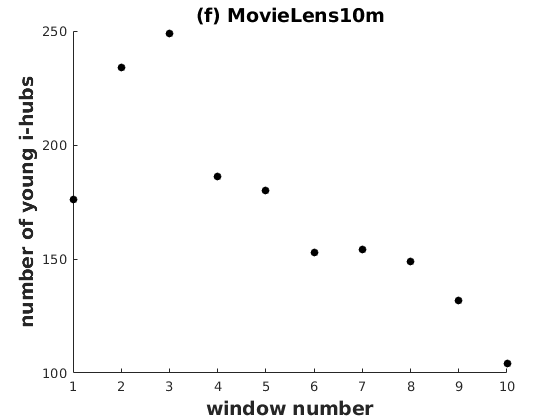}}
    
    \caption{The number of young (top) i-hubs and (bottom) j-hubs after arrival of each batch of edge insertion sgrs.}\label{fig:younghubs}
\end{figure*}


\begin{figure*}[h]
    \centering
    \subfigure{\includegraphics[width=0.14\textwidth]{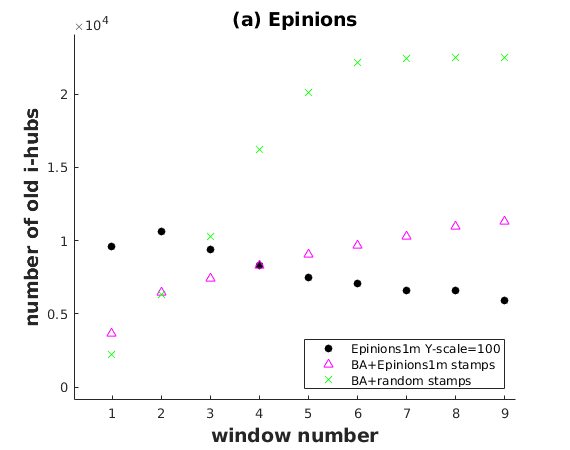}} 
    \subfigure{\includegraphics[width=0.14\textwidth]{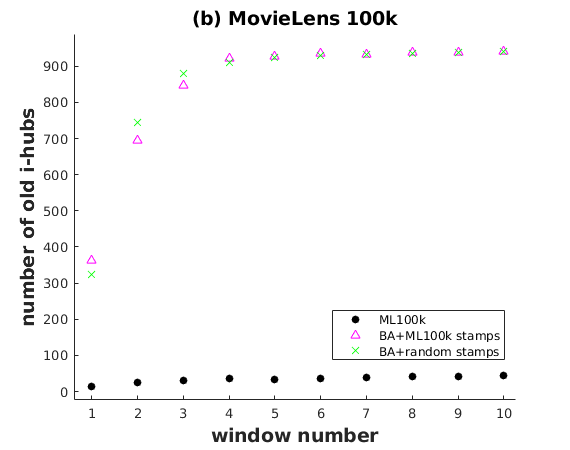}} 
    \subfigure{\includegraphics[width=0.14\textwidth]{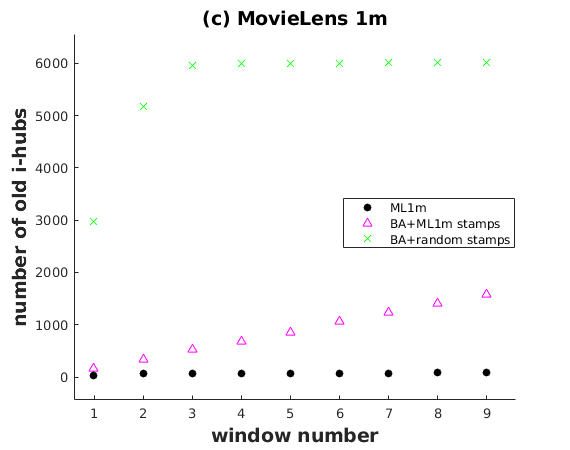}}
    \subfigure{\includegraphics[width=0.14\textwidth]{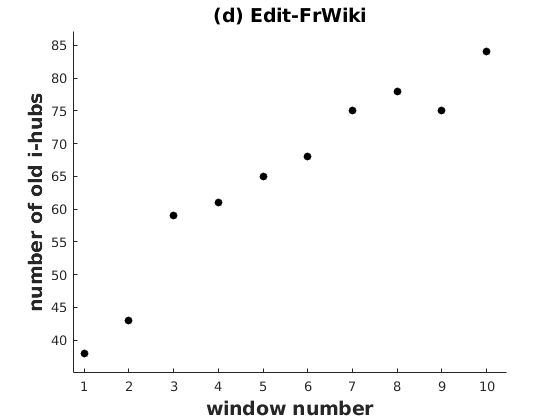}} 
    \subfigure{\includegraphics[width=0.14\textwidth]{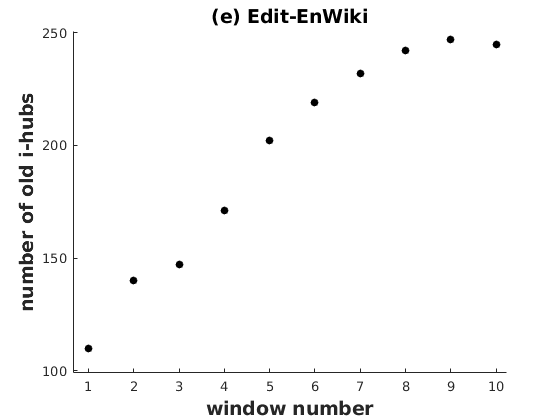}} 
    \subfigure{\includegraphics[width=0.14\textwidth]{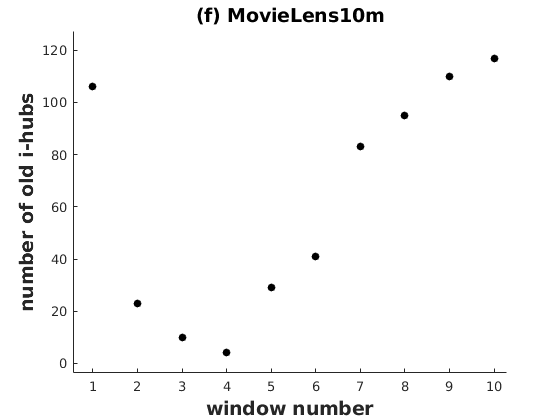}}
    
    \subfigure{\includegraphics[width=0.14\textwidth]{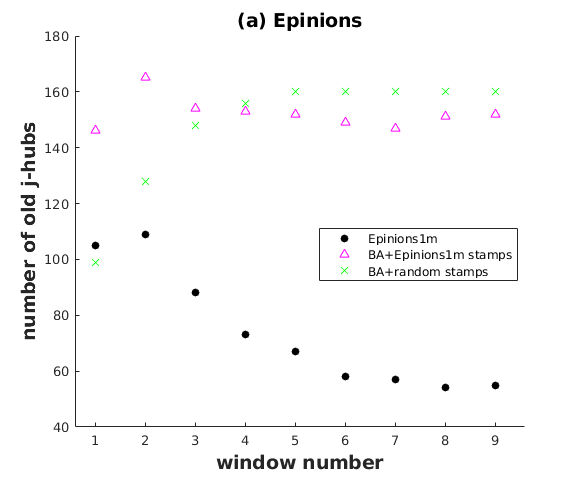}} 
    \subfigure{\includegraphics[width=0.14\textwidth]{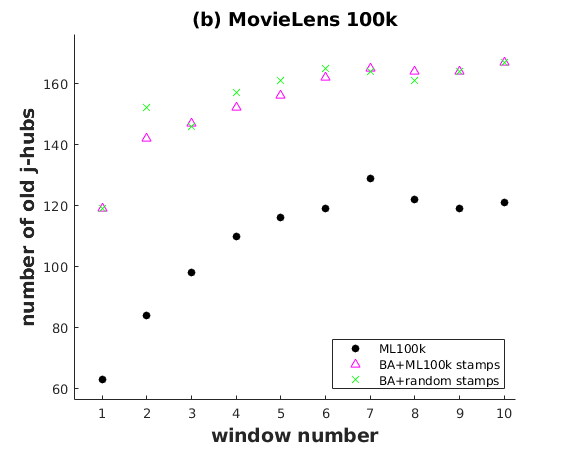}} 
    \subfigure{\includegraphics[width=0.14\textwidth]{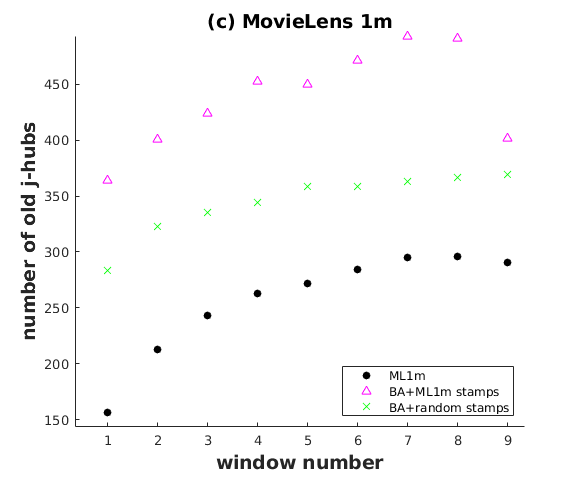}}
    \subfigure{\includegraphics[width=0.14\textwidth]{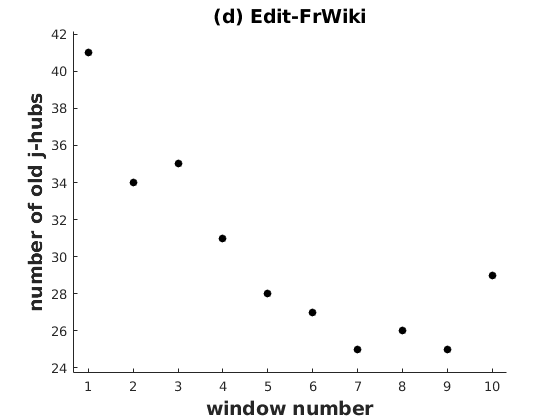}} 
    \subfigure{\includegraphics[width=0.14\textwidth]{oldihubs_enwiki.png}} 
    \subfigure{\includegraphics[width=0.14\textwidth]{oldihubs_ml10m.png}}
    
    \caption{The number of old (top) i-hubs and (bottom) j-hubs after arrival of each batch of edge insertion sgrs.}\label{fig:oldhubs}
\end{figure*}

\textbf{The inter-arrival of butterfly edges --}
Finally, we recheck the heavy tail of the inter-arrival distribution which is over-represented in BA graphs (Figure \ref{fig:interarrivalBA}). The heavy tail is related to the butterfly edges with high inter-arrival times. These highly frequent butterfly edges with high inter-arrivals reflect the connection between the young vertices and old vertices. We hypothesize that young vertices are ordinary vertices and old ones are hubs and we prove it since (a) we proved in the previous subsection that hubs are main contributors to butterfly emergence; and (b) the hubs forming the butterflies cannot be young hubs as BA graphs would be contradiction; BA graphs do not have young hubs (Figure \ref{fig:oldhubs}), while they have many butterfly edges with high inter-arrival(Figure \ref{fig:cliquesnum}), so butterflies cannot originate from young hubs. Therefore, old hubs signify the bursty butterfly emergence. Young hubs can exist, but they are not the hubs dominating the butterflies.

\subsection{Discussion}
In this section, we summarize our findings in this study of the emergence of butterflies in streaming graphs.  We observed that butterflies are network patterns across the time line of sgr arrivals since the number of butterflies increases significantly over time in real-world streaming bipartite graphs, and at each time point the number of butterfly occurrences in real-world graphs are significantly higher than random graphs. We formulated the emergence of butterfly interconnections as the \emph{butterfly densification power law}, stating that the number of butterflies at any time point $t$ is a power law function of the size of stream prefix seen until $t$.

In terms of \emph{how} these enormous number of butterflies emerge over time, our
 studies reveal the contribution of hubs in the streaming graphs. Further  investigation of the impact of hubs in terms of their age reveal that the older hubs contribute more to the densification of butterflies. 

An efficient streaming algorithm for butterfly counting can only deal with a subset of  the stream at any given point in time. Also, a precise streaming algorithm demands taking into account all existing butterflies regardless of how long they take to form and how much memory is available. The statistical analysis uncover the temporal organizing principles of butterflies that impact the identification of any potential butterfly that should be counted by the algorithm. Specifically, our study reveal the dominant contribution of old hubs with young neighbours on shaping butterfly structures over time. That is, a butterfly takes a long time to form as it takes a while before newly added vertices get connected to old hubs and the butterfly structure completes. In  window-based algorithms such as ours, care is required in windowing as butterflies may be split across windows, affecting the butterfly count -- it is important to take into account the butterflies that may fall between windows. Moreover, when counting the number of multiple-window-spanning butterflies, it is important to take advantage of the butterfly densification power law that  quantifies the butterfly count with respect to the number of edges seen so far. The total number of received edges is easy to track in streaming graphs. Analysis of real-world graph streams as we have done enabled us to design a data-driven butterfly counting algorithm discussed in next section. 
\section{ \texorpdfstring{\MakeLowercase{s}G\MakeLowercase{rapp}}{sGrapp}} \label{sec:approx}
The analysis results presented in the previous section, in particular the contribution of old hubs in bursty butterfly densification (the heavy tail of  the distribution of inter-arrival values in Figure \ref{fig:interarrivalreal}), provide insights to butterfly counting in streaming graphs. In view of these, the precise problem definition reads as  follows: \emph{Given a sequence of streaming graph records ordered by their timestamps, the goal is to compute  the total number of butterflies in  emerging graph $G$ at  time point $t$ -- denoted as $B(t)$.} In other words, the count is over the snapshot corresponding to the prefix of the stream seen so far. Computing $B(t)$ over a streaming graph is not feasible, since the stream is unbounded. It is known that without knowing the size of the streaming input data, it is not possible to determine the memory required for processing the data~\cite{arasu04}, and unless there is unbounded memory, it is not possible to compute exact answers for this data stream problem~\cite{babcock2002}. 
Butterfly counting is an example of streaming problems that are provably intractable if the available space is sub-linear in the number of stream elements~\cite{mcgregor14}. Windowing addresses this fundamental problem by providing an approximate result. However, as data enters and leaves the window as the graph emerges, the result is approximate. 
Approximation has been recognized as an important method for processing high speed data streams, and windows are known as a natural approximation method over data streams~\cite{babcock2002}.

Consequently, in this section we develop an approximate butterfly counting algorithm called sGrapp that uses windowing. The algorithm uses tumbling windows in order to avoid double counting of repeated butterflies. As defined in Section \ref{sec:overview}, tumbling windows do not overlap when windows move, thus avoiding the double-counting problem.
We  adopt a lazy time-based tumbling window model to compute the number of butterflies introduced by each window of disjoint edge insertions, $W_k$, at the end time of the window denoted by $B^{W_k}$, and increment the cumulative value accordingly: $B(t=W_k^e)=B(t=W_{k-1}^e)+B^{W_k}$. This processing is incremental. An issue that has to be addressed is that there may exist some butterflies that are formed by the edges with large inter-arrival times (heavy and long tail in Figure \ref{fig:interarrivalreal}). These butterflies are not captured within one window (unless it is sufficiently large) and we refer to these as \emph{inter-window butterflies}. However, setting the window length to a big value to cover the inter-window butterflies implies a high computational footprint in terms of memory and time. This conflicts with the goal of using a windowed approach to  lower this footprint by performing incremental processing over subsets of sgrs. 
sGrapp addresses this issue by not requiring lengthy windows but using tumbling windows with adaptive lengths. 

sGrapp estimates the number of butterflies from the beginning of the first window $t=W_0^b$ until the end of $k$th window denoted as $\hat{B}(t=W_k^e)=\hat{B_k}$ by counting the exact number of butterflies in the graph corresponding to the current window $W_k$ as $B_G^{W_k}$ and approximating the number of inter-window butterflies ($\hat{B}^{interW}$). The estimated cumulative value would be $\hat{B_k}=\hat{B}_{k-1}+B_G^{W_k}$+$\delta(k\neq0)\hat{B}^{interW}$, where the function $\delta(\cdot)$ returns 1 for true input and 0, otherwise. Note that the first window $W_0$ has no inter-window butterflies and hence the corresponding term would become zero by means of the delta function. In the following, we introduce our adaptive window framework to perform the butterfly approximation (Subsection \ref{subsec:adaptivewin}). Next, we explain how sGrapp approximates the $\hat{B}^{interW}$ and consequently $\hat{B_k}$  (Subsection \ref{subsec:approxalg}). Afterwards,  we discuss optimizations to sGrapp (Subsection \ref{subsec:optimization}). We end this section by analyzing the computational complexity and error bounds of sGrapp (Subsection \ref{subsec:bounds}).

\subsection{Adaptive time-based sliding windows}\label{subsec:adaptivewin}
A main challenge with time-based windows is \textit{how to set the length of windows}? A common approach in stream processing is setting the length of a window using a predetermined value $L$ ($|W_i|=L$, $\forall i$). However, different graph streams have different temporal distributions (frequency distribution of sgr timestamps -- Figure \ref{fig:temporaldist}) and the number of arrived sgrs is not uniform across all time intervals. Therefore, this approach would result in windows of sgrs that cover differing numbers of timestamps, which imposes unbalanced loads on the processing algorithms, particularly in the case of sgr arrivals with bursty characteristics and non-uniform temporal distribution.


\begin{figure*}[h]
    \centering
    \subfigure{\includegraphics[width=0.26\textwidth, height=0.11\textheight]{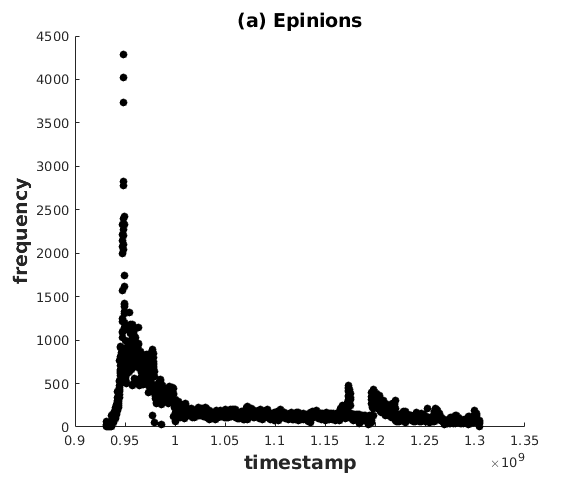}} 
    \subfigure{\includegraphics[width=0.26\textwidth, height=0.11\textheight]{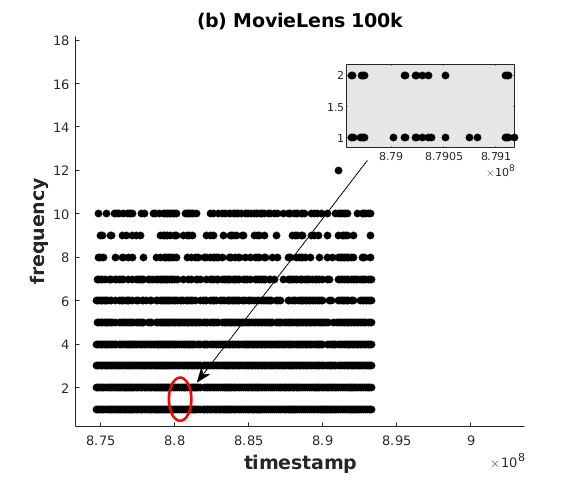}} 
    \subfigure{\includegraphics[width=0.26\textwidth, height=0.11\textheight]{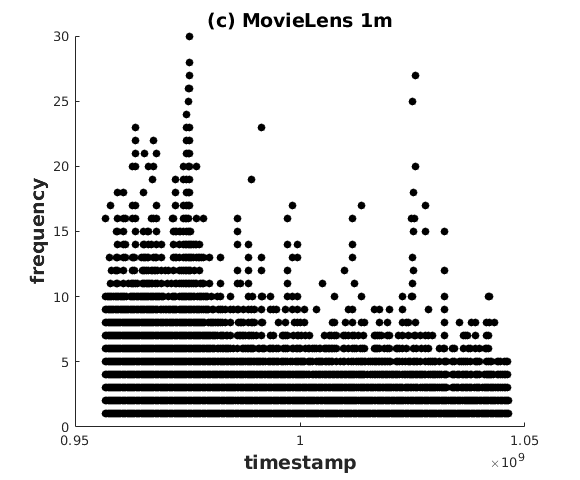}}
    \subfigure{\includegraphics[width=0.26\textwidth, height=0.11\textheight]{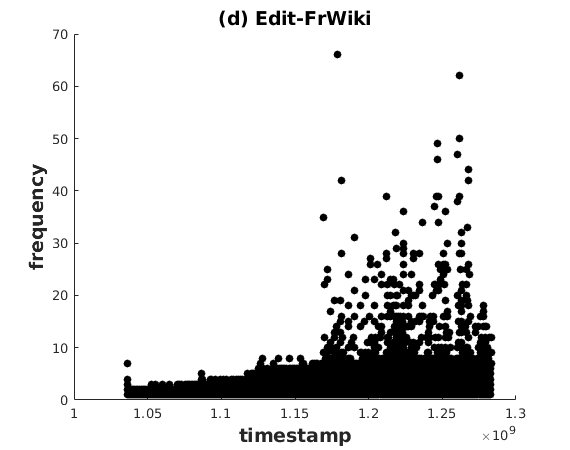}}
    \subfigure{\includegraphics[width=0.26\textwidth, height=0.11\textheight]{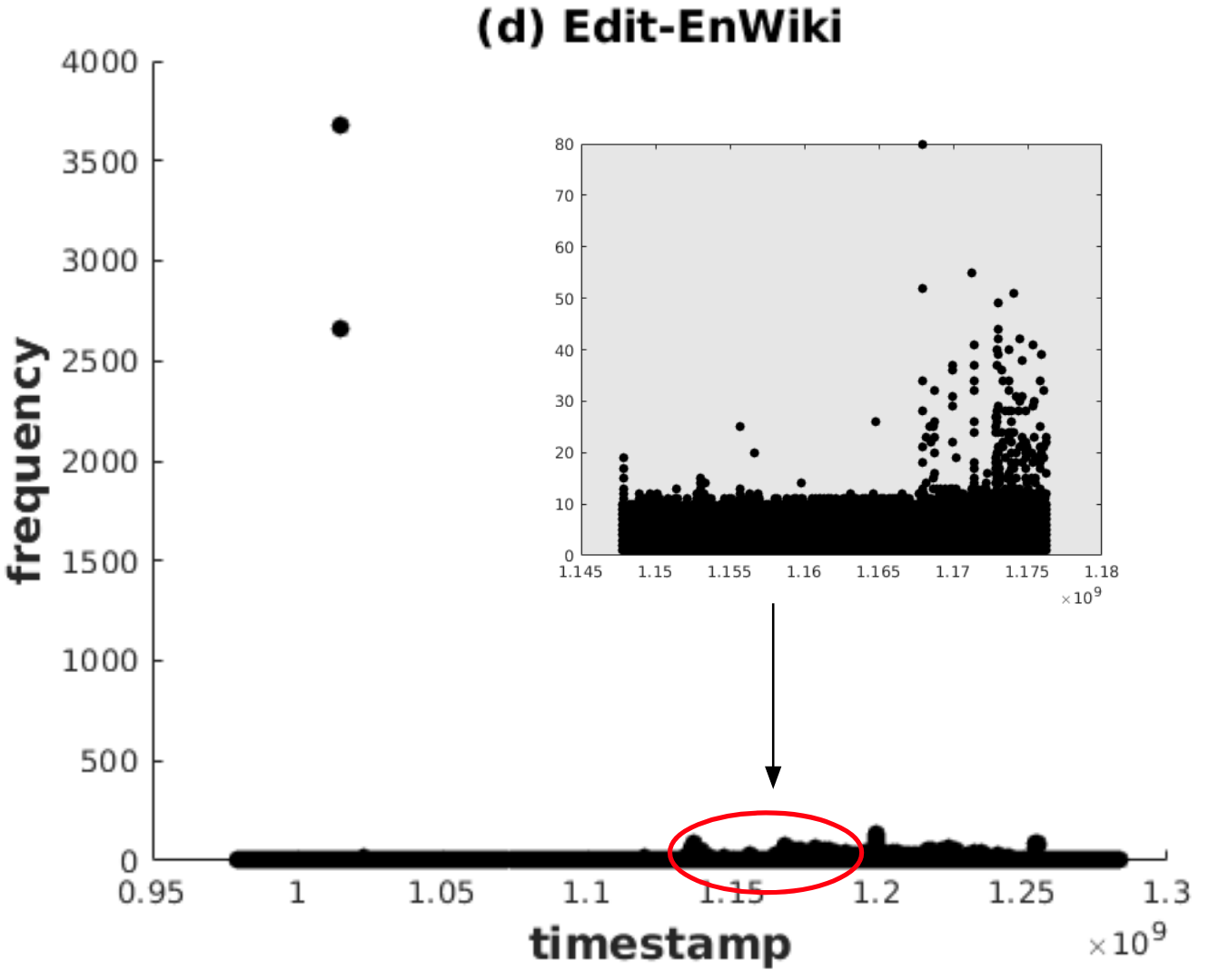}}
    \subfigure{\includegraphics[width=0.26\textwidth, height=0.11\textheight]{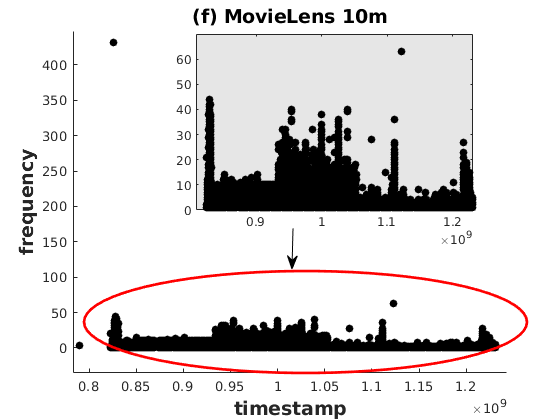}}
    
    \caption{ Temporal distribution of real-world graph streams. 
    }
   \label{fig:temporaldist}
\end{figure*}

To tackle this issue, we introduce an adaptive approach to set the window length. This approach determines the window length according to the timestamps of the graph stream and adapts to the temporal distribution of the stream (Algorithm \ref{alg:adaptivewin}) with no assumption about the order and number of arriving sgrs per time unit.Hence, graph streams with differing arrival rates and temporal distributions can be accommodated. Precisely, we use a number of time-based tumbling windows each including a variable number of sgrs but a certain number of \emph{unique} timestamps in the graph stream, $N_t^w$. For instance, in Subsection \ref{subsubsec:hubagecontribution} we used 10 windows each including variant number of sgrs that cover $10\%$ of unique timestamps ($N_t^w=0.1*N_t$) (Figures \ref{fig:younghubs} and \ref{fig:oldhubs} ). That is, given the number of unique timestamps per window $N_t^w$, we ingest sgrs to the window (lines $8$ to $11$ in Algorithm \ref{alg:adaptivewin}). When $N_t^w$ timestamps are seen, we close the window and perform the intended analysis over the corresponding snapshot (lines $12-13$ in Algorithm \ref{alg:adaptivewin}). The outputs of the analysis are streamed out correspondingly. 
Next, the window slides forward (line $14$ in Algorithm \ref{alg:adaptivewin}) and the retired edges are deleted from the computational graph (lines $15-16$ in Algorithm \ref{alg:adaptivewin}). In tumbling windows, all the edges are retired when the window slides, and the graph snapshot is renewed. 
The time-step is incremented and the algorithm continues until there is a sgr (i.e. continuously in real world streams).

This may appear as a count-based window, but it is not. A count-based window would contain a fixed number of sgrs, while we only fix the number of unique timestamps in the window, not the sgrs. Therefore, ours is time-based with adaptive width since the window borders adapt to the temporal distribution of the stream. In fact adaptive windowing would reduce to count-based windowing, if and only if the temporal distribution of stream is uniform and unique timestamps occur with equal frequency numbers. Therefore our windowing mechanism is general and conforms to real streams. Sequential adaptive windows cover the same fraction of distribution of the sgrs (load-balanced windows for efficient analysis) and also enables comparing the analysis over different windows of a graph stream as well as analysis over different graph streams having different temporal distributions (time-based windows for the accuracy of temporal analysis).

\begin{algorithm}\caption{Adaptive tumbling windows}\label{alg:adaptivewin}
  \DontPrintSemicolon
   \KwData{
    $\{r^i\}$,  sequence of time-ordered sgrs}
     \KwInput{ 
    \\$N_t^W$,  Number of unique timestamps in stream 
    }
    \KwOutput{$x$,  Analysis output collection}
    
    $G \gets \langle V=\emptyset,E=\emptyset \rangle$  \tcp*{initial empty graph}
    $t \gets 0$           \tcp*{time-step}
    $unqt \gets \emptyset$ \tcp*{an empty hashSet}
    $x \gets \emptyset$    \tcp*{output collection}
    $k \gets 0$            \tcp*{window number}
    $W_k^b \gets \tau^0$  \tcp*{begining time of $k$th window}
    \While{true}{
        $r^t=(\tau^t,p) \gets sgrIngest()$ \\
        \If{ $r^t \neq \emptyset$}{
            $unqt$.add($\tau^t$)\\   
            $G \gets updateG(r^t,G)$ 
        }

        \If{$unqt.size()==N_t^W$}{
                $x[k] \gets analysis(G)$ \\
            $k \gets k+1$\\
            $W_k^b \gets \tau^t$ 
            \For{$e \in G : e.timestamp \leq W_k^b$}{
                $G \gets Delete(e,G)$ 
            }
            
        }
        $t \gets t+1$
    }
\end{algorithm}

\subsection{Approximating the number of inter-window butterflies}\label{subsec:approxalg}
\begin{algorithm}[ht]\caption{sGrapp}\label{alg:approx}
  \DontPrintSemicolon
   \KwData{
    $\{r^i\}$,  sequence of time-ordered sgrs}
     \KwInput{ 
    \\$N_t^W$,  Number of unique timestamps per window \\
    $\alpha$, Approximation exponent 
    }
    \KwOutput{$timestep-Bcount$, Approximated number of butterflies at the end of each window}
    
    $G \gets \langle V=\emptyset,E=\emptyset \rangle$  \tcp*{initial empty graph}
    $t \gets 0$           \tcp*{time-step}
    $unqt \gets \emptyset$ \tcp*{an empty hashSet}
    $k \gets 0$            \tcp*{window number}
    
    $timestep-Bcount \gets \emptyset$ \tcp*{an empty hashMap}
    $B_G^{W_k} \gets 0$               \tcp*{number of butterflies in the graph of $k$th window}
    $\hat{B}_K \gets 1$                   \tcp*{cumulative number of butterflies until $t=W_k^e$}
    $E \gets 0$                     \tcp*{total number of edges since $t=0$}
    \While{true}{
        $r^t=(\tau^t,p) \gets sgrIngest()$ \\
        \If{ $r^t \neq \emptyset$}{
            $unqt$.add($\tau^t$)   \\
            $G \gets updateG(r^t,G)$ \\
            $E \gets updateE(r^t,E)$ 
        }
        
        \If{$unqt.size()==N_t^W$}{
                $B_G^{W_k} \gets countButterflies(G)$ \\
                $\hat{B}_K\gets B + B_G^{W_k} + \delta(k\neq0)E^\alpha$\\
                $timestep-Bcount.put(t,B_k)$
            $k \gets k+1$\\
            \tcc{Retire all the edges in the processing graph.}
            $G \gets \langle V=\emptyset,E=\emptyset \rangle$\\
        }
        $t \gets t+1$
    }
\end{algorithm}

Algorithm \ref{alg:approx} describes how sGrapp uses the adaptive windowing framework (Algorithm \ref{alg:adaptivewin}) to estimate the number of butterflies in the streaming graph. Note that sGrapp uses tumbling windows, therefore instead of checking the timestamp of windowed edges to decide on the retirement (lines $15-16$ of Algorithm \ref{alg:adaptivewin}), the processing graph is renewed in sGrapp (line $19$ of Algorithm \ref{alg:approx}). As mentioned earlier in this section the total number of butterflies (line $17$ of Algorithm \ref{alg:approx}) is calculated as total number of butterflies computed at the end of previous window plus the exact number of butterflies in the current window (computed by invoking Algorithm \ref{alg:countB} in line $16$ of Algorithm \ref{alg:approx}) plus the estimated number of inter-window butterflies contributed by current window. According to the butterfly densification power law discussed in the previous subsection, the number of butterflies follows a power-law function of the number of existing edges in the graph. Moreover, recall the observation that butterflies are formed by hubs. Thus, we propose to approximate the number of inter-window butterflies as $\hat{B}^{interW}=|E(t=W_k^e)|^\alpha$, where $|E(t=W_k^e)|$ is the total number of edges since $t=W_0^b$ until $t=W_k^e$. The total number of added edges are updated at ingestion time (line $14$ Algorithm \ref{alg:approx}) as $E$ is increased when the sgr is an edge insertion and decreased when sgr is an edge deletion  and $\alpha$ is the approximation exponent.



\subsection{Optimization}\label{subsec:optimization}
The approximation exponent used in sGrapp (Algorithm \ref{alg:approx}) is constant over windows. However, as we show in the experimental studies in Section \ref{sec:experiments}, the estimated number of butterflies using static exponent can be over or under the true value in subsequent windows. The reason is that the number of edges connecting to old hubs varies across different windows and consequently the estimation should not increase linearly with respect to the number of edges. 

To address this problem, we optimize sGrapp by changing the exponent over windows. To this end, we modify the unsupervised algorithm of sGrapp to a semi-supervised algorithm that we call sGrapp-x. We provide the algorithm with true value of butterflies for an initial subset of the stream. Based on the true value, in the corresponding window $W_K$ we compute the relative error $\frac{\hat{B_K}-B_K}{B_K}$ (line $27$ in Algorithm \ref{alg:approxOptimized}). If the relative error is lower than a user-specified negative tolerance value (in the experiments we use $-0.05$), that means there is an underestimation, therefore we increase the exponent by $0.005$ (line 23-24 in Algorithm \ref{alg:approxOptimized}). Similarly we decrease the exponent in case the relative error is above positive tolerance value to avoid over-estimation in the next window (line 21-22 in Algorithm \ref{alg:approxOptimized}). The exponent is stabilized when the error is tolerable and after the supervised search for the exponent is finished. In summary, the optimized version of sGrapp is an adaptive algorithm using reinforcement learning that learns the most accurate approximation exponent for any given window parameter $N_t^W$ in a subset of stream and generalizes the learned exponent to the rest of stream. sGrapp-x is semi-supervised with outstanding performance given limited ground truth. 

\begin{algorithm}[ht]\caption{sGrapp-x}\label{alg:approxOptimized}
  \DontPrintSemicolon
   \KwData{
    $\{r^i\}$,  sequence of time-ordered sgrs \\
    $B$, ground truths}
     \KwInput{ 
    \\$N_t^W$,  Number of unique timestamps per window \\
    $\alpha$, Approximation exponent
    }
    \KwOutput{$timestep-Bcount$, Approximated number of butterflies at the end of each window}
    
    $G \gets \langle V=\emptyset,E=\emptyset \rangle$\\ 
    $t \gets 0$\\          
    $unqt \gets \emptyset$\\ 
    $k \gets 0$ \\          
    $timestep-Bcount \gets \emptyset$\\ 
    $B_G^{W_k} \gets 0$  \\            
    $\hat{B}_K \gets 1$ \\                  
    $E \gets 0$  \\                  
    $error_0 \gets 0$            \tcp*{relative error for window $W_0$}
    \While{true}{
        $r^t=(\tau^t,p) \gets sgrIngest()$ \\
        \If{ $r^t \neq \emptyset$}{
            $unqt$.add($\tau^t$)   \\
            $G \gets updateG(r^t,G)$ \\
            $E \gets updateE(r^t,E)$ 
        }
        
        \If{$unqt.size()==N_t^W$}{
                $B_G^{W_k} \gets countButterflies(G)$ \\
                
                \If{$t<size(B) \And error>0.05$}{
                   $\alpha -= 0.005$ 
                }
                \If{$t<size(B) \And error<-0.05$}{
                   $\alpha += 0.005$ 
                }
                
                $\hat{B}_K\gets B + B_G^{W_k} + \delta(k\neq0)E^\alpha$\\
                $timestep-Bcount.put(t,B_k)$\\
                
                \If{$t<size(B)$}{
                   $error \gets \frac{\hat{B_K}-B_K}{B_K} $
                }
                 

            $k \gets k+1$\\
            $G \gets \langle V=\emptyset,E=\emptyset \rangle$\\
        }
        $t \gets t+1$
    }
\end{algorithm}

 \subsection{Analysis}\label{subsec:bounds}

Previous study of space bounds has shown that any butterfly counting algorithm, either randomized or deterministic, that returns an accurate (exact/approximate) answer (i.e. bounds the relative error to a small value $0<\delta<0.01$ for each computation round) requires storing the entire graph in $\theta(n^2)$ bits, where $n$ is the number of vertices~\cite{sanei2019fleet}. On the other hand, it is not possible to determine the size of stream (i.e. $n$) in real world streaming graphs. Hence, it is not possible to determine the memory required for processing the data without knowing the size of data~\cite{arasu04}. In the following we analyze the properties of our estimator in terms of computational and error bounds.
\subsubsection{Computational Bound}
 \begin{theorem}
 The upper bound of computational complexity of sGrapp for each window $W_k$ is $\mathcal{O}(\frac{K_{i,W_k}(K_{i,W_k}-1)}{2}K_{j,W_k}\mathcal{R}N_t^{W_k}) $, where $\mathcal{R}$ is the average stream rate and $K_{i,W_k}$($K_{j,W_k}$) is the lower bound of degree of i(j)-vertices in $W_k$.
 \end{theorem}

 \begin{proof}
sGrapp's computations at each window are dominated by the exact counting algorithm as calculating the number of inter-window butterflies is negligible and we ignore it as well as the summations. When i-vertices are the vertex set with lower average degree, the computational complexity of the core exact counting algorithm is the following. 

 \begin{equation}
     \mathcal{O}(\sum_{i_1\in V_i}\sum_{j_1,j_2\in N_{i_1}} Min(deg(j_1) , deg(j_2)))
 \end{equation}

 Let us assume that the lower bound i-degree and j-degree in the graph snapshot corresponding to the tumbling window $W_k$ are $K_{i,W_k}$ and $K_{j,W_k}$, respectively. Accordingly, the computational complexity for this window would be $O(\frac{K_{i,W_k}(K_{i,W_k}-1)}{2}K_{j,W_k}|V_{i,W_k}|)$, where $V_{i,W_k}$ denotes the set of i-vertices in the window $W_k$. Since the stream can include edges connecting already existing vertices, the total number of edges in $W_k$, denoted as 
 $E_{W_k}$, is greater than equal the total number of i-vertices in $W_k$, i.e. $|V_{i,W_k}|\leq |E_{W_k}|$. Therefore,

 \begin{equation}\label{eq:}
     \mathcal{O}(\frac{K_{i,W_k}(K_{i,W_k}-1)}{2}K_{j,W_k}|V_{i,W_k}|)   \leq \mathcal{O}(\frac{K_{i,W_k}(K_{i,W_k}-1)}{2}K_{j,W_k}|E_{W_k}|)    
     \end{equation}
sGrapp uses tumbling windows with adaptive lengths, therefore $|E_{W_k}| \approx \mathcal{R}N_t^{W_k}$, where $\mathcal{R}$ is the average stream rate (i.e. number of edges per timestamp) and $N_t^W$ is the number of unique timestamps in $W_k$. Hence, the upper bound of computational complexity of sGrapp for a tumbling window $W$ at $t$ is $\mathcal{O}(\frac{K_{i,W_k}(K_{i,W_k}-1)}{2}K_{j,W_k}\mathcal{R}N_t^{W_k}) $. Note that this stands for all sequential windows.
 \end{proof}

\subsubsection{Error Bound} 
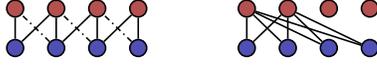
\begin{figure}[h]{}
\definecolor{myblue}{RGB}{80,80,180}
\definecolor{myred}{RGB}{180,80,80}


\resizebox{2cm}{!}{
\begin{tikzpicture}[thick,
  every node/.style={draw,circle},
  fsnode/.style={fill=myblue},
  ssnode/.style={fill=myred},
  every fit/.style={ellipse,draw,inner sep=-1pt,text width=1cm}
]

\begin{scope}[start chain=going right,node distance=4mm]
\foreach \i in {1,2,3,4}
  \node[fsnode,on chain] (f\i) {};
\end{scope}

\begin{scope}[xshift=0cm,yshift=0.7cm,start chain=going right,node distance=4mm]
\foreach \i in {1,2,3,4}
  \node[ssnode,on chain] (s\i) {};
\end{scope}

\draw (f1) -- (s1);
\draw (f1) -- (s2);
\draw [dash dot] (f2) -- (s1);
\draw (f2) -- (s2);
\draw (f2) -- (s3);
\draw [dash dot](f3) -- (s2);
\draw (f3) -- (s3);
\draw (f3) -- (s4);
\draw [dash dot](f4) -- (s3);
\draw (f4) -- (s4);
\end{tikzpicture}
}
\hspace{1cm}
\resizebox{2cm}{!}{
\begin{tikzpicture}[thick,
  every node/.style={draw,circle},
  fsnode/.style={fill=myblue},
  ssnode/.style={fill=myred},
  every fit/.style={ellipse,draw,inner sep=-1pt,text width=1cm}
]

\begin{scope}[start chain=going right,node distance=4mm]
\foreach \i in {1,2,3,4}
  \node[fsnode,on chain] (f\i) {};
\end{scope}

\begin{scope}[xshift=0cm,yshift=0.7cm,start chain=going right,node distance=4mm]
\foreach \i in {1,2,3,4}
  \node[ssnode,on chain] (s\i) {};
\end{scope}

\draw (f1) -- (s1);
\draw (f1) -- (s2);
\draw (f2) -- (s1);
\draw (f2) -- (s2);
\draw (f3) -- (s1);
\draw (f3) -- (s2);
\draw (f4) -- (s1);
\draw (f4) -- (s2);
\end{tikzpicture}
}

    \caption{Schematic butterfly formation. 
    i(j)-vertices are blue (red) in the bottom (top)
    .}\label{fig:minBmaxB}
\end{figure}
\begin{theorem}
The absolute error of sGrapp at the end of each window $W_k$ is bounded as $\Sigma_{l=1}^{k} |E_l|^\alpha - \binom{|V_{i,W_k}|}{2} \leq Err \leq \Sigma_{l=1}^{k} |E_l|^\alpha-|E_{W_k}|+2|V_{i,W_k}|$ where $E_k$, $E_{W_k}$, and $V_{i,W_k}$ denote the number of edges in the interval $[W_0^b,W_k^e)$, the number of edges in the interval $[W_k^b,W_k^e)$, and the number of i-vertices in the interval $[W_k^b,W_k^e)$, respectively.
\end{theorem}

\begin{proof}
sGrapp estimates the total number of butterflies at the end of each window $W_k$, $\forall k>0$, as 
$\hat{B}_k= \hat{B}_{k-1} + B_G^{W_k} +|E_k|^\alpha$ with initial term $\hat{B}_0=B_G^{W_0}$. 
Expanding this recursive equation would yield $\hat{B}_k= \Sigma_{l=0}^{k} B_G^{W_l} + \Sigma_{l=1}^{k} E_l^\alpha$.
On the other hand, according to the lemma \ref{lem:trueVals}, the true value of the total number of butterflies at the end of each window $W_k$, $\forall k>0$, denoted as $B_k$ lies in the range 
$\Sigma_{l=0}^{k} B_G^{W_l} + E_k-2|V_{i,W_k}| < B_k < \Sigma_{l=0}^{k} B_G^{W_l} + \binom{|V_{i,W_k}|}{2} $, where $V_{i,W_k}$ is the set of all seen i-vertices in the interval $[W_k^b,W_k^e)$.
Therefore, the absolute error of sGrapp $Err=|B_k-\hat{B}_k|$ falls in the range $\Sigma_{l=1}^{k} |E_l|^\alpha - \binom{|V_{i,W_k}|}{2} \leq Err \leq \Sigma_{l=1}^{k} |E_l|^\alpha-|E_{W_k}|+2|V_{i,W_k}|$.
\end{proof}

\begin{lemma}\label{lem:trueVals}
The exact number of inter-window butterflies at the end of each window $W_k$, $\forall k>0$, denoted as $B^{interW}$ is bounded as 
$|E_{W_k}|-2|V_{i,W_k}| \leq B^{interW} \leq \binom{|V_{i,W_k}|}{2} $, where $V_{i}$ is the set of all i-vertices in the $W_k$.
\end{lemma}

\begin{proof}
The number of inter-window butterflies contributed by window $W_k$ denoted as  $B^{interW}$, is minimum when the $W_k$'s edges $E_{W_k}$ are uniformly distributed over vertices by connecting each i-vertex in $W_k$ to at least 2 j-neighbors in $W_k$ and previous windows forming a series of caterpillars (solid edges in Figure~\ref{fig:minBmaxB}--left). In this case, according to the pigeonhole principle, the number of edges that complete the caterpillars (dashed edges in Figure~\ref{fig:minBmaxB}--left) will determine the number of inter-window butterflies: $B^{interW}=|E_{W_k}|-2|V_{i,W_k}|$. $B^{interW}$ is maximum when all of the $W_k$'s i-vertices are connected to two j-vertices such that at least one of them is not in $W_k$ (Figure~\ref{fig:minBmaxB}--right).  (Note, when all of j-neighbors are in previous windows, there wouldn't be any in-window butterfly in $W_k$). In this case, the number of inter-window butterflies reduces to the number of ways we can choose two i-vertices from the entire set of i-vertices:  $B^{interW}=\binom{|V_{i,W_k}|}{2}$. Therefore, $|E_{W_k}|-2|V_{i,W_k}| \leq B^{interW} \leq \binom{|V_{i,W_k}|}{2} $.
\end{proof}
\section{Experiments}\label{sec:experiments}
We test the effectiveness and efficiency of sGrapp and its optimized version sGrapp-x where x is the percentage of the available ground truth. We use x=25, 50, 75, and 100. The ground truths are obtained by running the exact counting Algorithm \ref{alg:countB} over the graph streams. Due to the computational expense of Algorithm \ref{alg:countB}, we collect the truth values over a limited number of sgrs: $72344$ in Epinions, $12259$  in ML100k, $21696$  in ML1m, $21778$  in ML10m, $75000$  in Edit-EnWiki, and $75000$  in Edit-FrWiki. 
The data sets that we use are described in Section \ref{subsec:graphstreamdataset}.

We report the effectiveness and efficiency of sGrapp and sGrapp-x in Sections \ref{subsec:effectiveness} and \ref{subsec:efficiency}, respectively. We also compare the performance of our algorithms with that of baselines in Subsection \ref{subsec:comparison}. Our experiments as well as the analysis in Section \ref{sec:butterflyemergence} are conducted on a machine with $15.6$ GB native memory and Intel Core $i7-6770HQ CPU @ 2.60GHz * 8$ processor. We have implemented FLEET algorithms and sGrapp algorithms in Java (OpenJDK version $1.8.0-252$, OpenJDK Runtime Environment build $1.8.0-252-8u252-b09-1~16.04-b09$). 

\subsection{Effectiveness Evaluation}\label{subsec:effectiveness}
\subsubsection{sGrapp}
We compute the Mean Absolute Percentage Error (MAPE) of sGrapp for windows with variable number of unique timestamps ($N_t^W$, $y$ axis) and  different exponent values ($\alpha$, $x$ axis). These are shown in the Figure \ref{fig:mapes}. The number of unique timestamps per window, $N_t$, varies in different graph streams, therefore we set the value of $N_t^W$ differently for each graph stream. We cross-validated the values of $\alpha$ and $N_t^W$ to explore the region including the best accuracy (lowest MAPE illustrated by the lightest color) for sGrapp. $MAPE=\frac{1}{n}\Sigma\frac{|B_k-\hat{B_k}|}{B_k}$, where $B_k$ is the ground truth computed over the growing graph at $t=W_k^e$ by Algorithm \ref{alg:countB} and $\hat{B_k}$ is the approximated value at $t=W_k^e$, and n is the number of windows. The data tips in the figures demonstrate the pair of $\alpha$ and $N_t^W$ yielding the lowest MAPE. 

We observe that the approximation accuracy of sGrapp is not sensitive to window length and the exponent, since there exists a combination of approximation exponent and window length for each graph steam that yields appropriate MAPE (Figure \ref{fig:mapes}). In fact, the best MAPE of sGrapp is significantly lower than $0.1$ in all of the rating graph streams, demonstrating that sGrapp is a good approximator of actual butterfly count.

When the approximation exponent is high and the window is compact (bottom right corners in Figure \ref{fig:mapes}), the error is high. In this case, sGrapp overestimates the number of inter-window butterflies due to high exponent value. Also, when the exponent is low and the window includes a large number of sgrs (top left corner in Figure \ref{fig:mapes}), the error is high. The reason in this case is that sGrapp underestimates the number of inter-window butterflies. An appropriate parameter region to gain a reasonable accuracy is where $\alpha$ and $N_t^W$ are both high or low (middle diameter from top right corner to bottom left corner in Figure \ref{fig:mapes}). The best accuracy is always obtained for higher exponent values. For rating networks, an appropriate exponent value for sGrapp is $\alpha=1.4$.
        
        
        
       
        

As we investigated the contribution of hubs to the emergence of butterflies (Section \ref{sec:butterflyemergence}), we relate the value of approximation exponent to the probability of having at least one i-hub ($P(N_{iHub}^t>=1)$) plus the probability of having at least one j-hub ($P(N_{jHub}^t>=1)$) in the butterflies at  time $t$, i.e. $\alpha=P(t)=P(N_{iHub}^t=1)+P(N_{iHub}^t=2)+P(N_{jHub}^t=1)+P(N_{jHub}^t=1)$ (Table \ref{tab:ijhubs}). That is, the value of $\alpha$ can be determined based on the probability of i- or j-hubs forming butterflies at a certain time point $t$. The time point $t$ is likely a \emph{tipping point} where the number of hub connections in the graph is stabilized (Figures \ref{fig:nihubdegree} and \ref{fig:njhubdegree}). To check this, we calculate the value of $P(t)$ for $t\in \{ 1000, 2000, .., 9000, 10000\}$  in the Epinions graph stream. We compute the value of MAPE for sGrapp($N_t^W$, $\alpha$). We set $\alpha =P(t)$ and $N_t^W\in \{ 0.006N_t, 0.007N_t, 0.008N_t, 0.009N_t, 0.01N_t\}$. 
In Table \ref{tab:mapePhubs}, we report the value of MAPE for the approximations with different exponent values and different fraction of unique timestamp per adaptive window. 
We observe that, at $t=6000$, where the exponent is equal to $\alpha=P(t=6000)=\sim1.03$, the approximation error is the lowest. This time point is a \emph{tipping point} where the fraction of average hub degree is steadily low afterward and high backward (Figures \ref{fig:nihubdegree} and \ref{fig:njhubdegree}). Moreover, in Figure \ref{fig:mapes}, we  see that the best accuracy is obtained when the exponent is equal to $P(t=6000)=1.03$. We leave further investigation of the significance of these values as future work.
 
\begin{table*}[ht]\caption{Epinions - The approximation MAPE for different adaptive window lengths (columns) and different exponents calculated as the probability of one or two i-hub plus the probability of one or two j-hub at different time points (rows).}

    \begin{tabular}{|p{3.5cm}|p{1.5cm}|p{1.5cm}|p{1.5cm}|p{1.5cm}|p{1.5cm}|}\hline
        
        \makecell{MAPE} & \makecell{$0.006*N_t$}& \makecell{$0.007*N_t$} & \makecell{$0.008*N_t$} & \makecell{$0.009*N_t$} & \makecell{$0.01*N_t$}\\ \hline\hline 
        
        \makecell{$\alpha=P(t=1k)=1.2178$} &\makecell{$3.0036$} &\makecell{$2.5461$} &\makecell{$2.5005$} &\makecell{$2.2996$} &\makecell{$2.2602$} \\ \hline 
        
        \makecell{$\alpha=P(t=2k)=1.077$} &\makecell{$0.4472$} &\makecell{$0.3291$} &\makecell{$0.3318$} &\makecell{$0.2359$} &\makecell{$0.2632$} \\ \hline 
         
        \makecell{$\alpha=P(t=3k)=1.1274$} &\makecell{$1.0295$} &\makecell{$0.8281$} &\makecell{$0.8212$} &\makecell{$0.6954$} &\makecell{$0.7079$} \\ \hline 
        
        \makecell{$\alpha=P(t=4k)=1.0806$}&\makecell{$0.4778$} &\makecell{$0.3551$} &\makecell{$0.3574$} &\makecell{$0.2597$} &\makecell{$0.2864$} \\ \hline 
        
        \makecell{$\alpha=P(t=5k)=1.0389$} &\makecell{$0.14286$} &\makecell{$0.1016$} &\makecell{$0.0778$} &\makecell{$0.0864$} &\makecell{$0.0456$} \\ \hline 
        
       \makecell{\textbf{$\alpha=P(t=6k)=1.0296$}} &\makecell{\textbf{0.0953}} &\makecell{\textbf{0.0723}} &\makecell{\textbf{0.524}} &\makecell{\textbf{0.0709}} &\makecell{\textbf{0.0315}} \\ \hline
        
        \makecell{$\alpha=P(t=7k)=1.0438$} &\makecell{$0.1760$}&\makecell{$0.1176$} &\makecell{$0.1054$} &\makecell{$0.1014$} &\makecell{$0.0597$}  \\ \hline 
        
        \makecell{$\alpha=P(t=8k)=1.0591$} &\makecell{$0.2897$} &\makecell{$0.1950$} &\makecell{$0.2000$} &\makecell{$0.1525$} &\makecell{$0.1446$} \\ \hline 
        
        \makecell{$\alpha=P(t=9k)=1.0546$} &\makecell{$0.2553$} &\makecell{$0.1658$} &\makecell{$0.1713$} &\makecell{$0.1370$} &\makecell{$0.1188$} \\ \hline 
         
        \makecell{$\alpha=P(t=10k)=1.0420$} &\makecell{$0.1639$} &\makecell{$0.1189$} &\makecell{$0.0953$} &\makecell{$0.0959$} &\makecell{$0.0508$} \\ \hline 
        
    \end{tabular}\label{tab:mapePhubs}
\end{table*} 
 
After  evaluating sGrapp in terms of the average window errors (MAPE), we delve into its performance evolution over windows so that we can track the origins of the accuracy gain. We pick the most accurate $\alpha$ and $N_t^W$ (highlighted data points in Figure \ref{fig:mapes}) and plot the signed value of relative error $\frac{|B_k-\hat{B_k}|}{B_k}$ for each window $W_k$ in the Figure \ref{fig:realtiveerrors}. Depending on the value of $N_t^W$, the number of windows vary in different graph streams. Positive errors (depicted by red upward triangles ) reflect over-estimations and negative errors (depicted by blue downward triangles) reflect under-estimations. In ML10m, Edit-EnWiki and Edit-FrWiki, the approximation begins with over-estimation and ends up with under-estimation. The underlying reason is the static exponent over sequential windows with different number of connections to the old hubs and consequently different number of inter-window butterflies. 
 
\subsubsection{sGrapp-x} 
We also evaluate the accuracy of sGrapp-x in terms of MAPE  in the region that sGrapp displays lowest errors in Figures \ref{fig:mapes25} --
\ref{fig:mapes100}. This enables a fair comparison of sGrapp with its optimized version sGrapp-x. Note that, sGrapp-x begins with a given exponent value and ends up with a modified value after the supervision phase reaches an error below $0.05$. Therefore we fed sGrapp-x with same input values of $\alpha$ and $N_t^W$ as sGrapp. The values shown in Figures \ref{fig:mapes25} -- 
\ref{fig:mapes100} reflect the inputs. 

It is evident from these figures that sGrapp-x improves the accuracy, which can be summarized as (a) improving the minimum MAPE (Figure \ref{fig:mapemin}), (b) improving the maximum MAPE (Figure \ref{fig:mapemax}), as well as (c) expanding the coverage of MAPE$\leq0.15$ and MAPE$\leq0.2$ (Figures \ref{fig:mapeprob0.15} and \ref{fig:mapeprob0.2}). As illustrated in Figure \ref{fig:mapemin}, the minimum MAPE value in the studied parameter space is roughly the same for both sGrapp and sGrapp-x  $x=25-100$ in all rating graph streams. sGrapp-x lowers  the minimum MAPE with respect to sGrapp in Edit-EnWiki graph from $0.681$ to $0.376$ (via $x=25$), $0.105$ (via $x=75$), $0.101$ (via $x=50$), and $0.097$ (via $x=100$); in Edit-FrWiki graph from $0.201$ to $0.235$ (via $x=25$), $0.137$ (via $x=100$), $0.134$ (via $x=75$), and $0.130$ (via $x=50$). That is, the minimum MAPE is lowered ranging from $44.79\%$ to $85.76\%$ in Edit-EnWiki and $31.84\%$ to $35.32\%$ in Edit-FrWiki. As illustrated in Figure \ref{fig:mapemax},the maximum MAPE related to the over-estimations (bottom right corners in Figures \ref{fig:mapes25} -- \ref{fig:mapes100}) is notably decreased in all graph streams. The most significant decrease corresponds to Edit-FrWiki stream with the highest change from $2$ to $0.26$ (via $x=75,100$) and Edit-EnWiki stream with highest change from $0.715$ to $0.15$  (via $x=100$).

In Figures \ref{fig:mapeprob0.15} and \ref{fig:mapeprob0.2}, we present the probability of approximation with MAPE$\leq0.15$ and MAPE$\leq 0.2$ ($P(MAPE\leq 0.15(0.2))$) by calculating the fraction of approximations that satisfy MAPE$\leq 0.15$ and MAPE$\leq 0.2$. That is the relative coverage of light blue areas in Figures \ref{fig:mapes} -- \ref{fig:mapes100}. When the approximation MAPE is above $0.15$ or $0.2$ the corresponding bars are omitted in Figures \ref{fig:mapeprob0.15} and \ref{fig:mapeprob0.2}. Since sGrapp-100 approximates the number of butterflies in Edit-EnWiki with highest MAPE equal to $0.15$, the corresponding bar has a height of $1$. sGrapp-25 improves the accuracy of sGrapp in MovieLens10m better than other sGrapp-x versions. For the other graph streams, when $x\geq 50$, sGrapp-x displays fairly well accuracy improvement as the probability of accurate approximation (i.e. average window error below 0.15 and 0.2) is amplified. As expected sGrapp-100 has the most improvement, however sGrapp-75 and sGrapp-50 are reliable improvement alternatives for Edit-FrWiki and the rest of graph streams, respectively. sGrapp-x, $x=25, 50, 75,$ and $100$ can achieve the $P(MAPE\leq 0.15(0.2))$  equal to $67.13\%$ (78.53\%), $60.94\%$ ($94.55\%$), $79.74\%$ ($84.27\%$), and $99.31\%$ ($100\%$). Most notably, sGrapp-50(75) increases $P(MAPE\leq 0.2)$ from $0$ to $94.55(100)\%$ in Edit-EnWiki.

We check the evolution of the signed value of relative error over windows for the data points with the lowest sGrapp-x MAPE. 
As shown in Figures \ref{fig:realtiveerrors25}, \ref{fig:realtiveerrors50}, \ref{fig:realtiveerrors75}, and \ref{fig:realtiveerrors100}, 
dynamically changing the approximation exponent heals the under/over-estimation problem; Hence the average window error is diminished. There is always a value of x by which sGrapp-x can yield average approximation error less than equal $0.05$ in rating graphs and $0.14$ in Wikipedia graphs.
\begin{center}
\begin{table*}
\caption{Throughput of different algorithms for \boldsymbol{$\gamma$}=0.7.}
\small
    \begin{tabular}{|p{1.7cm}||
    p{1.3cm}|p{1cm}|
    p{1.3cm}|p{1cm}|
    p{1.3cm}|p{1cm}|
    p{1.3cm}|p{1cm}|
    p{1.3cm}|p{1.5cm}|}\hline
        
        Throughput & FLEET2 M=75k & FLEET3 M=75k &
        FLEET2 M=150k & FLEET3 M=150k &
        FLEET2 M=300k & FLEET3 M=300k &
        FLEET2 M=600k & FLEET3 M=600k &
        sGrapp & sGrapp-100 \\ \hline\hline
       
        Epinions & 89 575 & 137 411 &
        59 336 & 53 077 &
        16 912 & 16 360 &
        11 028 & 10 907 & 
        \textbf{182 427} & 166 895 \\ \hline
        
        ML100k & 3 664 & 5 652  &
        4 691 & 4 717 &
        3 509 & 3 424 & 
        4 268 & 4 378 & 
        8 026 & \textbf{8 629} \\ \hline
        
        ML1m & 23 490 & 23 292 &
        12 038 & 7 355 &
        2 383 & 1 673 &
        1 004 & 857 &
        \textbf{26 698} & 26 487 \\ \hline
        
        ML10m & 147 665 & 72 918 &
        62 905 & 23 536 &
        16 719 & 5 358 & 
        4 410 & 2 337 &
        \textbf{234 571} & 228 021 \\ \hline
        
        Edit-FrWiki &  554 741 & 155 343 &
        298 019 & 57 477 &
        116 917 & 16 856 & 
        41 051 & 6 240 &      
        \textbf{1 000 861} & 985 265 \\ \hline
       
        Edit-EnWiki &  \textbf{2 564 565} & 719 375 &
        1 373 708 & 305 347 &
        911 170 & 114 806 & 
        324 183 & 34 283 &
        1 085 185 & 1 098 382 \\ \hline
        
\end{tabular} \label{tab:baselineThroughputs}
\normalsize
\end{table*}
\begin{table*}
\caption{MAPE of different algorithms for \boldsymbol{$\gamma$}=0.7 and M=0.1S and same \boldsymbol{$N_t^W$}.}
\small

    \begin{tabular}{|p{1.7cm}||
    p{1.3cm}|p{1.3cm}|p{1.3cm}|
    p{1.3cm}|p{1.7cm}|p{1.7cm}|p{1.7cm}|p{1.7cm}|}\hline
        
        \makecell{MAPE} & 
        FLEET1 & FLEET2 & FLEET3 &
        sGrapp & sGrapp-25 & sGrapp-50 & sGrapp-75 & sGrapp-100 \\ \hline\hline
       
        Epinions &  
        0.058 & 13.789 & 0.336 &
          \textbf{0.022} &   \textbf{0.022} &   0.028 &   0.028 &   0.028 
        \\ \hline
        
        ML100k &  
        0.959 &  2.287 & 0.399 &
          \textbf{0.009} &   \textbf{0.009} &   \textbf{0.009} &   \textbf{0.009} &   \textbf{0.009}
        \\ \hline
        
        ML1m & 
        0.085 & 5.261 & 0.047 &
          0.043 &   \textbf{0.043} &   \textbf{0.053} &   0.067 & 0.055 
        \\ \hline
        
        ML10m &  
        0.156 &  0.839 & \textbf{0.086} &
          0.143 &   0.247 &   0.162 &   0.180 &   0.170
        \\ \hline
        
        Edit-FrWiki &  
        1.575 &  49.165 & 57.563 &
          0.201 &   0.313 &   0.217 &   \textbf{0.134} &   0.137 
        \\ \hline
       
        Edit-EnWiki &   
        2.689 & 467.747 & 178.702 &
        0.684 &  0.494 & 0.161 & 0.141 & \textbf{0.137} 
        \\ \hline
        
\end{tabular} \label{tab:baselineErrors}
\normalsize
\end{table*}
\end{center}
\subsection{Efficiency Evaluation}\label{subsec:efficiency} 
We evaluate the efficiency of sGrapp and sGrapp-100 by averaging over 50 independent cases. We do not report the efficiency metrics for sGrapp-x for $x<100$ since their efficiency is close to that of sGrapp-100. For each graph stream we study the performance for the parameter settings that yield the best accuracy (highlighted data points in Figures \ref{fig:mapes} and \ref{fig:mapes100}) to see the overhead of a highly accurate approximation. Note that parameter values do not affect the efficiency. 

We check the latency of sGrapp and sGrapp-100 for each processing window (Figures \ref{fig:wlatency} and \ref{fig:wlatency100}). We observe that the window latency of all the graph streams (except the Epinions) is not decreasing.  The window latency of each graph stream follows its temporal distribution pattern (Figure \ref{fig:temporaldist}). Therefore, to omit the effect of temporal distribution, we study the performance by considering both the processing time (latency) and the number of processed elements. To this end, at the end point of each window, we check the window throughput (i.e. the number of processed edges in the window divided by the elapsed time in seconds, Figures \ref{fig:wthroughput} and \ref{fig:wthroughput100})) as well as the total throughput (i.e. the total number of processed edges since the first window until the end of the current window divided by the total elapsed time in seconds, Figures \ref{fig:totalthroughput} and \ref{fig:totalthroughput100}). 

The window throughput displays fluctuations due to variant number of sgrs in each window; however in overall it is higher in later windows for both sGrapp and sGrapp-100. The total throughput of both sGrapp and sGrapp-100 displays an increasing pattern. As mentioned in previous section, the old hubs are the main contributors to the butterfly formation. Since old hubs occur in the early windows, the later windows mostly include butterfly vertices with lower degree. That is, there are fewer windowed butterflies in later windows  than the inter-window butterflies. Therefore, the exact counting algorithm that computes the number of windowed butterflies finishes quicker. Also, rapid approximation of the inter-window butterflies plays the main role in reducing the processing time, enhancing the total throughput. An evidence is the throughput for MovieLens100k that has almost uniform temporal distribution: we observe an increasing total throughput over windows. This is important since the number of sgrs in the windows is not decreasing while the throughput is increasing. This confirms (1) the algorithm's power is independent of the structural/temporal characteristics of the input data and (2) the algorithm is efficient particularly in dense graph streams. 
\subsection{Comparison with Baselines}\label{subsec:comparison}
We compare the effectiveness and efficiency of sGrapp suit and FLEET suit.  Experimental results of FLEET suit show that FLEET3, FLEET2 and FLEET1 have the best performance (in that order), so we use those as baselines. While sGrapp has the $\alpha$ (approximation exponent) and $N_t^W$ (number of unique timestamps per window) parameters,  FLEET has the $M$ (reservoir size) and $\gamma$ (sub-sampling probability) parameters. Since the performance of FLEET algorithms is sensitive to its parameters, we compare our algorithms against the FLEET settings which achieve the best performance. We set the sub-sampling probability as $\gamma =0.7$ as suggested by FLEET authors \cite{sanei2019fleet}.

We observe that when the reservoir size $M$ is greater than the entire stream, latency is negatively impacted since sub-sampling does not occur and all the edges are added to the reservoir and for each new edge the exact butterfly counting is executed. Hence, for evaluating the accuracy over the prefix of a stream, we set $M=0.01S$, where $S$ is the size of available stream. For evaluating the efficiency, we also use a range of values $M\in \{75k, 150k, 300k, 600k \}$ to examine the throughput over the entire stream; these values are the ones offered in the original paper \cite{sanei2019fleet}. We use the approximation exponent values yielding lowest MAPE in sGrapp, which do not necessarily yield the best MAPE in the optimized variant sGrapp-x. Since FLEET algorithms use different window semantics than sGrapp, we use virtual time-based adaptive windows over FLEET algorithms to extract the estimated values at the end of virtual windows for accuracy evaluations only (not for efficiency tests). We use the same value of $N_t^W$ for sGrapp and FLEET suits to compute  MAPE: $N_t^W\in[42, 912, 1050, 80, 290, 500]$ for Epinions, ML100k, Ml1m, Ml10m, Edit-EnWiki, and Edit-FrWiki, respectively. For efficiency comparisons, we used the same value used in effectiveness experiments since our goal is to check the efficiency cost of the most accurate approximation. For each $N^W_t$, there exists an alpha yielding a high precision estimate. $N^W_t$ does not affect accuracy.

In Table \ref{tab:baselineThroughputs}, we report the total throughput over the entire graph streams for sGrapp and FLEET suits. Since FLEET1's throughput is very low, we do not include it in this experiment. 
By increasing the size of reservoir the throughput of all FLEET algorithms decreases since the frequency of exact butterfly counting per edge increases. It is always the case that $M=75k$ and $M=600k$ yields the highest and the lowest throughput, respectively. sGrapp outperforms FLEET for every setting: minimum (maximum) ratios of sGrapp to FLEET throughput are $1.32$ ($16.7$), $1.5$ ($2.5$), $1.13$ ($31.1)$, $1.58$  ($100.3$), $1.8$ ($160.4$), and $0.4$  ($32$) in Epinions, ML100k, ML1m, ML10m, Edit-FrWiki, and Edit-EnWiki, respectively. sGrapp and its optimized version outperforms FLEET suit within a range of $[1.13$ $160.4]$, with the performance improvement increasing as graph streams become larger (i.e., Edit-FrWiki, ML10m, and Edit-Enwiki).
\begin{figure}[ht]
 \begin{tikzpicture}[node distance=1cm]
    \node [draw, circle, thick](1) {$\gamma$};
    \node [draw, circle, thick, right=of 1](2) {$P$};
    \node [draw, circle, thick, fill=cyan, right=of 2](3) {$\hat{B}$};
    \node [draw, circle, thick, below=of 1](4) {$M$};
    \node [draw, circle, thick, below=of 3](5) {$F$};

    \draw [thick, ->] (1) -- (2);
    \draw [thick, ->] (4) -- (2);
   \draw [thick, ->] (2) -- (3);
    \draw [thick, ->] (2) -- (5);
    \draw [thick, ->] (5) -- (3);
\end{tikzpicture}
    
    \caption{Impact of FLEET parameters on estimate.}
   \label{fig:fleetParameters}
\end{figure}
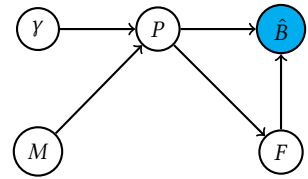
In Table \ref{tab:baselineErrors}, we report accuracy (in terms of MAPE) of sGrapp and FLEET suits over the subset of stream with available true values. We observe that sGrapp and sGrapp-x achieve MAPE values equal to $0.022$, $0.009$, $0.043$, $0.143$, $0.134$, and $0.137$ in Epinions, ML100k, ML1m, ML10m, Edit-FrWiki, and Edit-EnWiki which are significantly lower than those of FLEET -- sGrapp errors are $0.38\times$, $0.02\times$, $0.91\times$, $1.66\times$, $0.08\times$, and $0.05\times$ of FLEET for these graphs.
Table \ref{tab:baselineErrors}  (Table \ref{tab:baselineThroughputs})  shows that for ML10m, FLEET3’s accuracy (throughput) is $0.057$ better (up to $100x$ lower) than sGrapp explaining the high computational cost of FLEET3 in this specific dataset. FLEET3 updates the estimate for each new edge by enumerating butterflies incident to that edge. This increases the probability of detecting the incident butterflies by a factor of $P$ (i.e. sampling probability), however the computations are much increased. This technique is more impactful in ML10m with high butterfly density.
Butterfly estimate $\hat{B}$ is updated as soon as an edge arrives in FLEET3 or during the sampling and (or) sub-sampling phase in FLEET1 (FLEET2).  
 In FLEET1, when $P$ is not high or $M$ is small and $\gamma$ is low, $\hat{B}$ is not frequently updated and error goes up. In FLEET2, 
 many butterflies are missed due to sampling. 
 Moreover, FLEET has poor accuracy when the butterflies are distributed across the edges uniformly (e.g. Edit-EnWiki with a low butterfly density of $9.1\times 10^{-21}$ according to the statistics in \cite{sanei2019fleet}). The reason is that $\hat{B}$ is updated for some edges only. 
In summary, the accuracy of FLEET algorithms highly depend on $M$, $\gamma$, and the frequency of updating $\hat{B}$, because $\hat{B}$ is updated wrt the $P$; and $P$ is updated as $p \gets p*\gamma$ in each sampling round, which in turn increases $\hat{B}$ more. As depicted in Figure \ref{fig:fleetParameters}, 
$M$ and $\gamma$ (confounding variables) impact $P$ and $P$ impacts $\hat{B}$ directly through the formula and indirectly through the frequency of updates. A high frequency of butterfly counting and high sub-sampling come at the cost of low throughput. A large $M$ comes at the cost of memory consumption as well as latency issues. FLEET suit cannot guarantee both efficiency and effectiveness at the same time. sGrapp does not suffer from the aforementioned issues since it does not rely on exact counting and sampling; rather it relies on counting the inter-window butterflies. sGrapp keeps the computational footprint of exactly counting the in-window butterflies low by means of the load-balanced adaptive windows and then, effectively estimates the number of inter-window butterflies which are the dominant butterflies based on the butterfly densification power law formalism.
\begin{figure*}[h]\centering
    \subfigure{\includegraphics[width=0.3\textwidth]{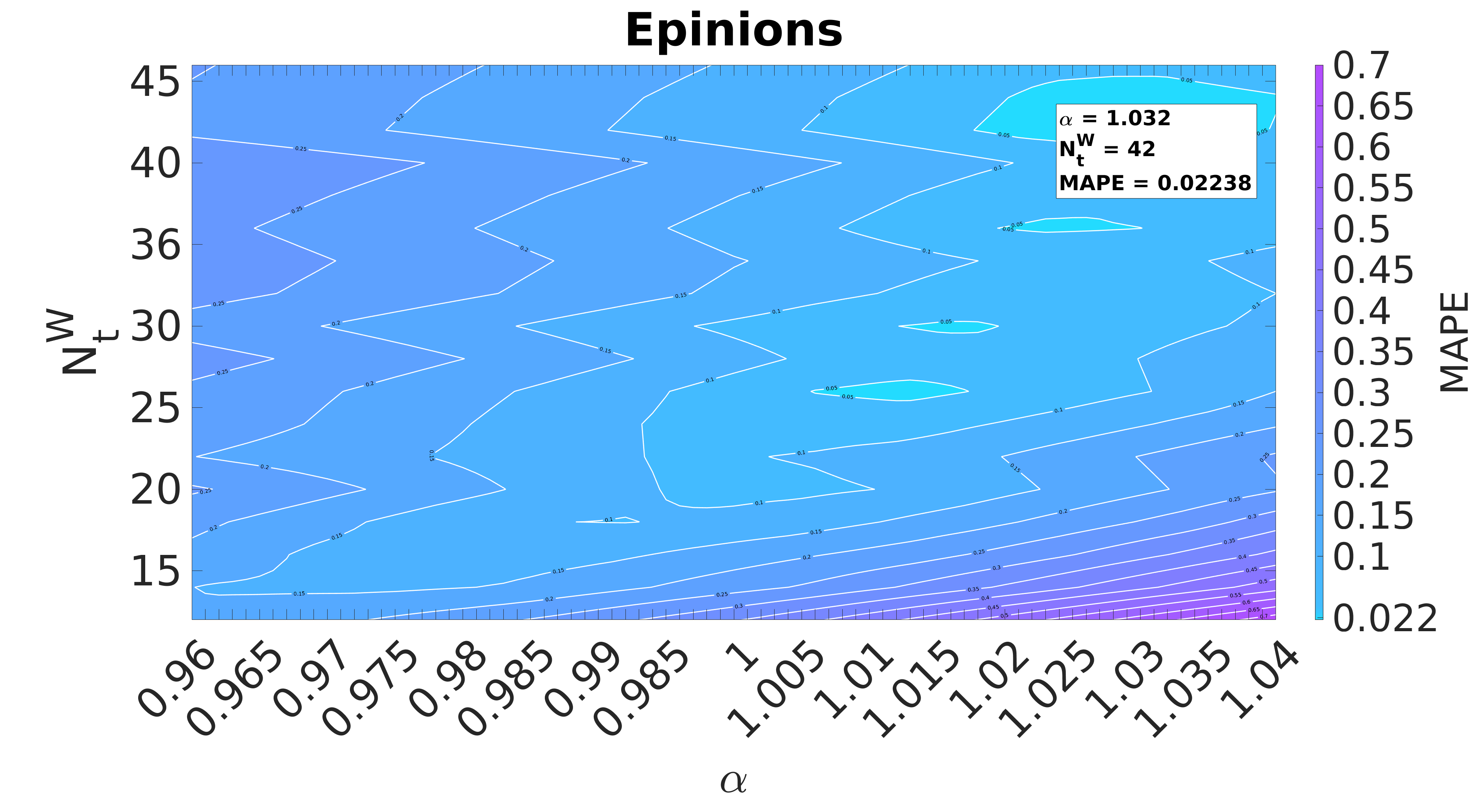}} 
    \subfigure{\includegraphics[width=0.3\textwidth]{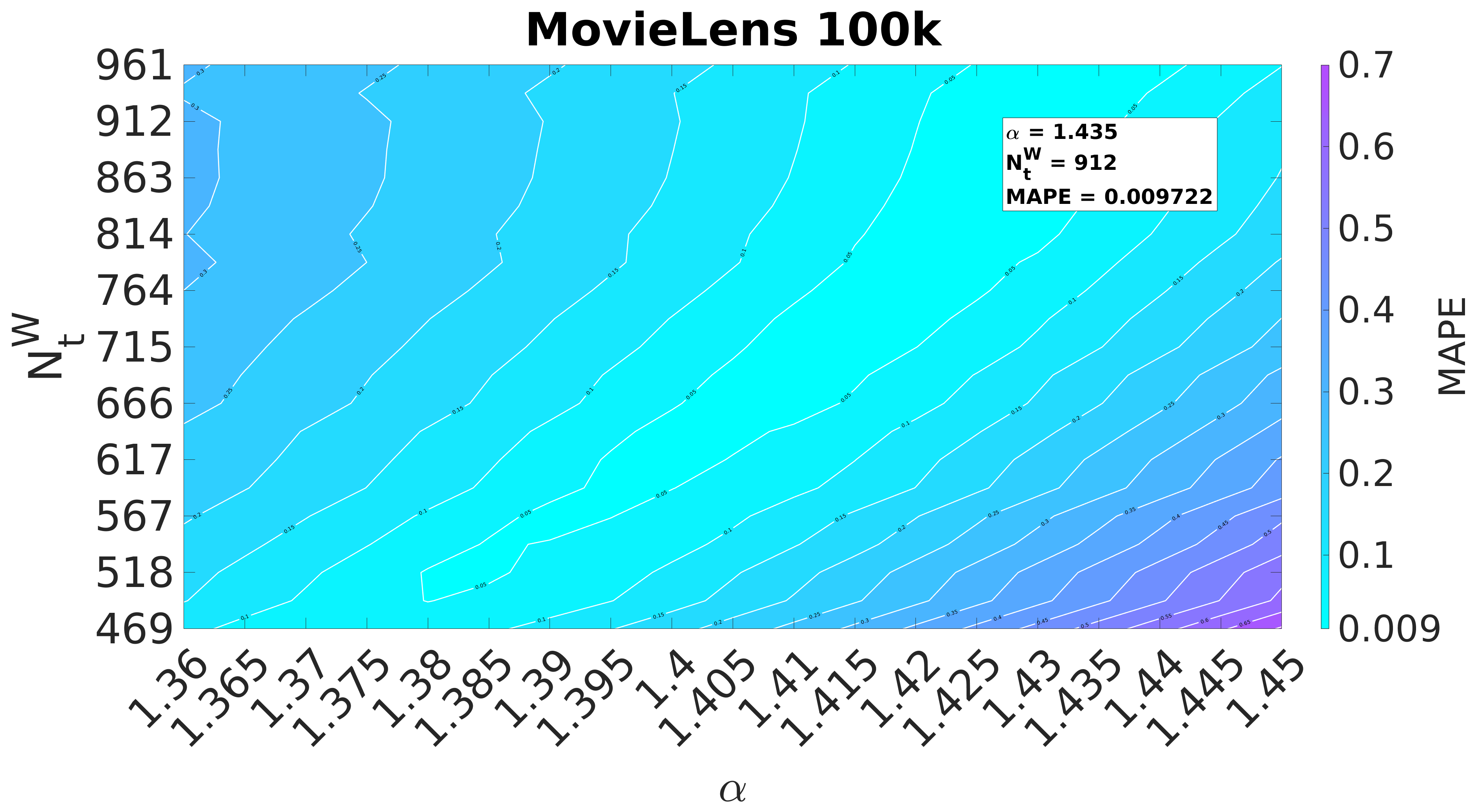}} 
    \subfigure{\includegraphics[width=0.3\textwidth]{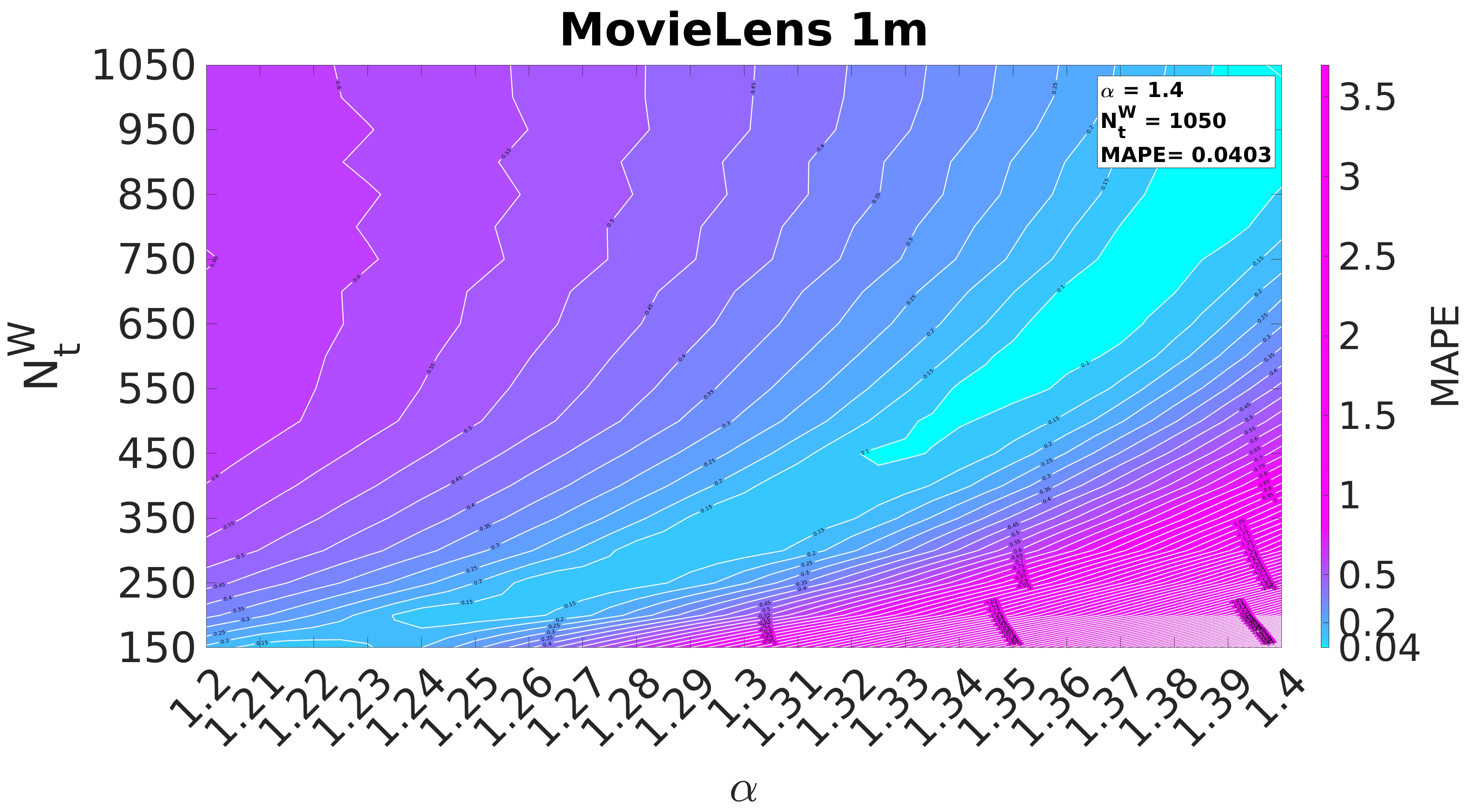}}
    \subfigure{\includegraphics[width=0.3\textwidth]{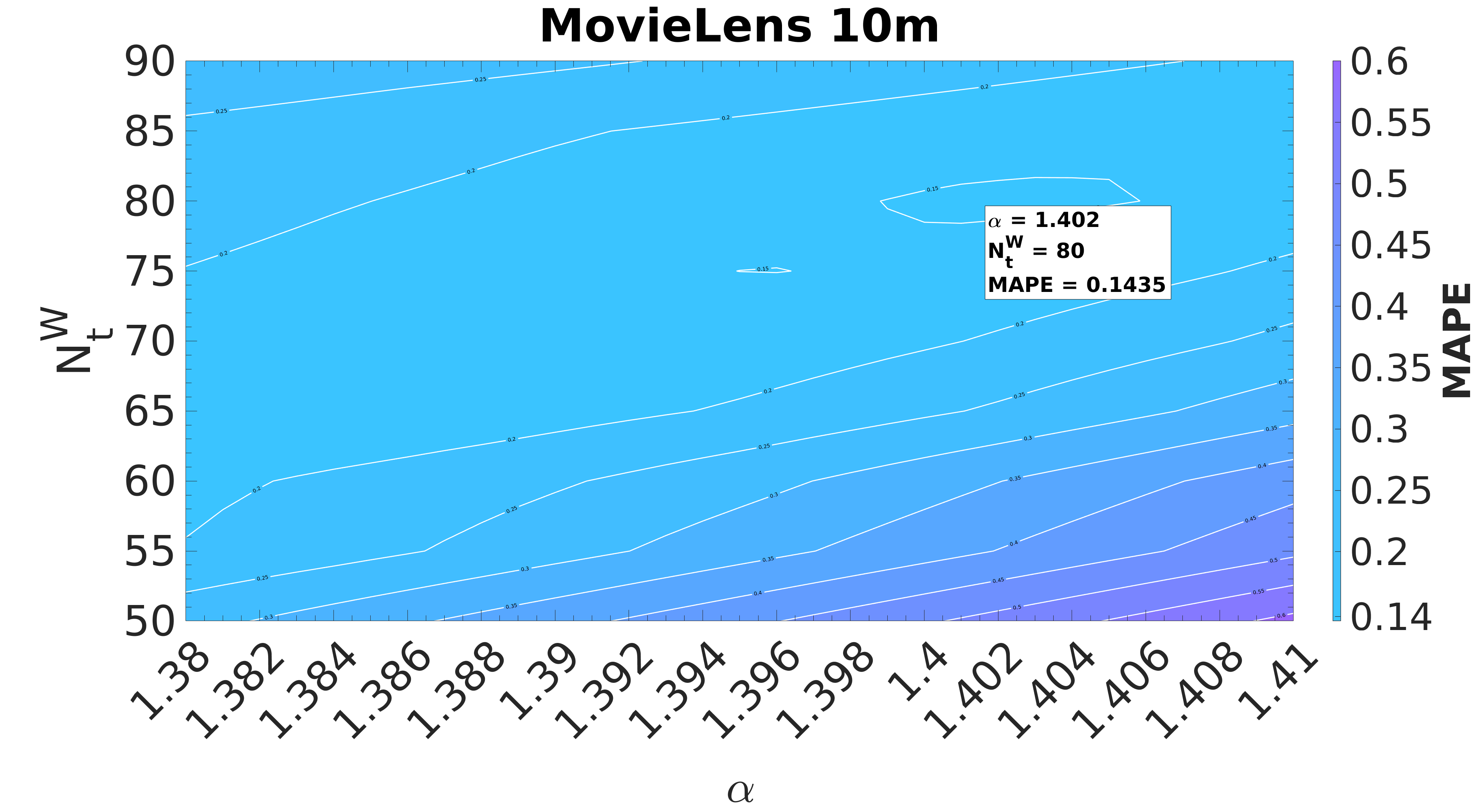}}
    \subfigure{\includegraphics[width=0.3\textwidth]{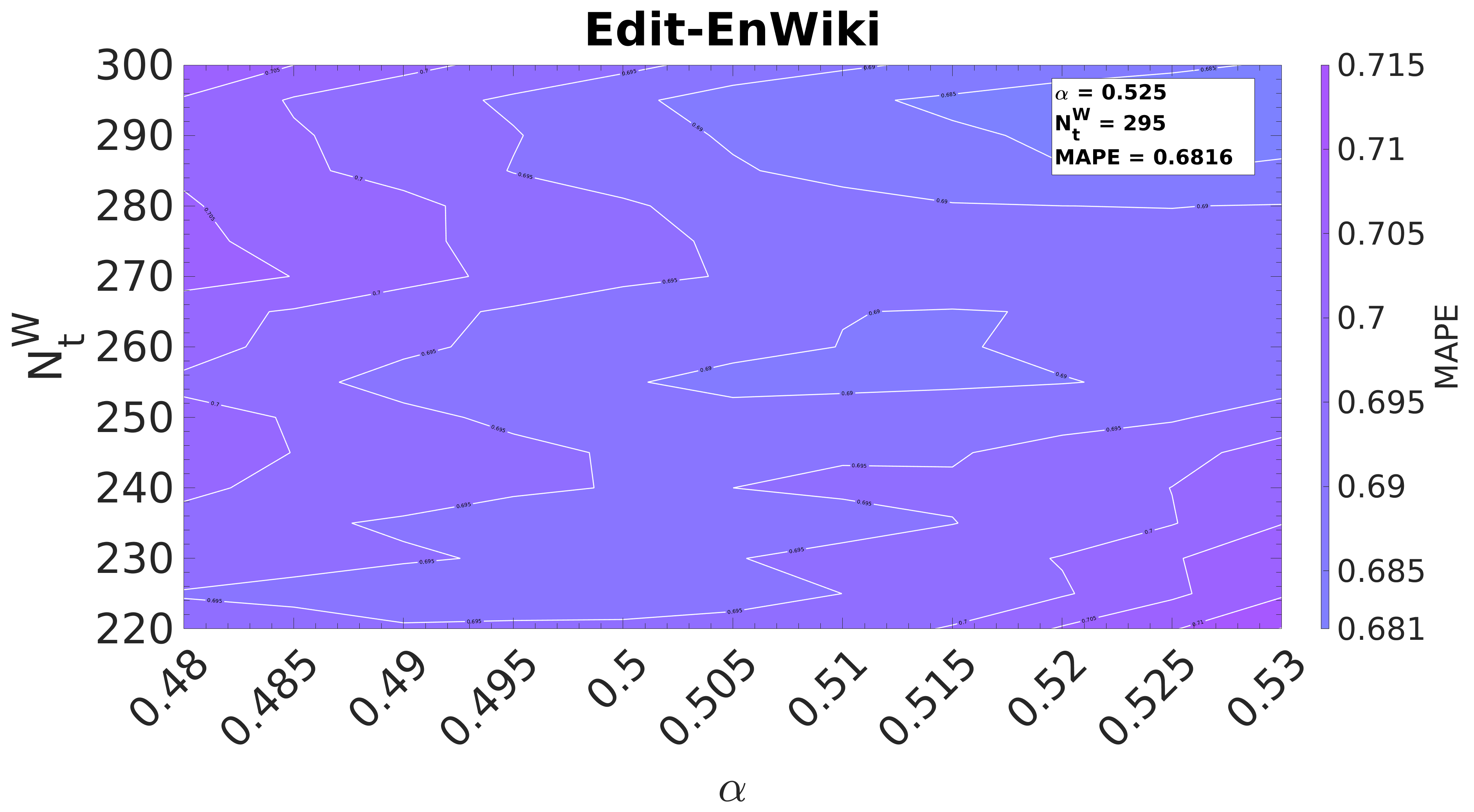}}
    \subfigure{\includegraphics[width=0.3\textwidth]{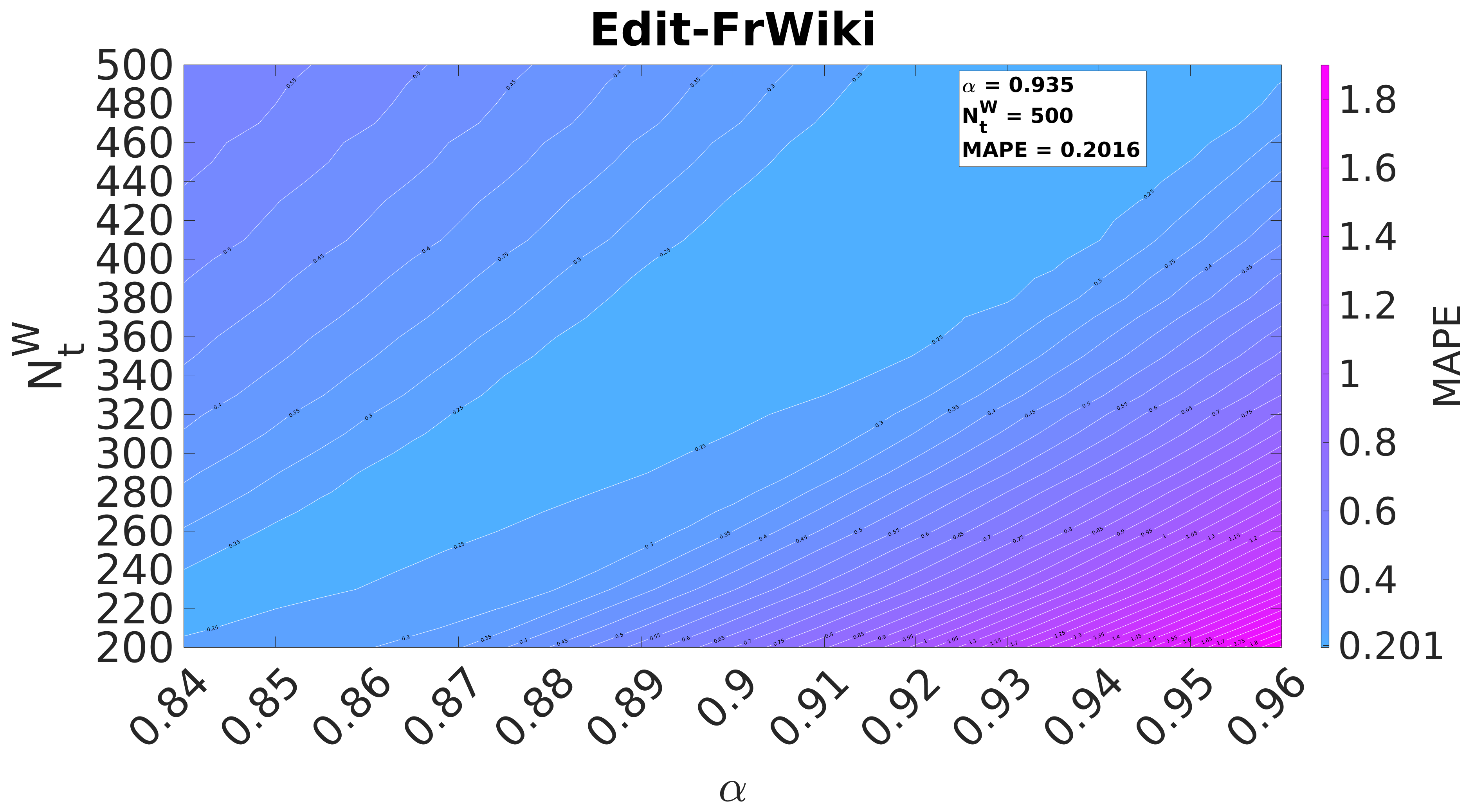}}
    \caption{[Best viewed in colored.] Accuracy of sGrapp for different values of $\alpha$ and $N_t^W$ }\label{fig:mapes}\end{figure*}
\begin{figure*}[h]\centering
    \subfigure{\includegraphics[width=0.3\textwidth]{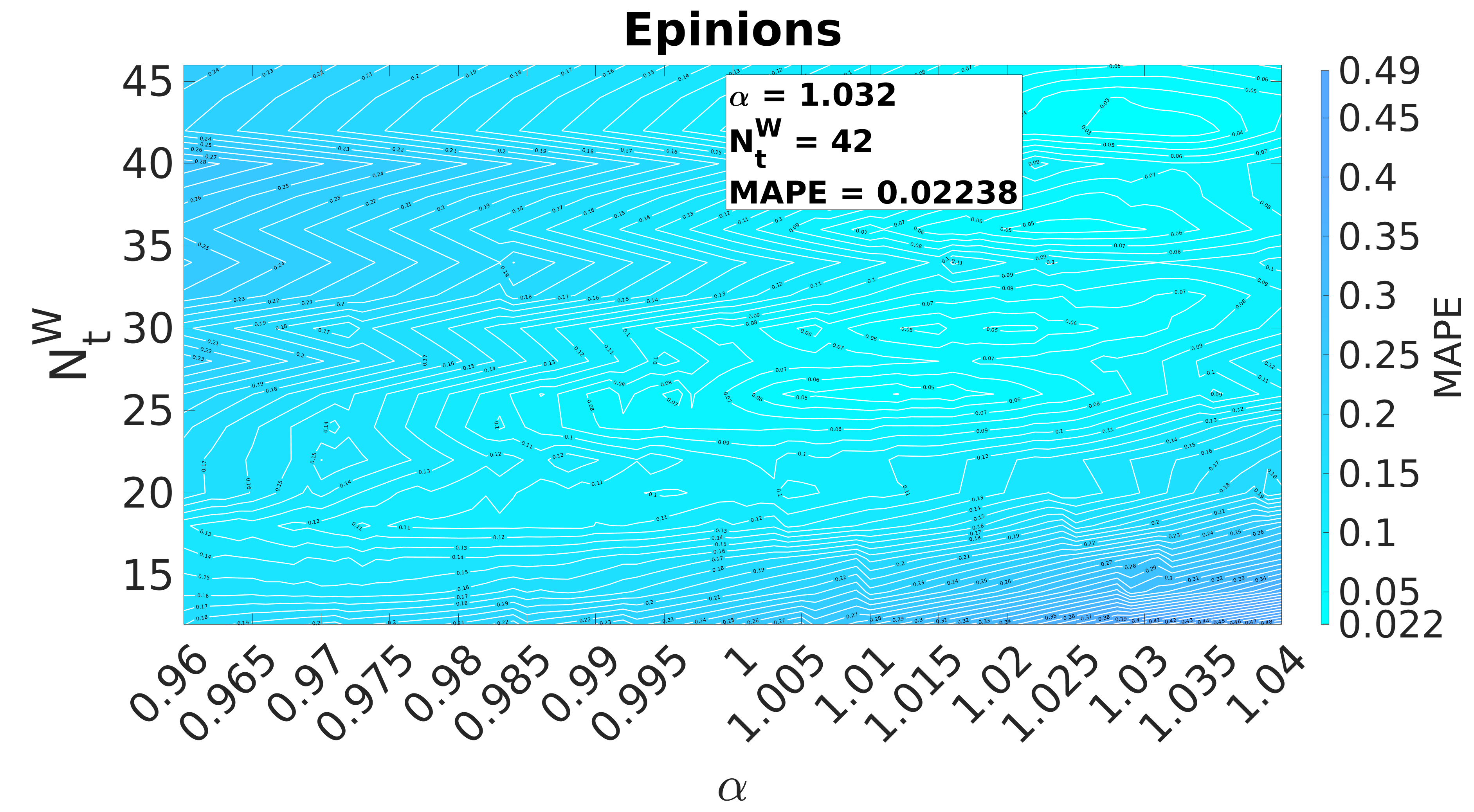}} 
    \subfigure{\includegraphics[width=0.3\textwidth]{ 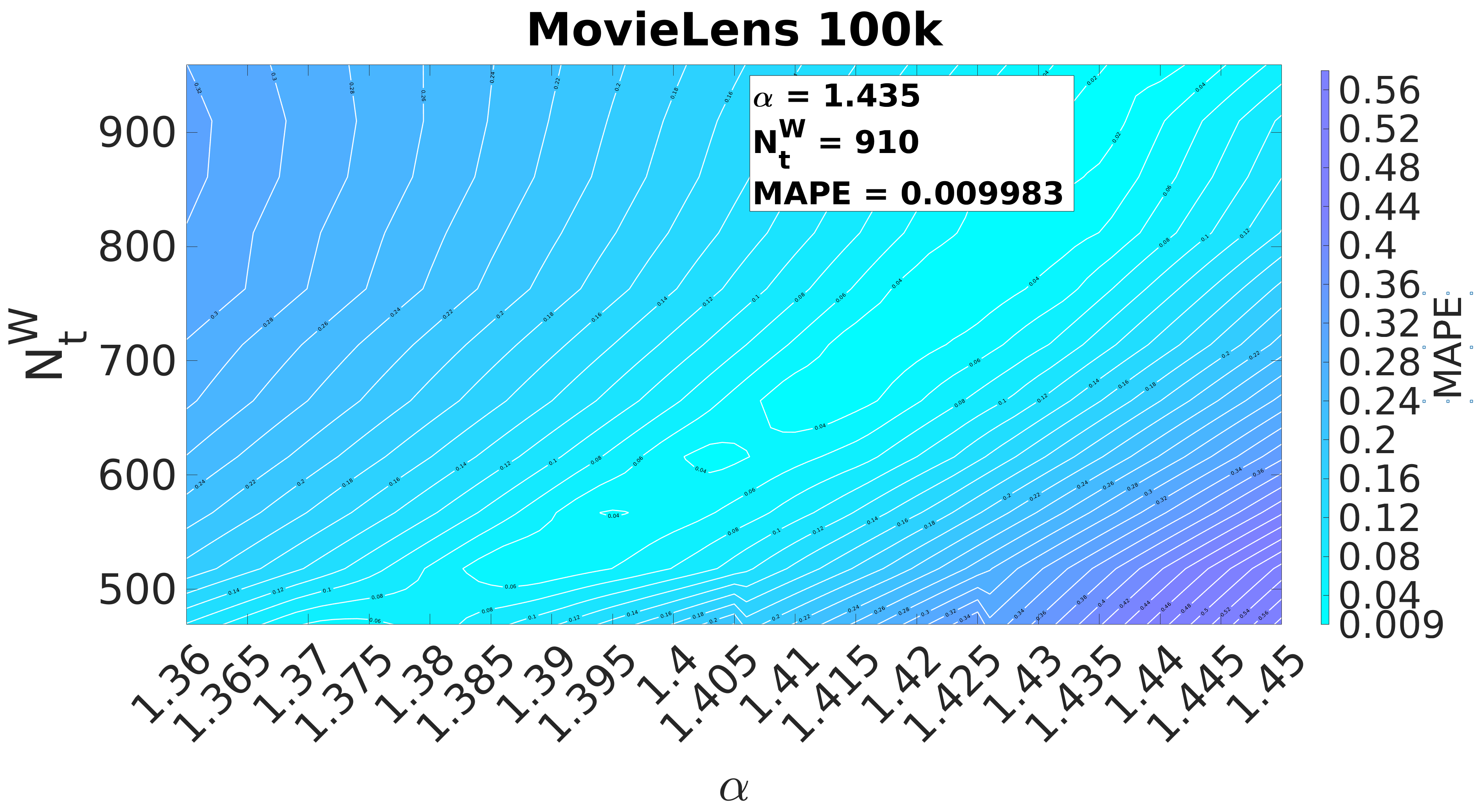}} 
    \subfigure{\includegraphics[width=0.3\textwidth]{ 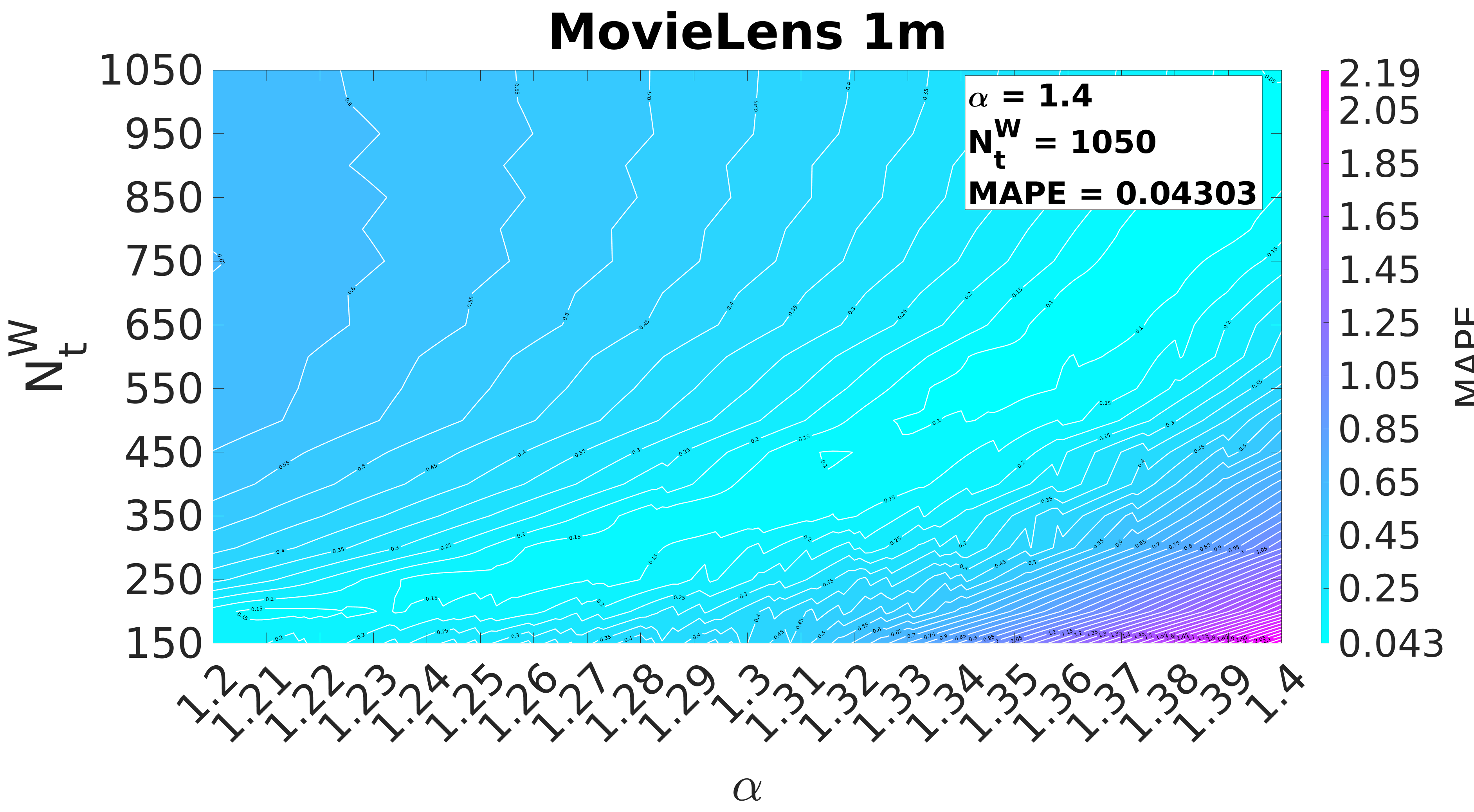}}
    \subfigure{\includegraphics[width=0.3\textwidth]{ 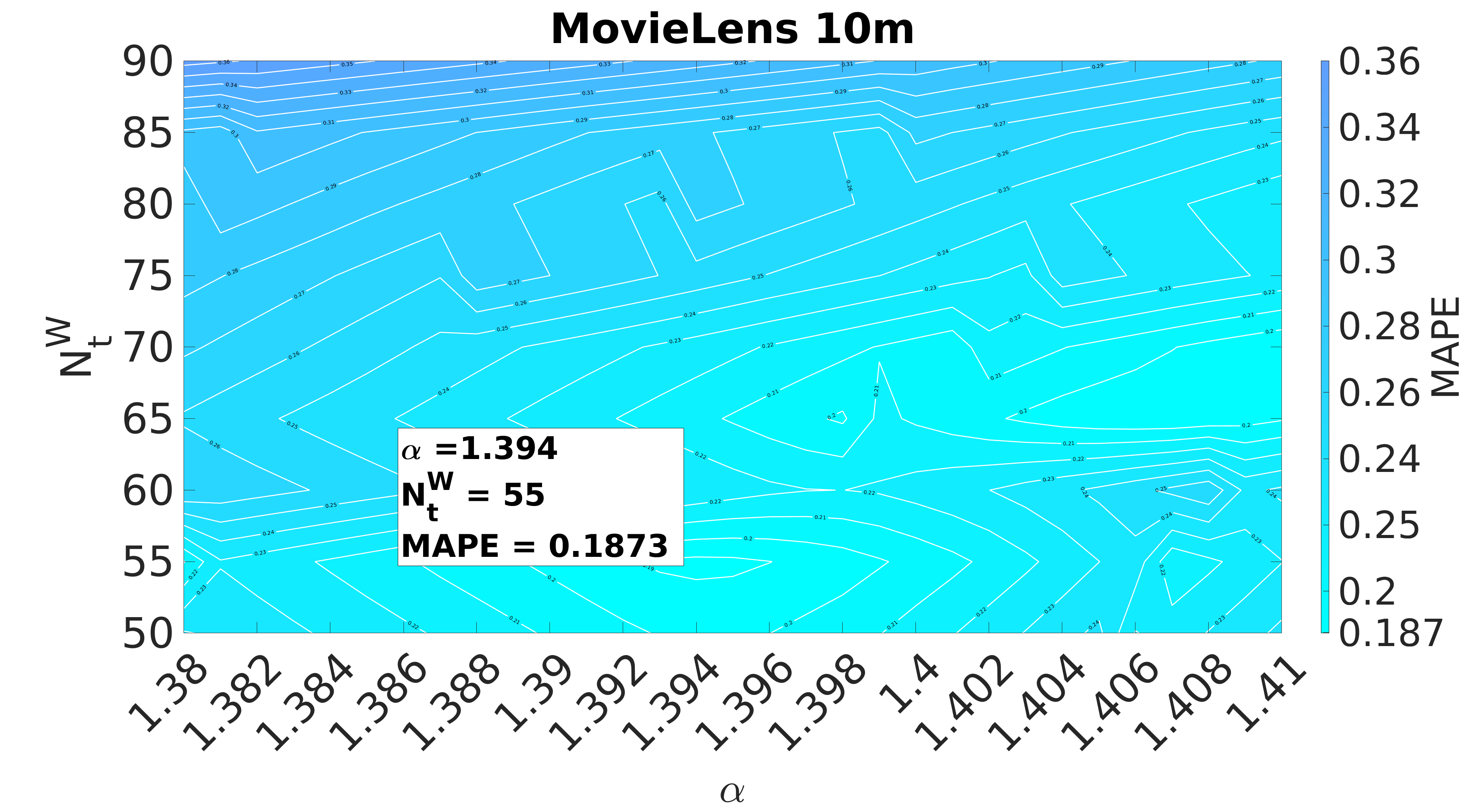}}
    \subfigure{\includegraphics[width=0.3\textwidth]{ 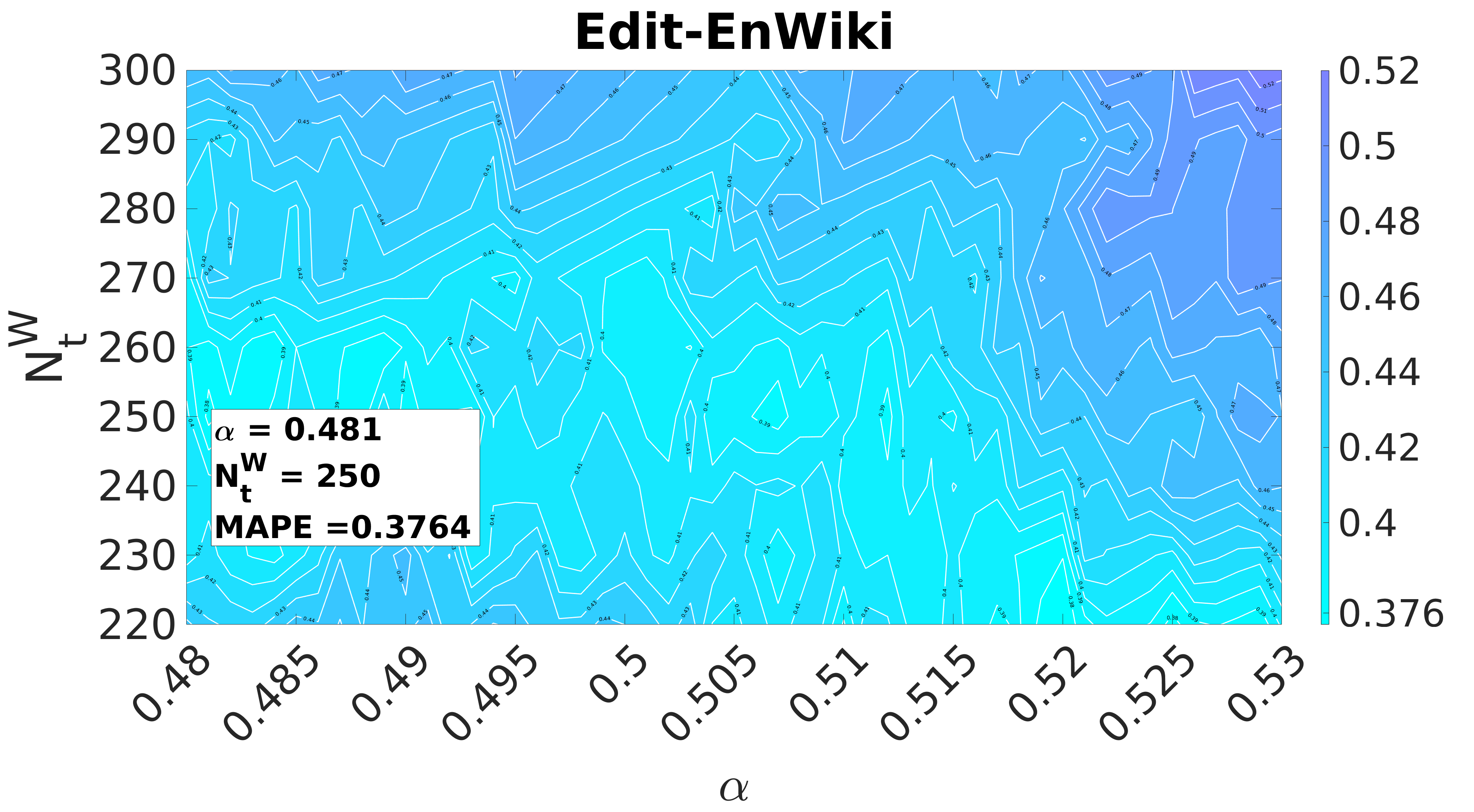}}
    \subfigure{\includegraphics[width=0.3\textwidth]{ 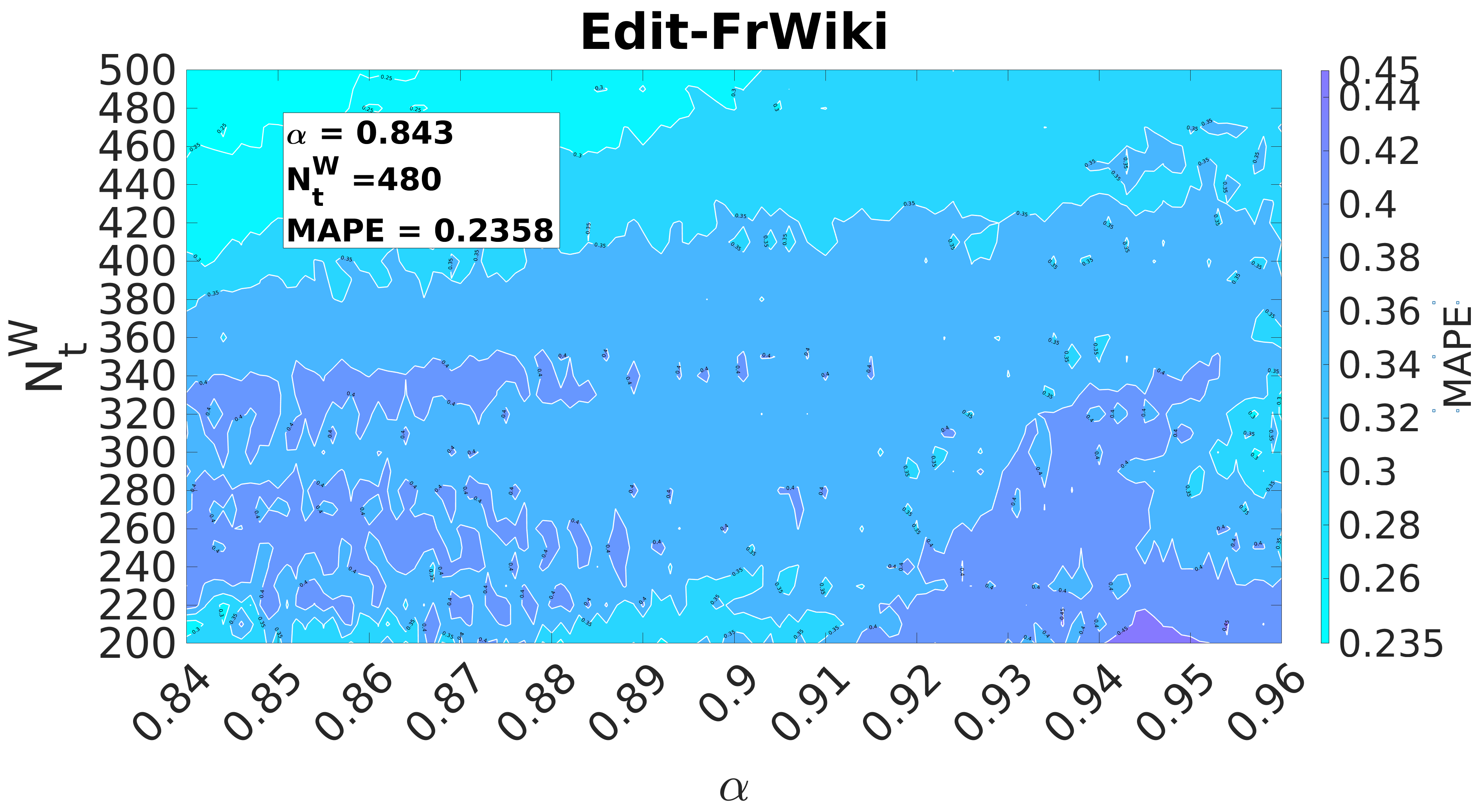}}
    \caption{[Best viewed in colored.] Accuracy of sGrapp-25 for different values of $\alpha$ and $N_t^W$.}
   \label{fig:mapes25}
\end{figure*}
\begin{figure*}[h]\centering
    \subfigure{\includegraphics[width=0.3\textwidth]{ 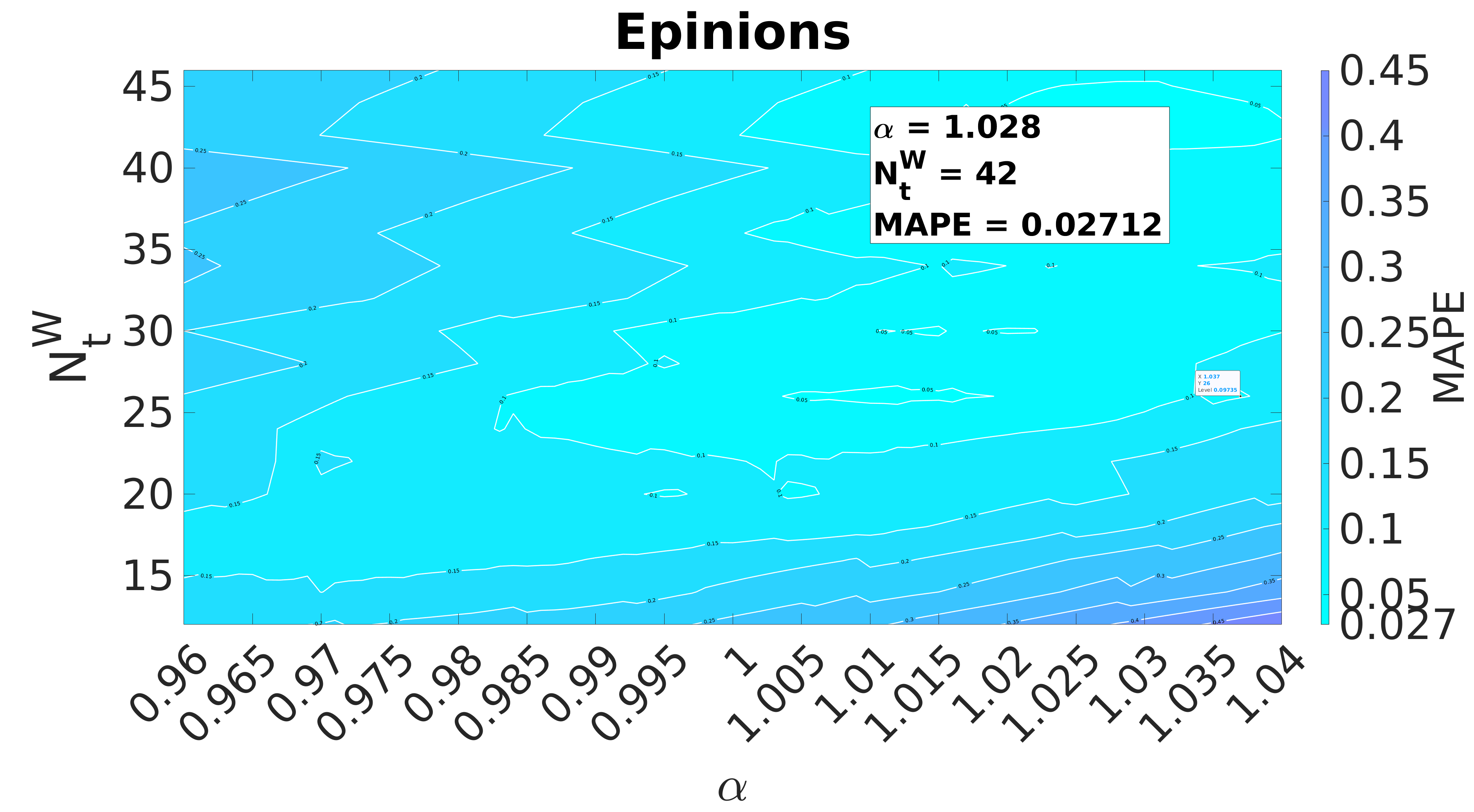}} 
    \subfigure{\includegraphics[width=0.3\textwidth]{ 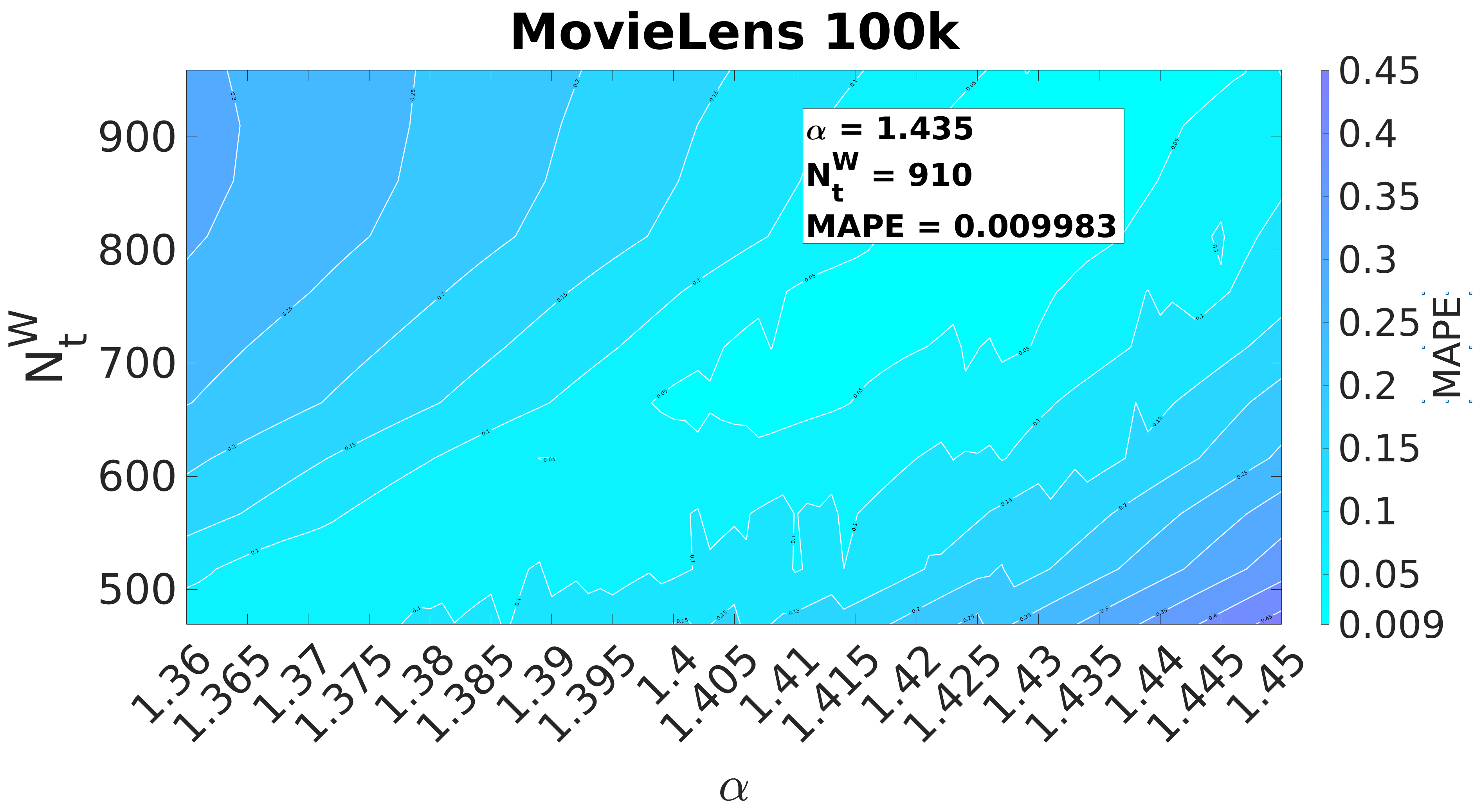}} 
    \subfigure{\includegraphics[width=0.3\textwidth]{ 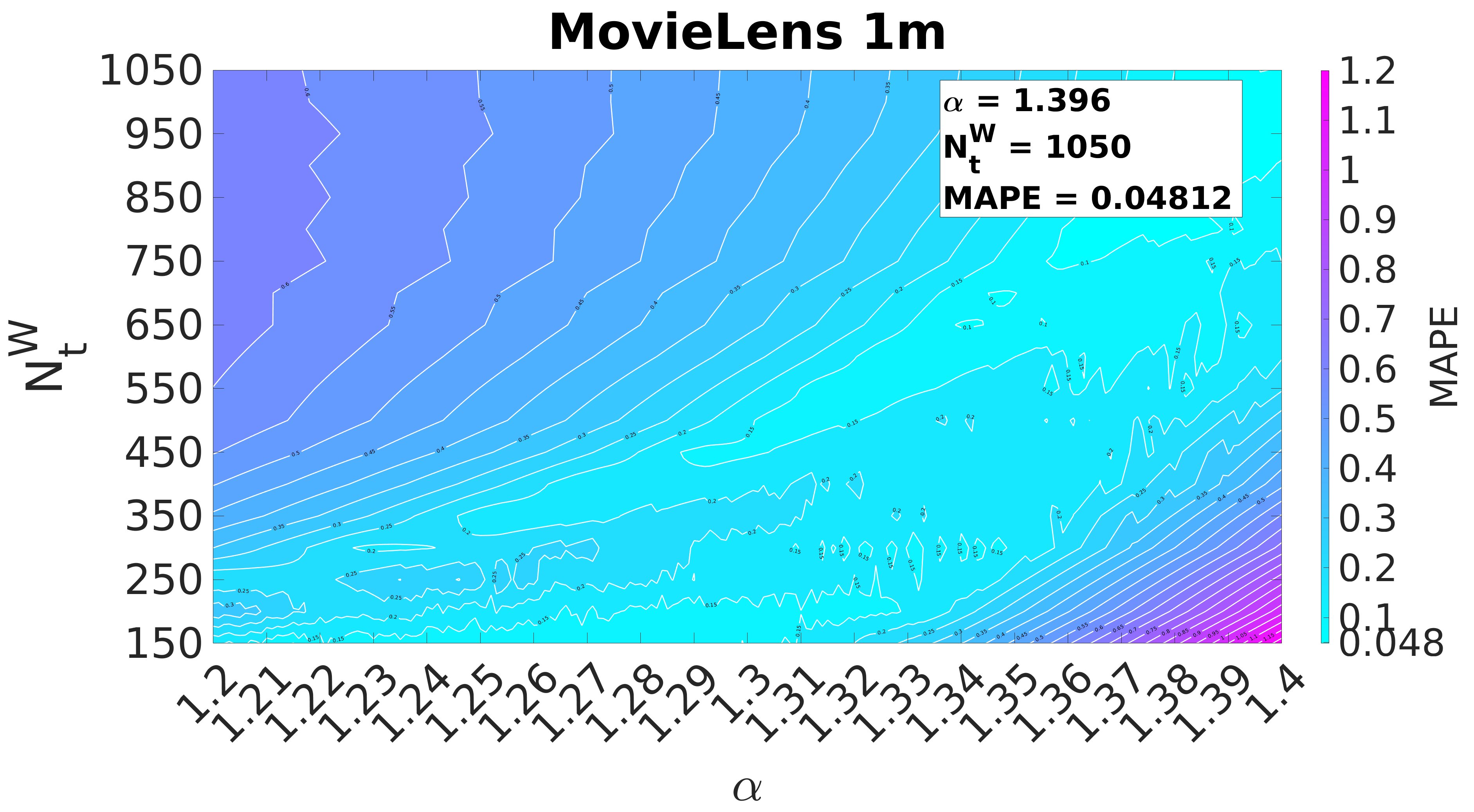}}
    \subfigure{\includegraphics[width=0.3\textwidth]{ 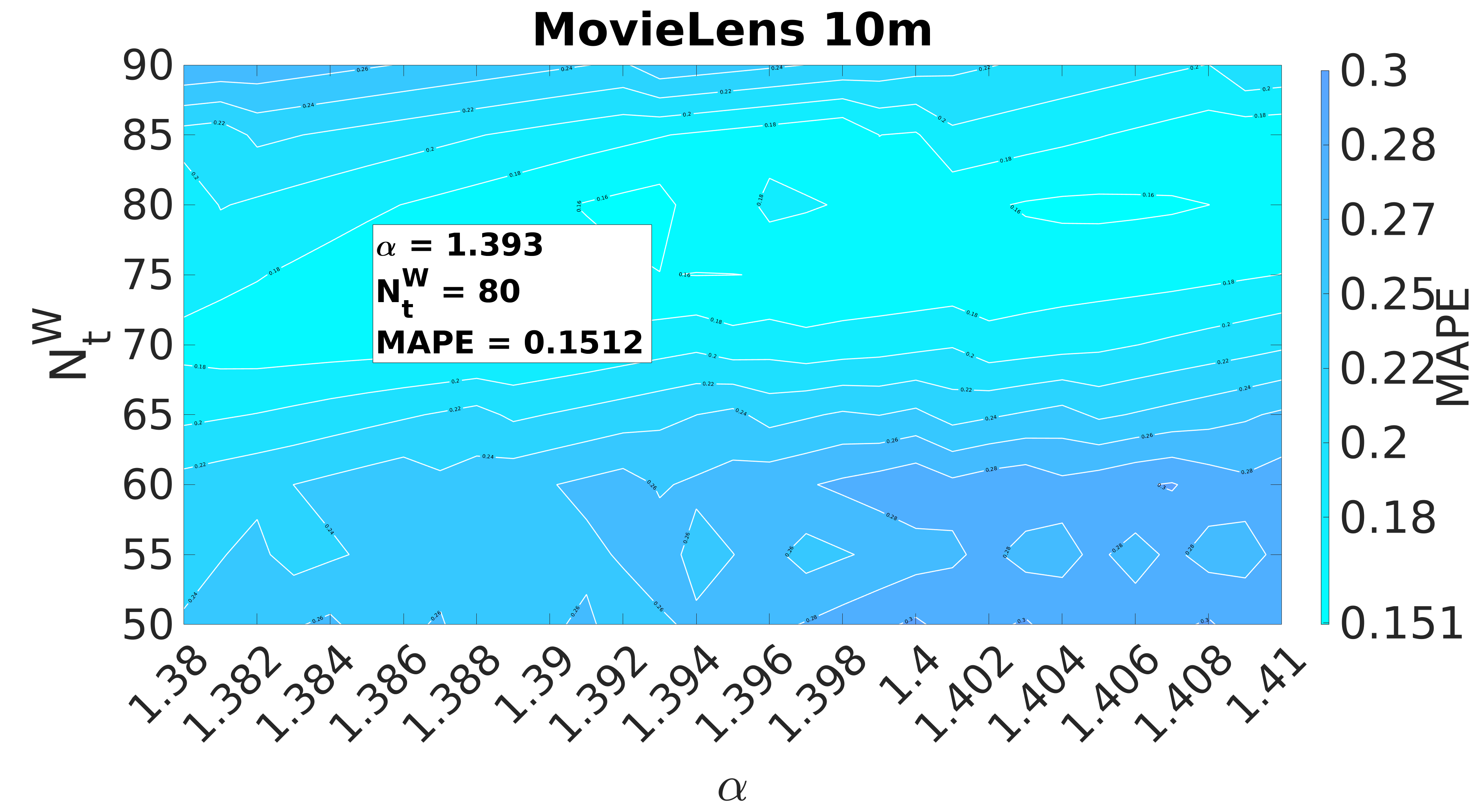}}
    \subfigure{\includegraphics[width=0.3\textwidth]{ 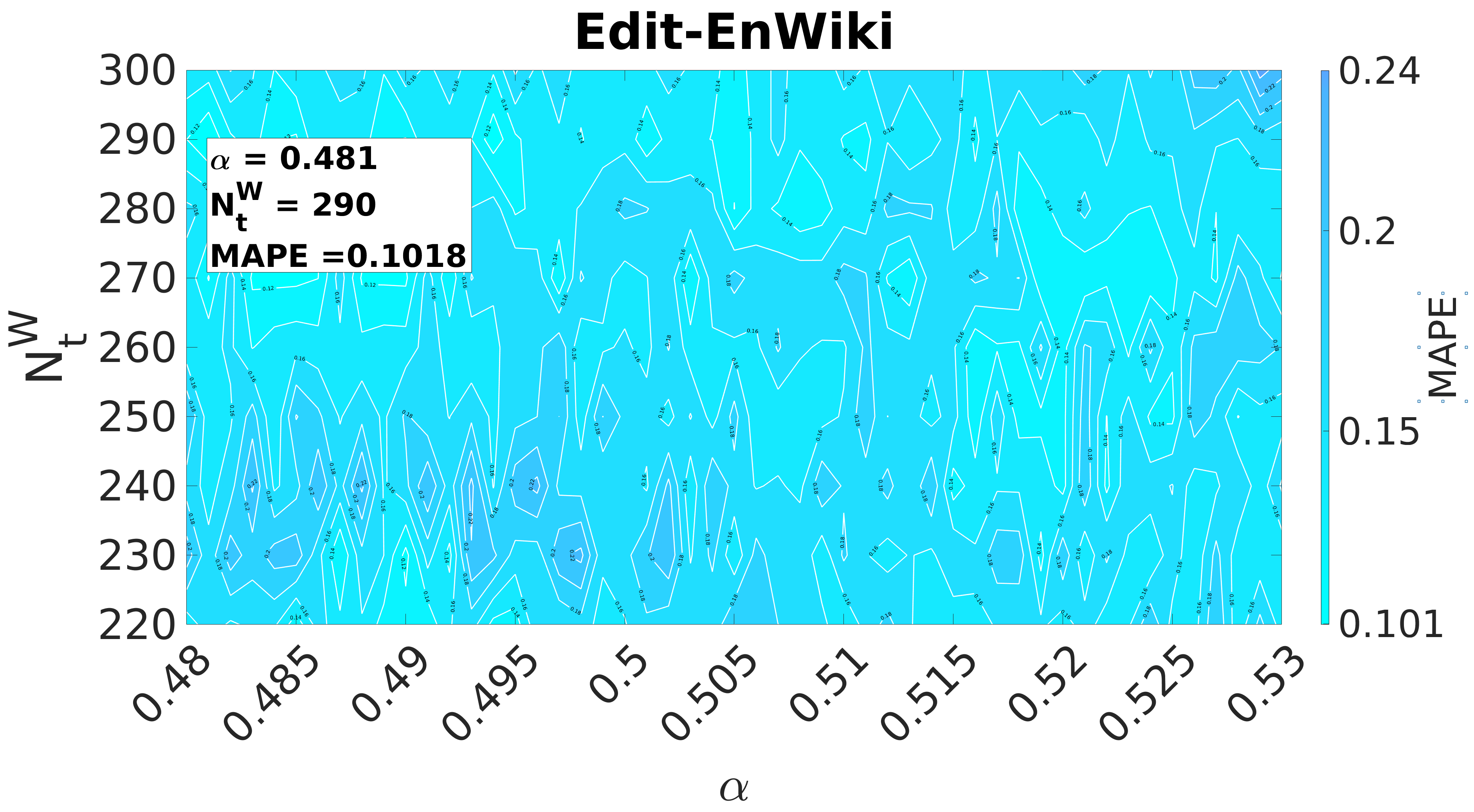}}
    \subfigure{\includegraphics[width=0.3\textwidth]{ 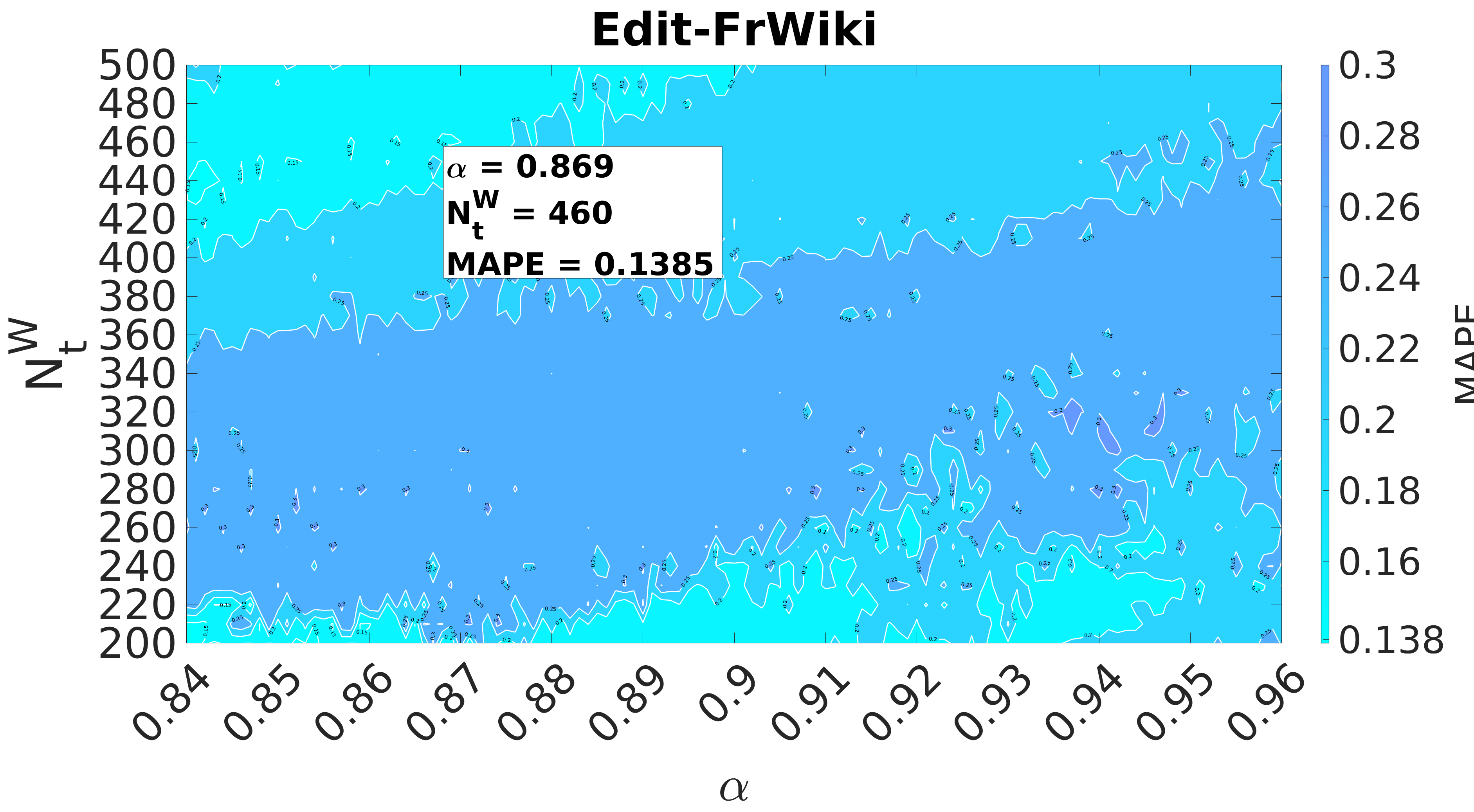}}
    \caption{[Best viewed in colored.] Accuracy of sGrapp-50 for different values of $\alpha$ and $N_t^W$ }
   \label{fig:mapes50}
\end{figure*}
\begin{figure*}[h]\centering
    \subfigure{\includegraphics[width=0.3\textwidth]{ 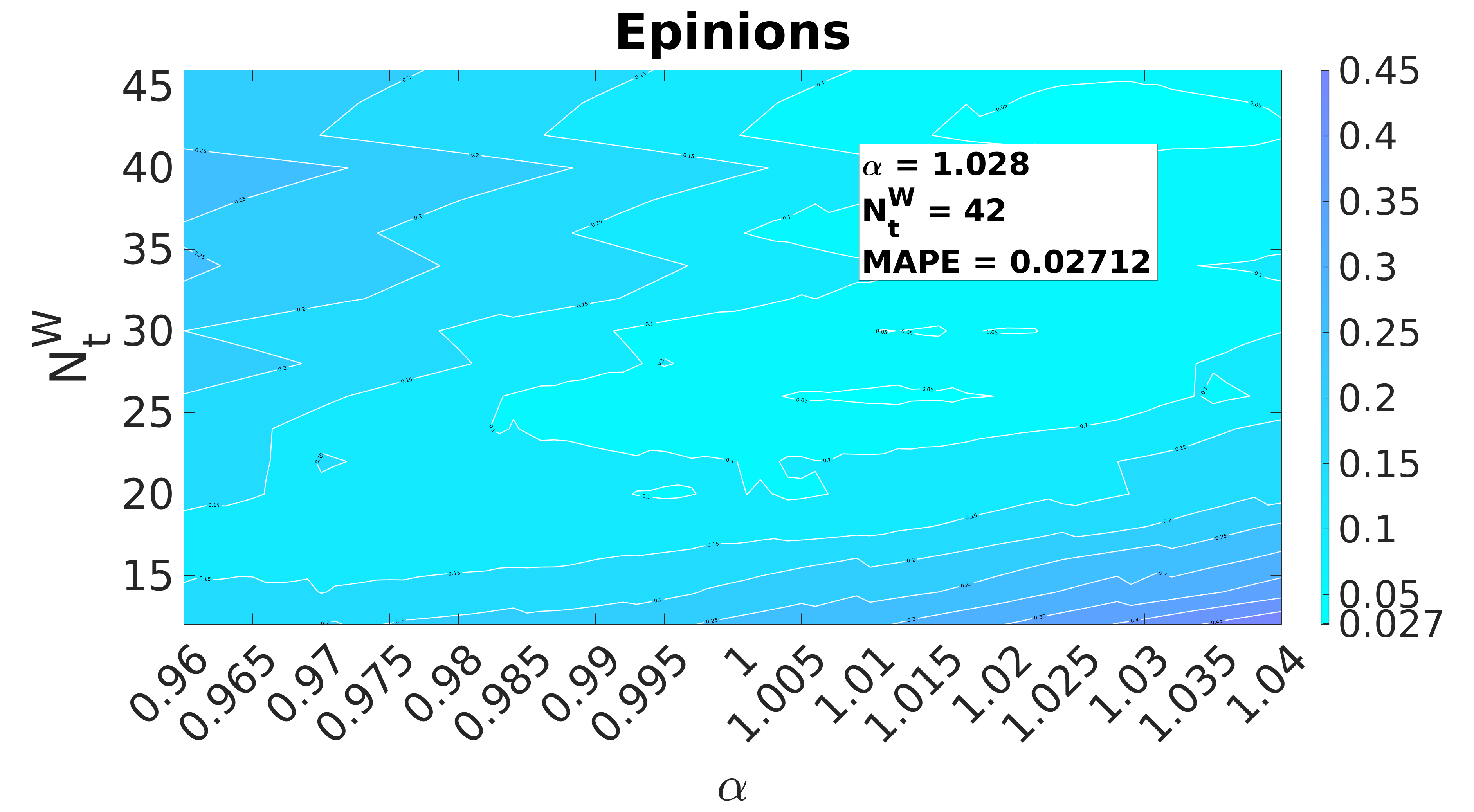}} 
    \subfigure{\includegraphics[width=0.3\textwidth]{ 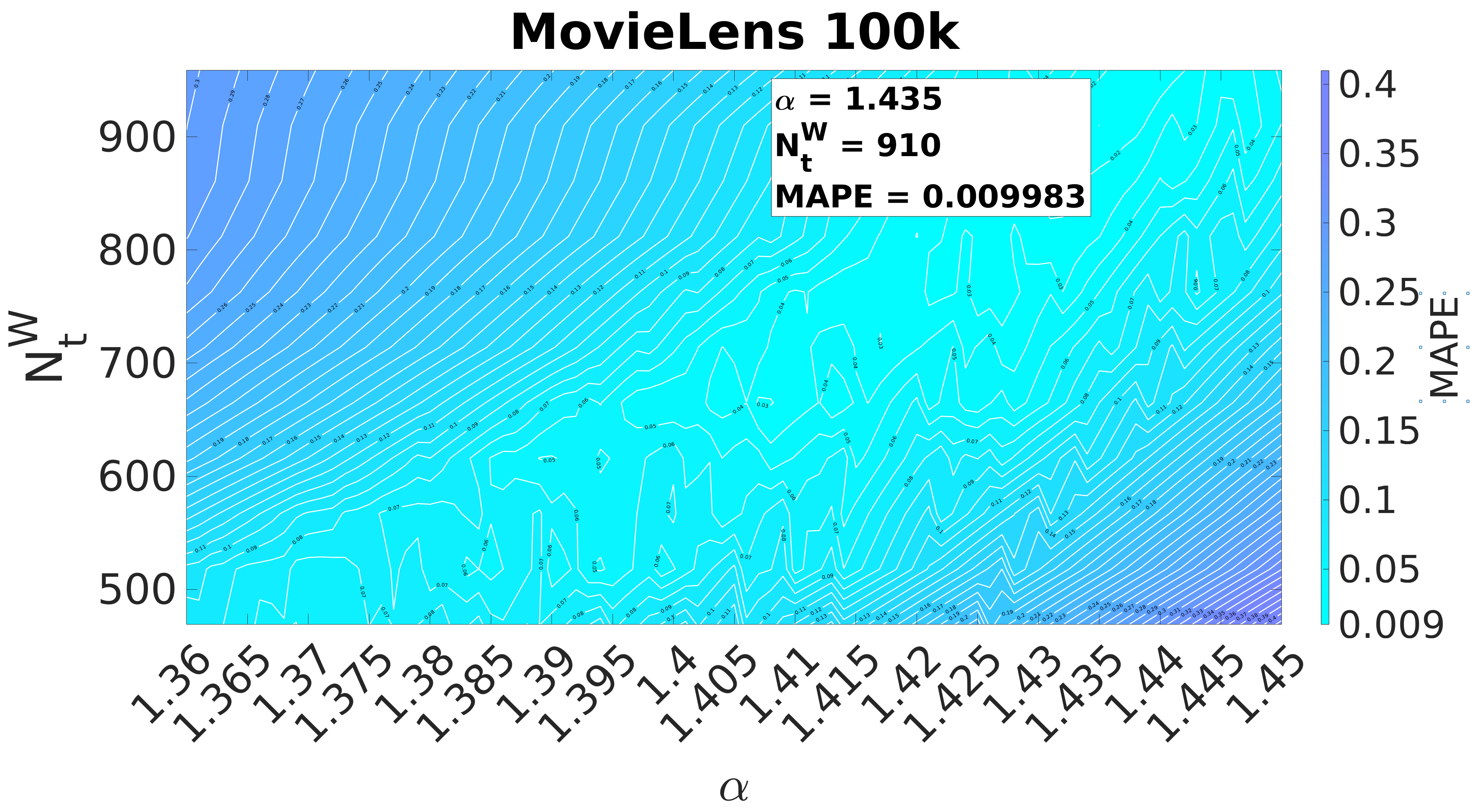}} 
    \subfigure{\includegraphics[width=0.3\textwidth]{ 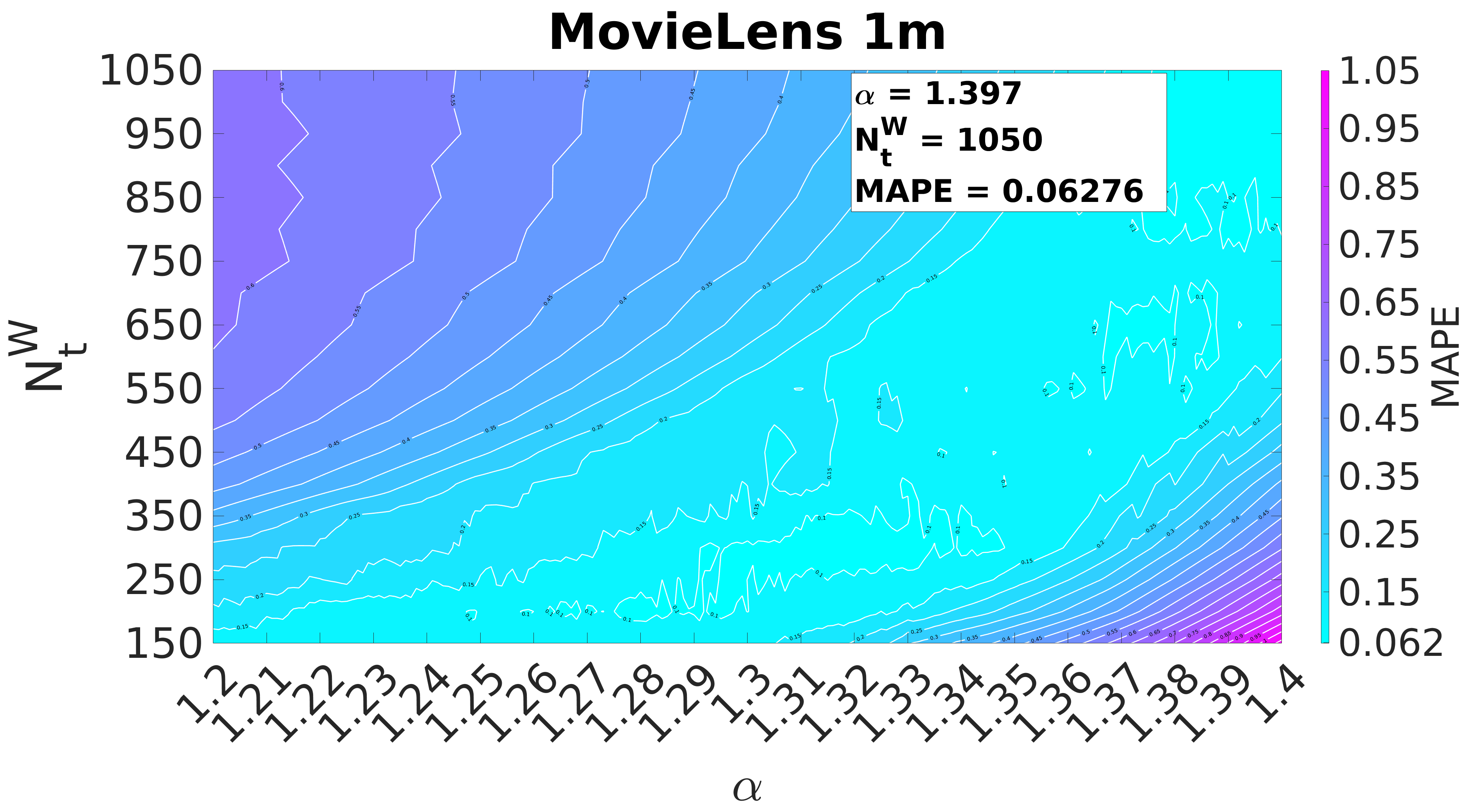}}
    \subfigure{\includegraphics[width=0.3\textwidth]{ 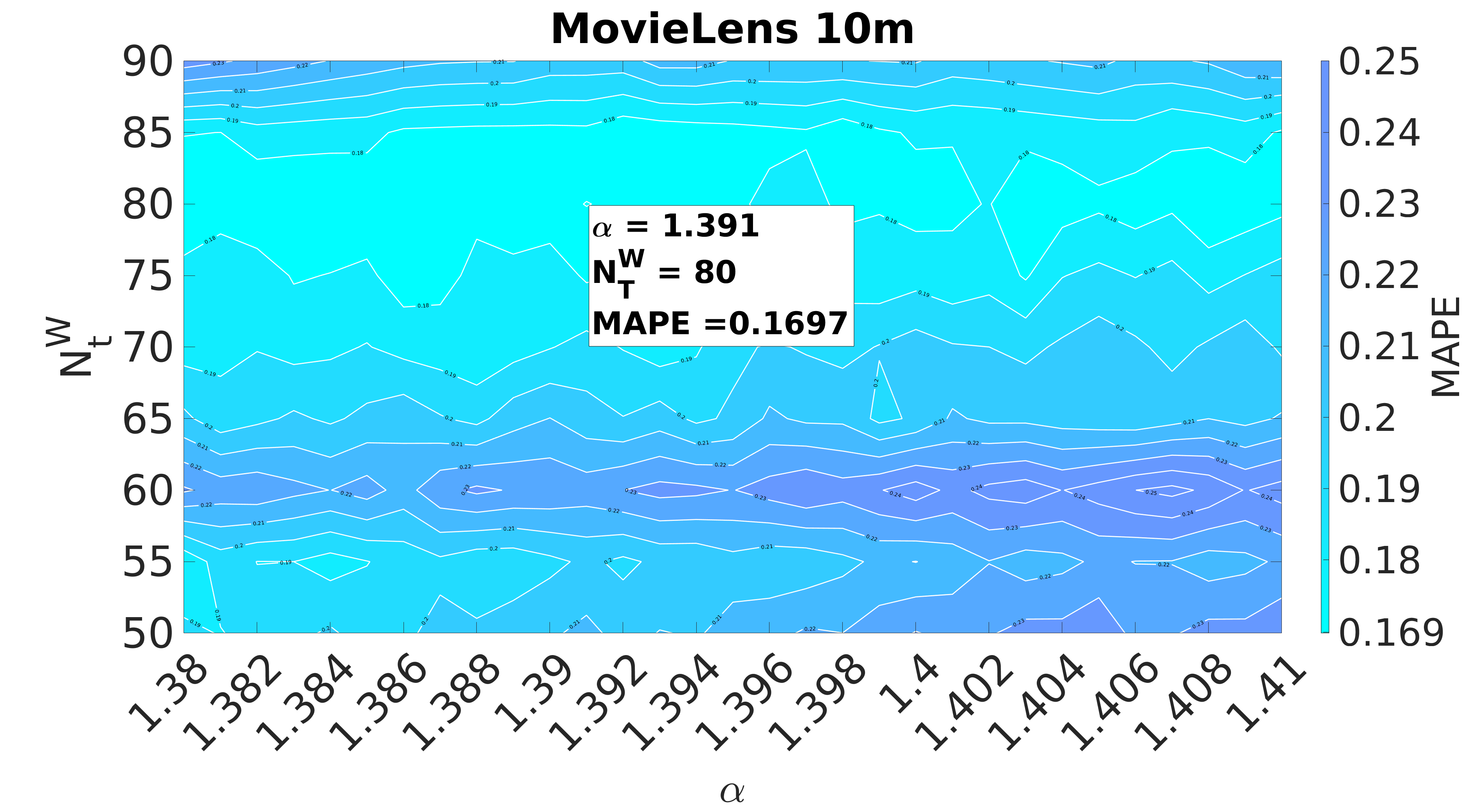}}
    \subfigure{\includegraphics[width=0.3\textwidth]{ 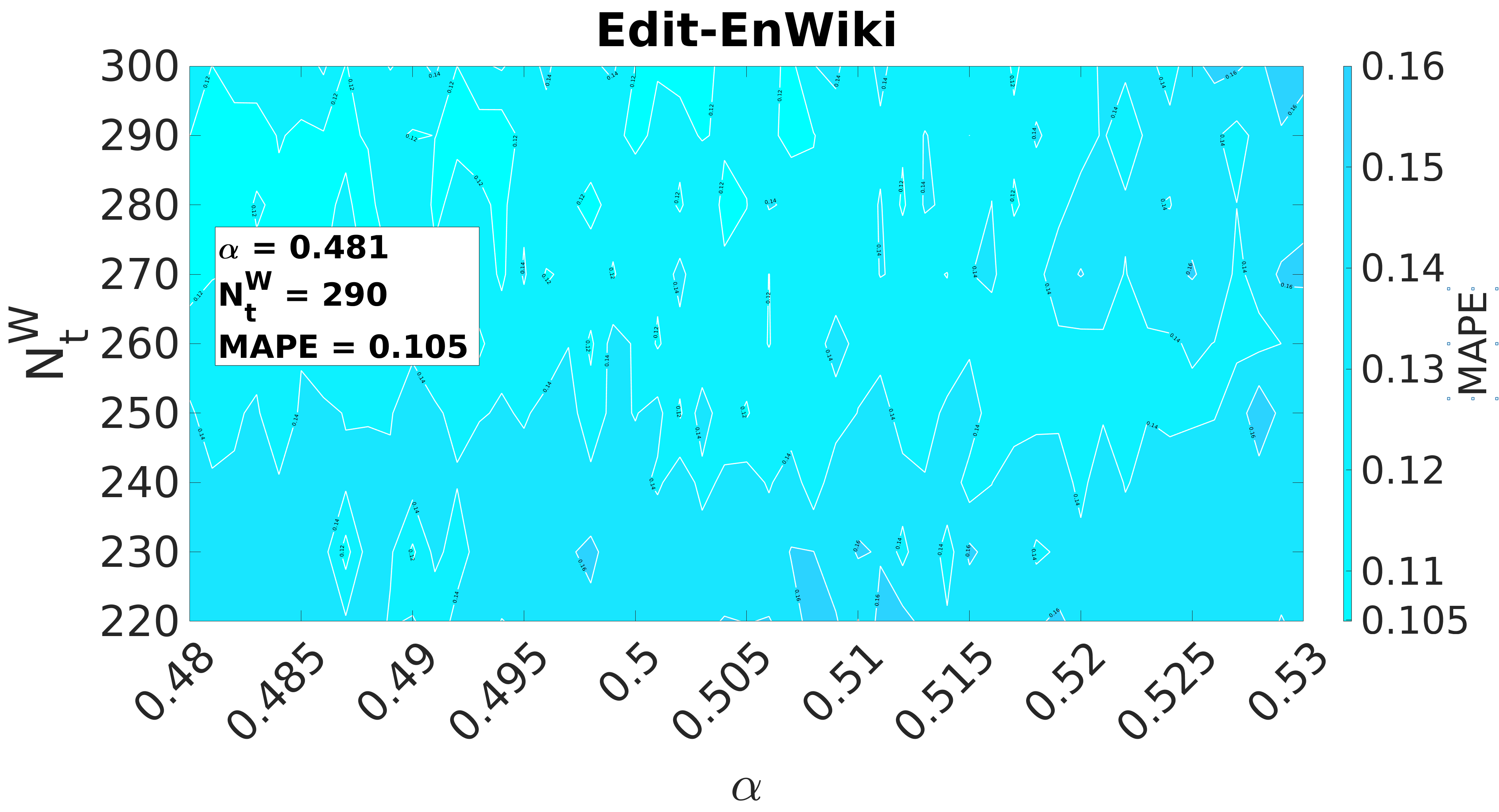}}
    \subfigure{\includegraphics[width=0.3\textwidth]{ 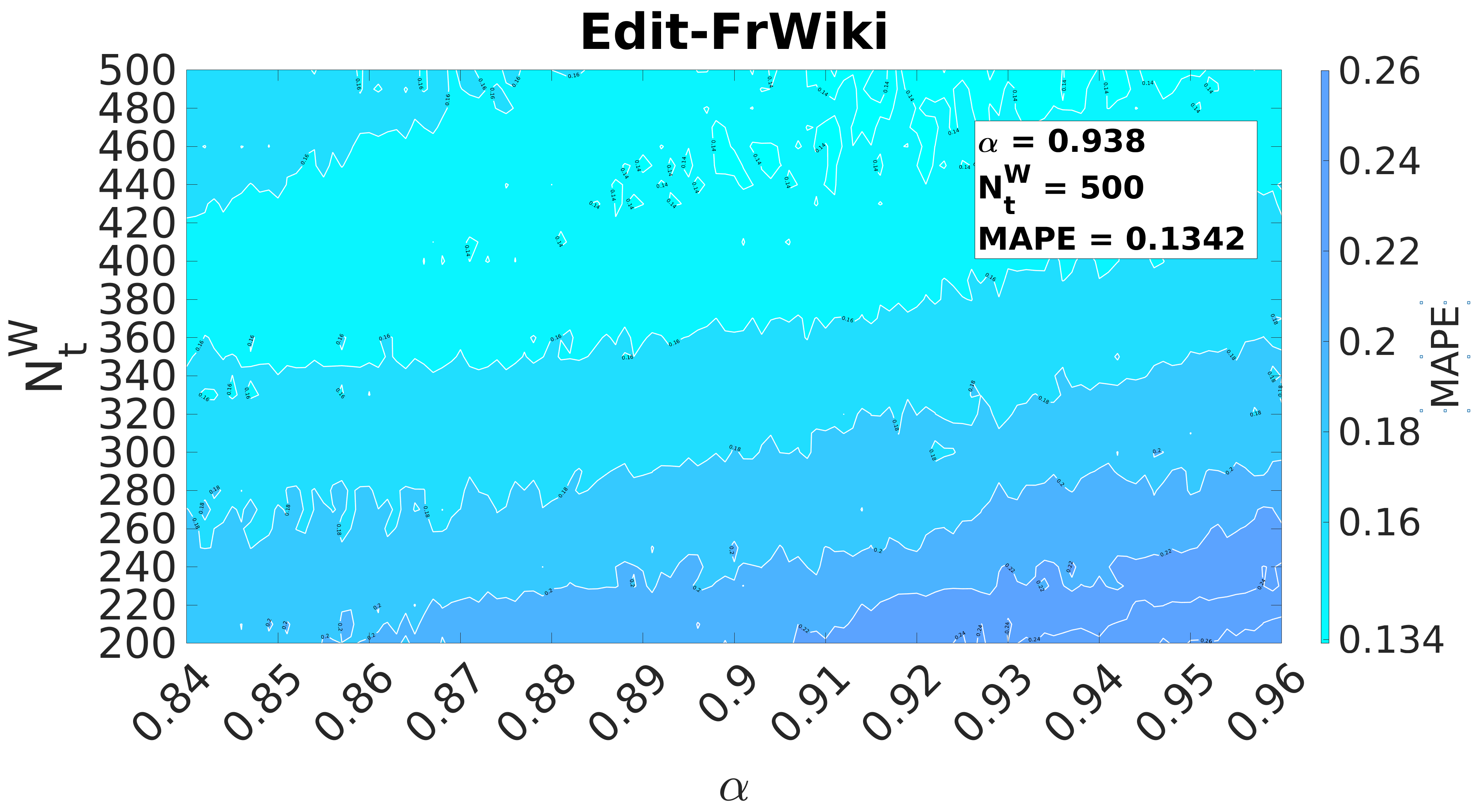}}
    \caption{[Best viewed in colored.] Accuracy of sGrapp-75 for different values of $\alpha$ and $N_t^W$.}
   \label{fig:mapes75}
\end{figure*}
\begin{figure*}[h]\centering
    \subfigure{\includegraphics[width=0.3\textwidth]{ 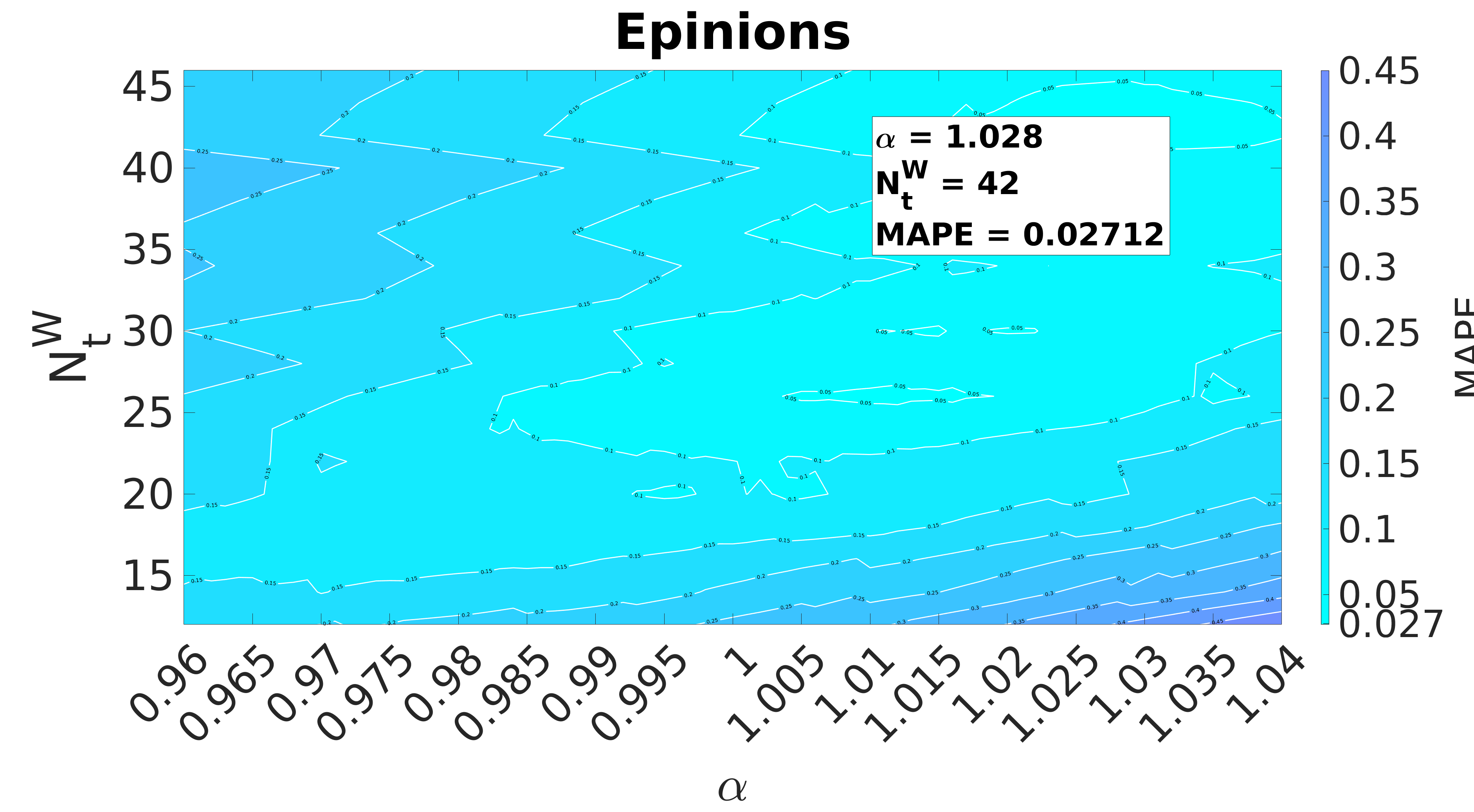}} 
    \subfigure{\includegraphics[width=0.3\textwidth]{ 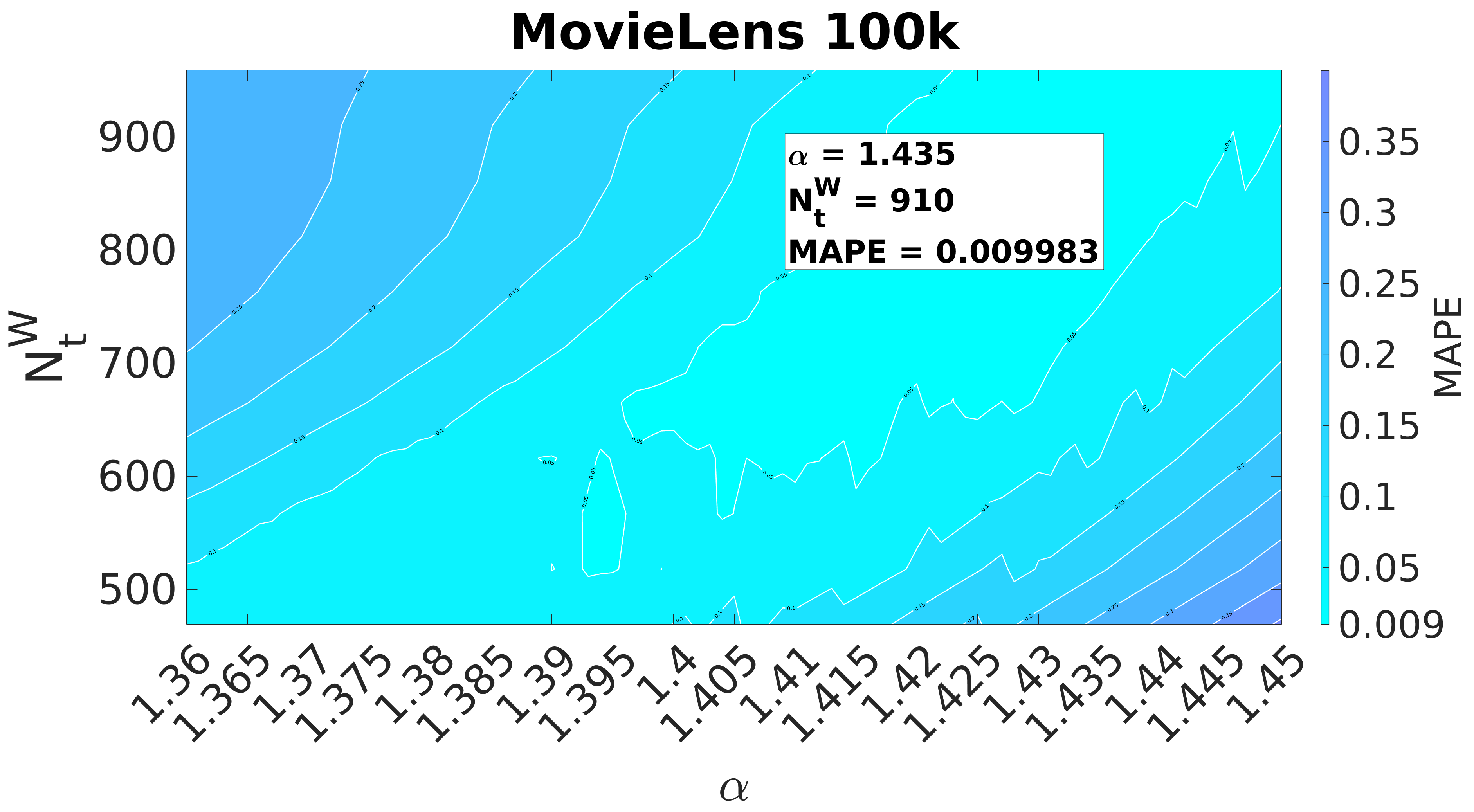}} 
    \subfigure{\includegraphics[width=0.3\textwidth]{ 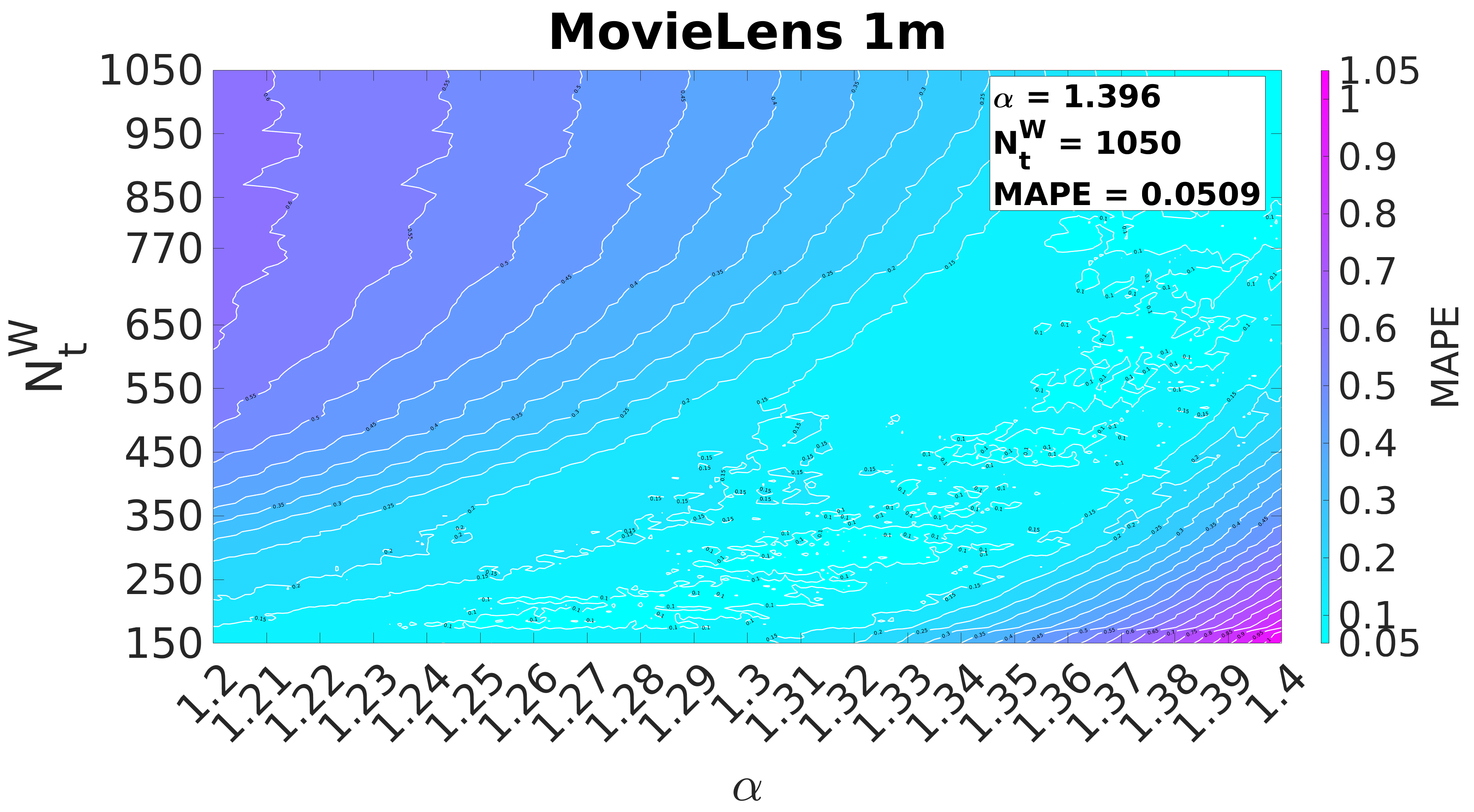}}
    \subfigure{\includegraphics[width=0.3\textwidth]{ 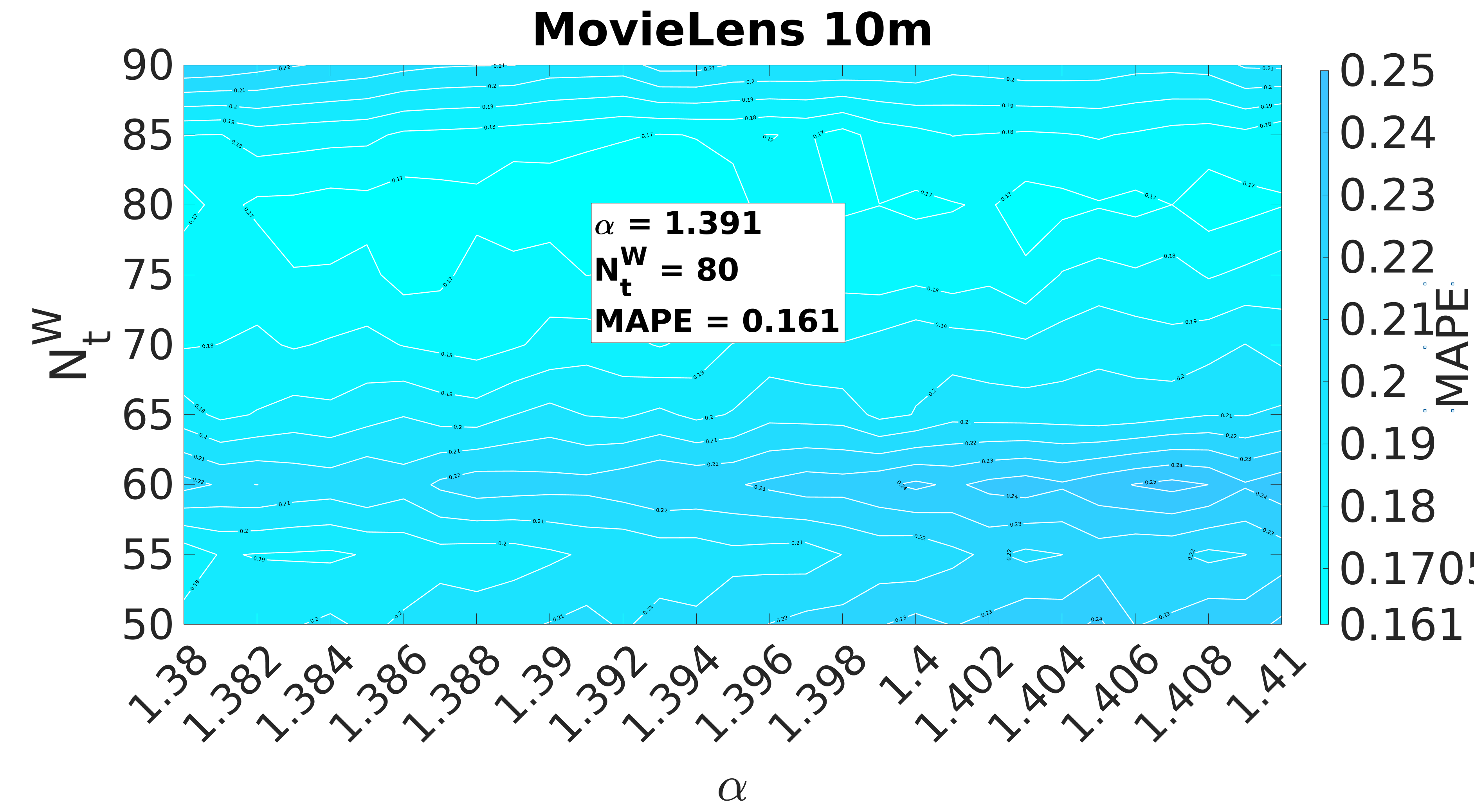}}
    \subfigure{\includegraphics[width=0.3\textwidth]{ 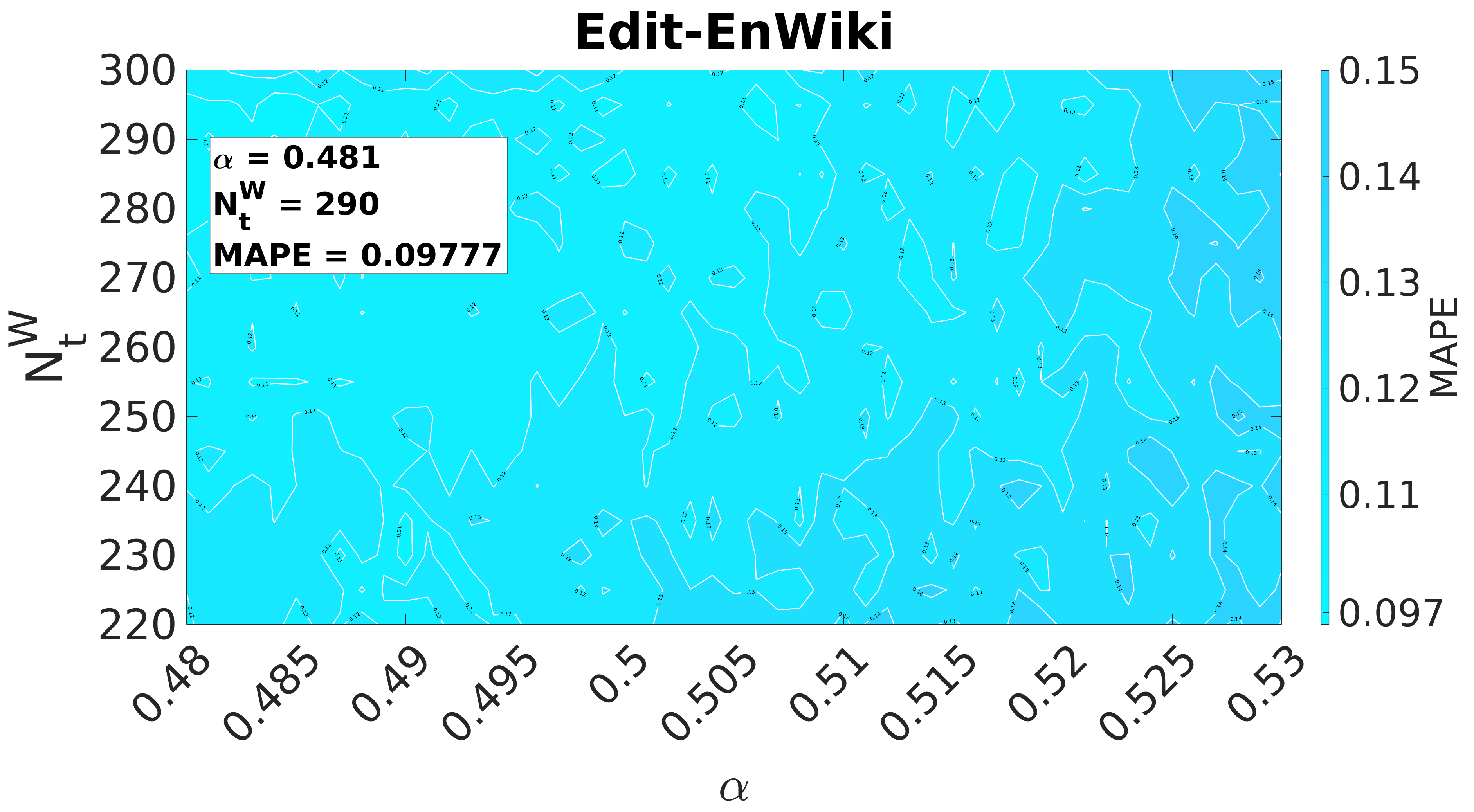}}
    \subfigure{\includegraphics[width=0.3\textwidth]{ 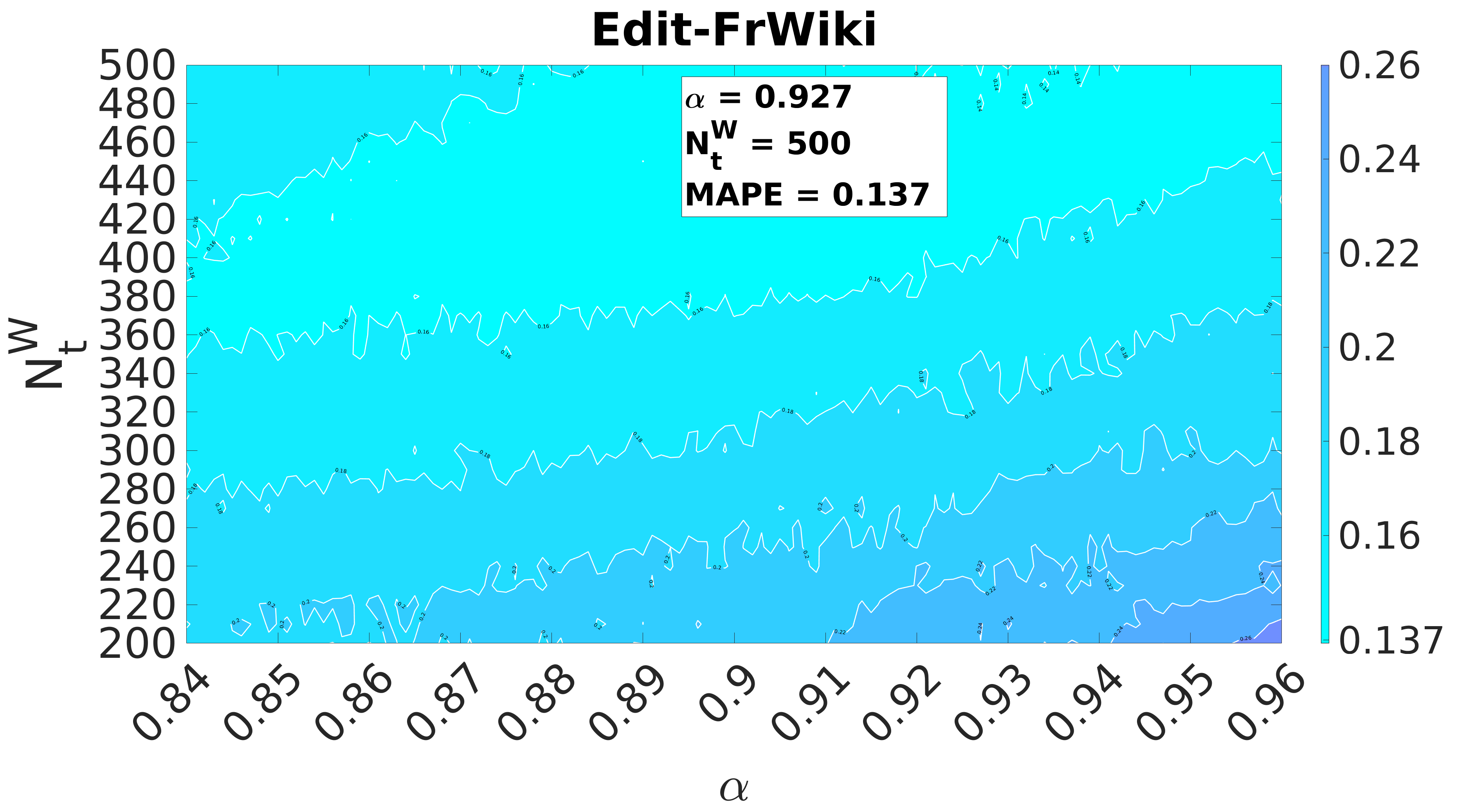}}
    \caption{[Best viewed in colored.] Accuracy of sGrapp-100 for different values of $\alpha$ and $N_t^W$. }
   \label{fig:mapes100}
\end{figure*}
\begin{figure*}[h]\centering
    \includegraphics[width=0.9\textwidth]{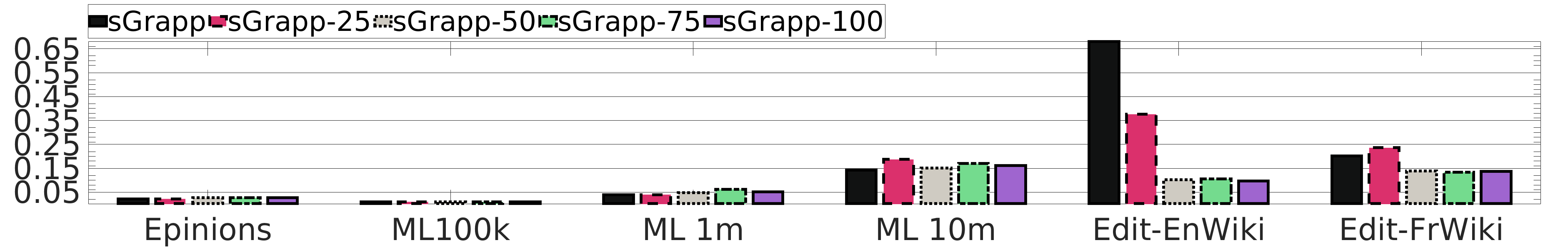}
\caption{Minimum approximation MAPE. }
   \label{fig:mapemin}
\end{figure*}
\begin{figure*}[h]\centering
    \includegraphics[width=0.9\textwidth]{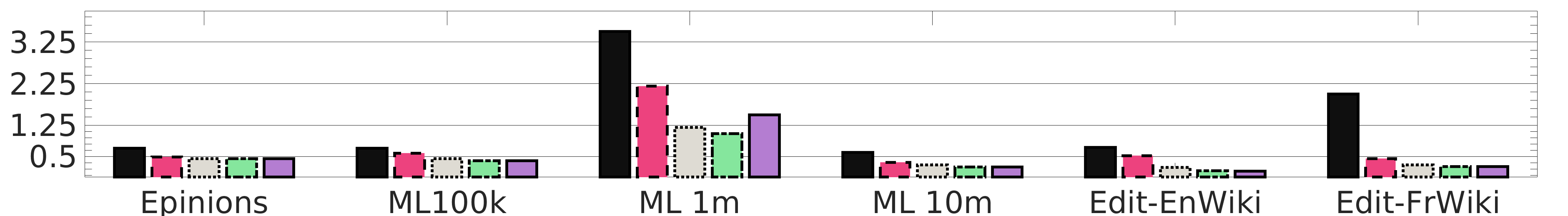}
\caption{Maximum approximation MAPE. }
   \label{fig:mapemax}
\end{figure*}
\begin{figure*}[h]\centering
    \includegraphics[width=0.9\textwidth]{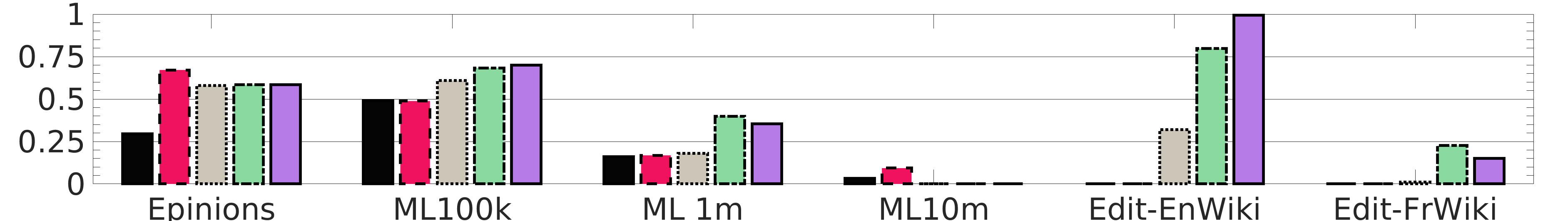}
\caption{Probability of approximation with MAPE less than equal 0.15. }
   \label{fig:mapeprob0.15}
\end{figure*}
\begin{figure*}[h]\centering
    \includegraphics[width=0.9\textwidth]{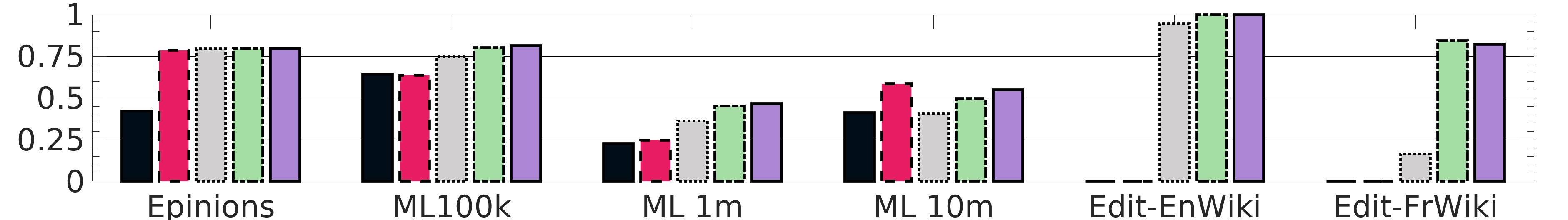}
\caption{Probability of approximation with MAPE less than equal 0.2. }
   \label{fig:mapeprob0.2}
\end{figure*}
\begin{figure*}[h]\centering
    \subfigure{\includegraphics[width=0.3\textwidth]{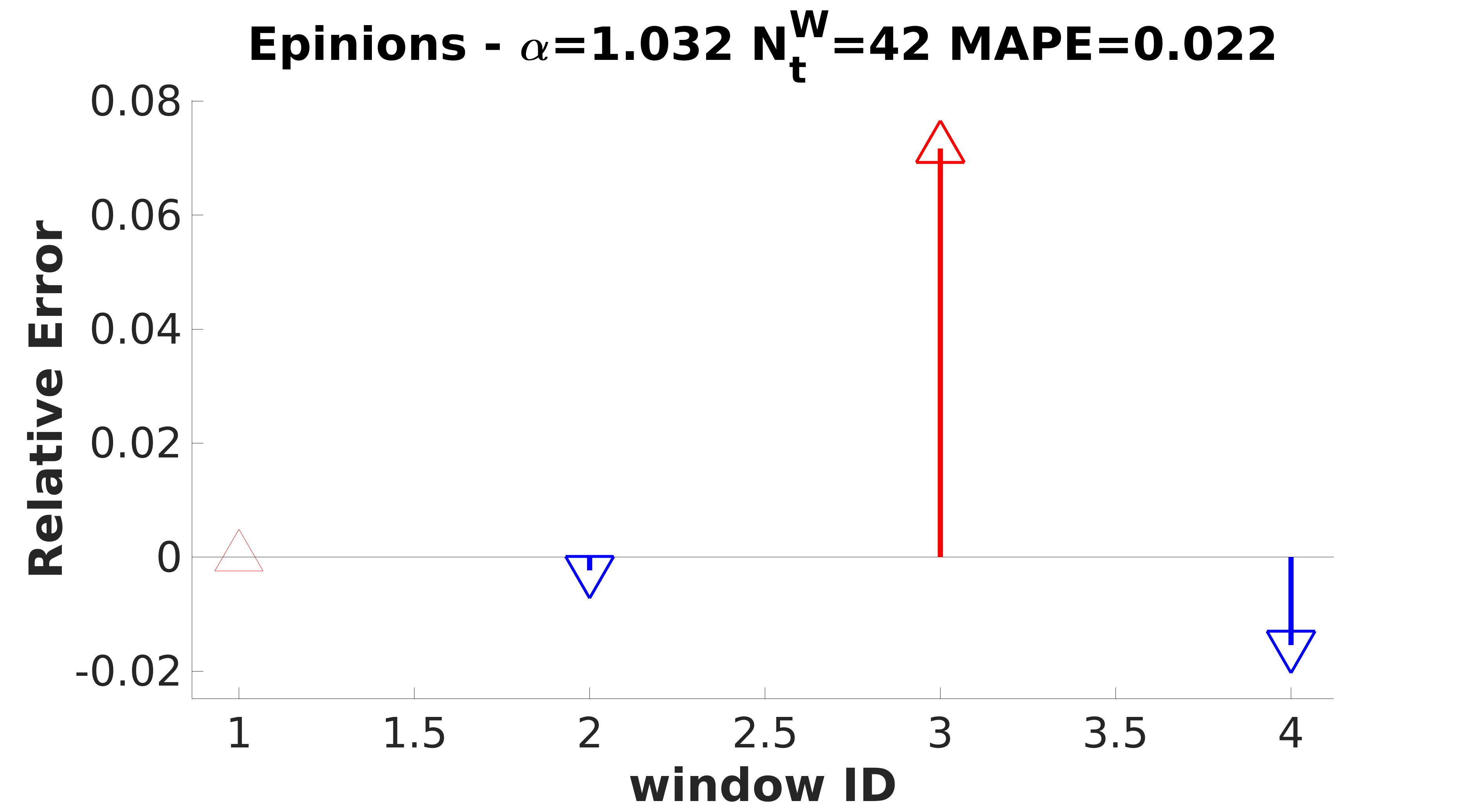}} 
    \subfigure{\includegraphics[width=0.3\textwidth]{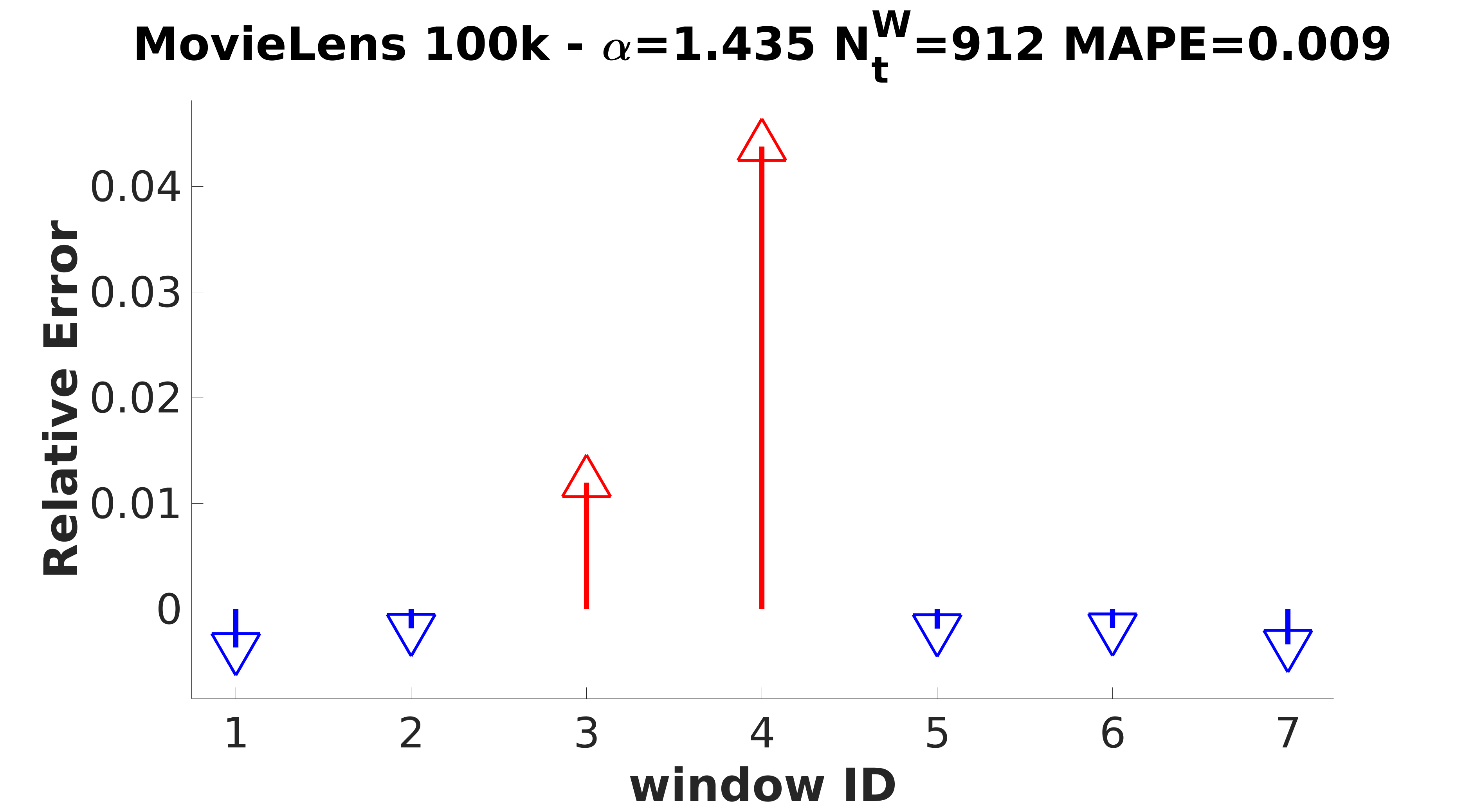}}
    \subfigure{\includegraphics[width=0.3\textwidth]{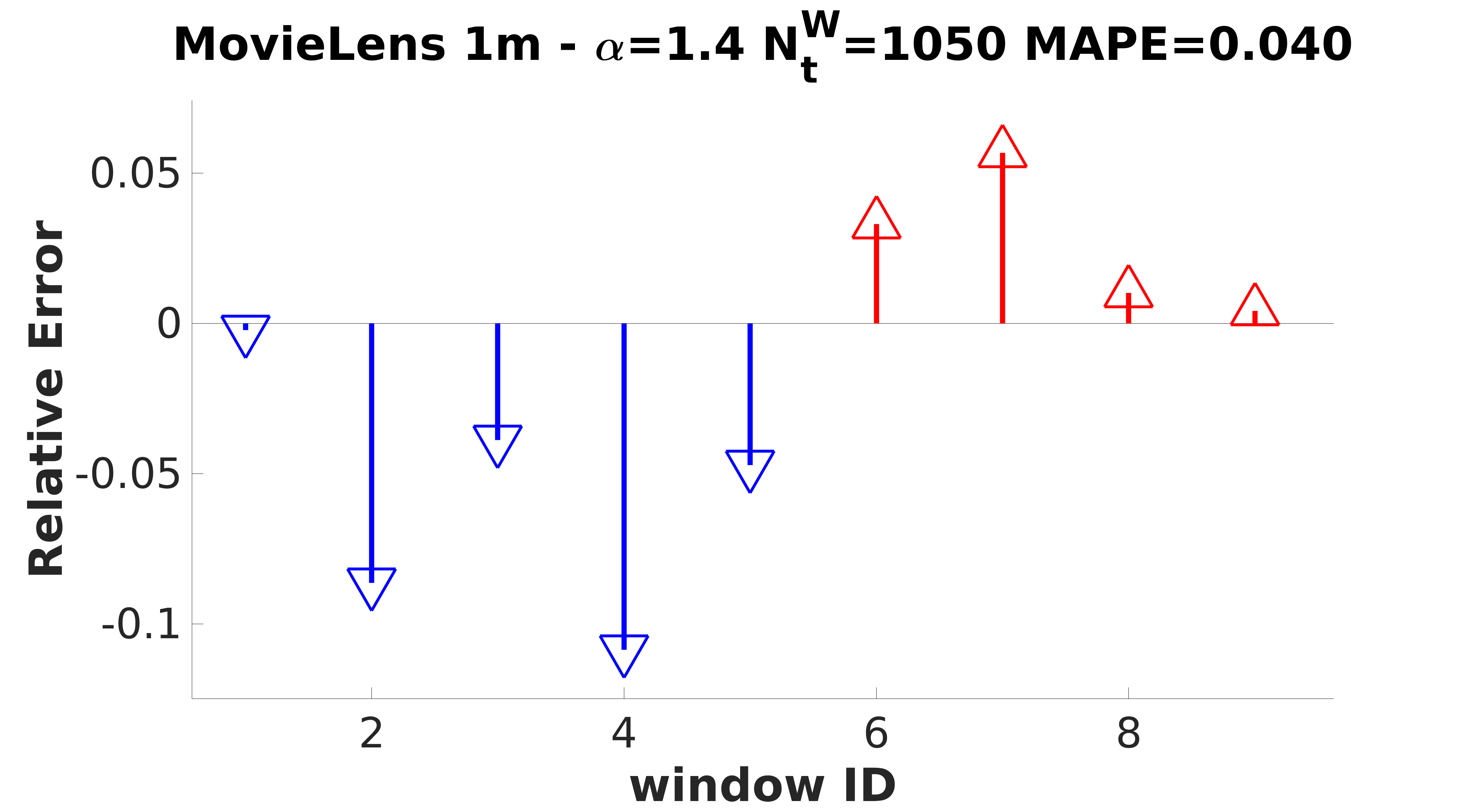}}
    \subfigure{\includegraphics[width=0.3\textwidth]{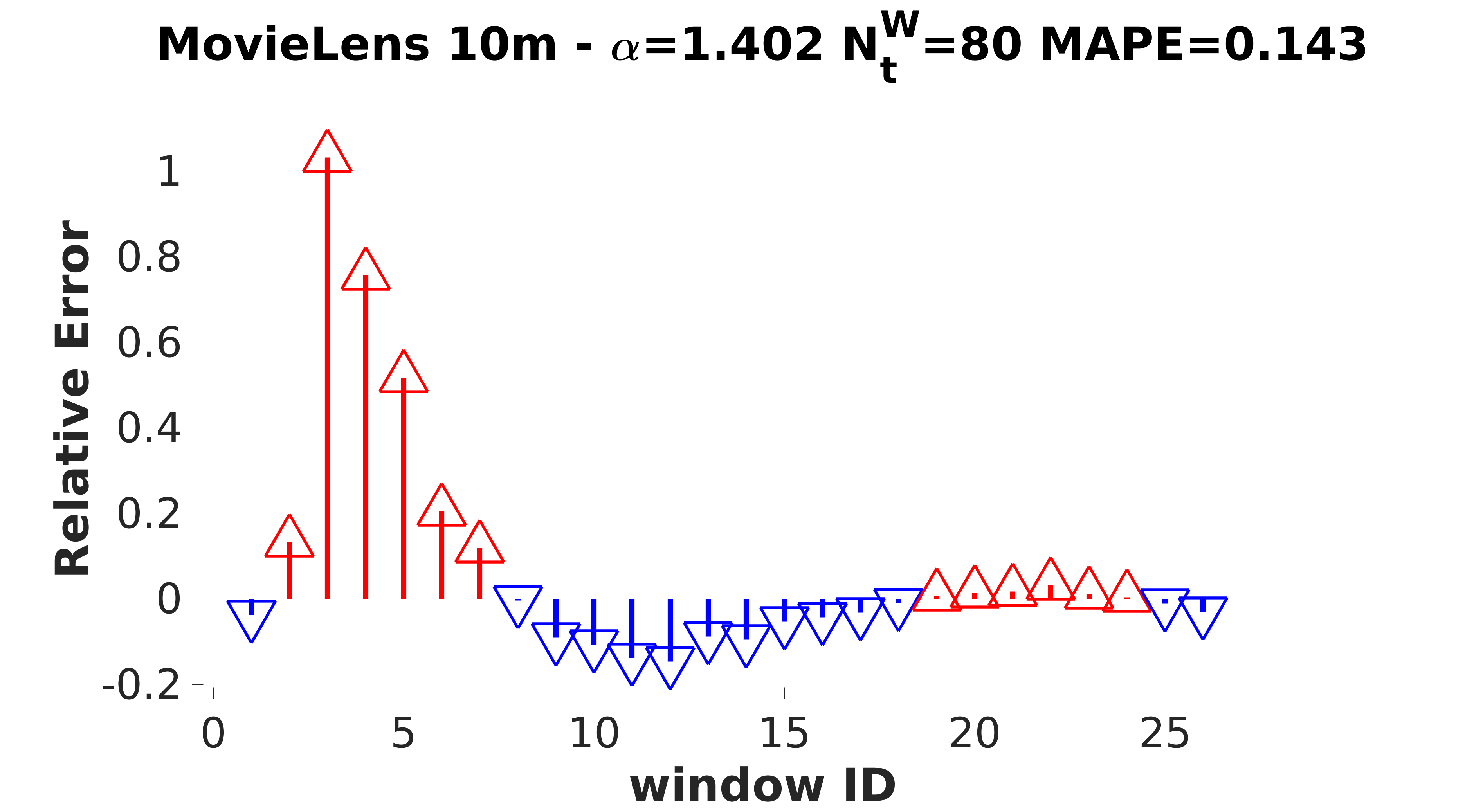}}
    \subfigure{\includegraphics[width=0.3\textwidth]{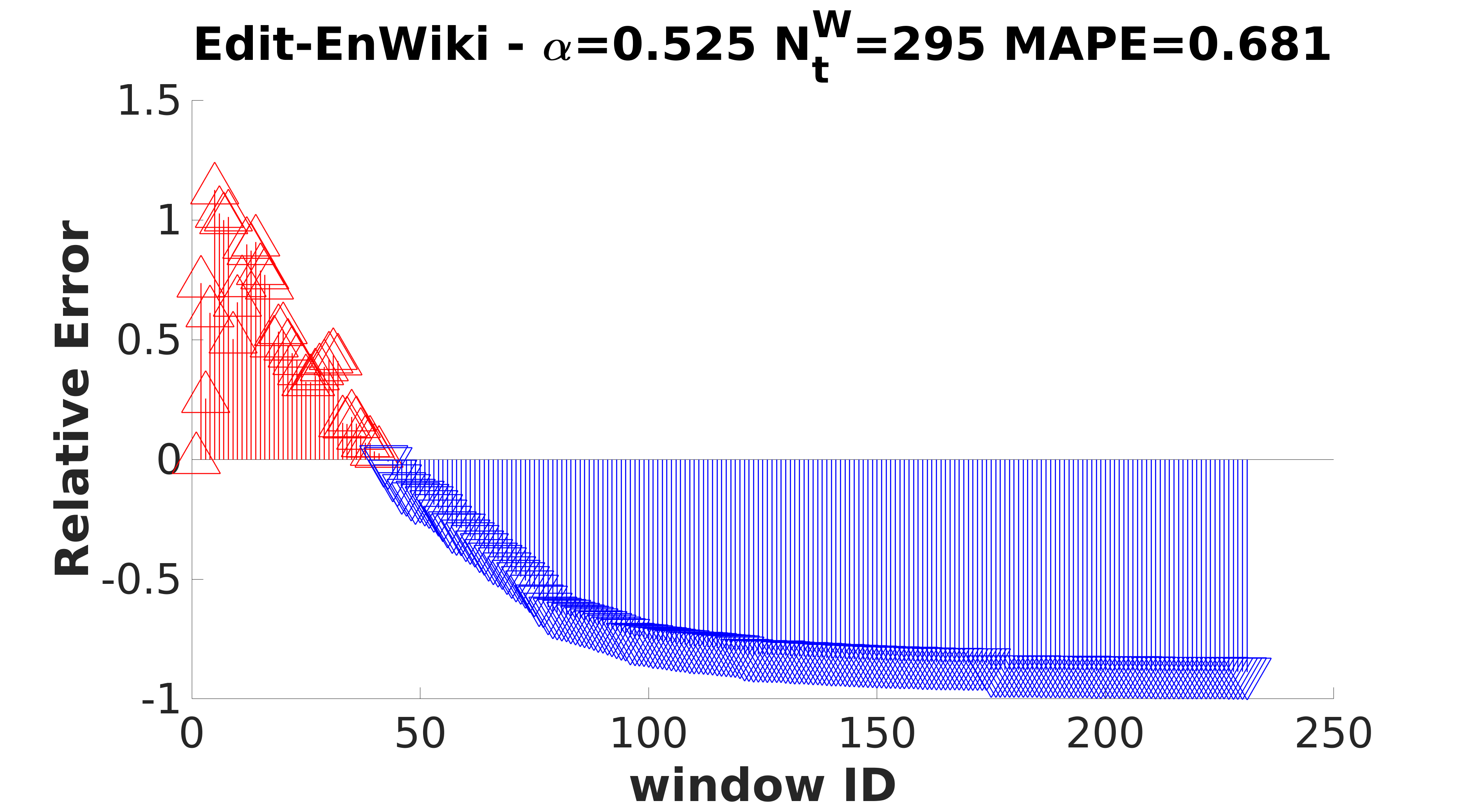}}
    \subfigure{\includegraphics[width=0.3\textwidth]{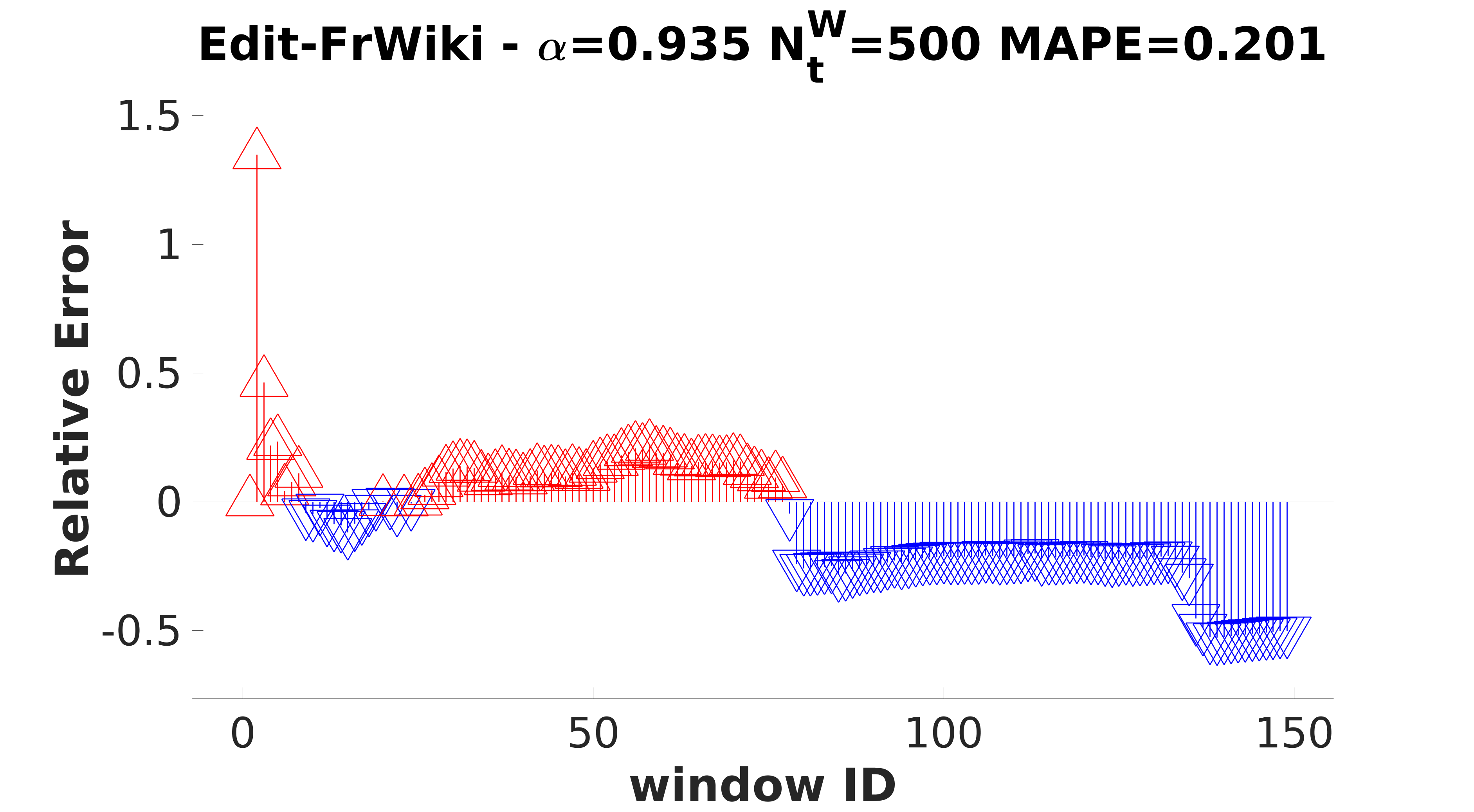}}
    \caption{Relative Error of sGrapp over windows for the best obtained MAPE.}\label{fig:realtiveerrors}
\end{figure*}
\begin{figure*}[h]\centering
    \subfigure{\includegraphics[width=0.3\textwidth]{ 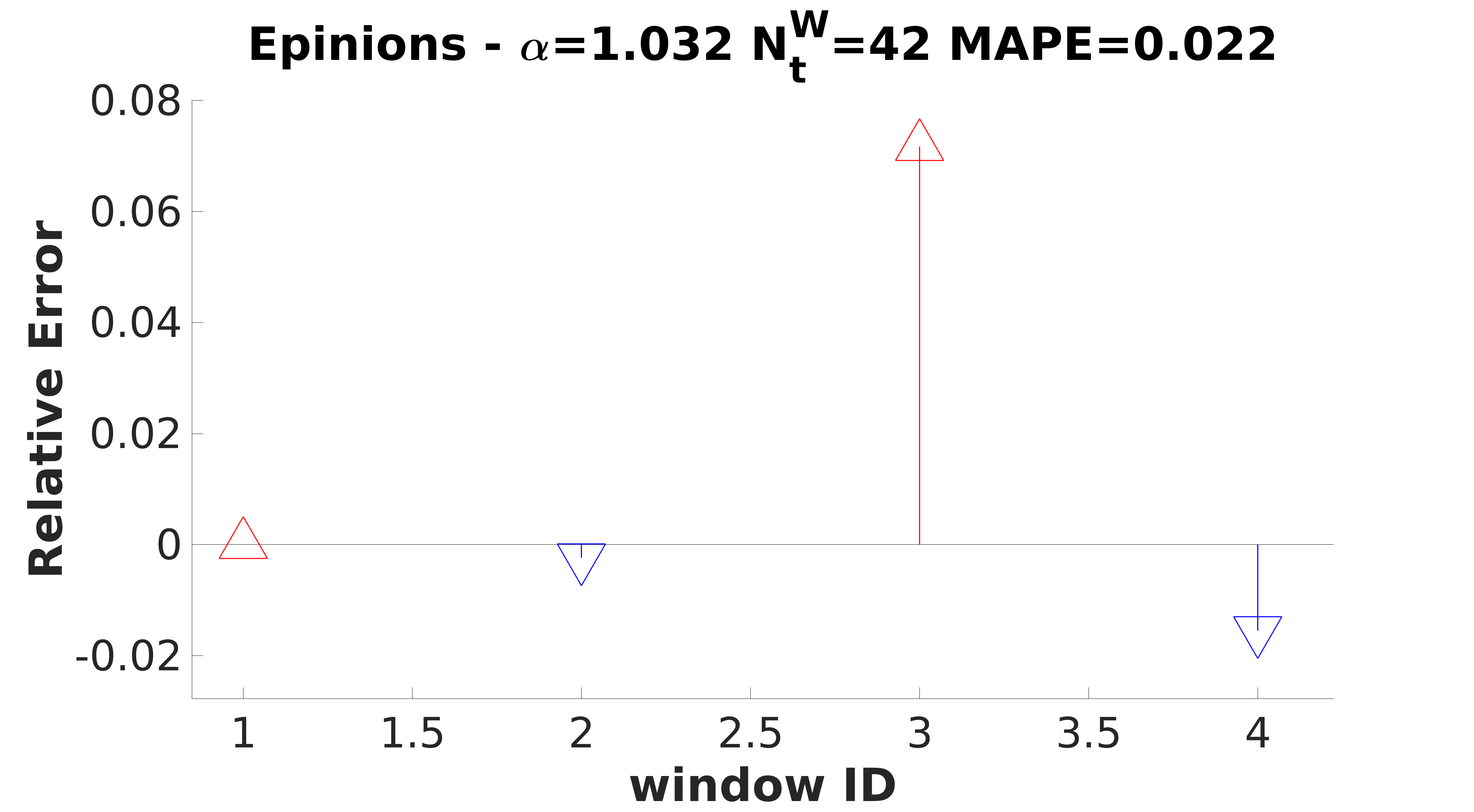}} \subfigure{\includegraphics[width=0.3\textwidth]{ 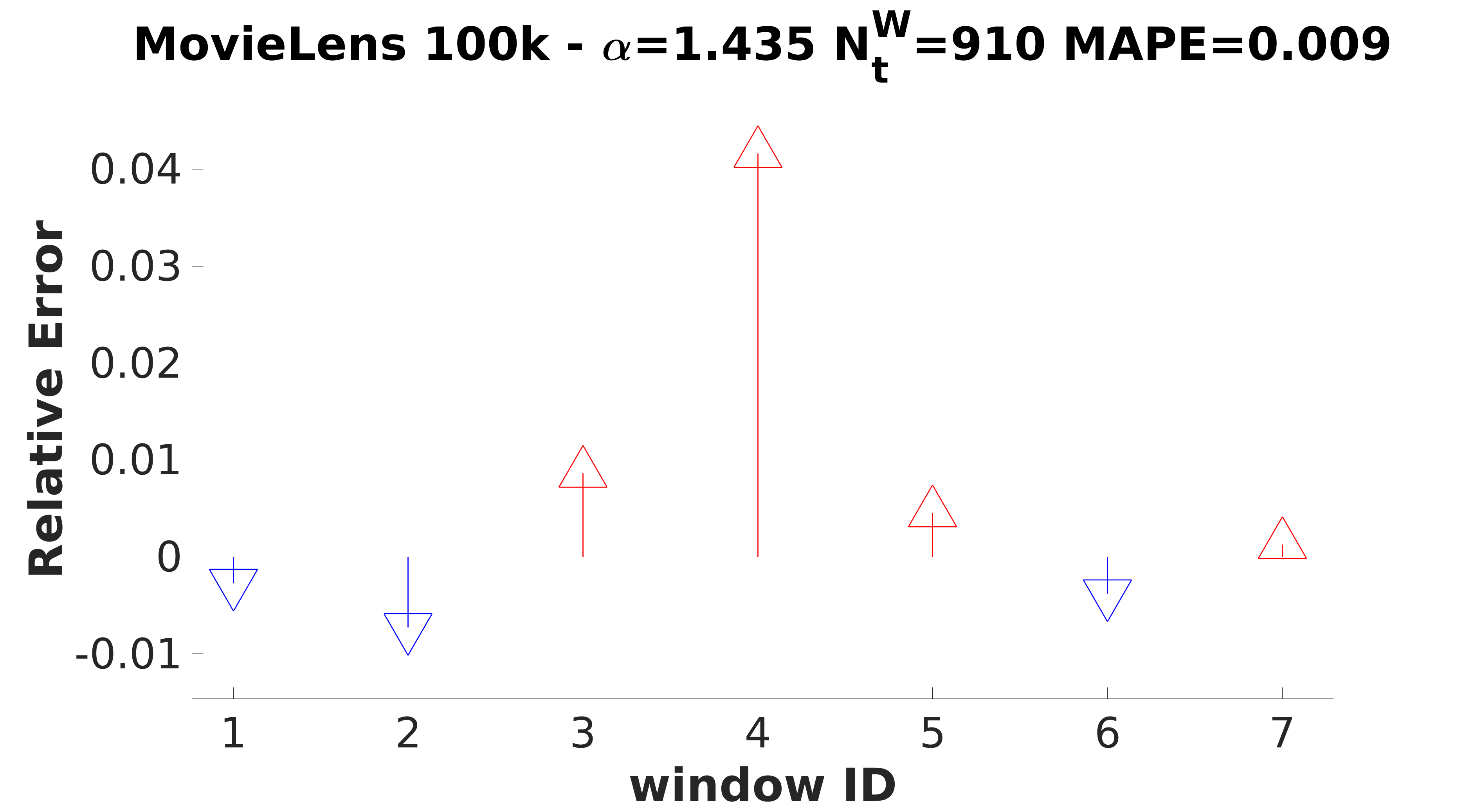}}
    \subfigure{\includegraphics[width=0.3\textwidth]{ 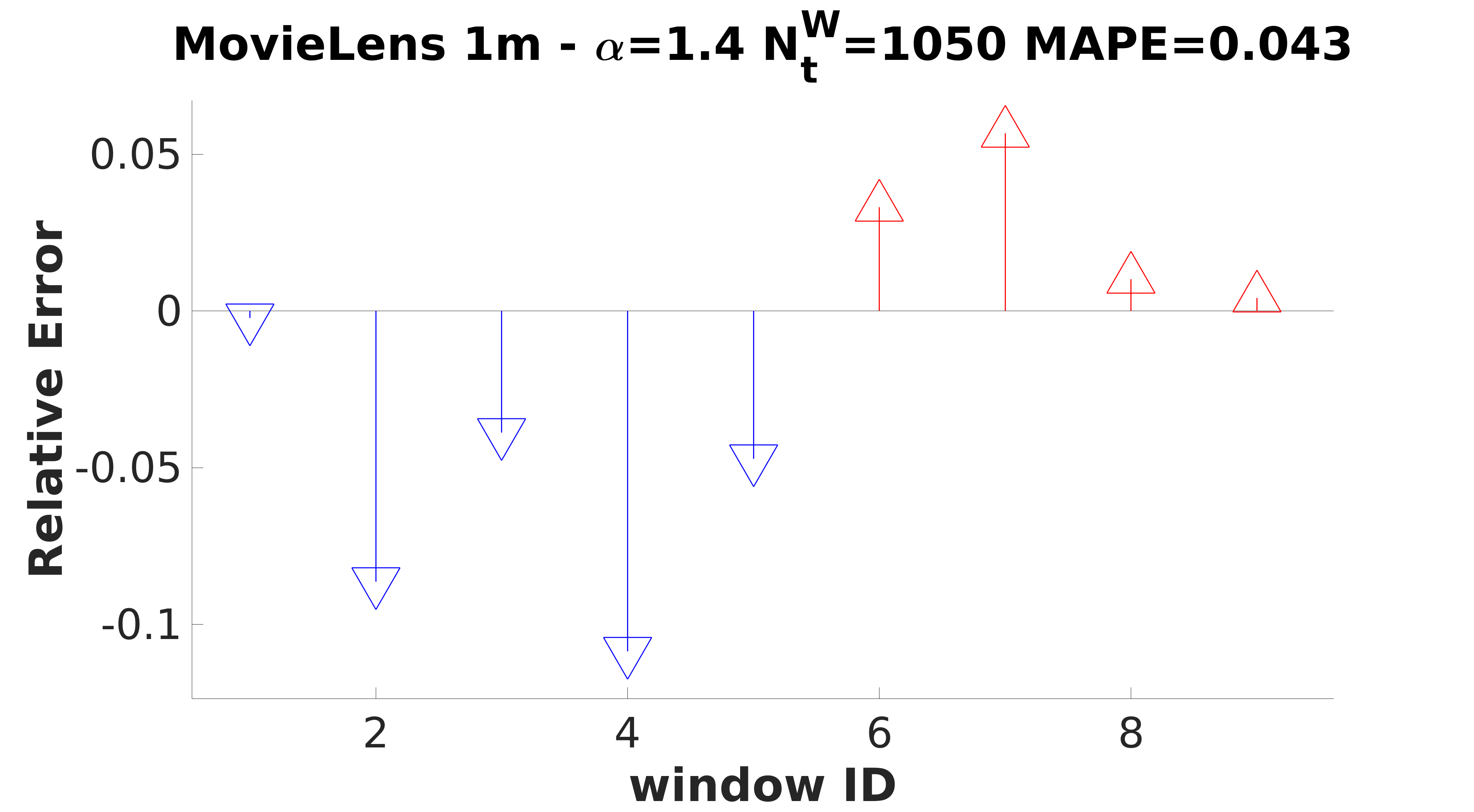}}
    \subfigure{\includegraphics[width=0.3\textwidth]{ 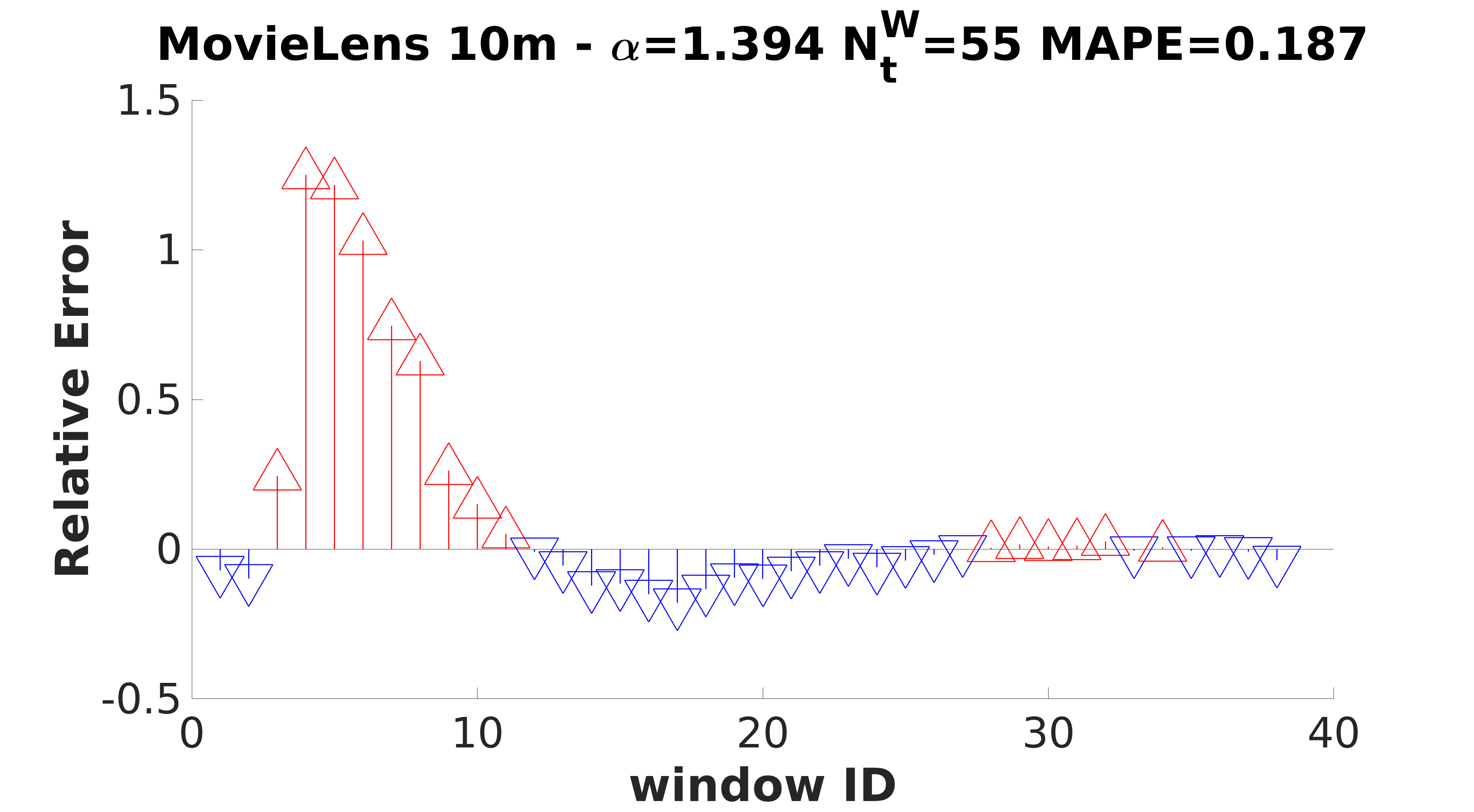}}
    \subfigure{\includegraphics[width=0.3\textwidth]{ 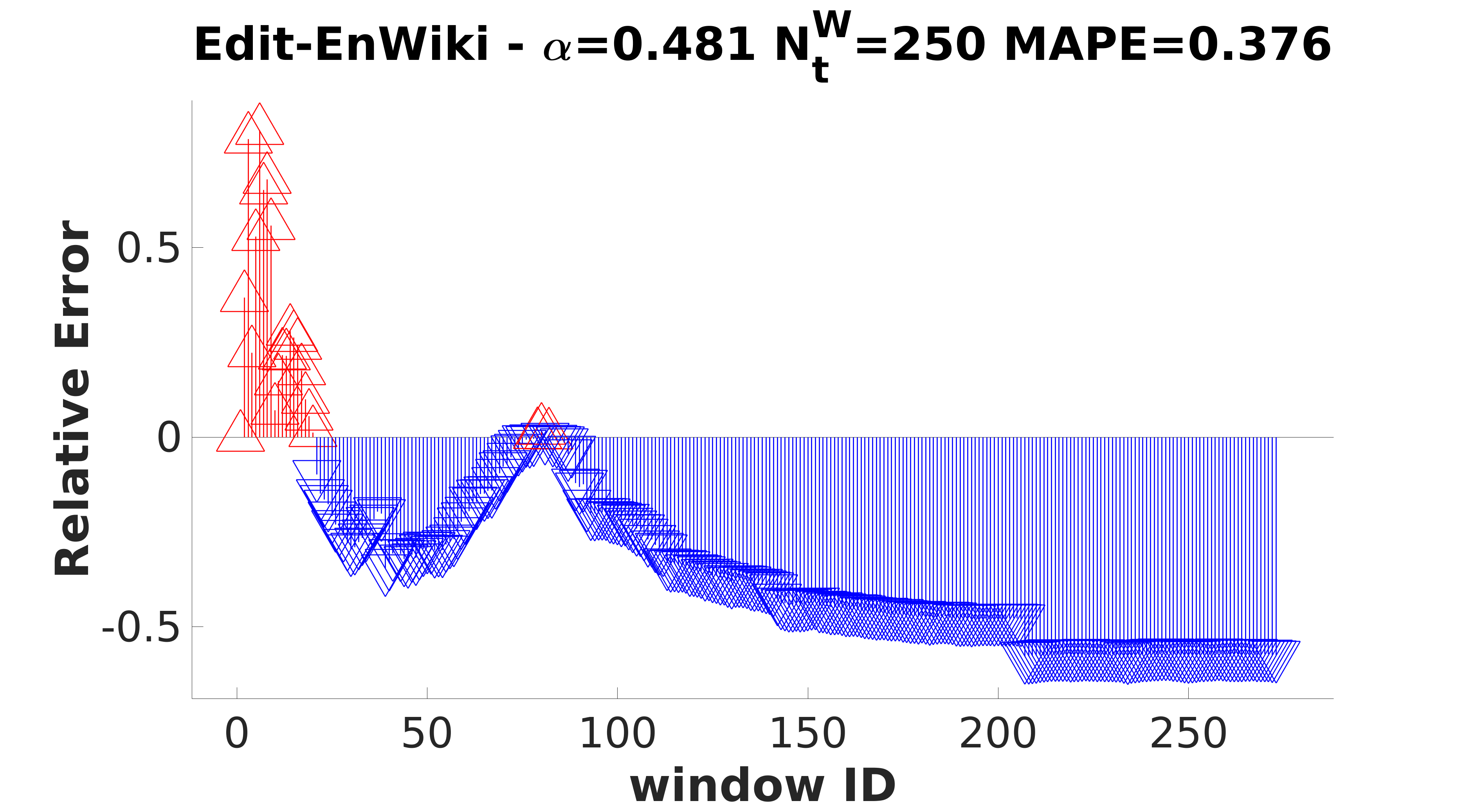}}
    \subfigure{\includegraphics[width=0.3\textwidth]{ 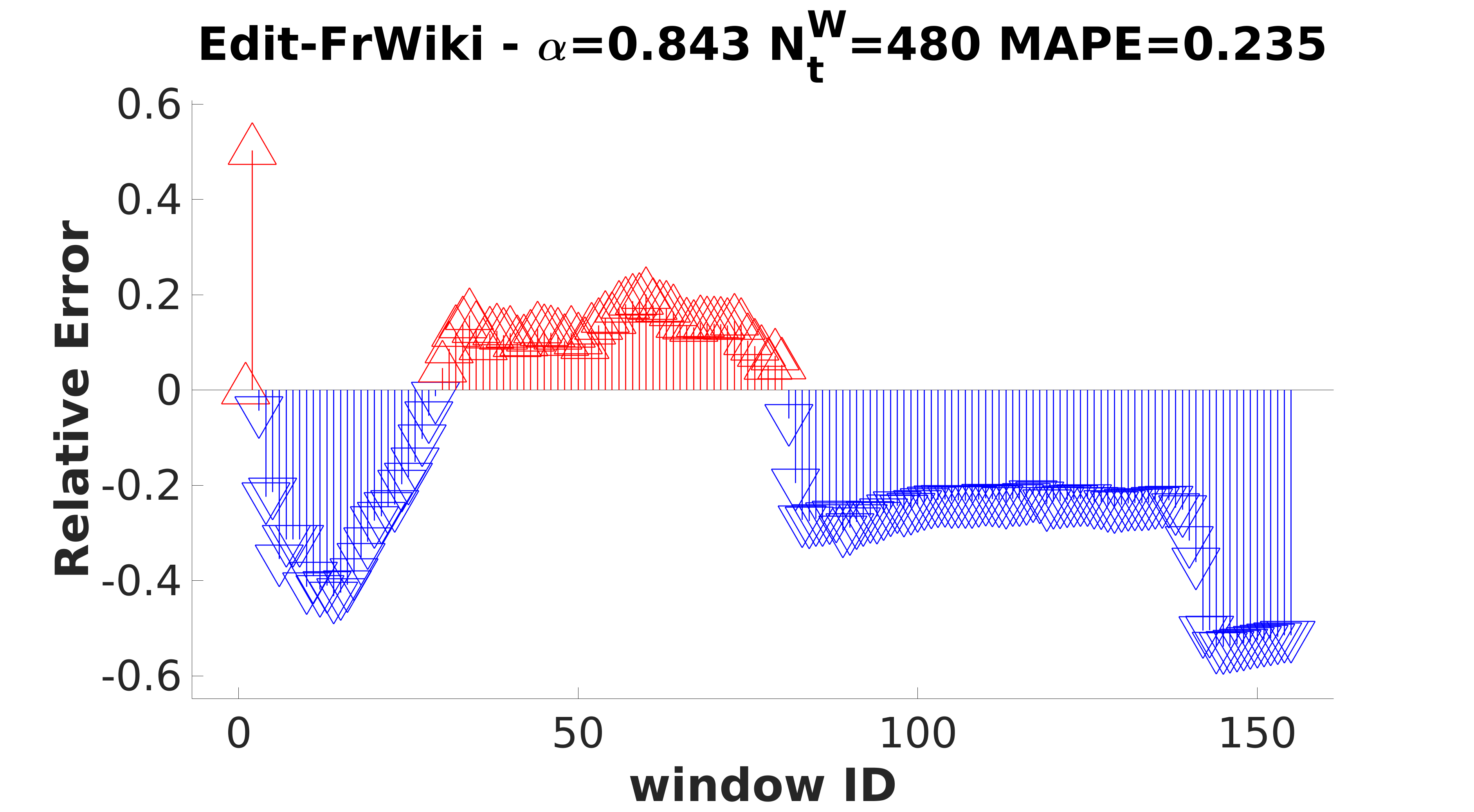}}
\caption{Relative Error of sGrapp-25 over windows for the best obtained MAPE.}\label{fig:realtiveerrors25}\end{figure*}
\begin{figure*}[h]\centering
   \subfigure{\includegraphics[width=0.3\textwidth]{ 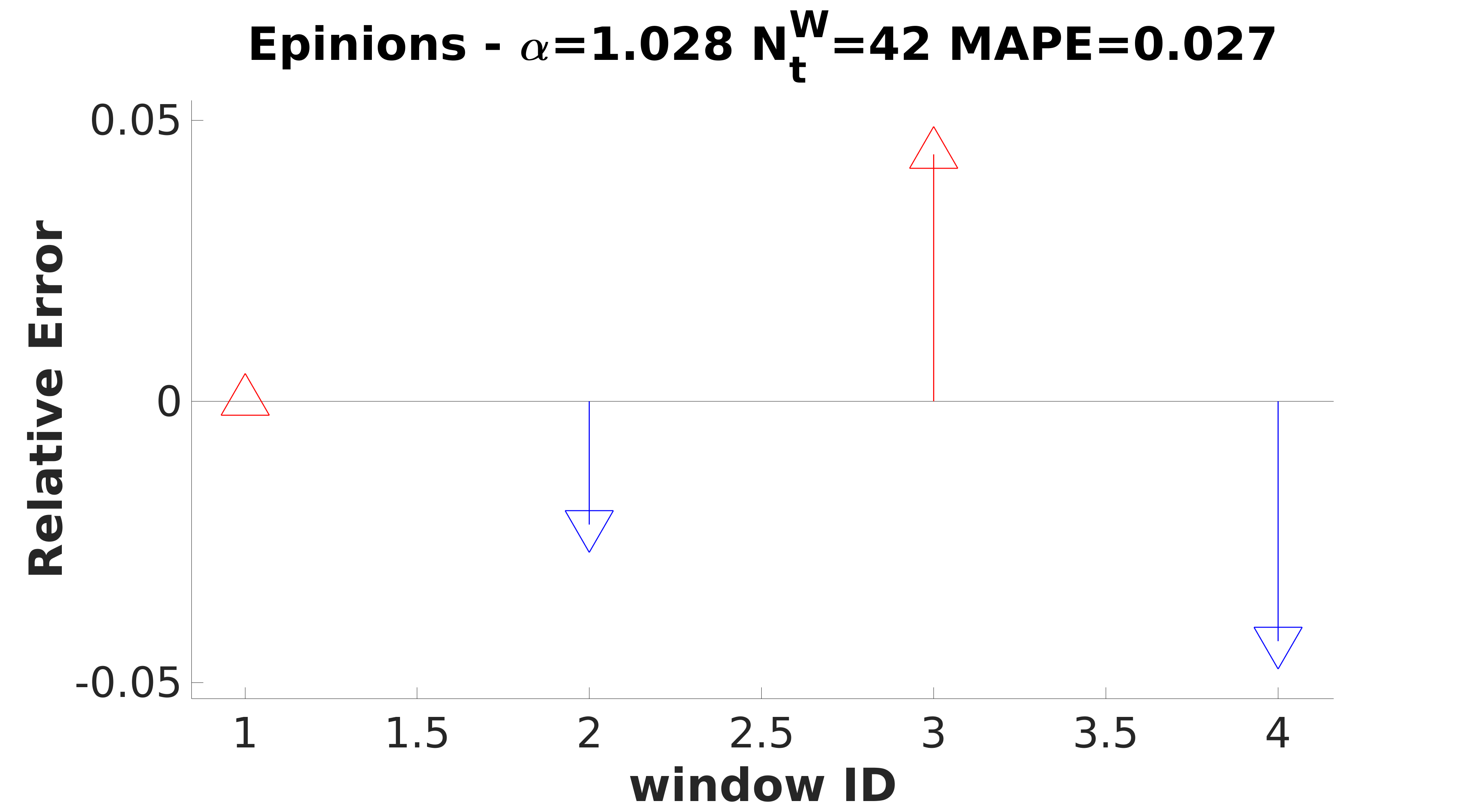}} \subfigure{\includegraphics[width=0.3\textwidth]{ 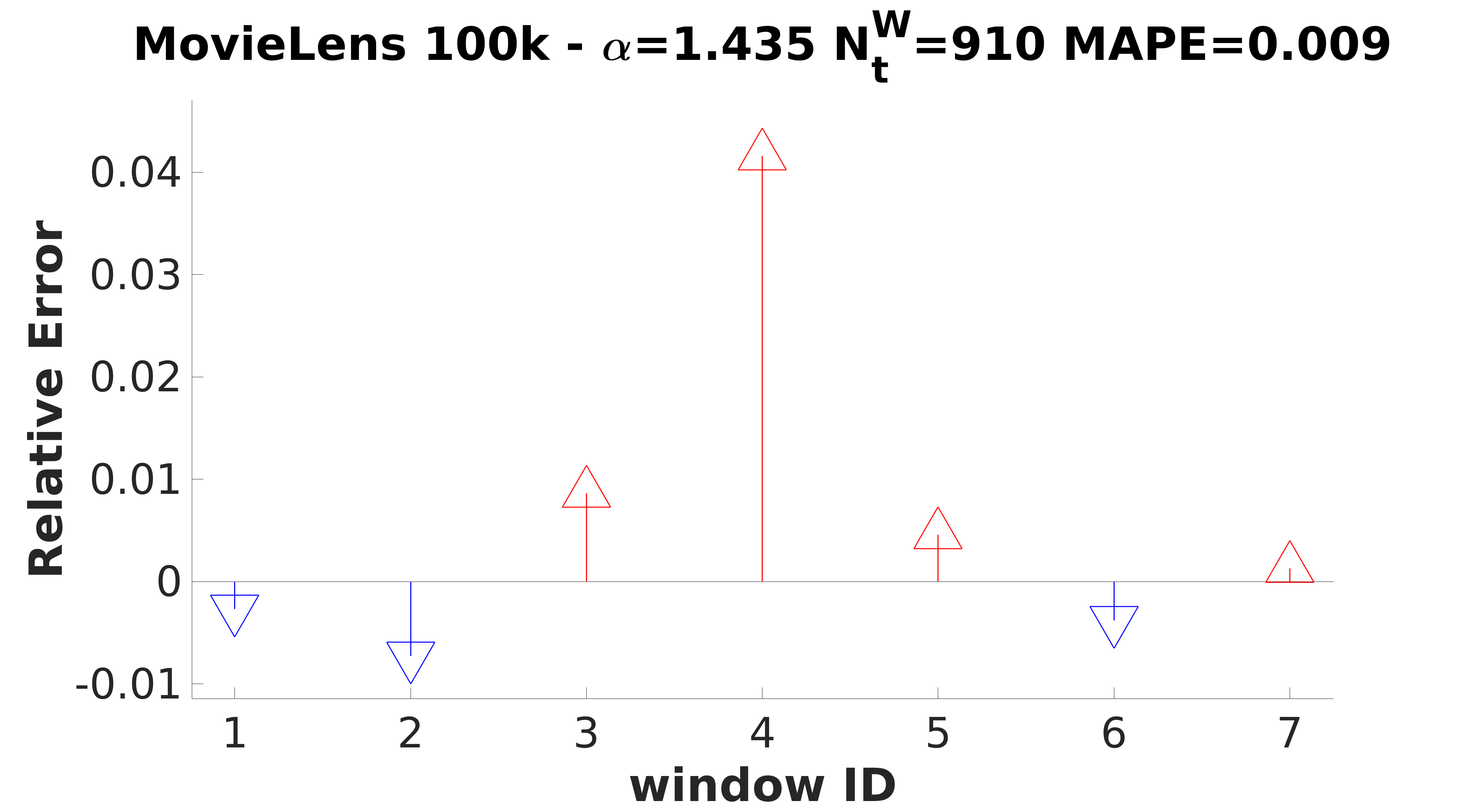}}
    \subfigure{\includegraphics[width=0.3\textwidth]{ 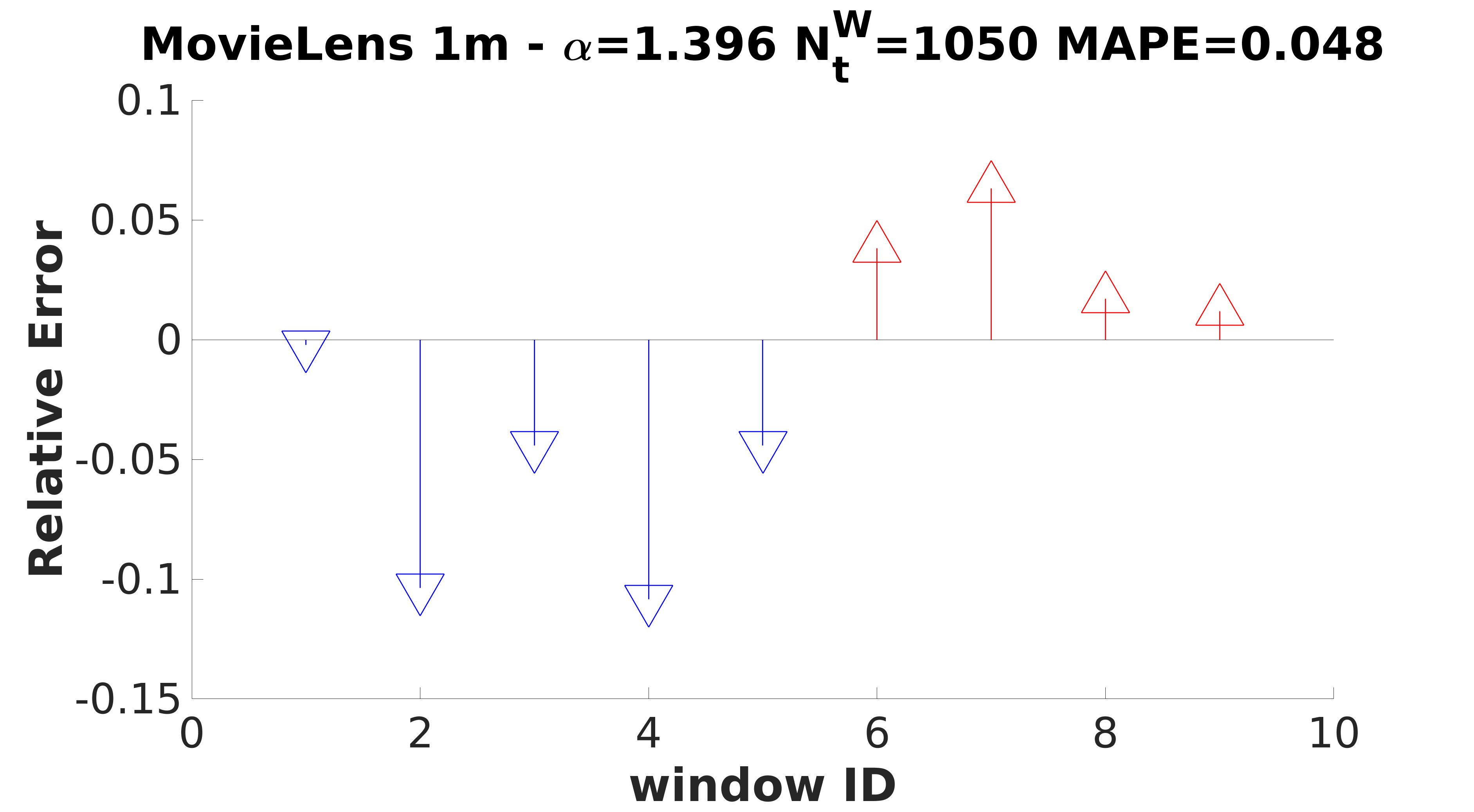}}
    \subfigure{\includegraphics[width=0.3\textwidth]{ 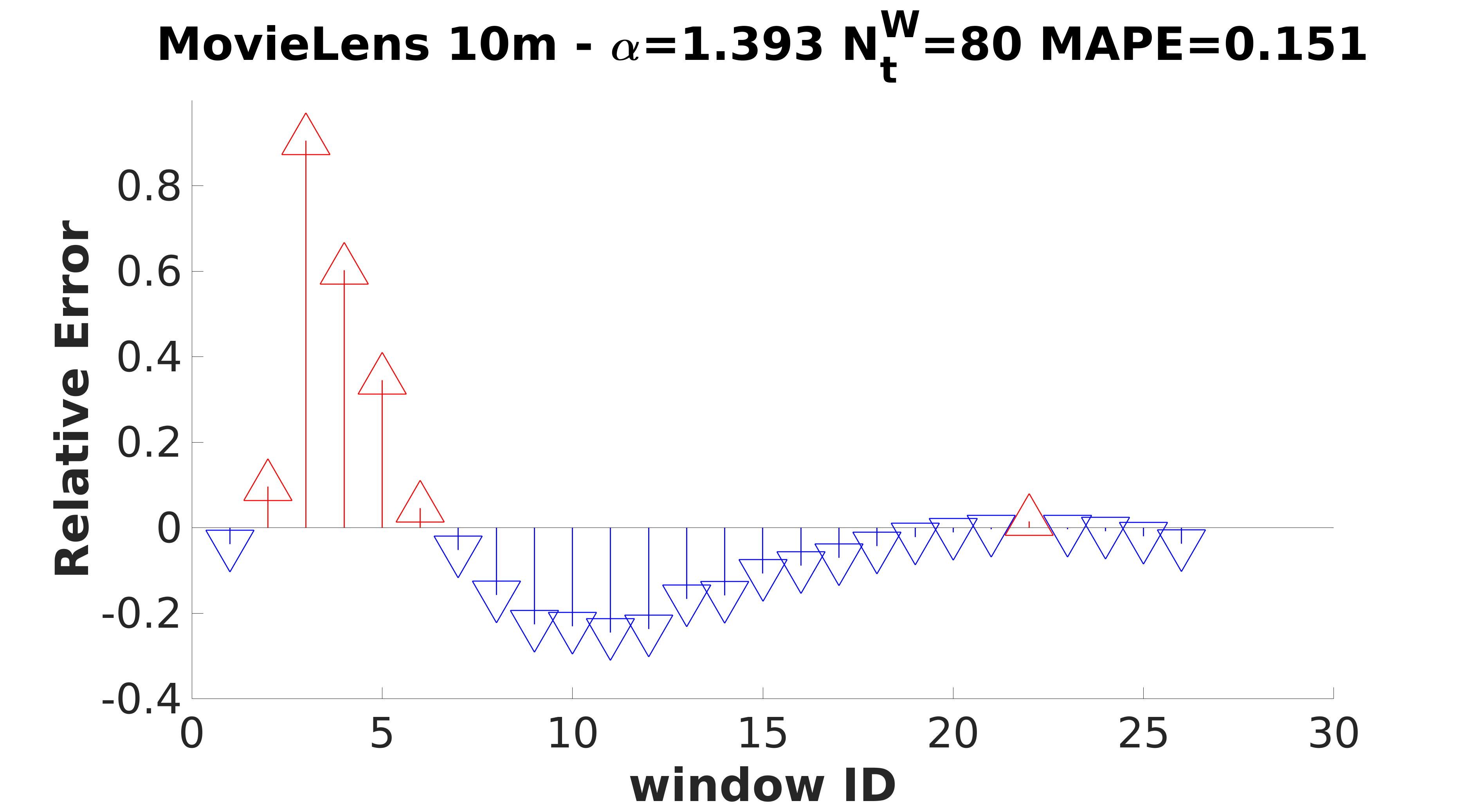}}
    \subfigure{\includegraphics[width=0.3\textwidth]{ 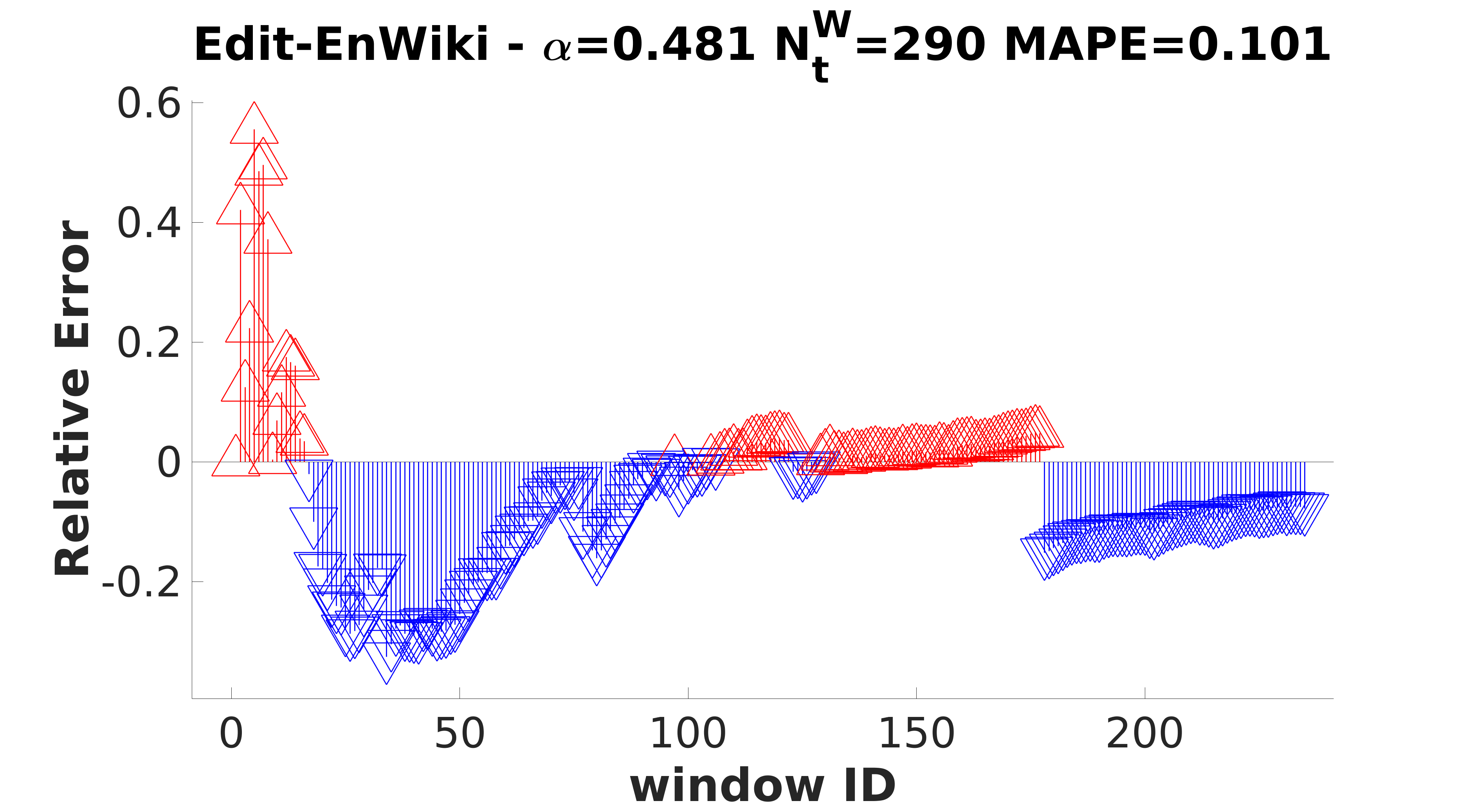}}
    \subfigure{\includegraphics[width=0.3\textwidth]{ 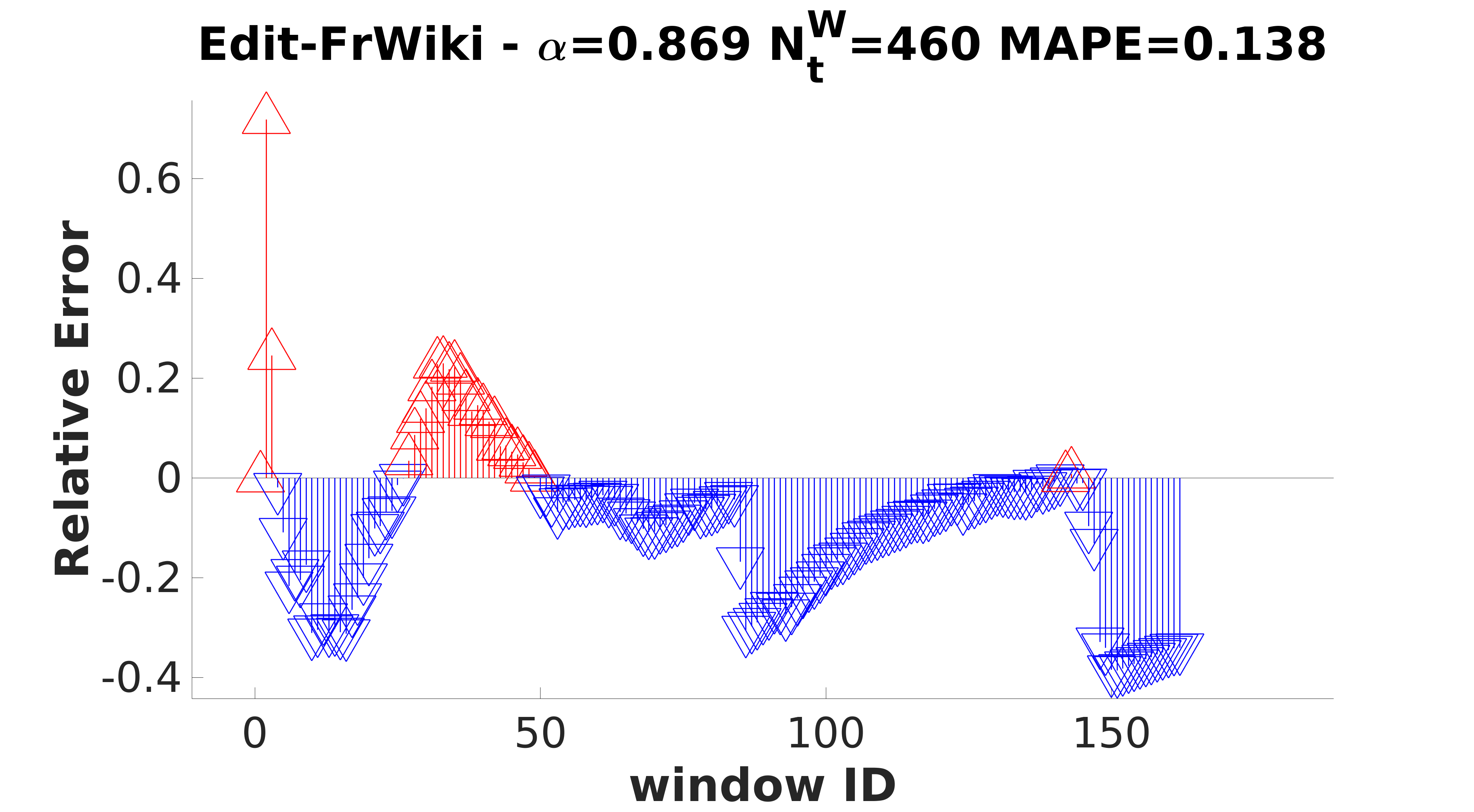}}
 \caption{Relative Error of sGrapp-50 over windows for the best obtained MAPE.}\label{fig:realtiveerrors50}\end{figure*}
\begin{figure*}[h]\centering
    \subfigure{\includegraphics[width=0.3\textwidth]{ 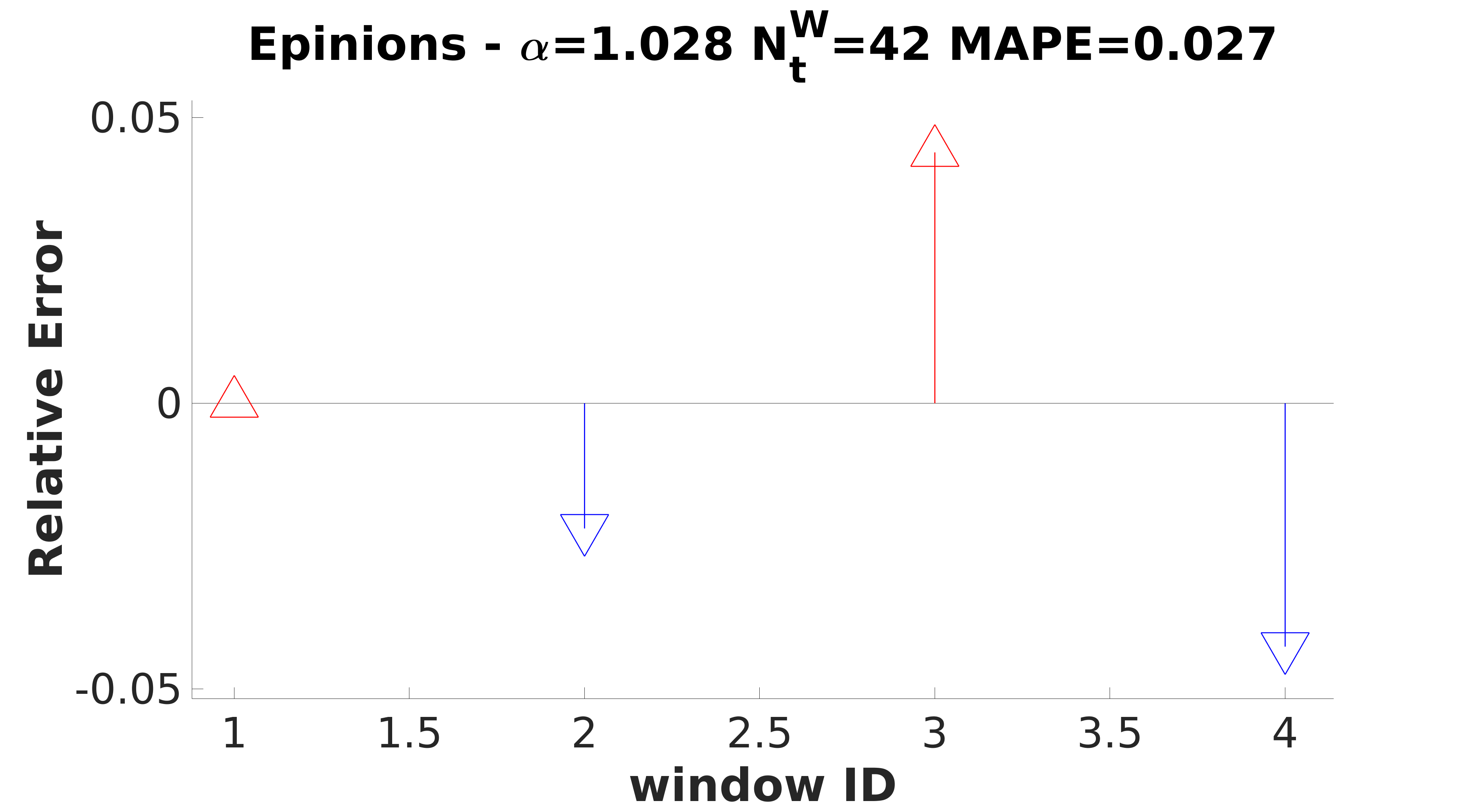}} \subfigure{\includegraphics[width=0.3\textwidth]{ 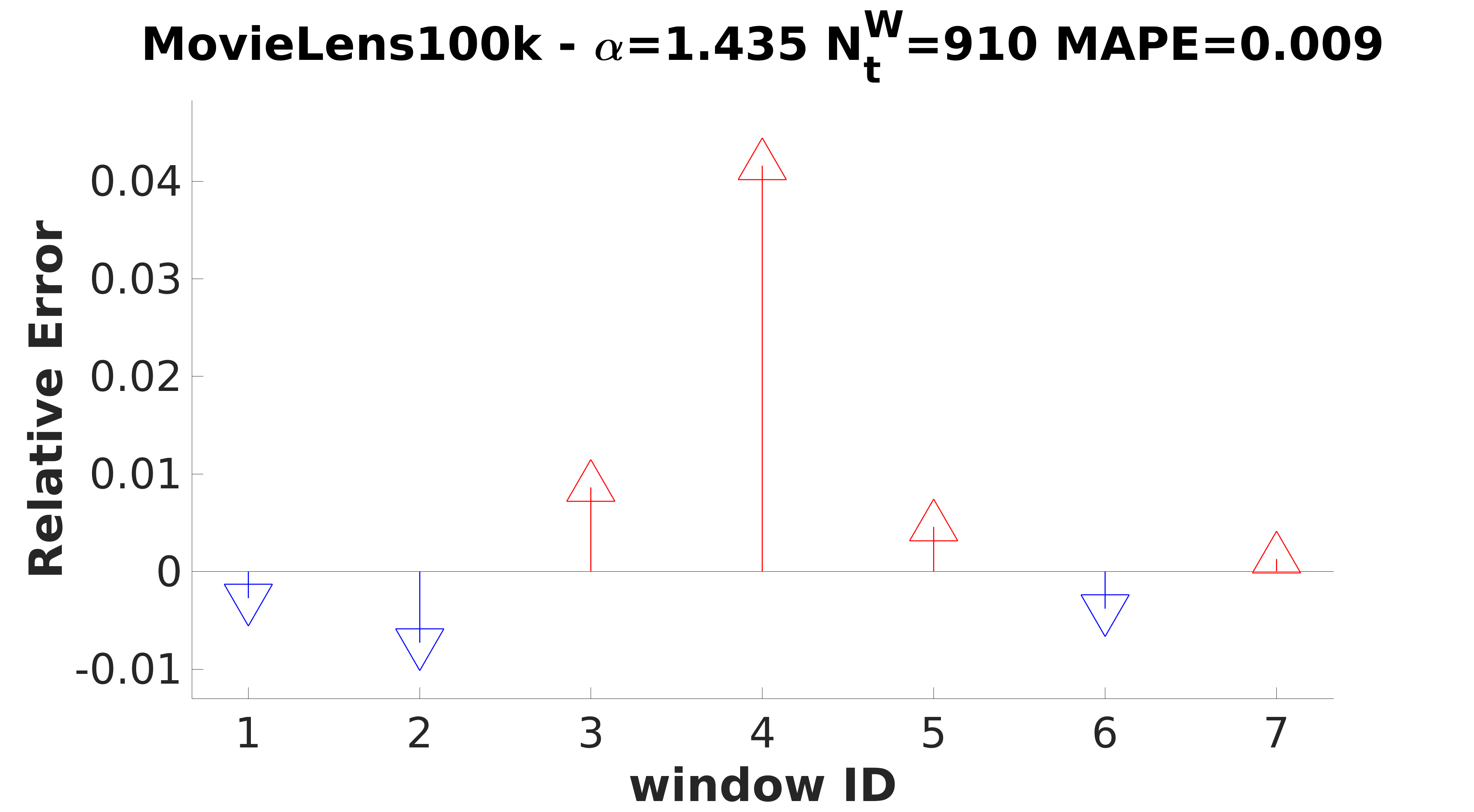}}
    \subfigure{\includegraphics[width=0.3\textwidth]{ 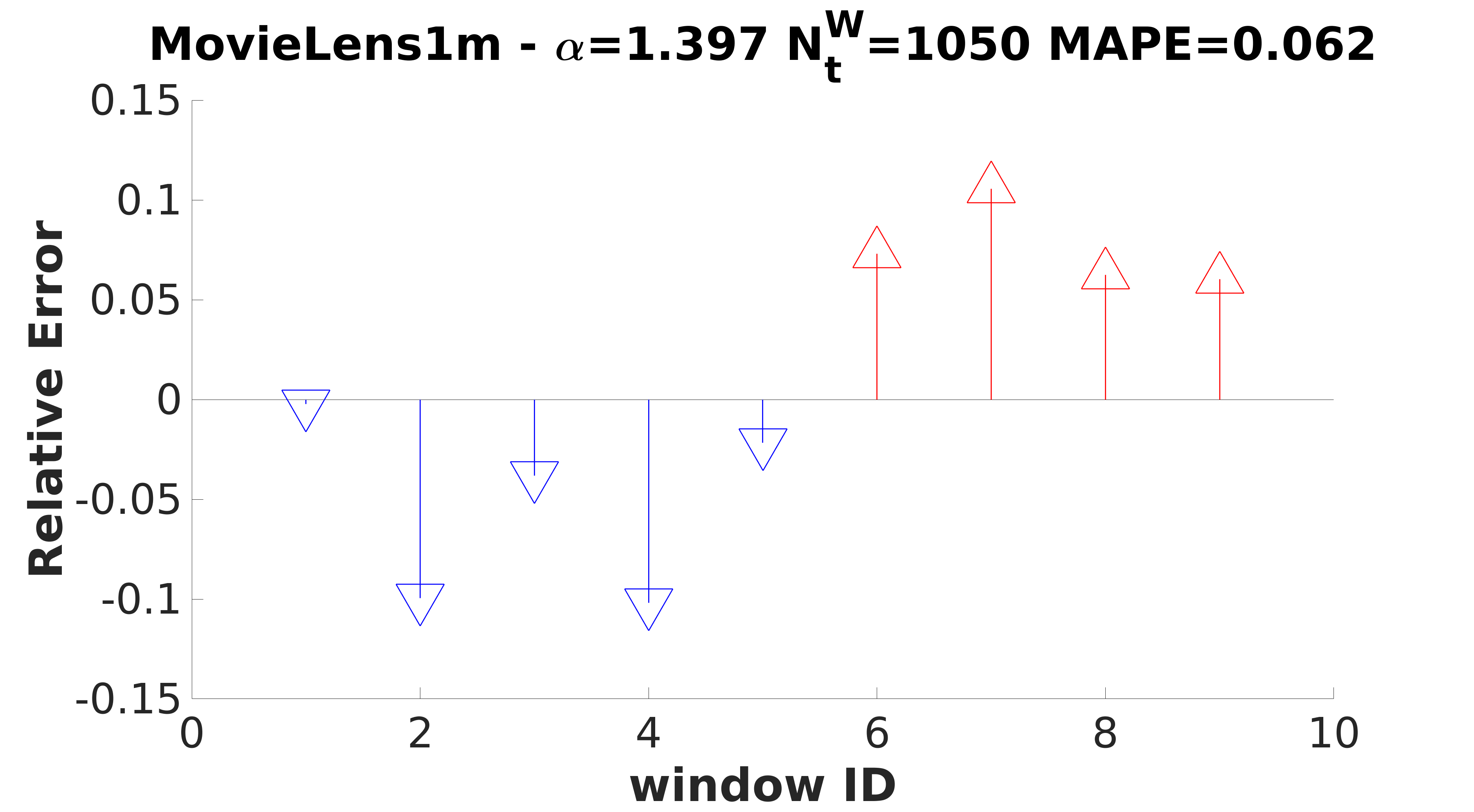}}
    \subfigure{\includegraphics[width=0.3\textwidth]{ 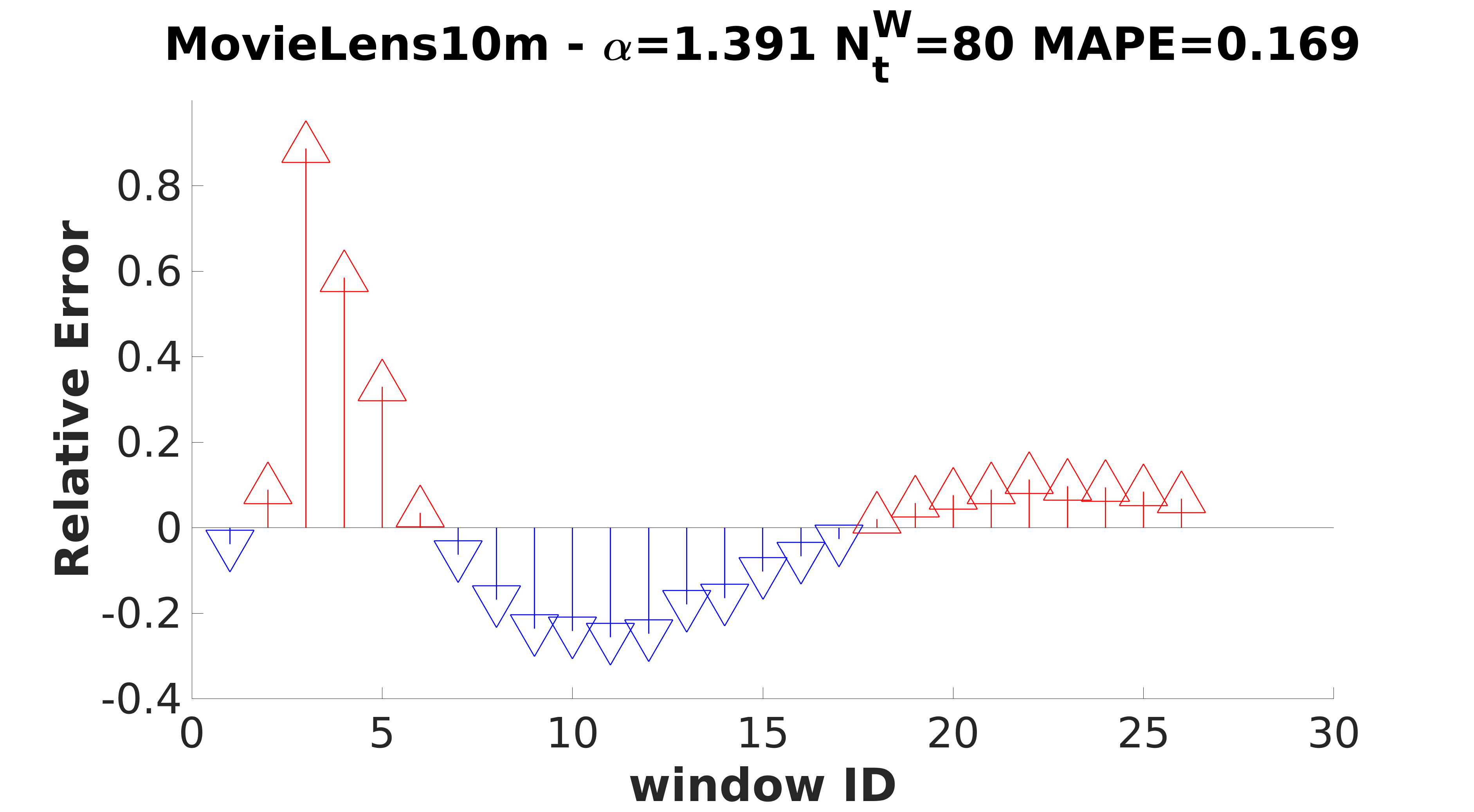}}
    \subfigure{\includegraphics[width=0.3\textwidth]{ 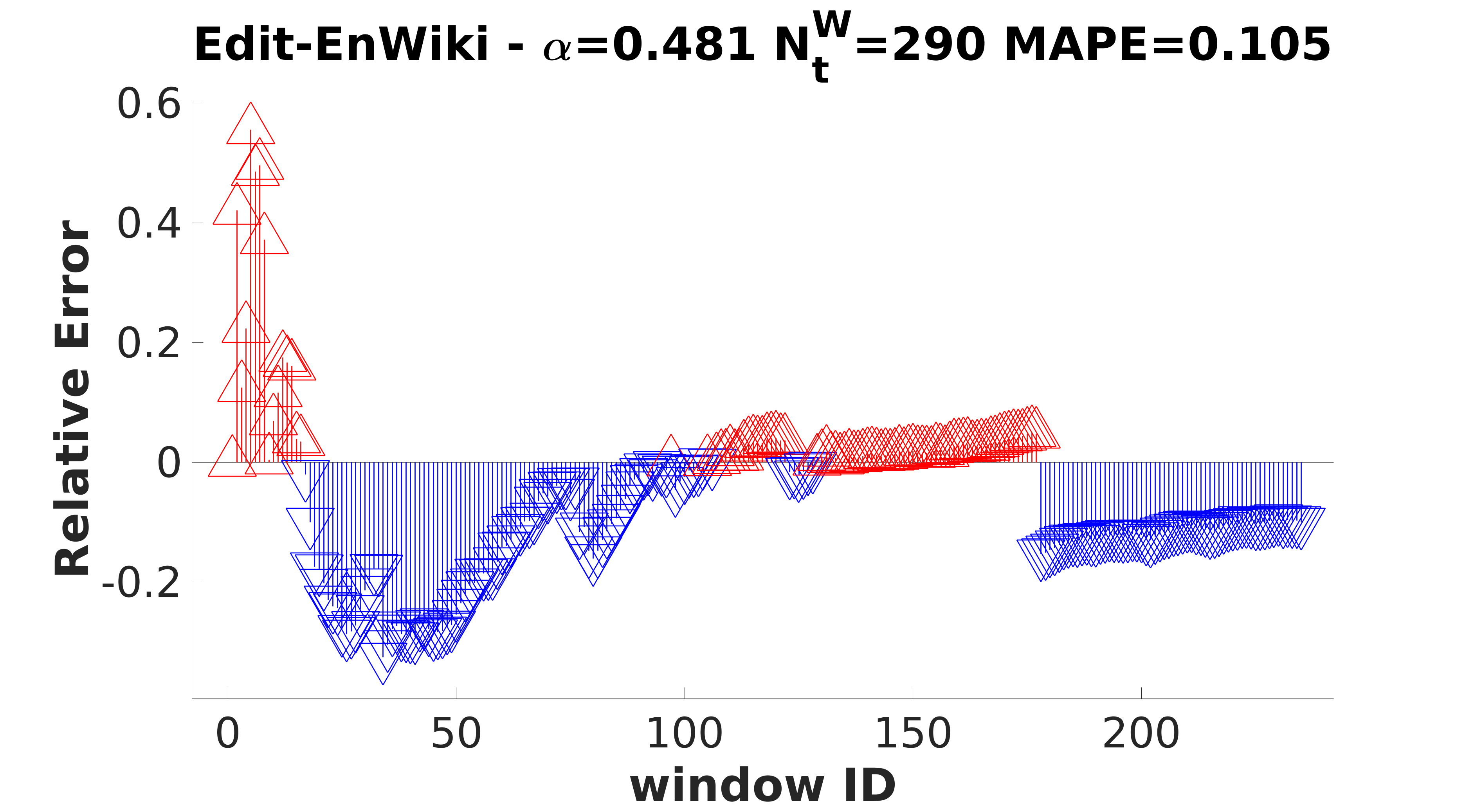}}
    \subfigure{\includegraphics[width=0.3\textwidth]{ 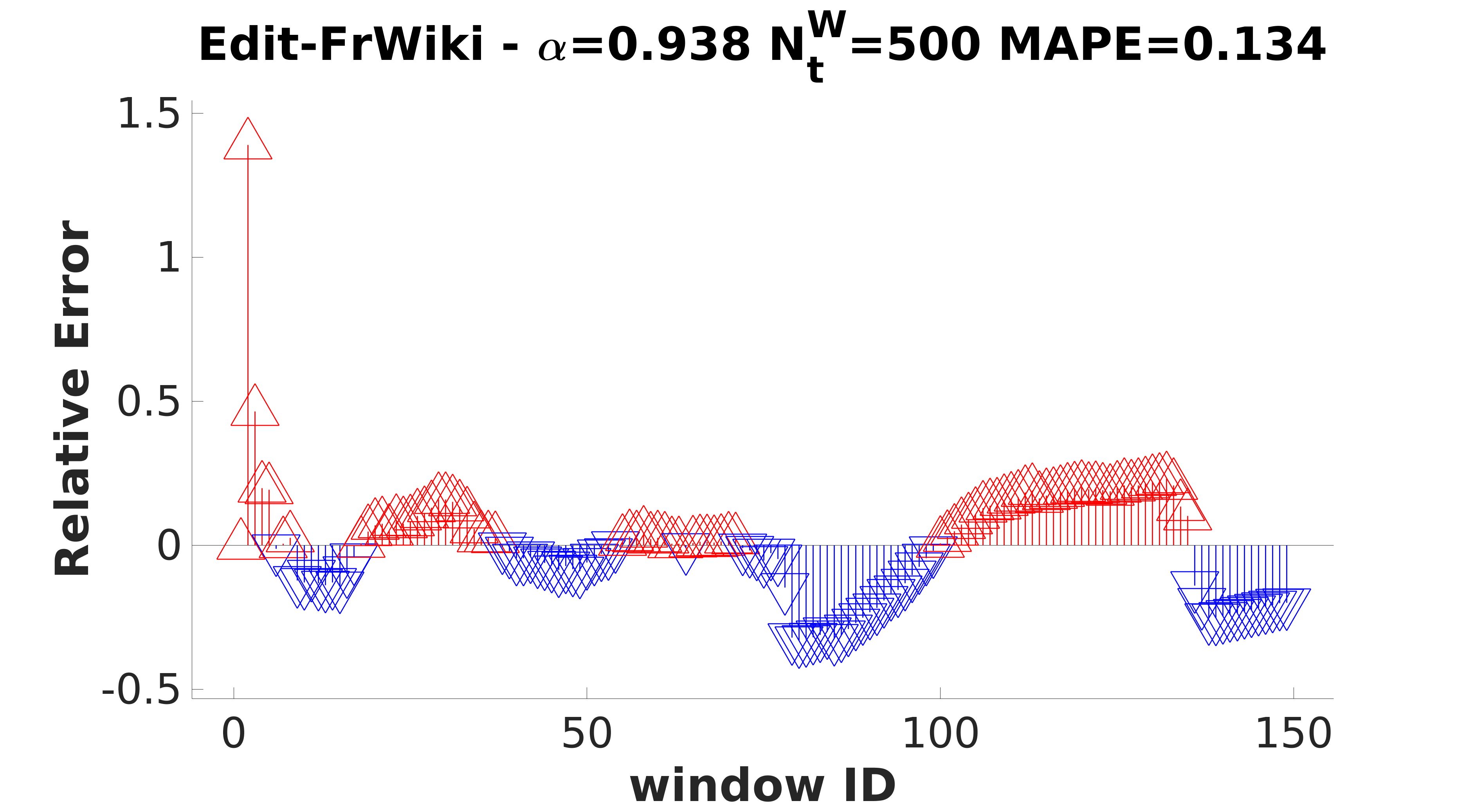}}
\caption{Relative Error of sGrapp-75 over windows for the best obtained MAPE.}\label{fig:realtiveerrors75}\end{figure*}
\begin{figure*}[h]\centering
    \subfigure{\includegraphics[width=0.3\textwidth]{ 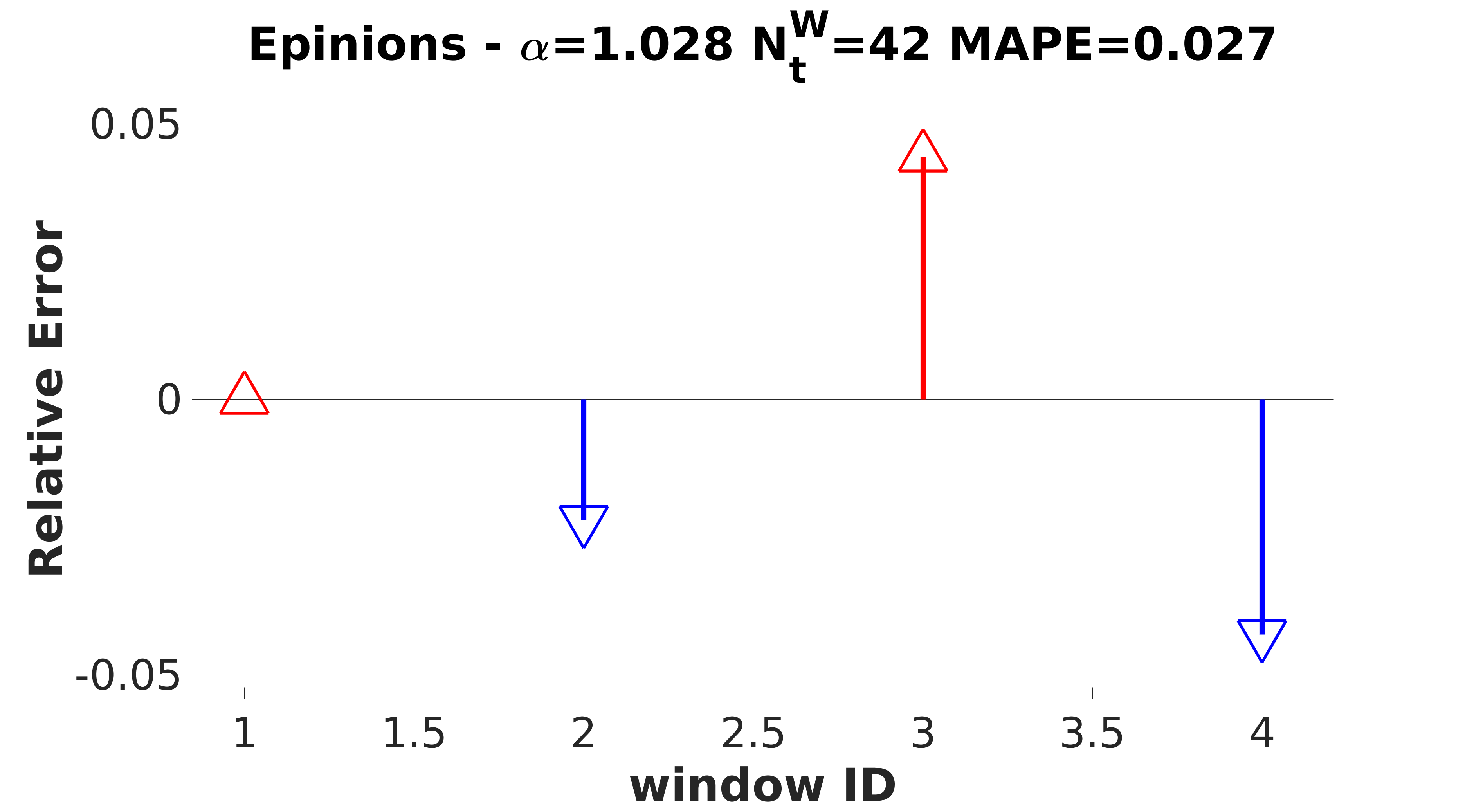}} \subfigure{\includegraphics[width=0.3\textwidth]{ 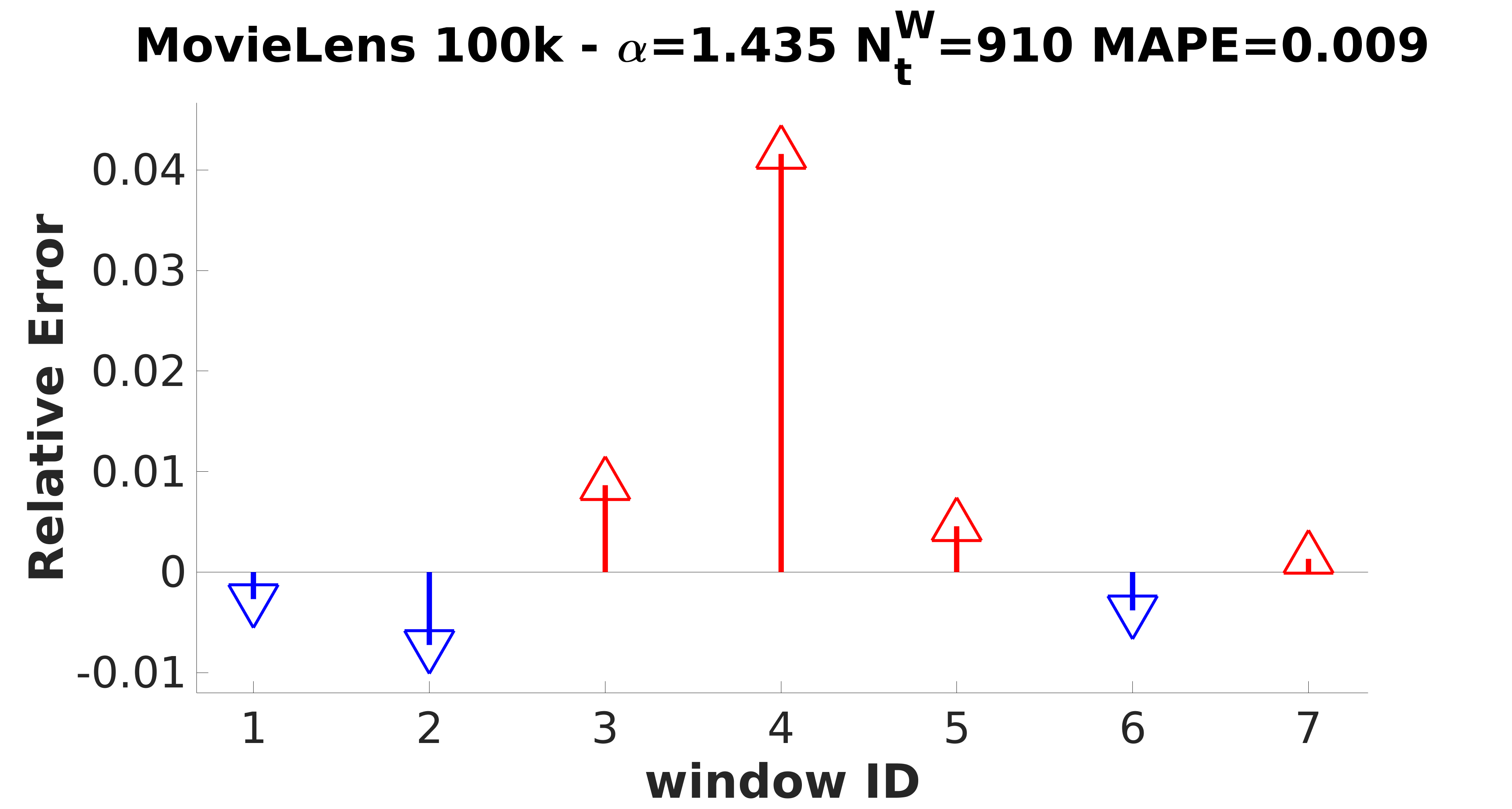}}
    \subfigure{\includegraphics[width=0.3\textwidth]{ 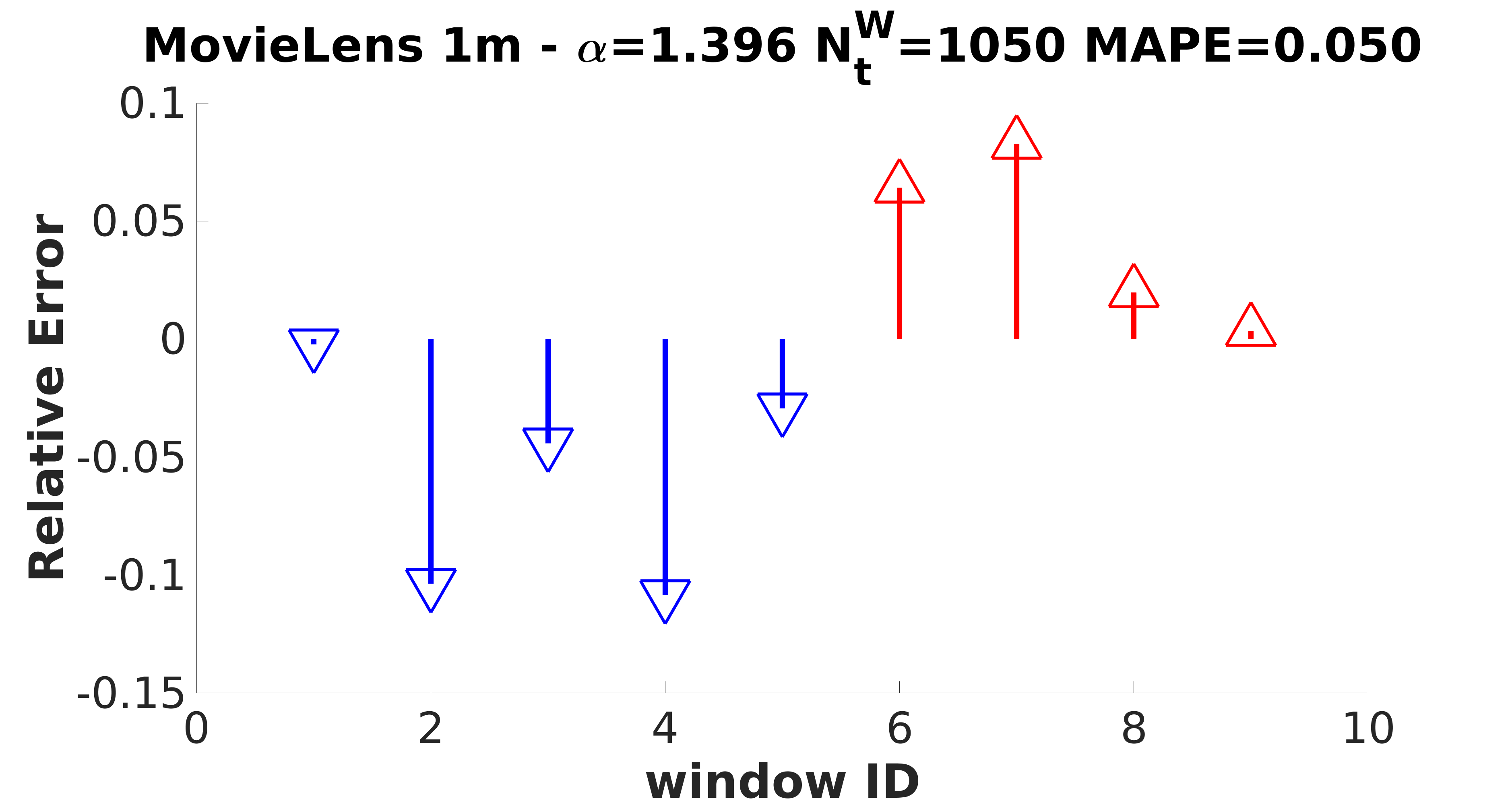}}
    \subfigure{\includegraphics[width=0.3\textwidth]{ 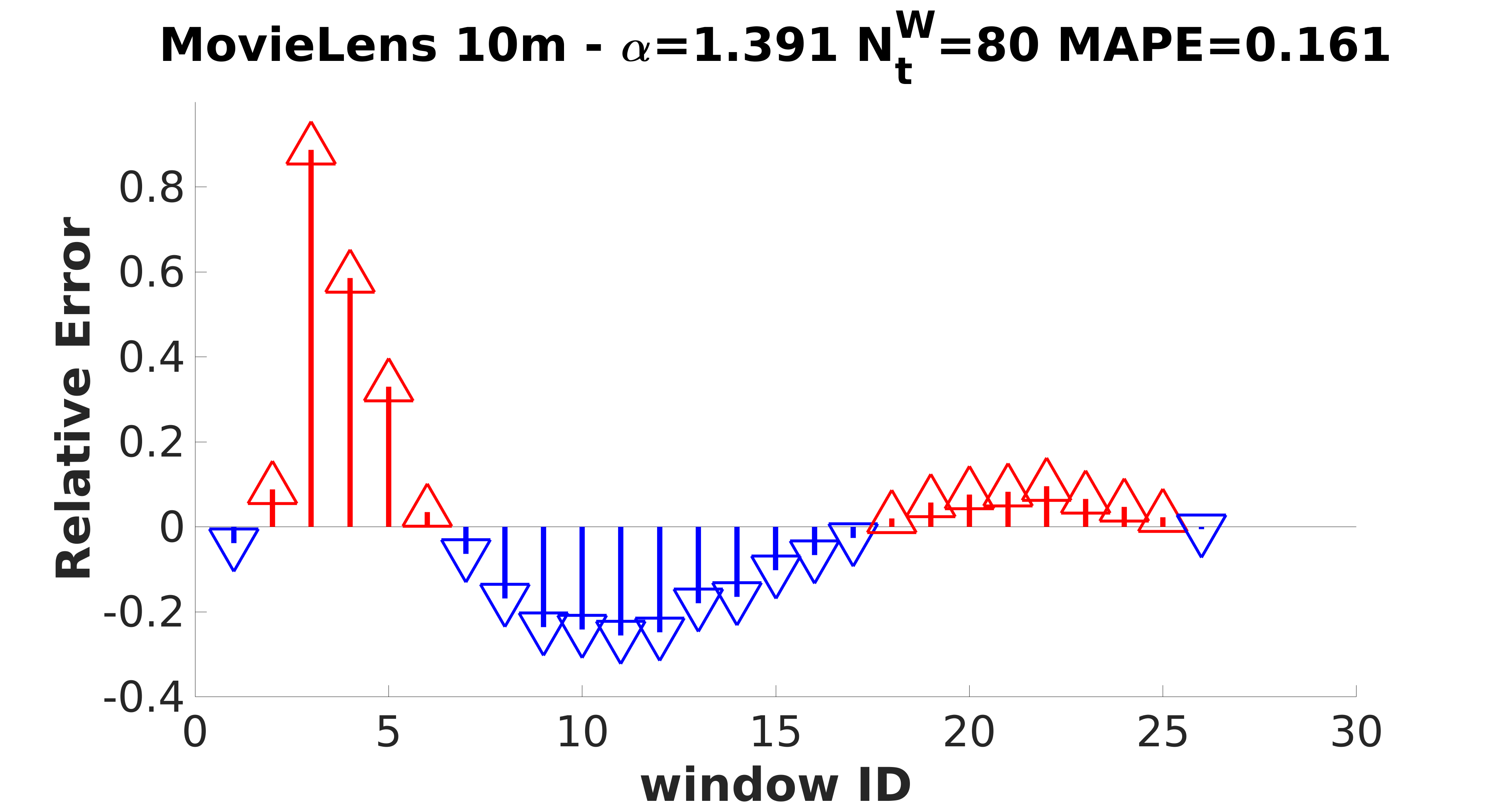}}
    \subfigure{\includegraphics[width=0.3\textwidth]{ 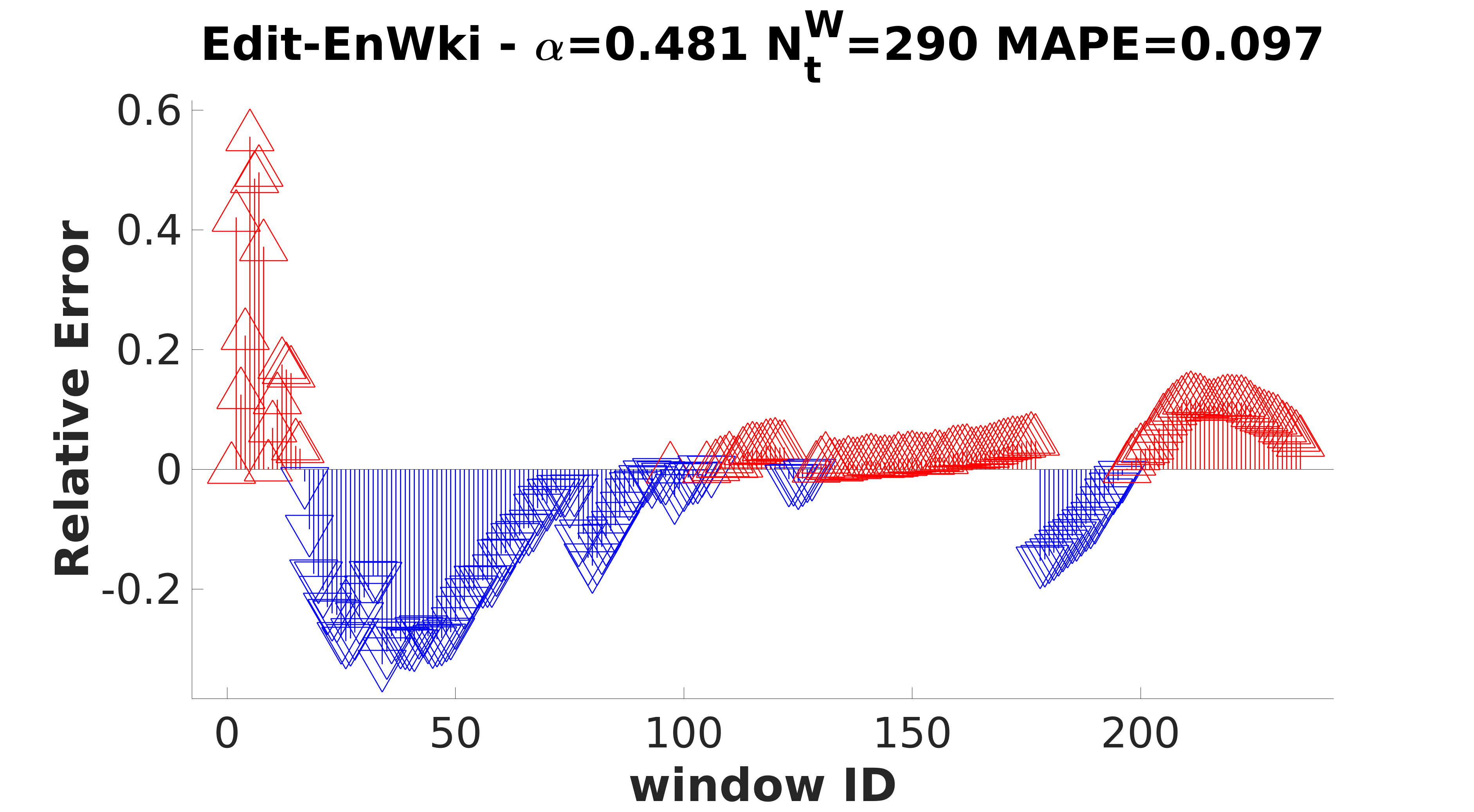}}
   \subfigure{\includegraphics[width=0.3\textwidth]{ 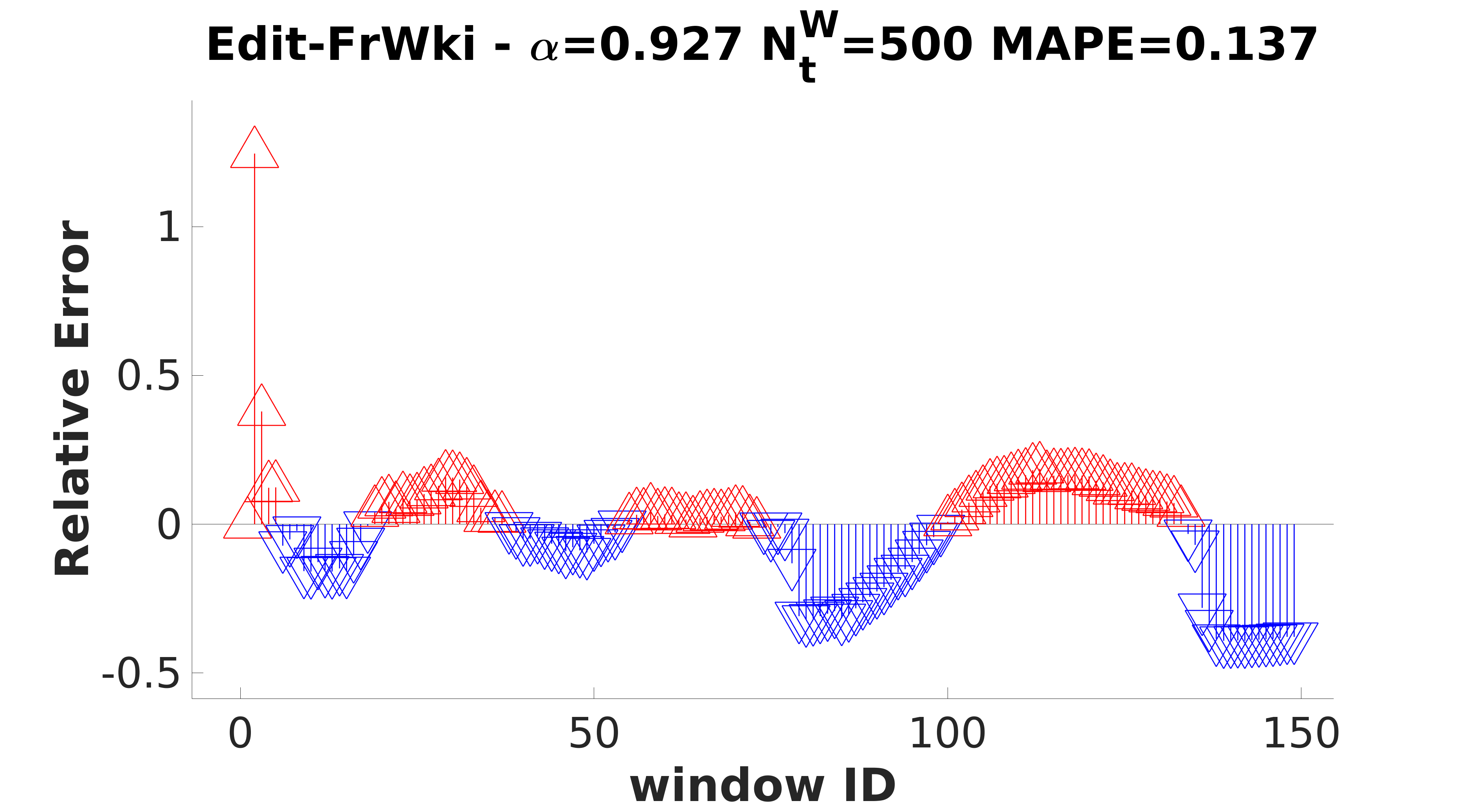}}
 \caption{Relative Error of sGrapp-100 over windows for the best obtained MAPE.}\label{fig:realtiveerrors100}\end{figure*}
\begin{figure*}[h]\centering
    \includegraphics[width=0.9\textwidth]{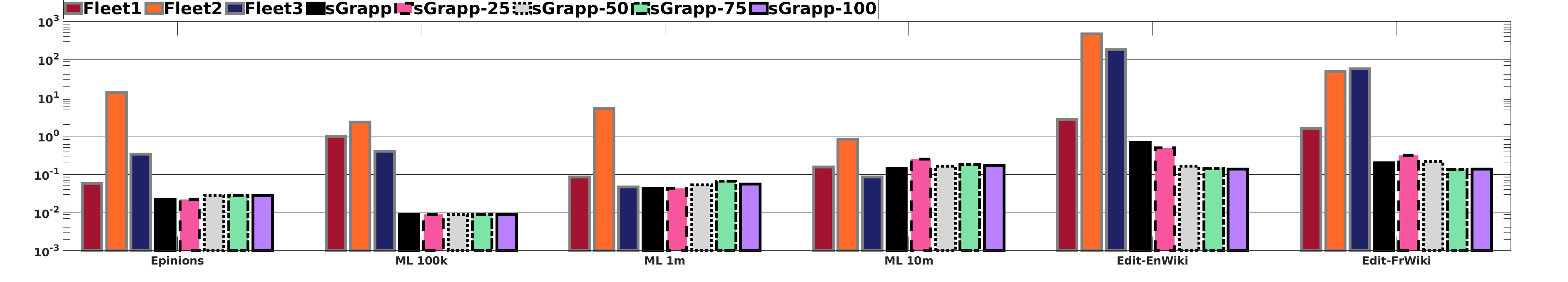}
\caption{MAPE of different algorithms.}
   \label{fig:mapecomparison}
\end{figure*}
\begin{figure*}[h]\centering
    \subfigure{\includegraphics[width=0.3\textwidth]{ 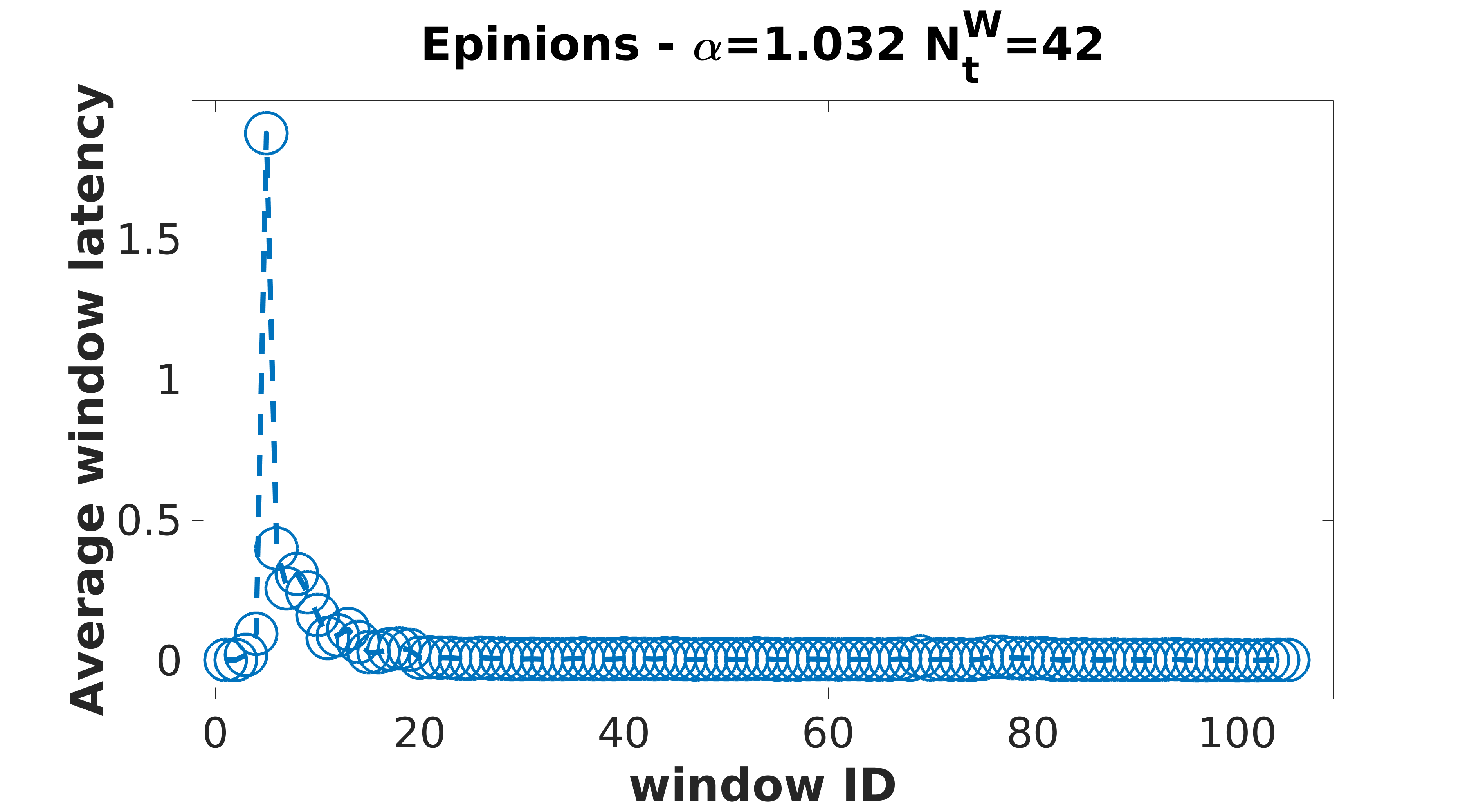}} 
    \subfigure{\includegraphics[width=0.3\textwidth]{ 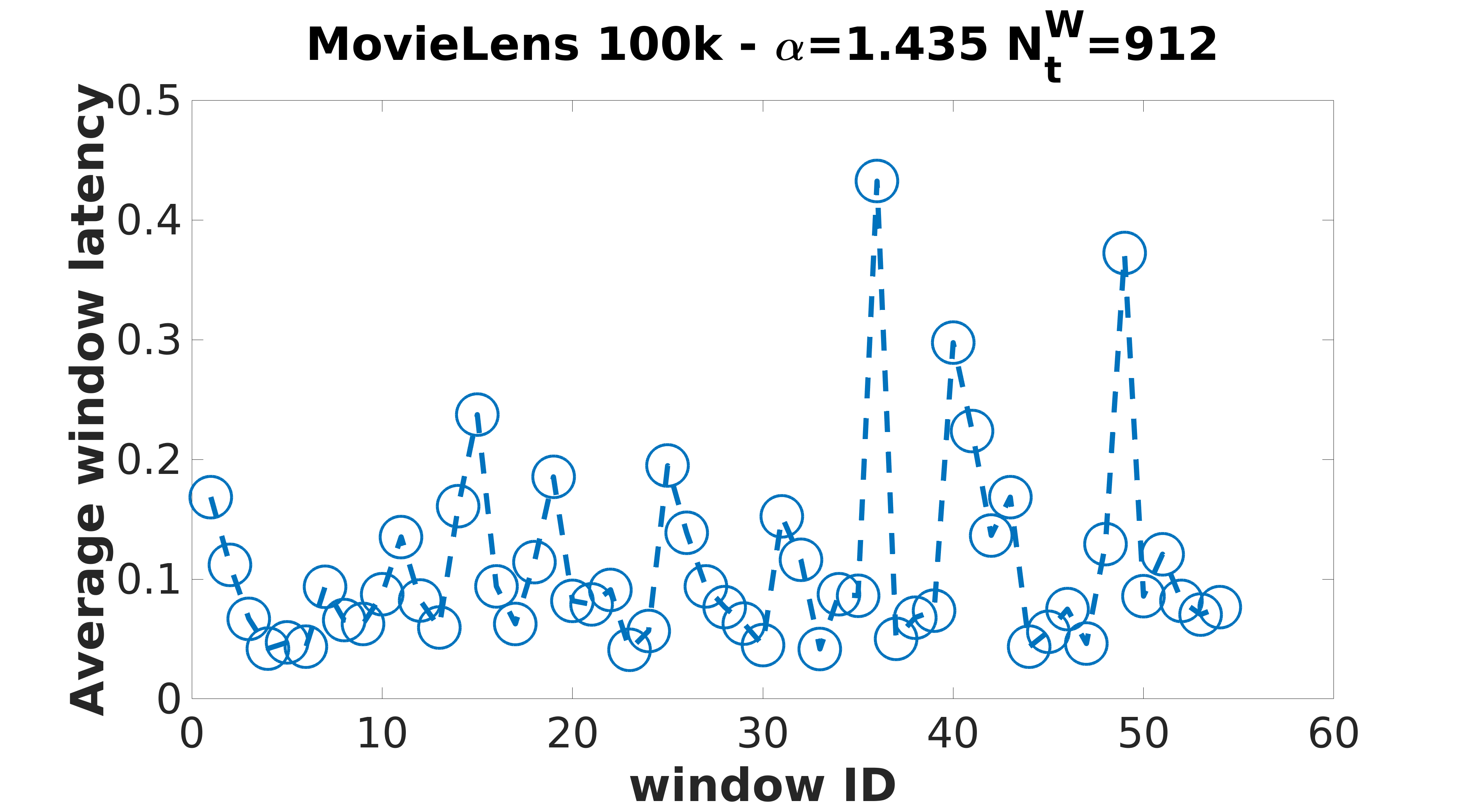}}
    \subfigure{\includegraphics[width=0.3\textwidth]{ 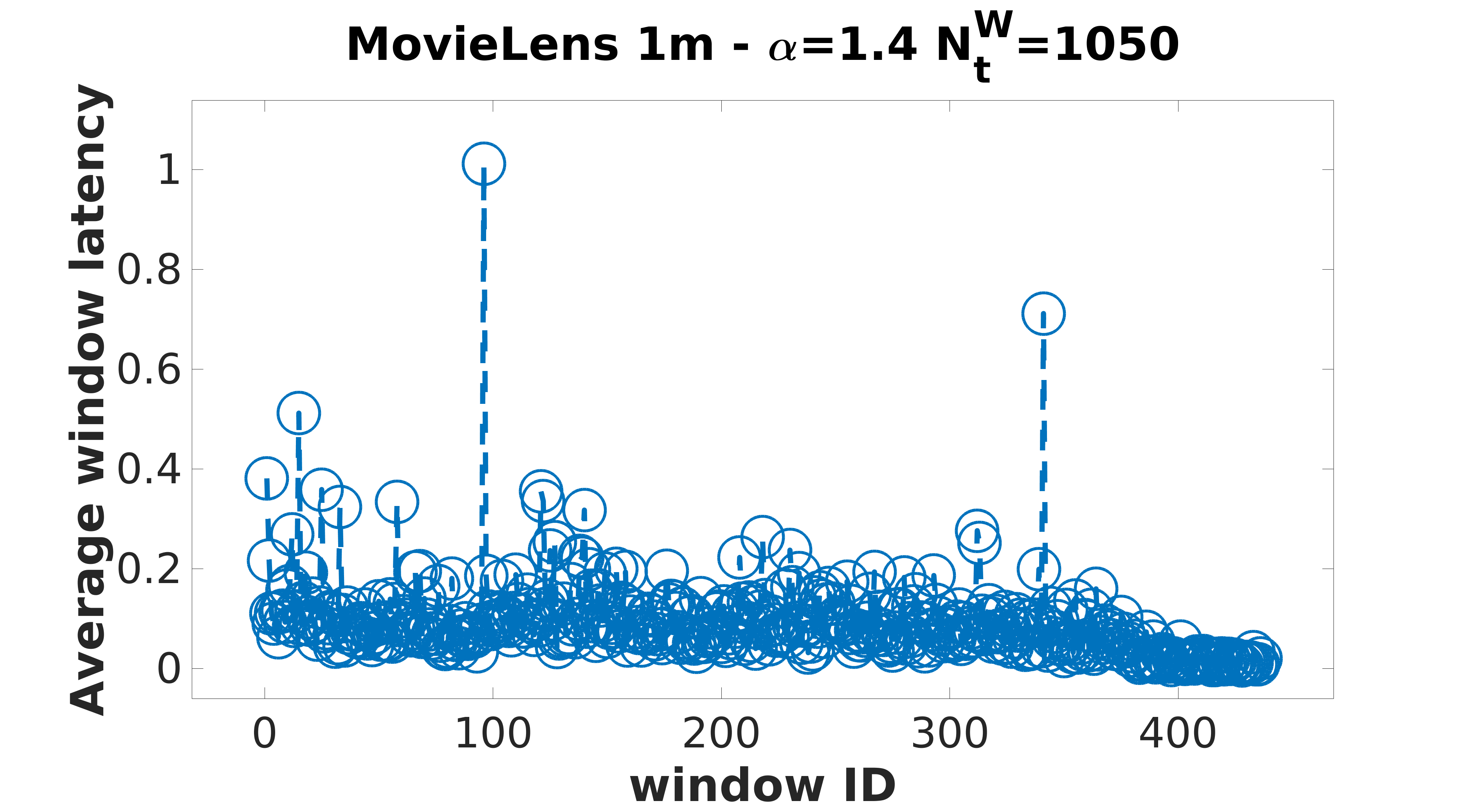}}
    \subfigure{\includegraphics[width=0.3\textwidth]{ 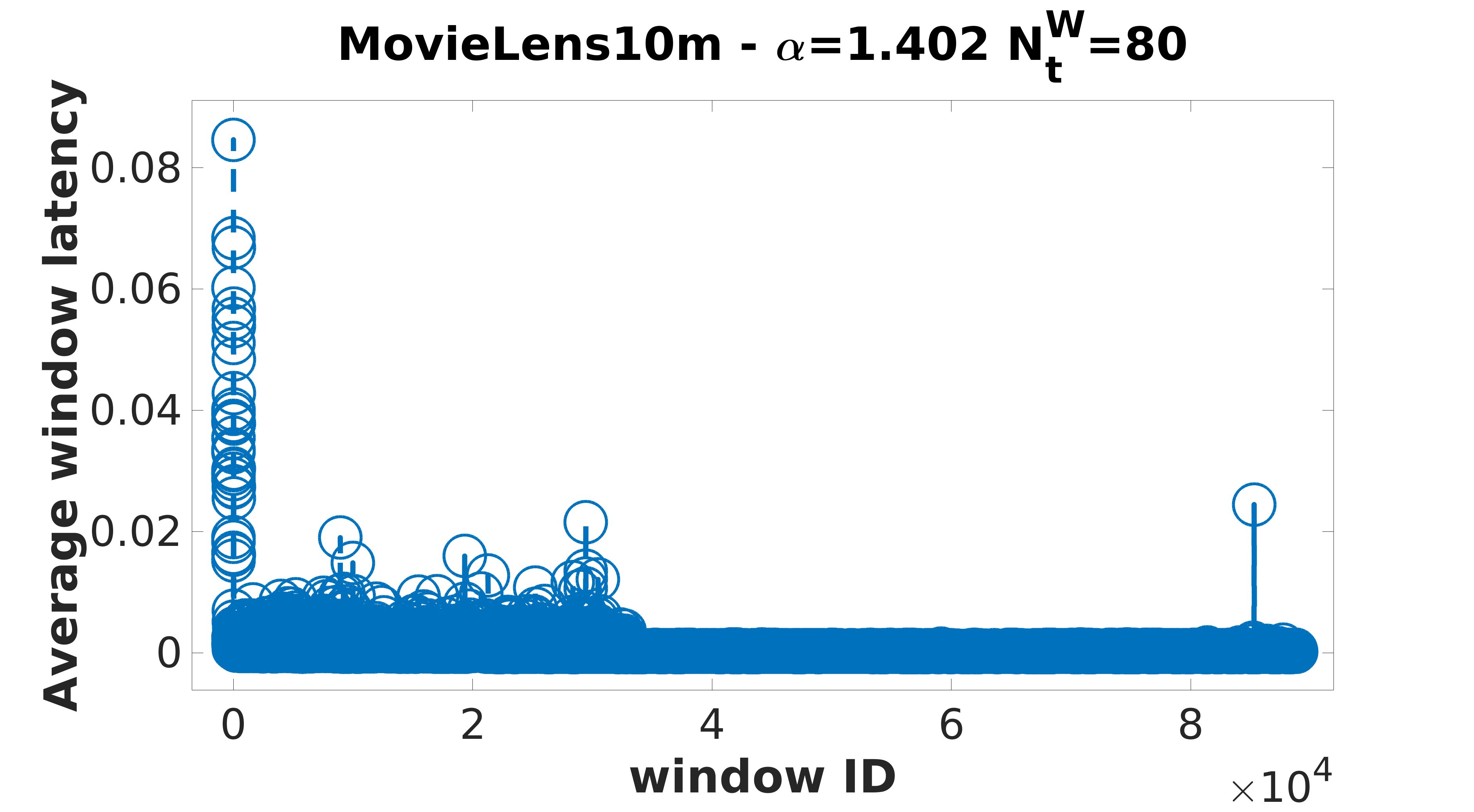}}
   \subfigure{\includegraphics[width=0.3\textwidth]{ 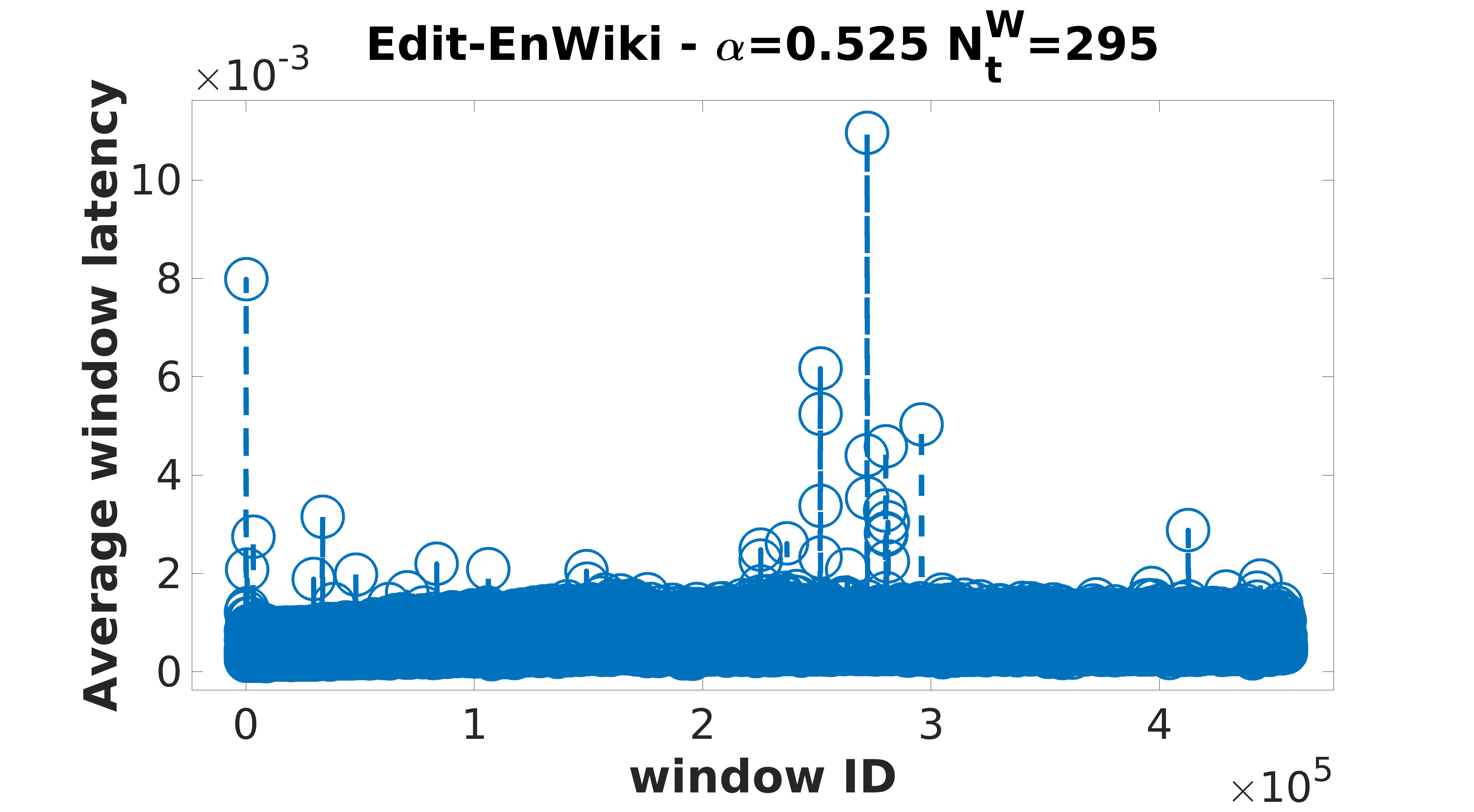}}
   \subfigure{\includegraphics[width=0.3\textwidth]{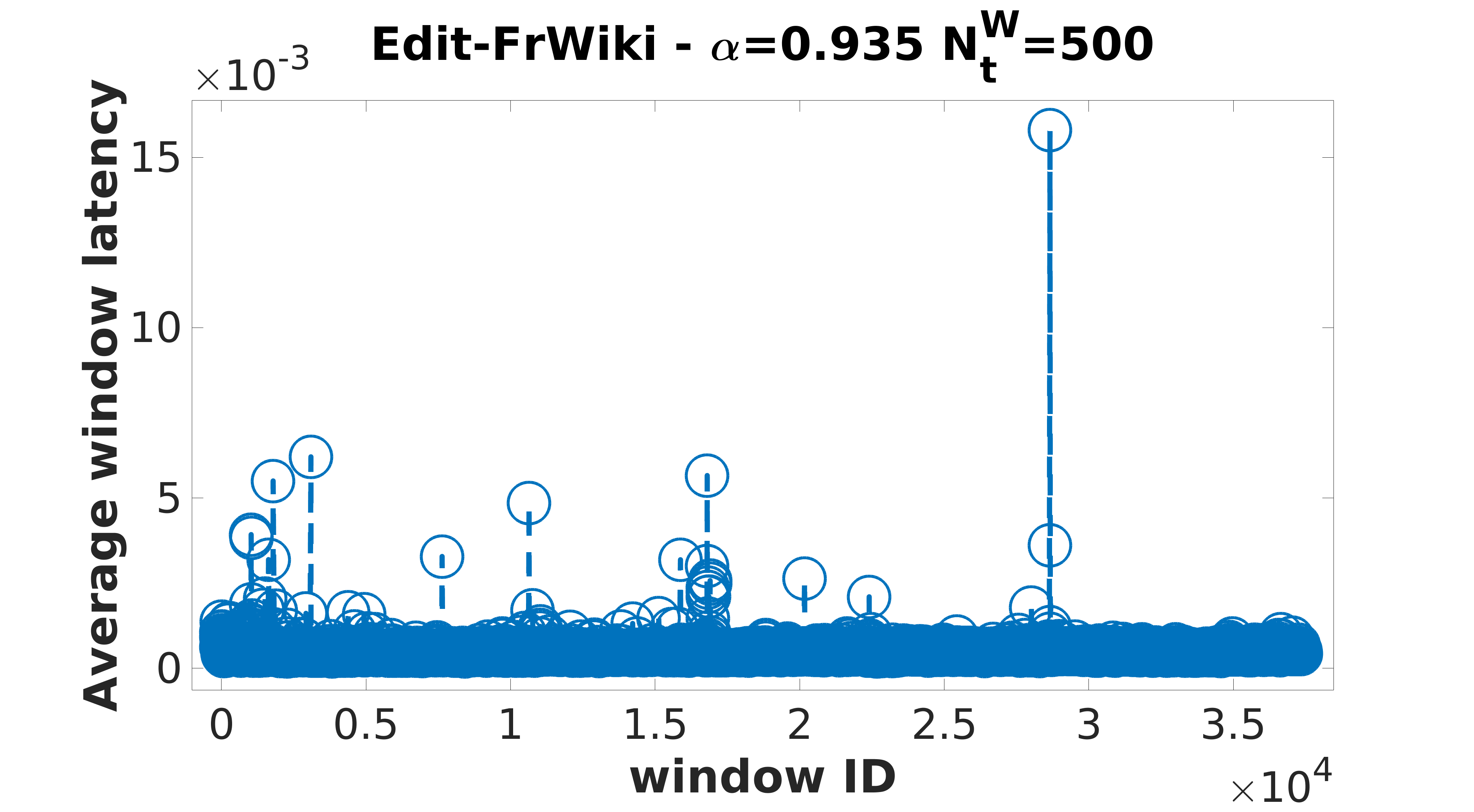}} 
 \caption{ Average window latency (s) of sGrapp.} \label{fig:wlatency}\end{figure*}
\begin{figure*}[h]\centering
    \subfigure{\includegraphics[width=0.3\textwidth]{ 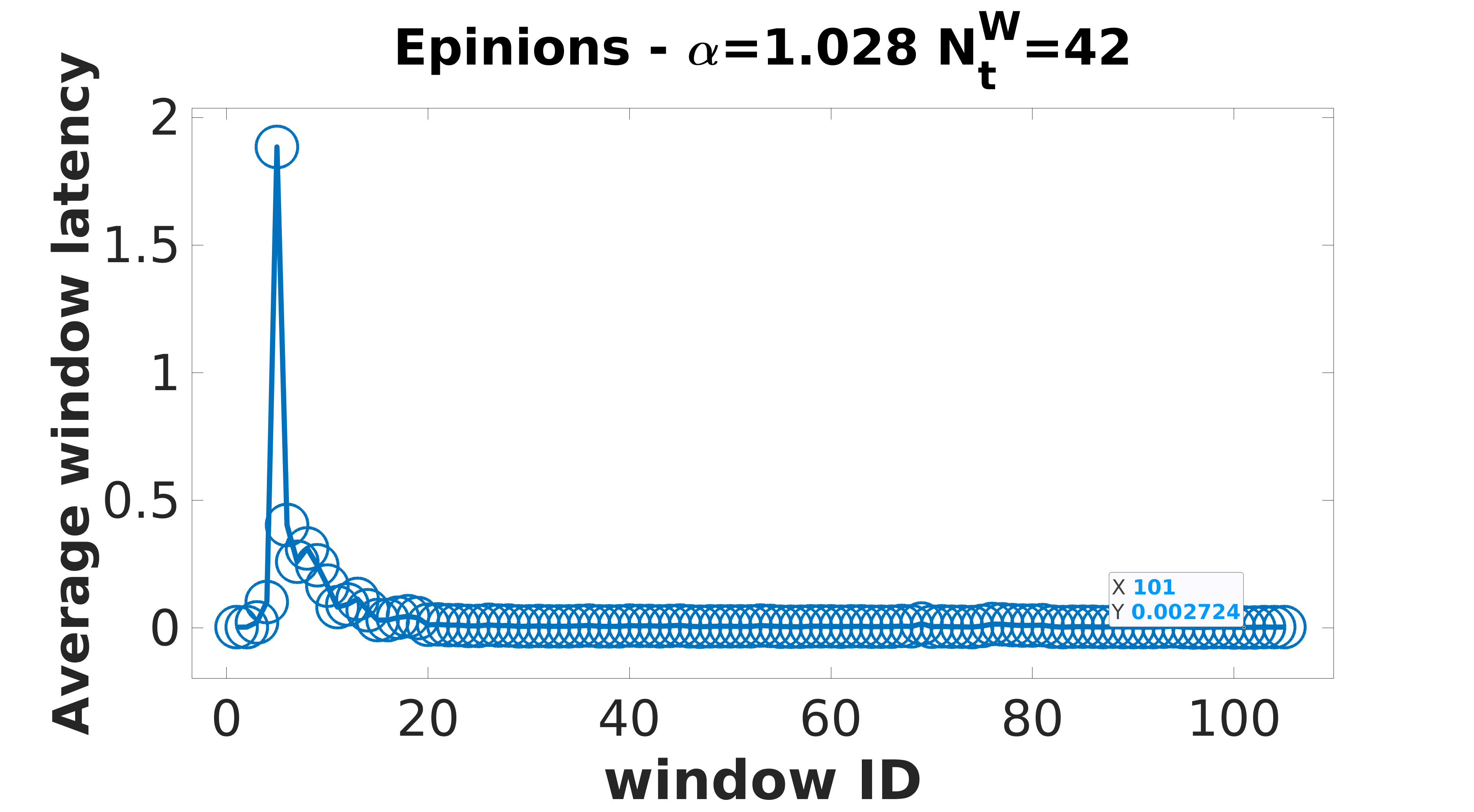}} 
    \subfigure{\includegraphics[width=0.3\textwidth]{ 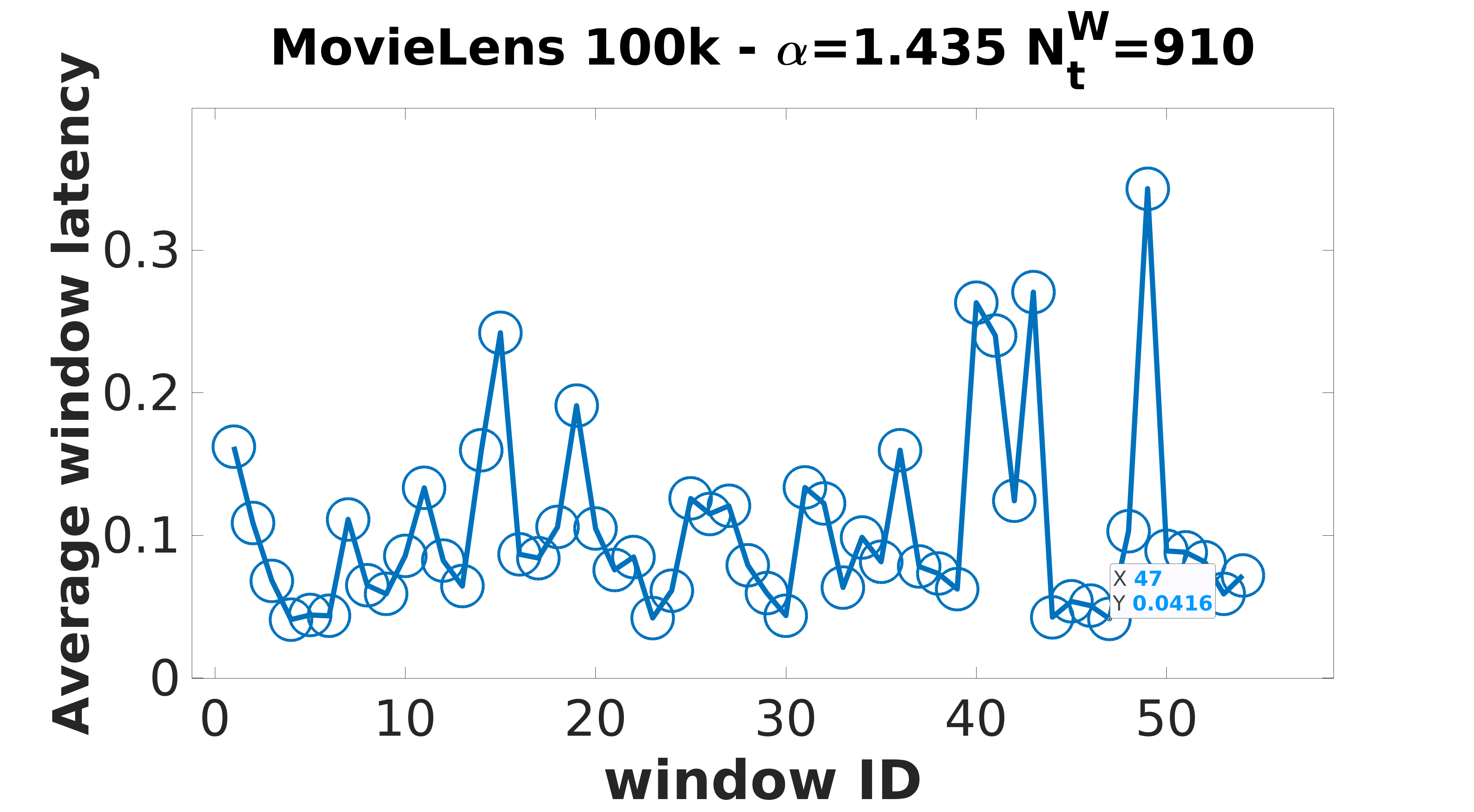}}
    \subfigure{\includegraphics[width=0.3\textwidth]{ 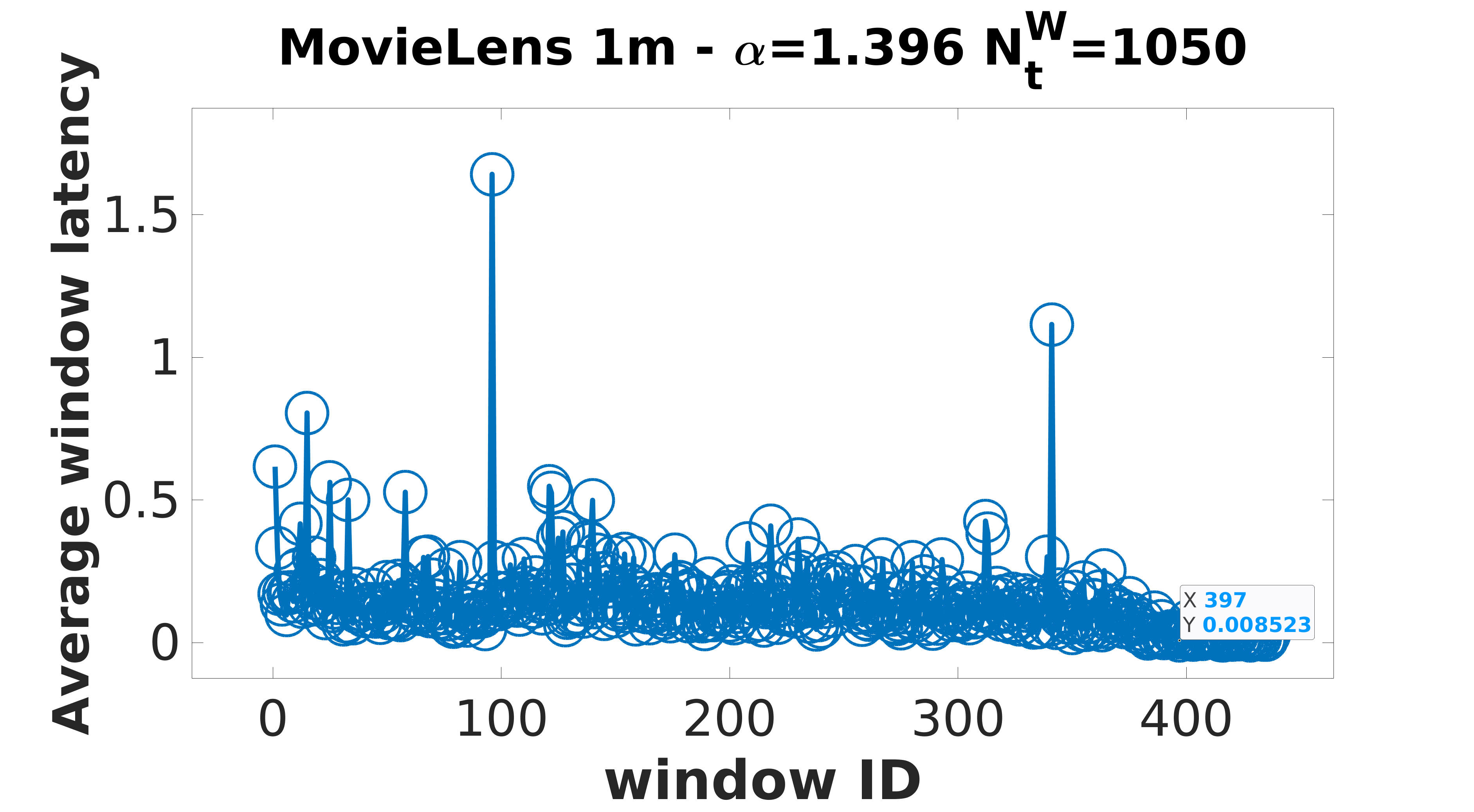}}
    \subfigure{\includegraphics[width=0.3\textwidth]{ 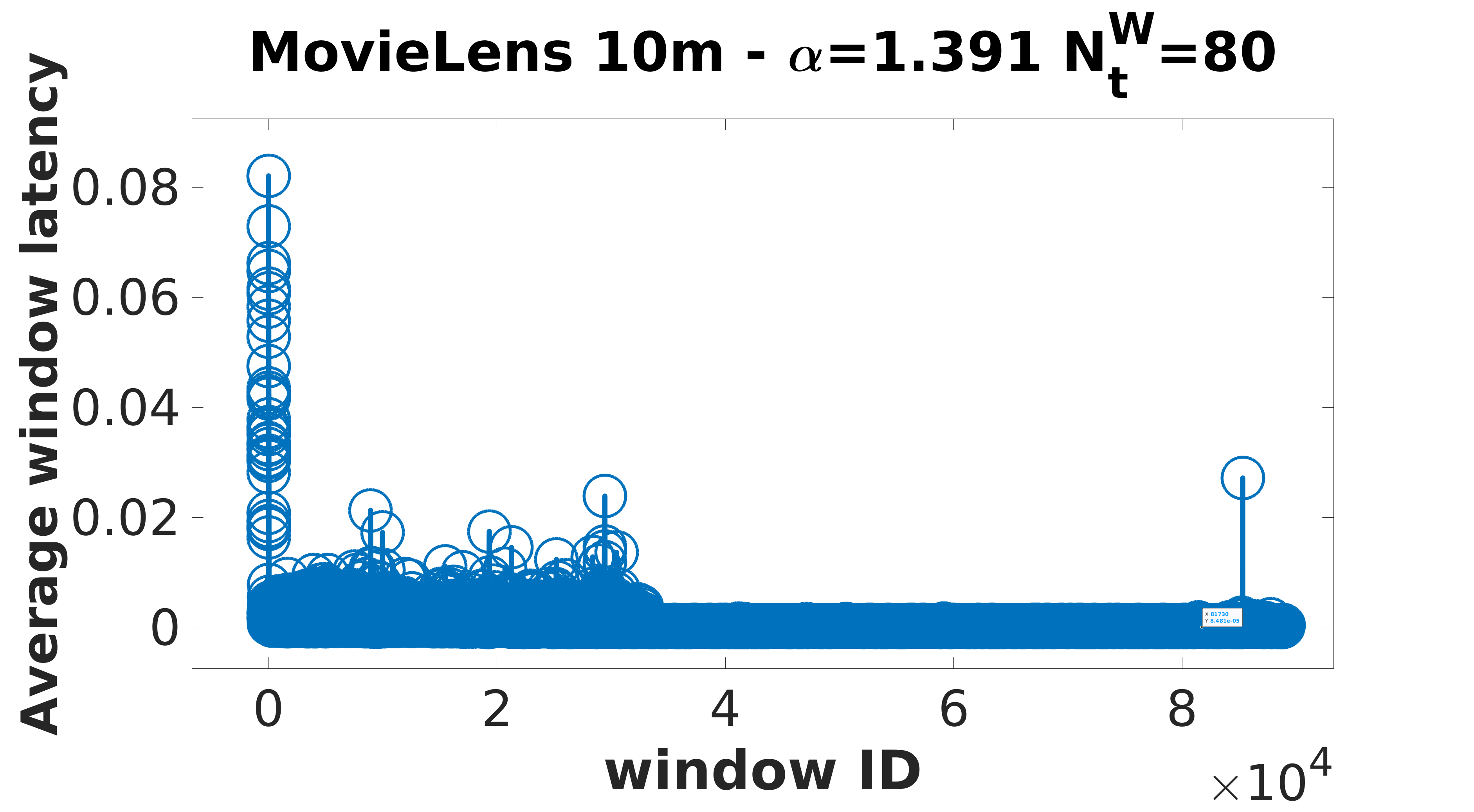}}
    \subfigure{\includegraphics[width=0.3\textwidth]{ 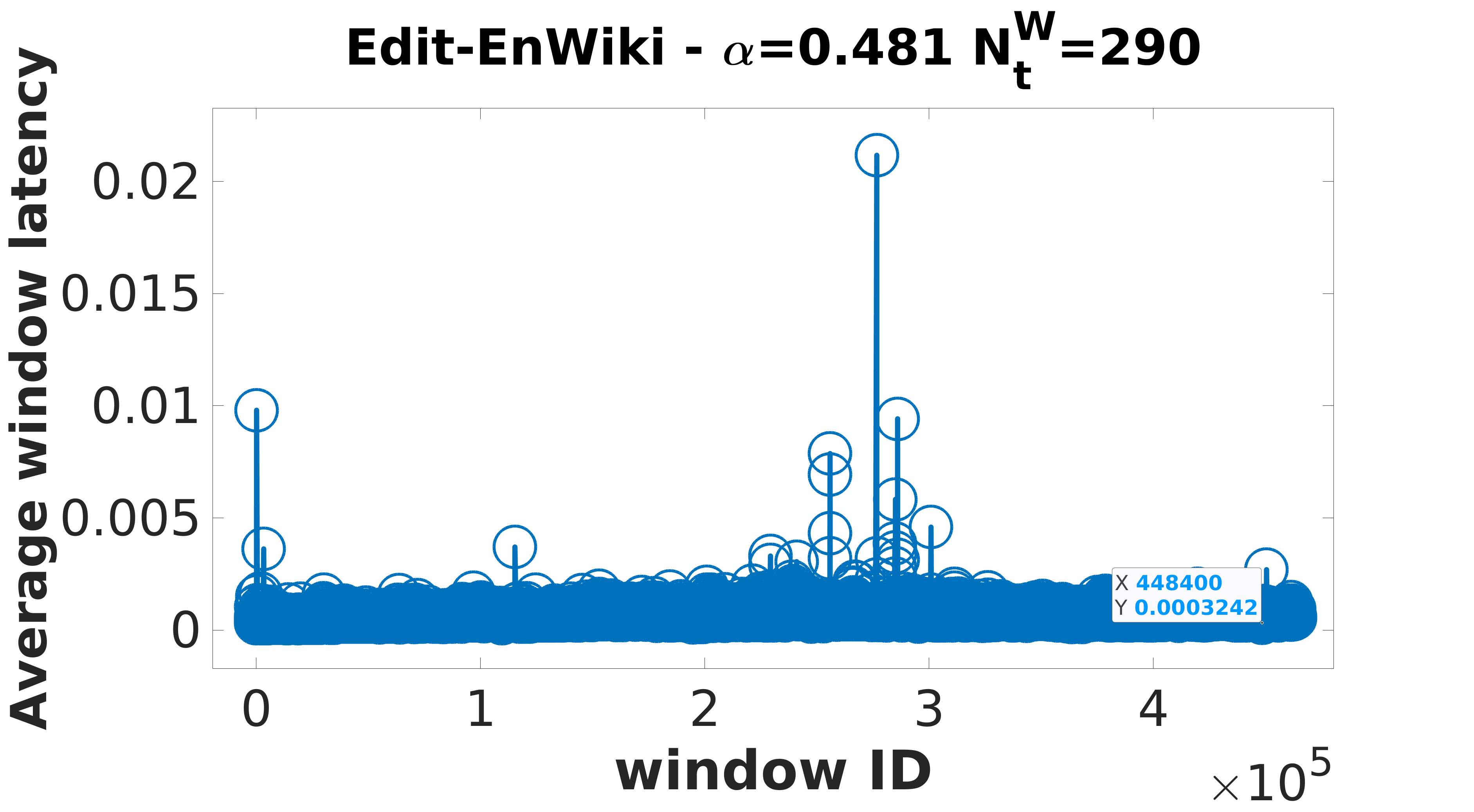}}
    \subfigure{\includegraphics[width=0.3\textwidth]{ 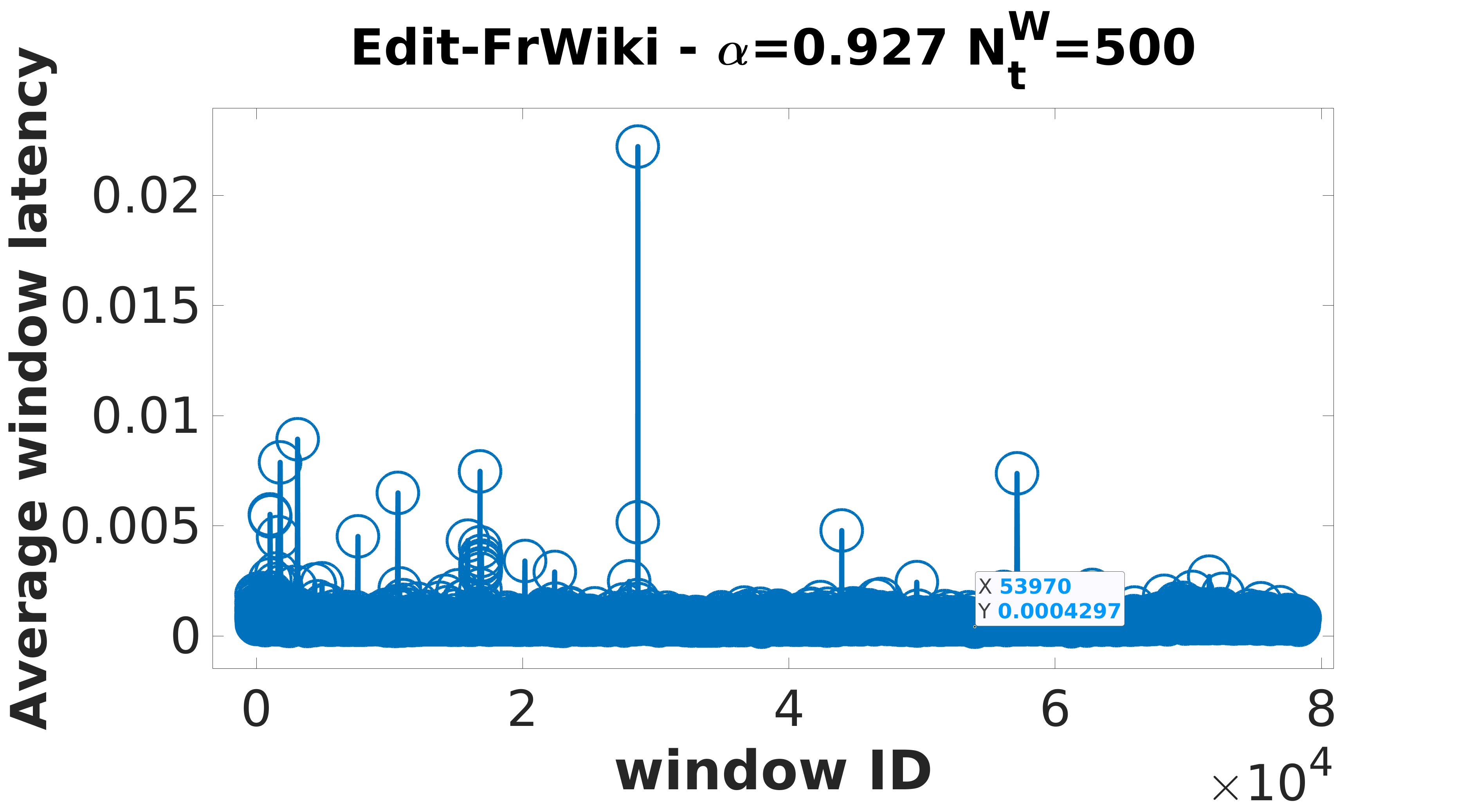}} 
\caption{ Average window latency (s) of sGrapp-100.} \label{fig:wlatency100}\end{figure*}
   
\begin{figure*}[h]\centering
    \subfigure{\includegraphics[width=0.3\textwidth]{ 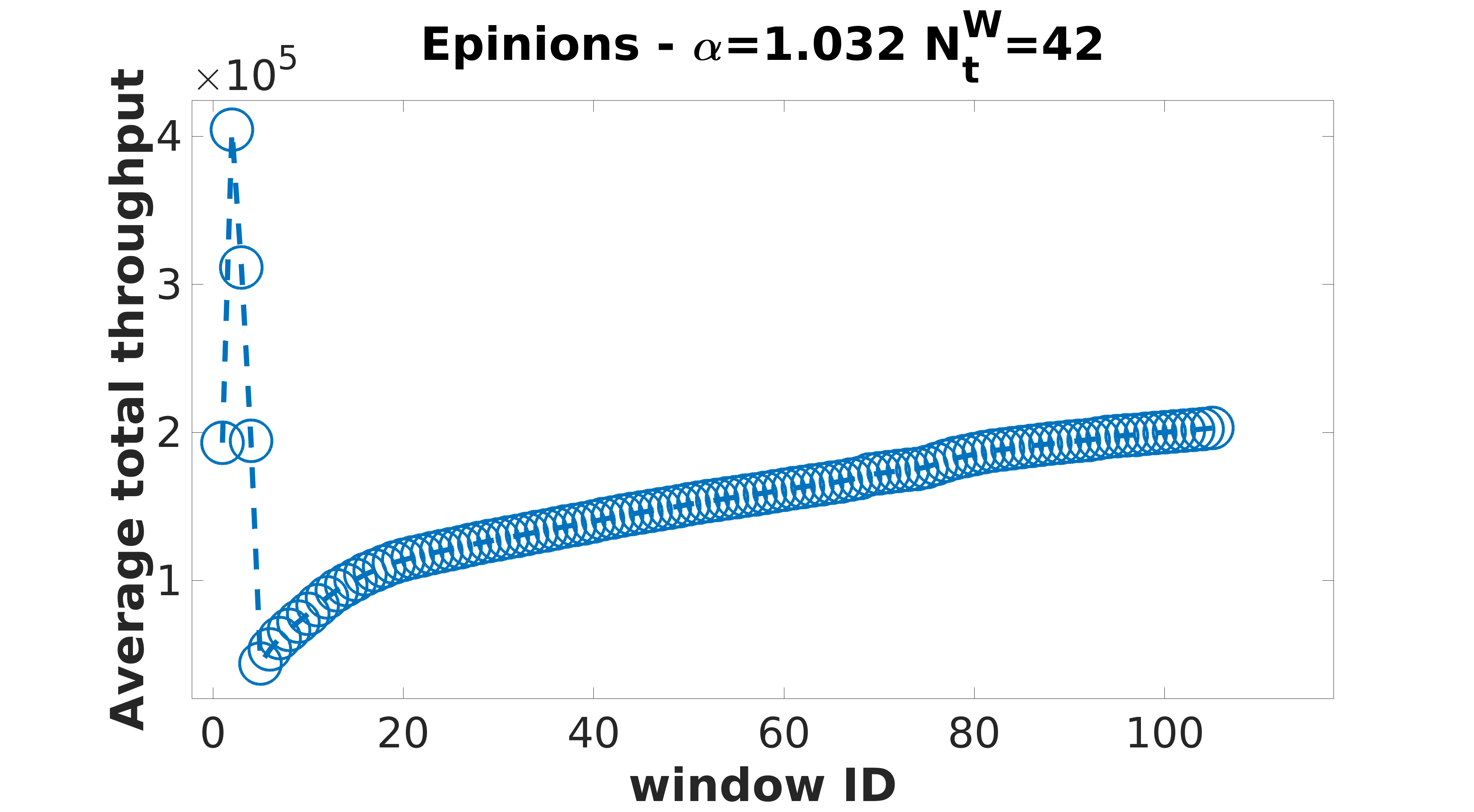}} 
    \subfigure{\includegraphics[width=0.3\textwidth]{ 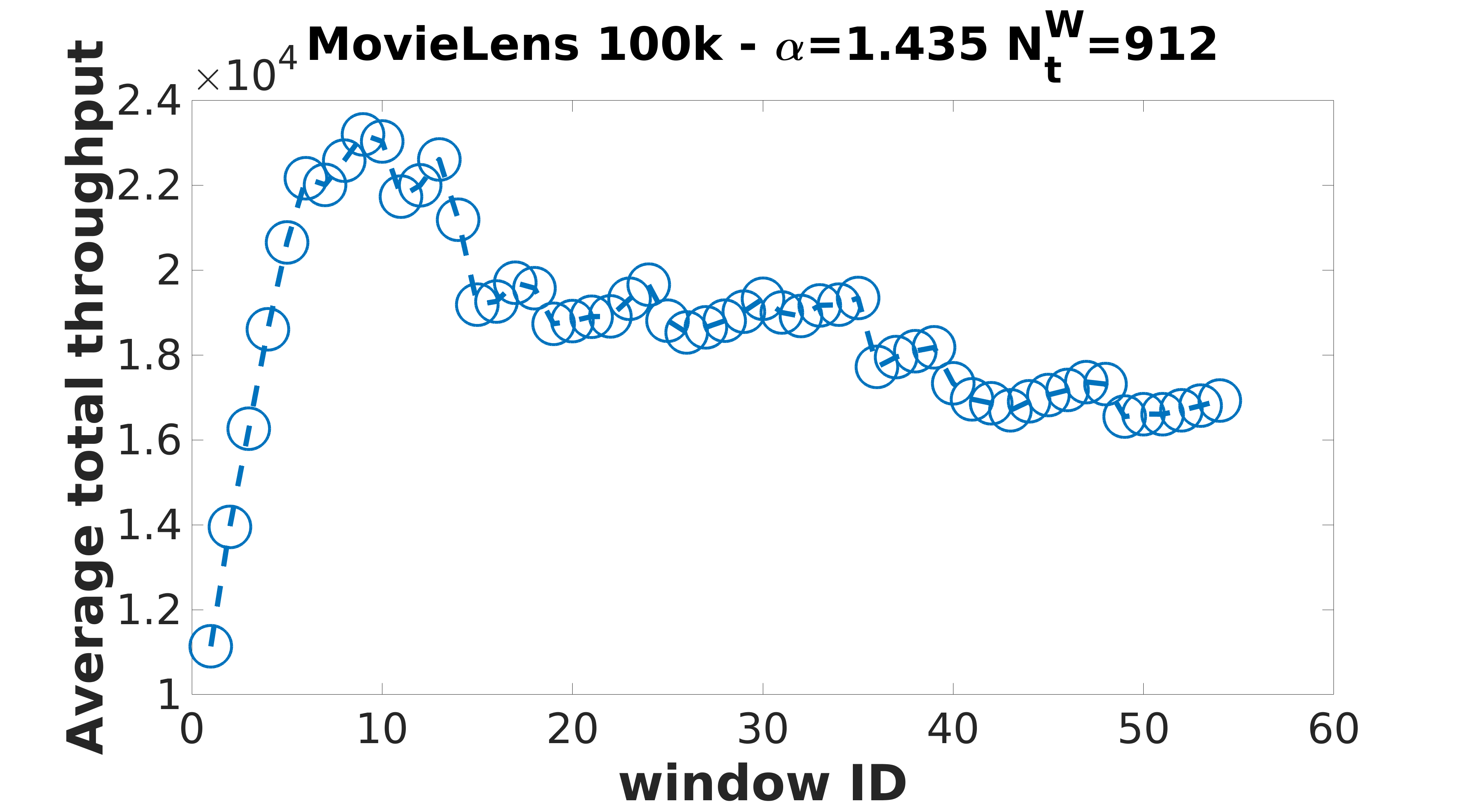}} 
    \subfigure{\includegraphics[width=0.3\textwidth]{ 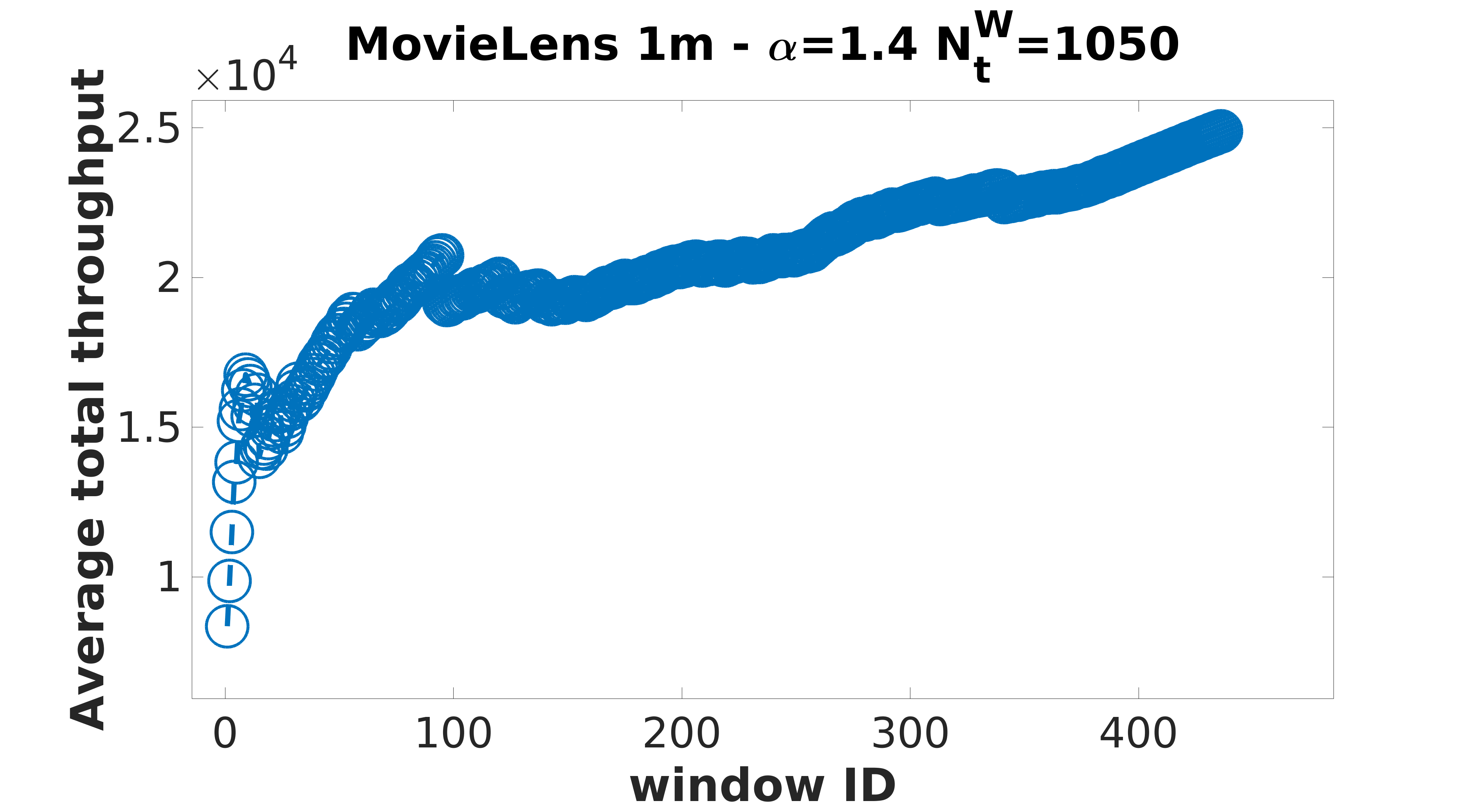}}
    \subfigure{\includegraphics[width=0.3\textwidth]{ 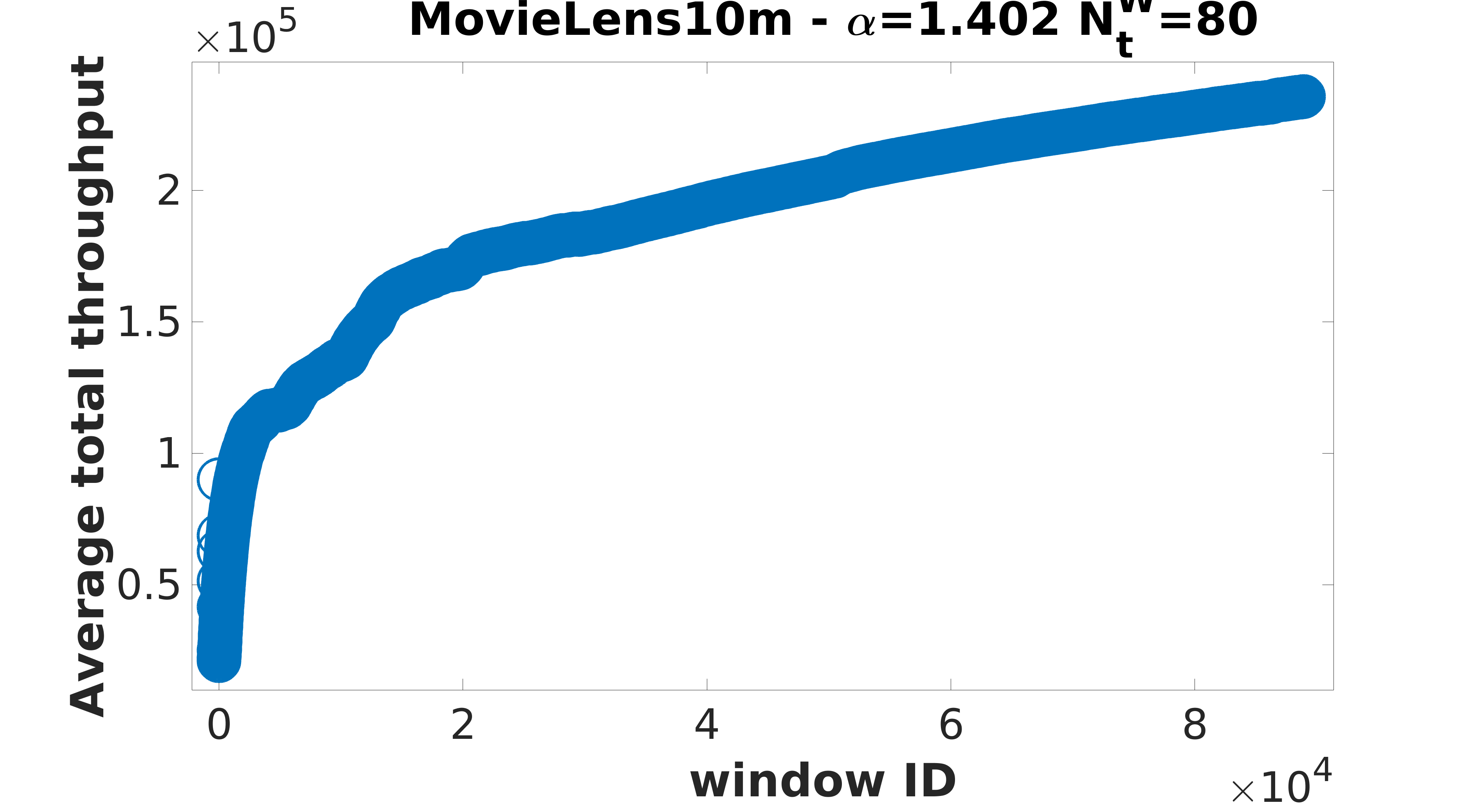}}
    \subfigure{\includegraphics[width=0.3\textwidth]{ 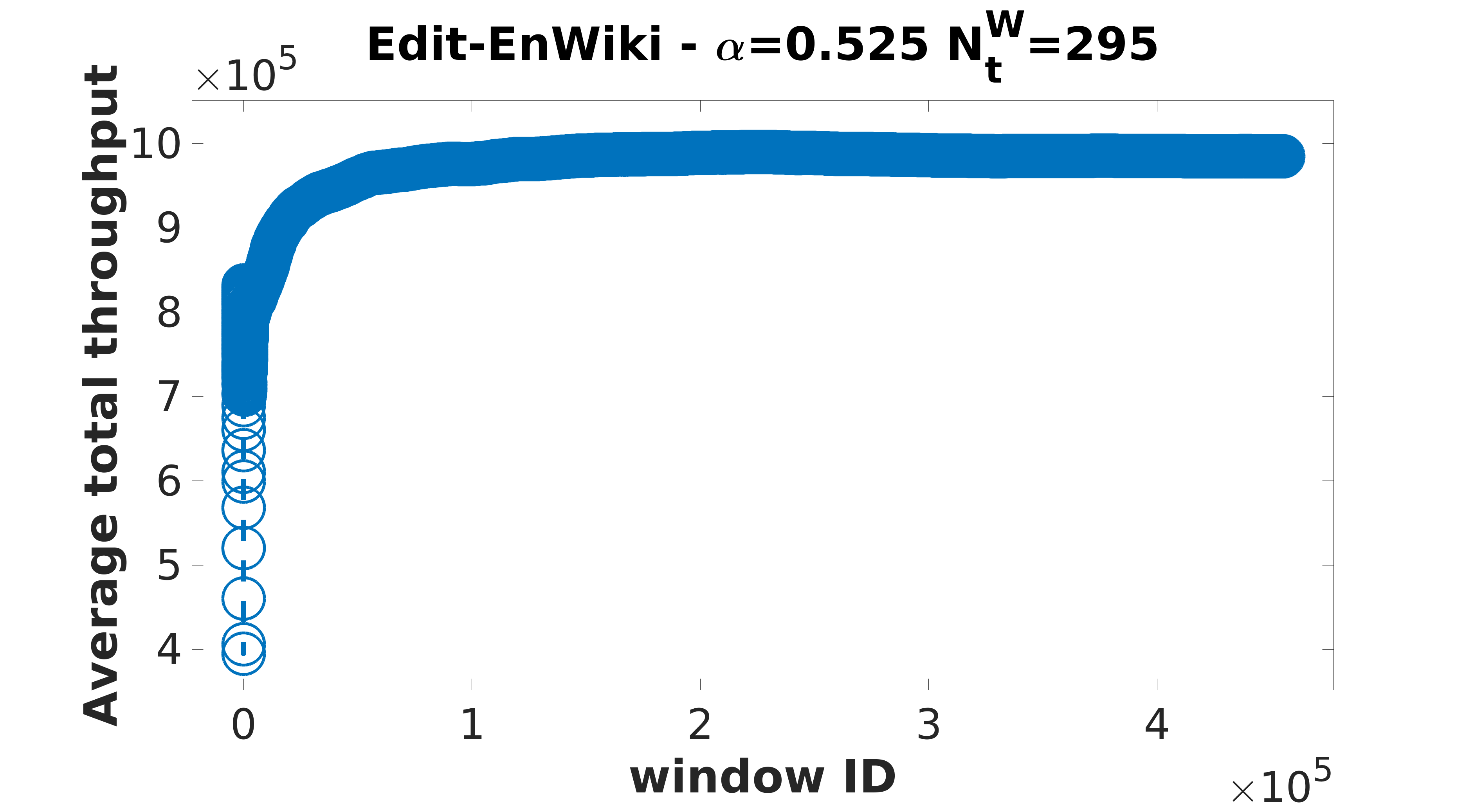}}
    \subfigure{\includegraphics[width=0.3\textwidth]{ 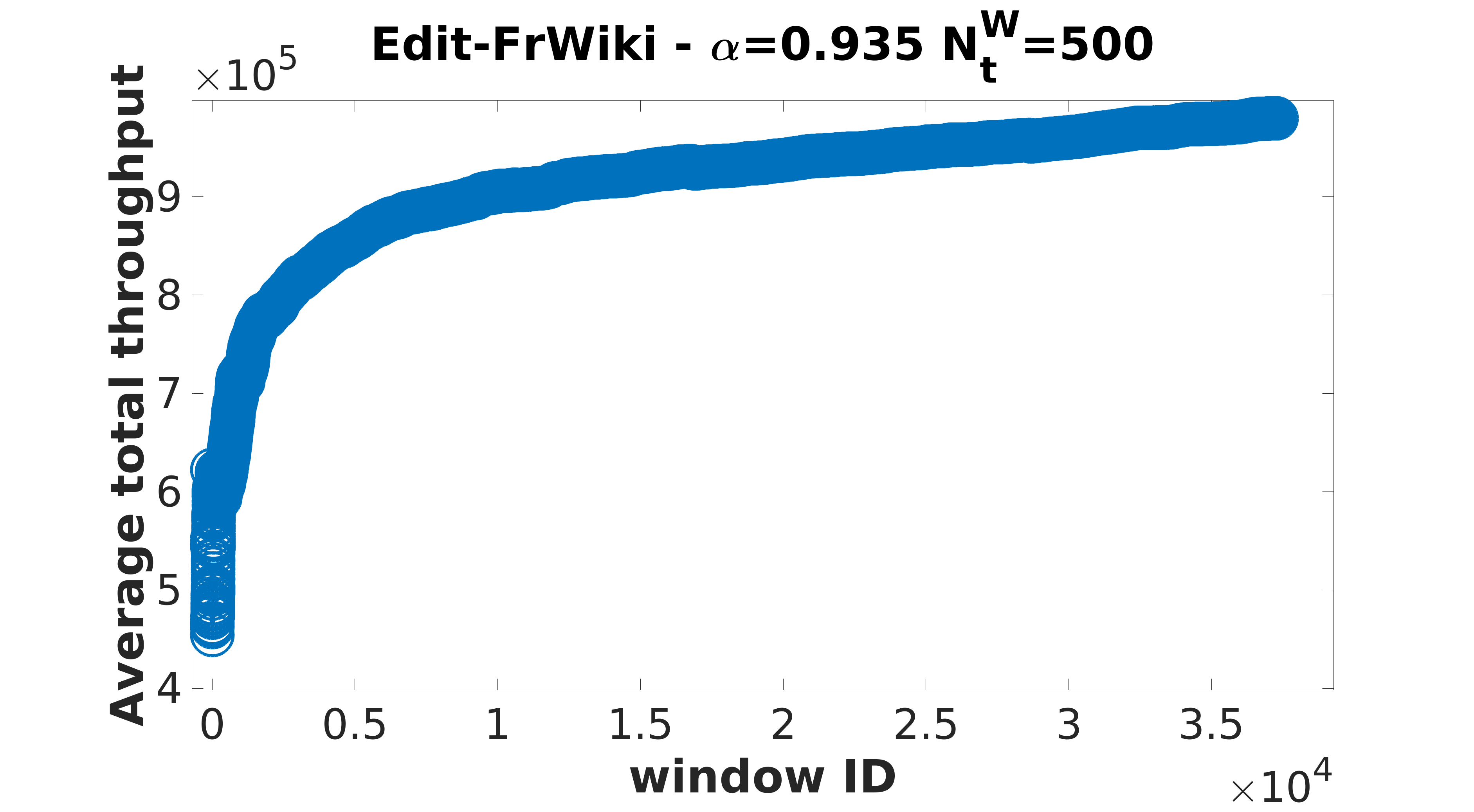}}
    \caption{ Average total throughput (edge/s) of sGrapp at the end of each window.}
   \label{fig:totalthroughput}
\end{figure*}
\begin{figure*}[h]\centering
     \subfigure{\includegraphics[width=0.3\textwidth]{ 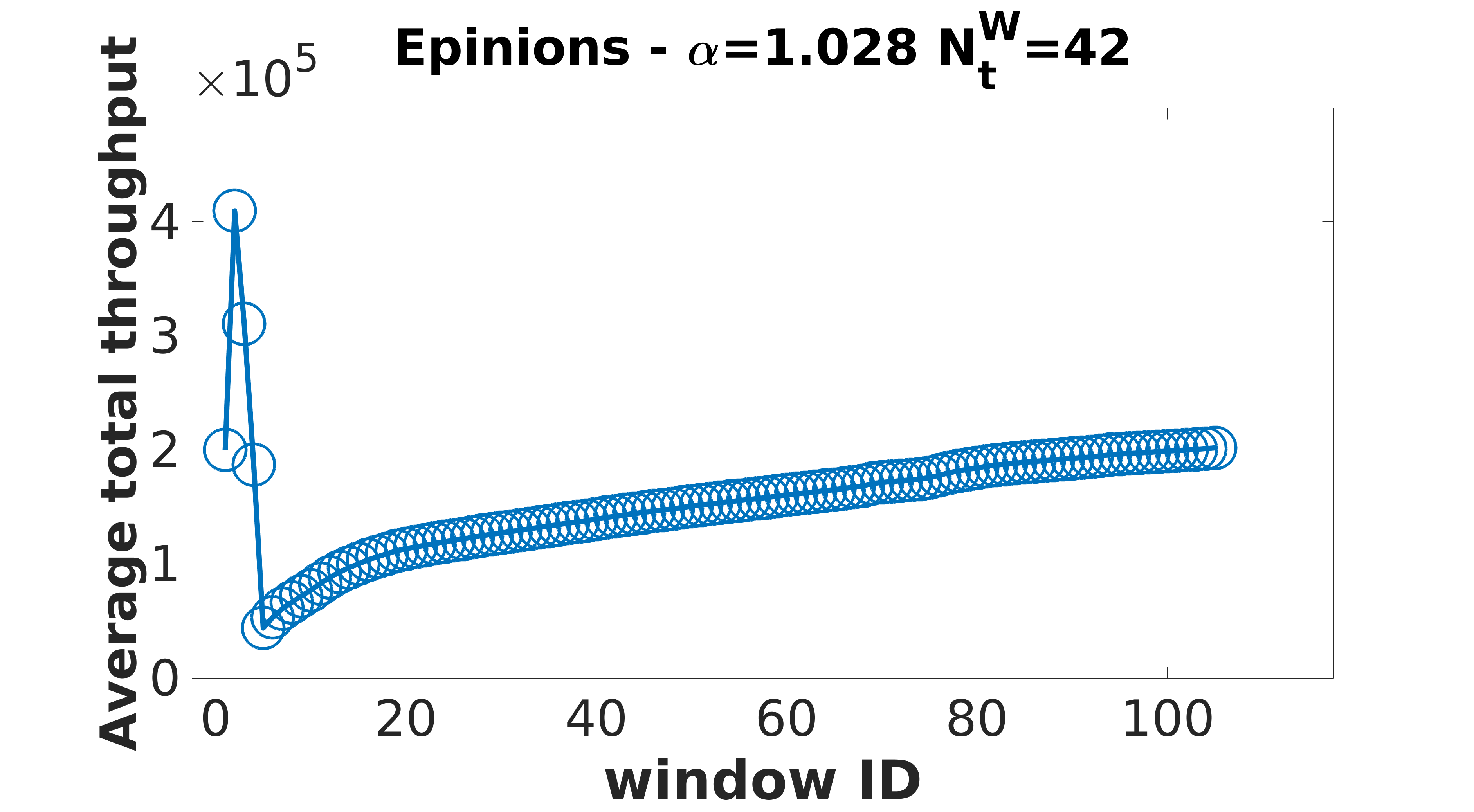}} 
    \subfigure{\includegraphics[width=0.3\textwidth]{ 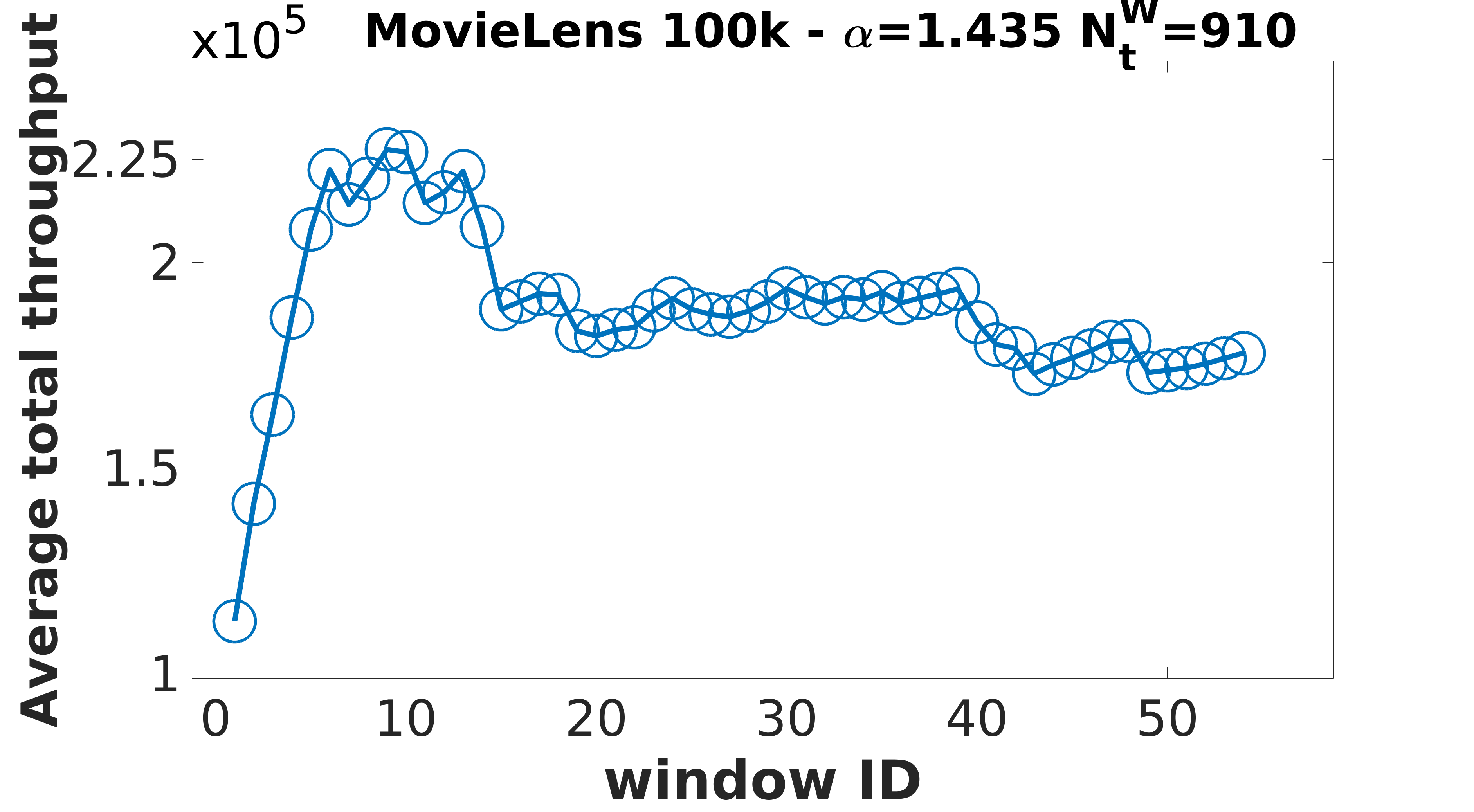}}
    \subfigure{\includegraphics[width=0.3\textwidth]{ 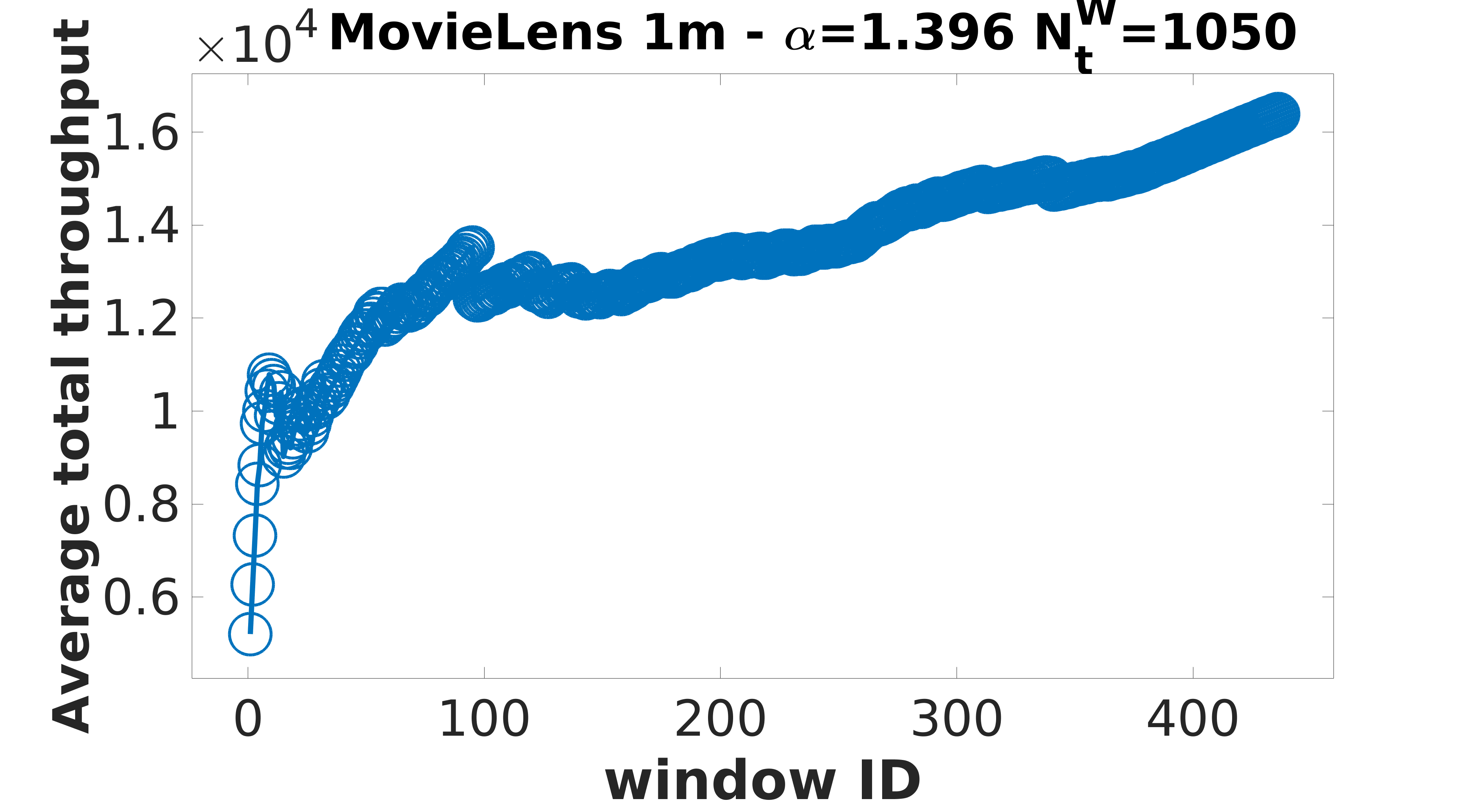}}
    \subfigure{\includegraphics[width=0.3\textwidth]{ 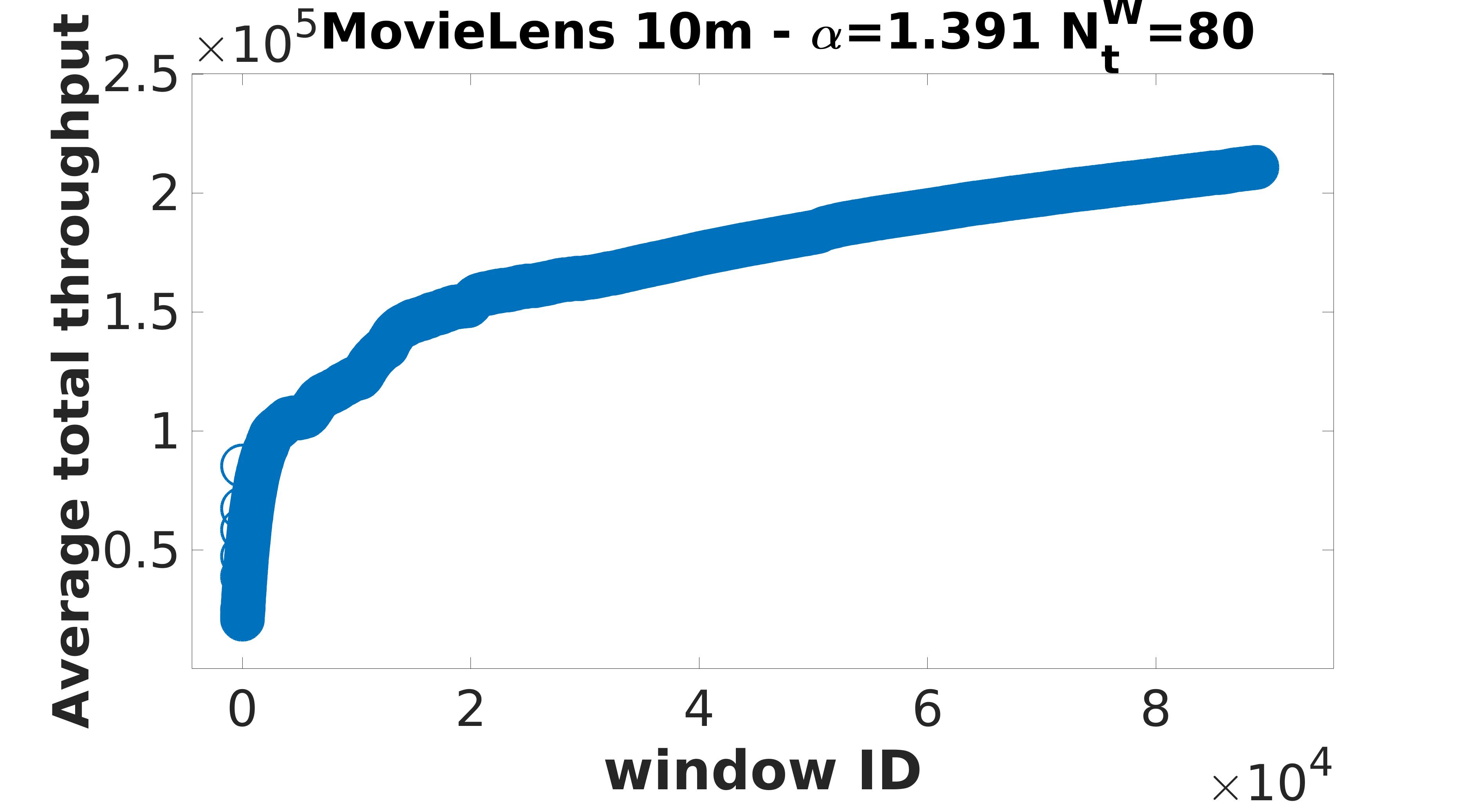}}
    \subfigure{\includegraphics[width=0.3\textwidth]{ 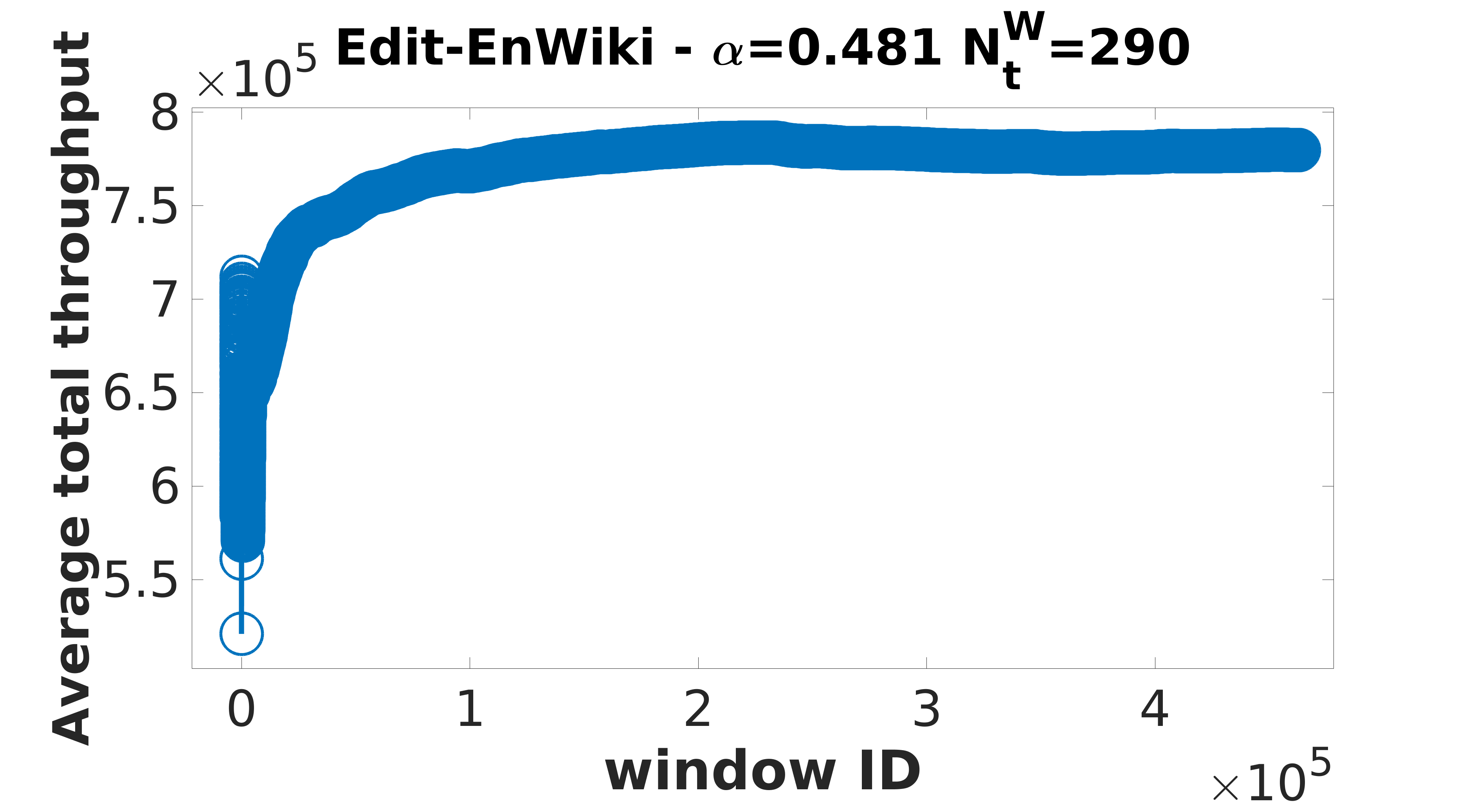}}
    \subfigure{\includegraphics[width=0.3\textwidth]{ 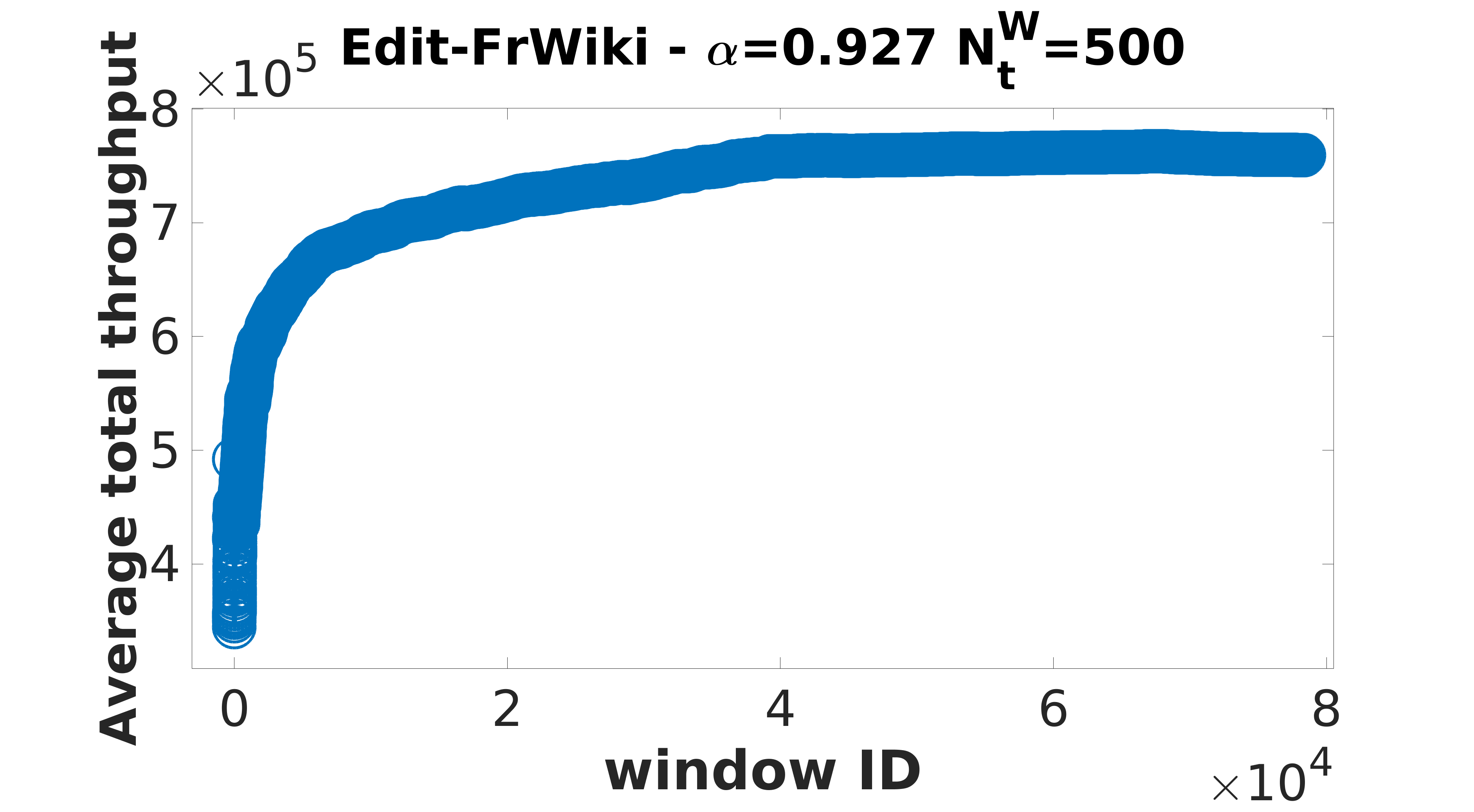}} 
    \caption{ Average total throughput (edge/s) of sGrapp-100 at the end of each window.}
   \label{fig:totalthroughput100}
\end{figure*}
\begin{figure*}[h]\centering
    \subfigure{\includegraphics[width=0.3\textwidth]{ 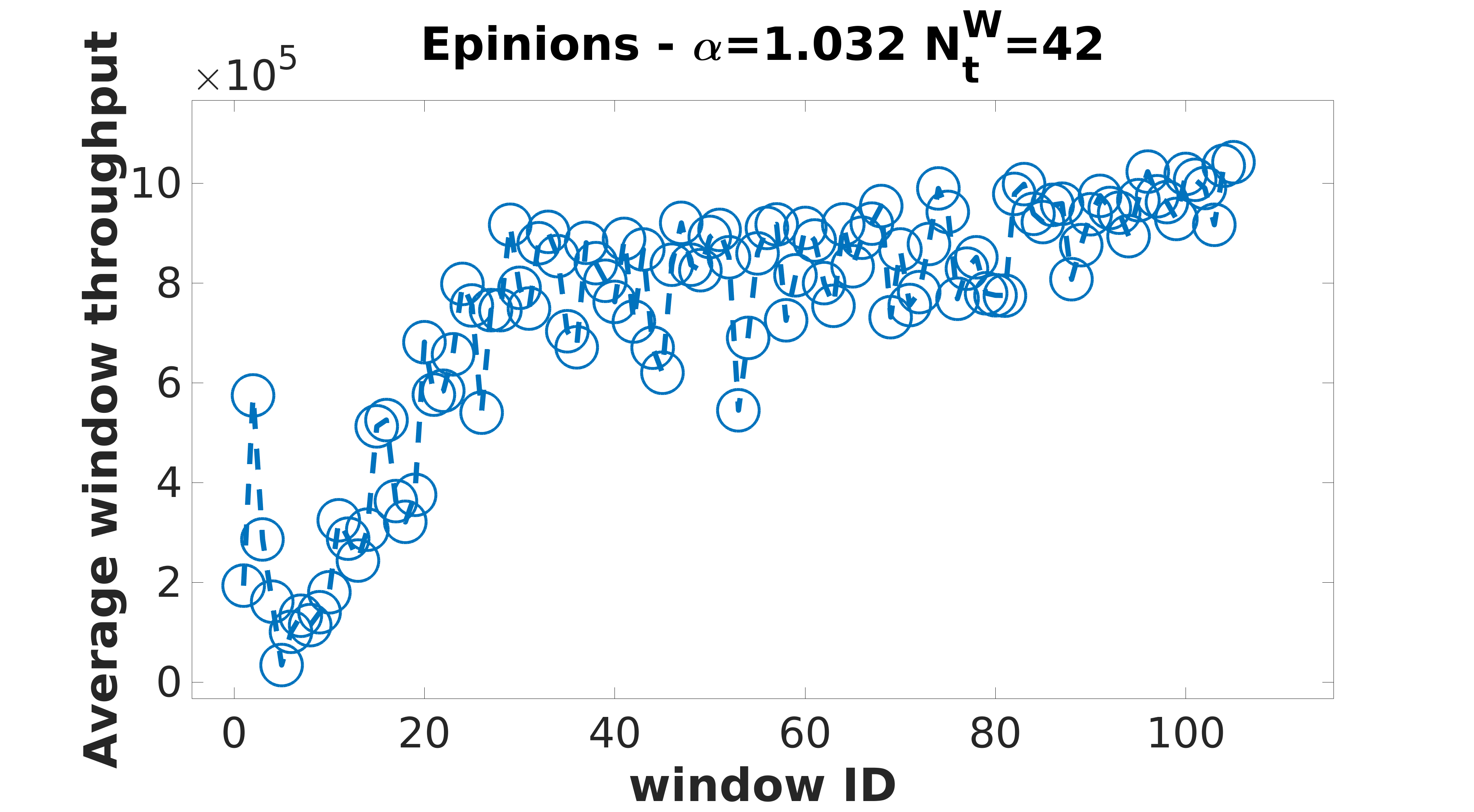}} 
    \subfigure{\includegraphics[width=0.3\textwidth]{ 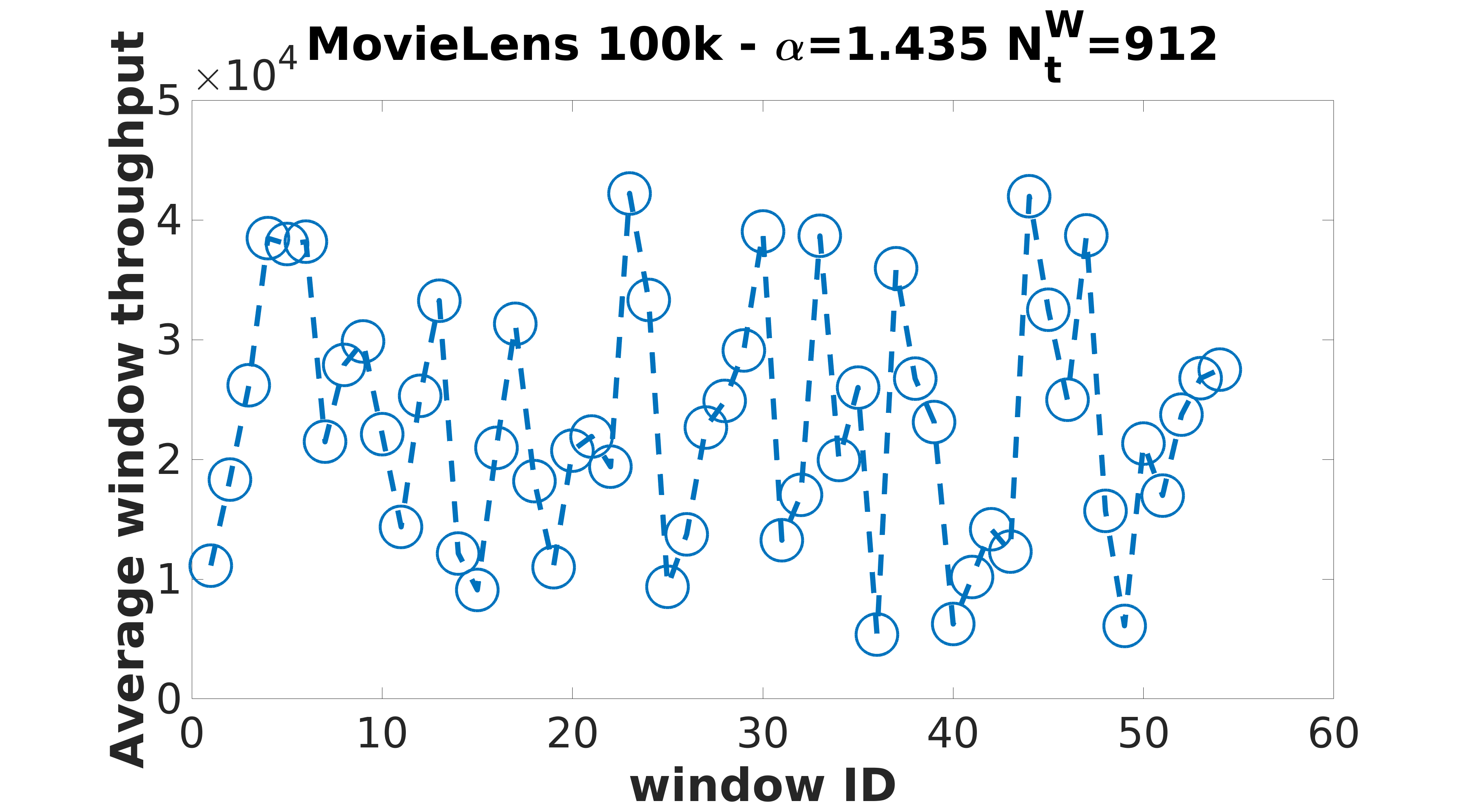}} 
    \subfigure{\includegraphics[width=0.3\textwidth]{ 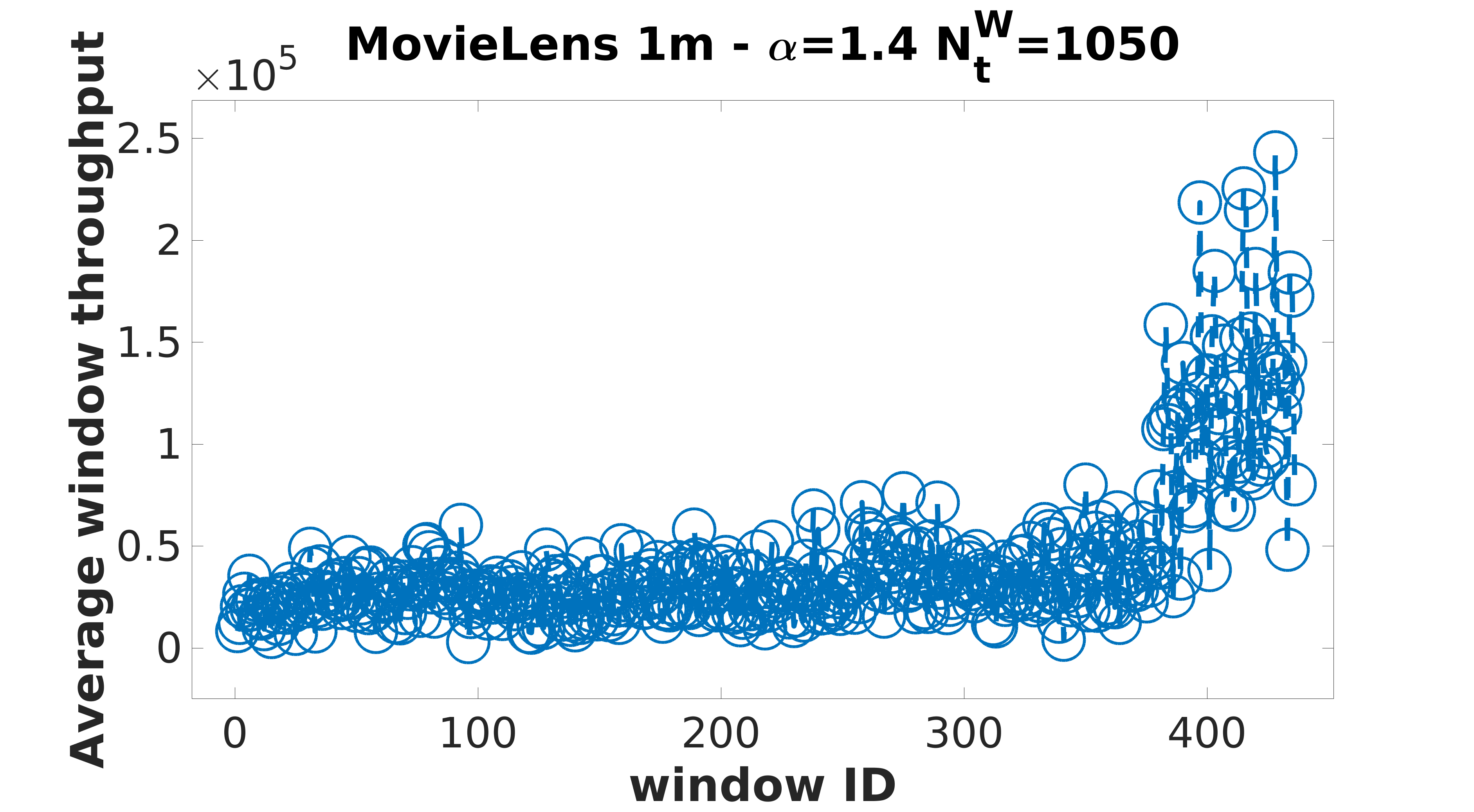}}
    \subfigure{\includegraphics[width=0.3\textwidth]{ 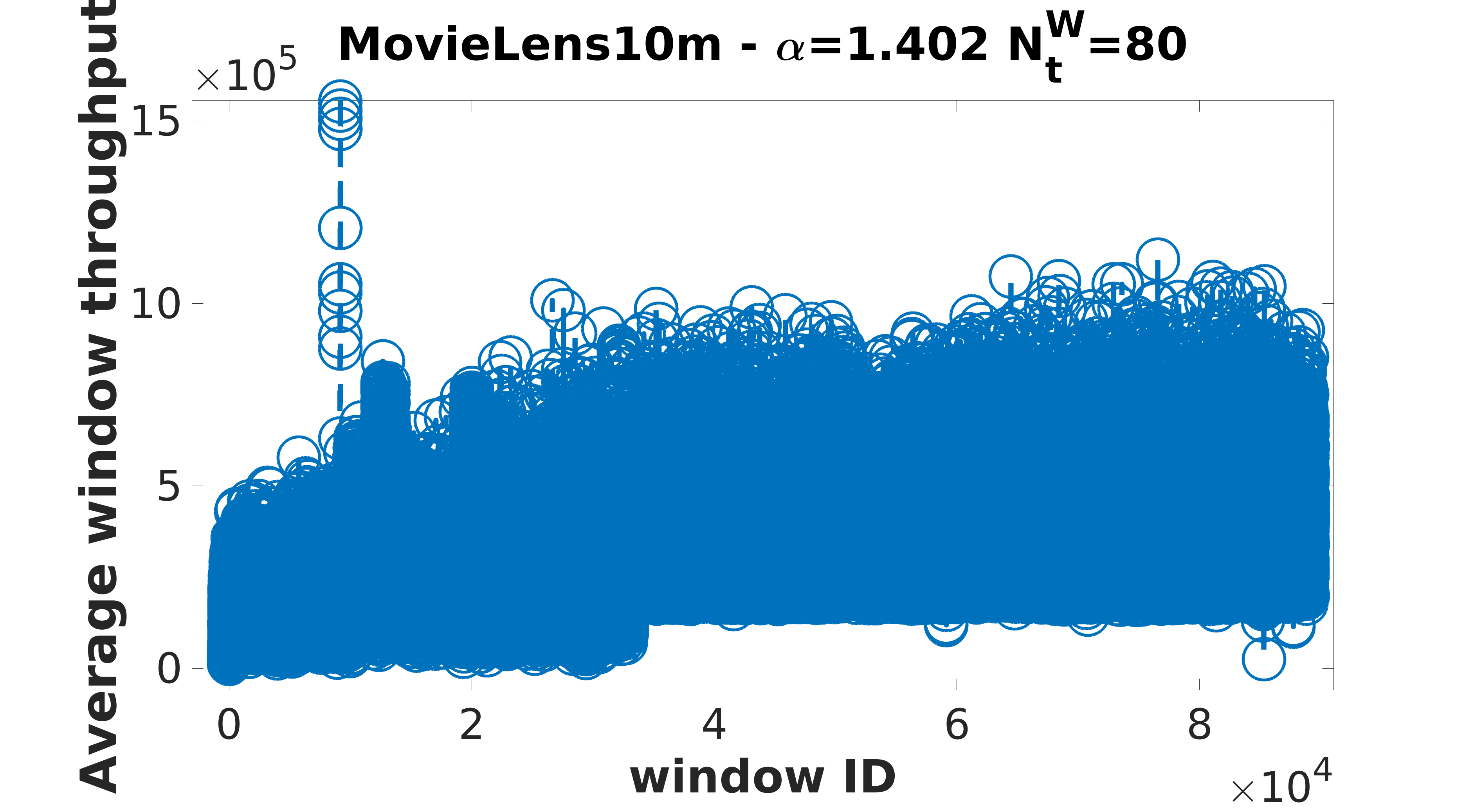}}
    \subfigure{\includegraphics[width=0.3\textwidth]{ 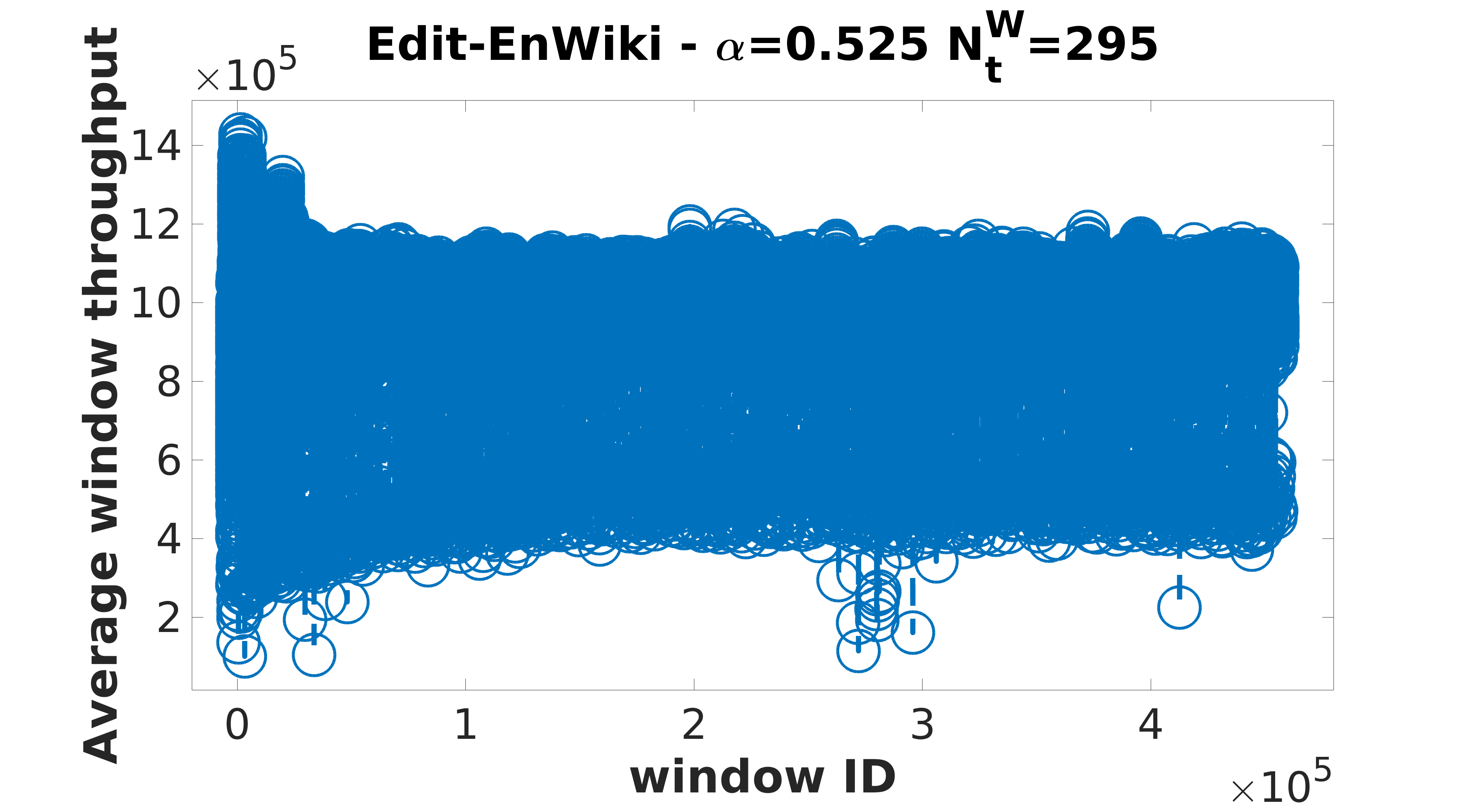}}
    \subfigure{\includegraphics[width=0.3\textwidth]{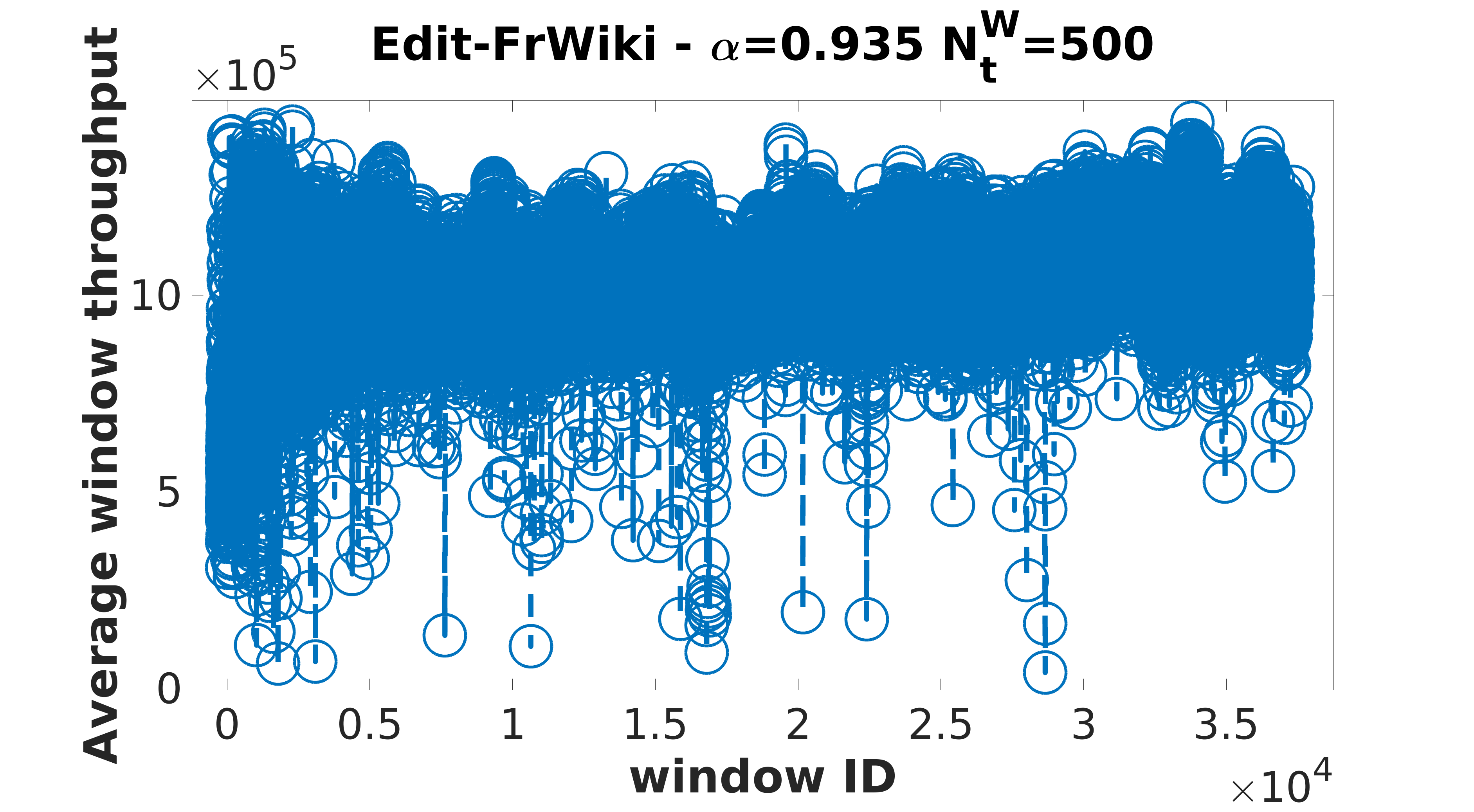}}
    \caption{Average window throughput (edge/s) of sGrapp at the end of each window.}
   \label{fig:wthroughput}
\end{figure*}
\begin{figure*}[h]\centering
    \subfigure{\includegraphics[width=0.3\textwidth]{ 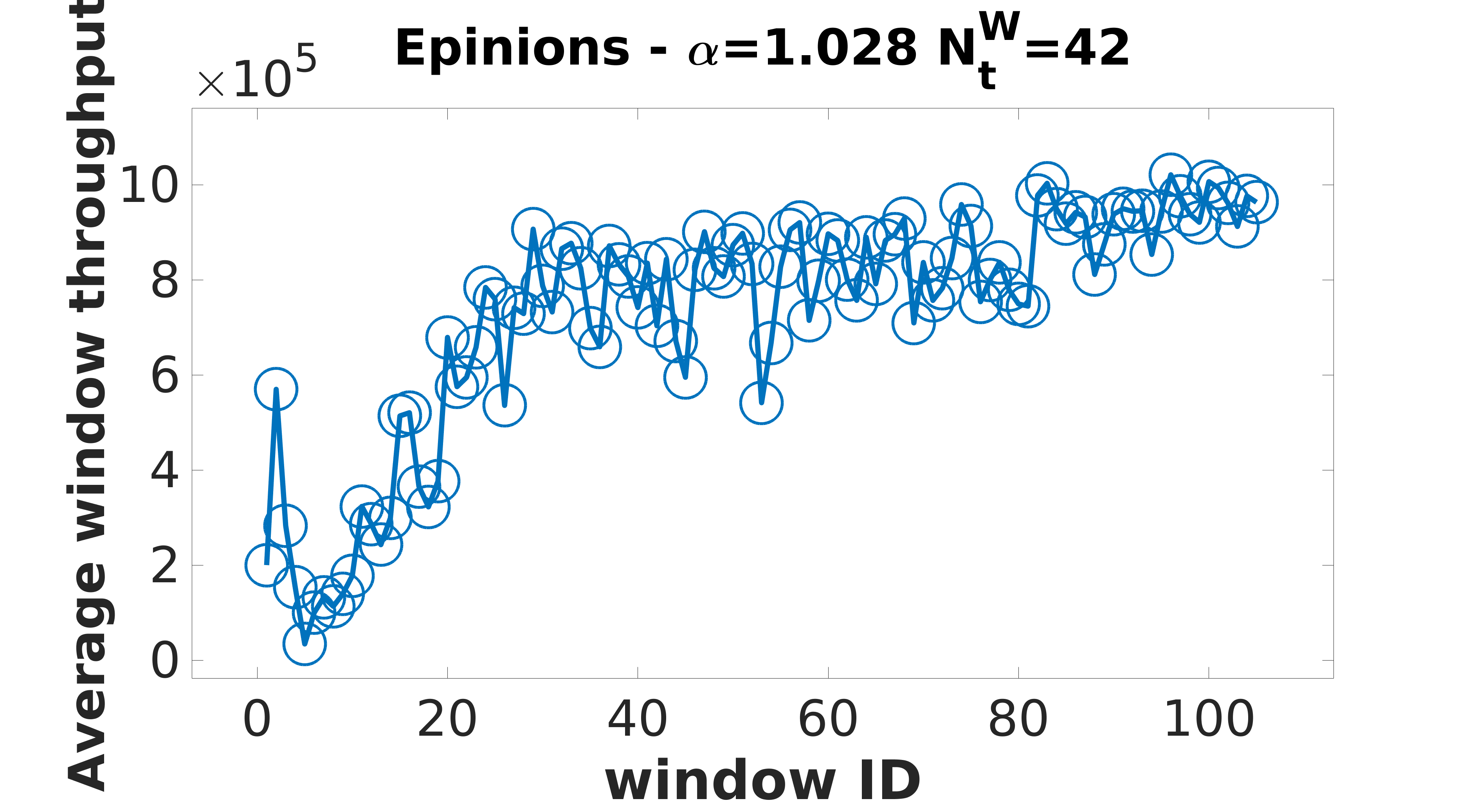}} \subfigure{\includegraphics[width=0.3\textwidth]{ 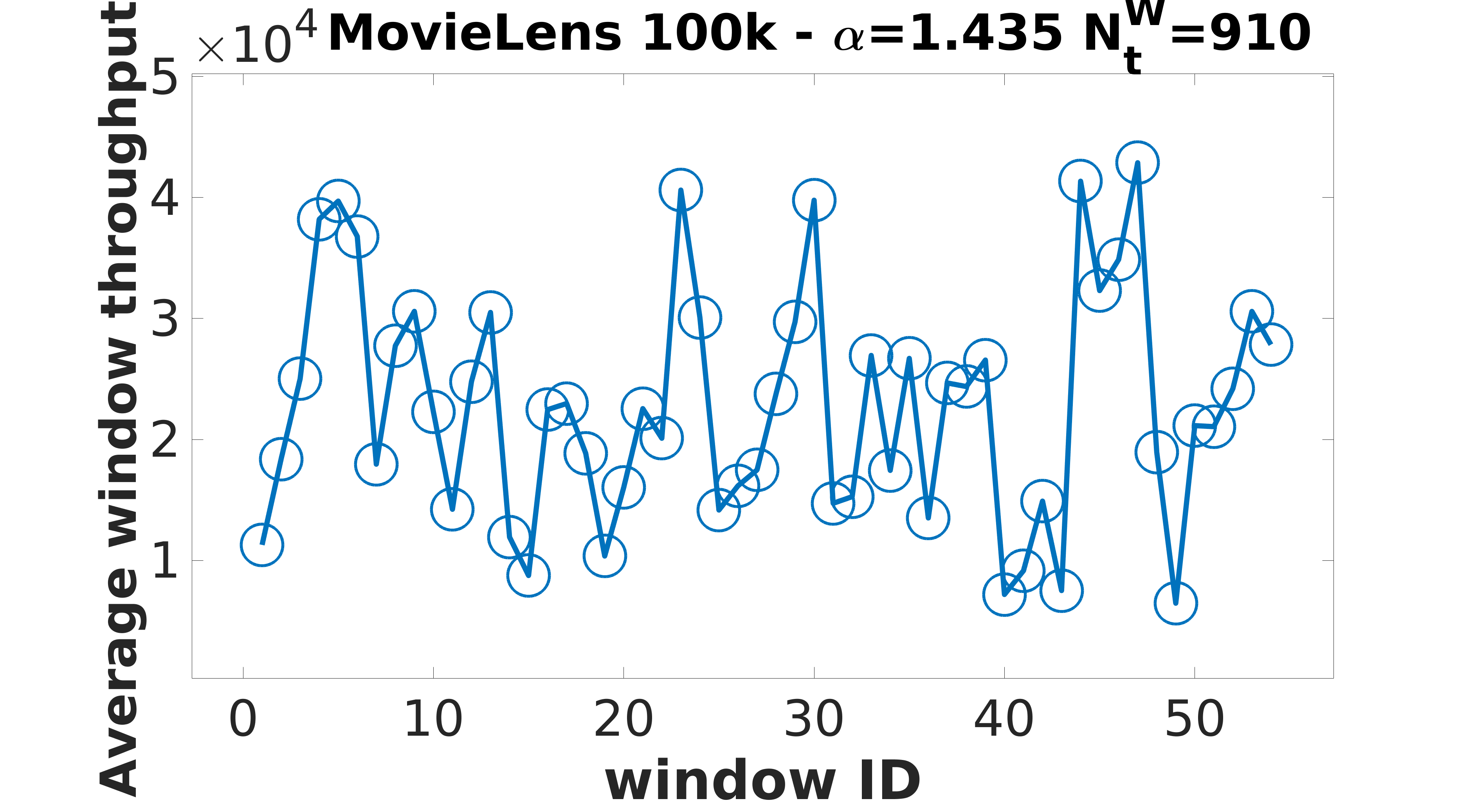}}
    \subfigure{\includegraphics[width=0.3\textwidth]{ 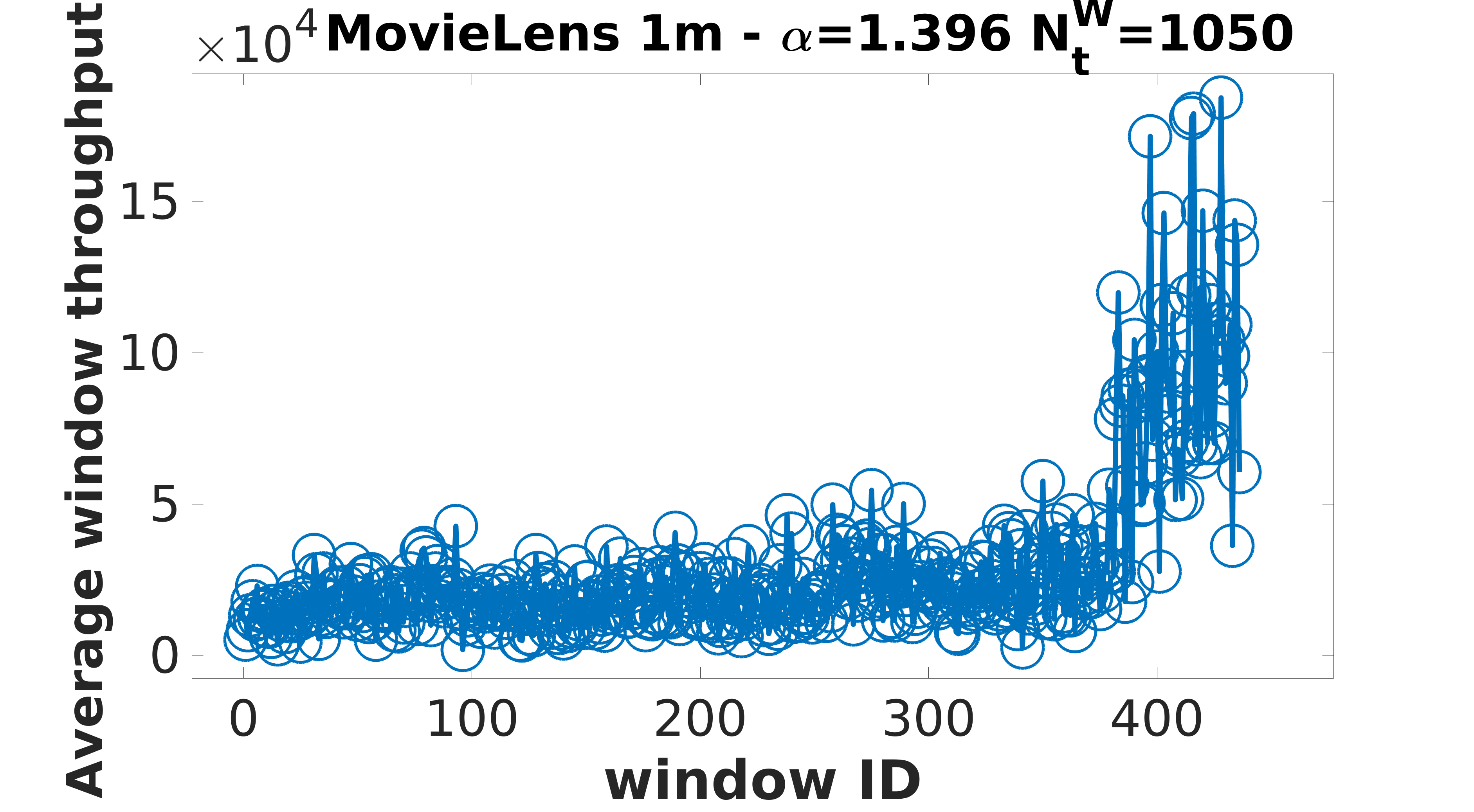}}
    \subfigure{\includegraphics[width=0.3\textwidth]{ 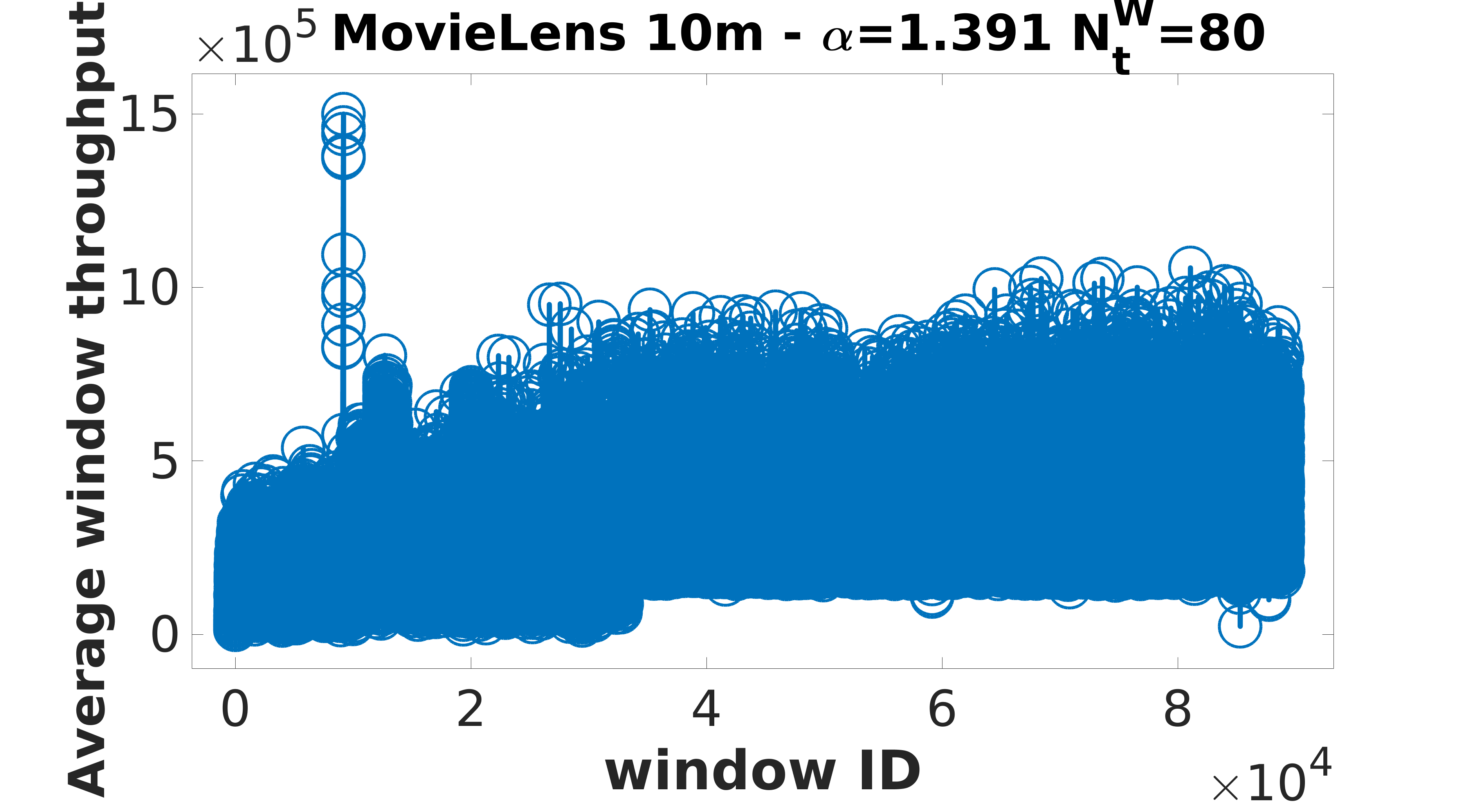}}
    \subfigure{\includegraphics[width=0.3\textwidth]{ 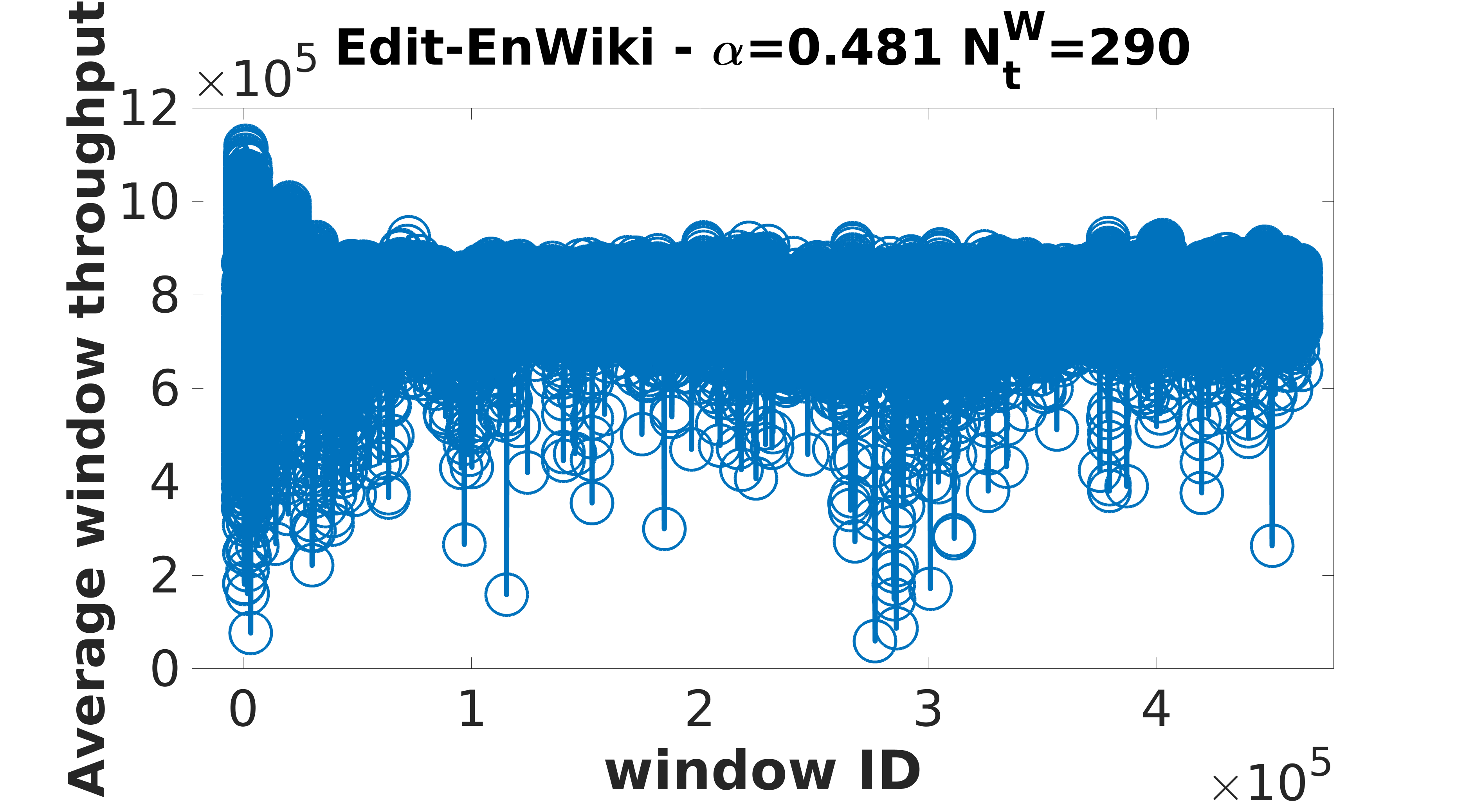}}
    \subfigure{\includegraphics[width=0.3\textwidth]{ 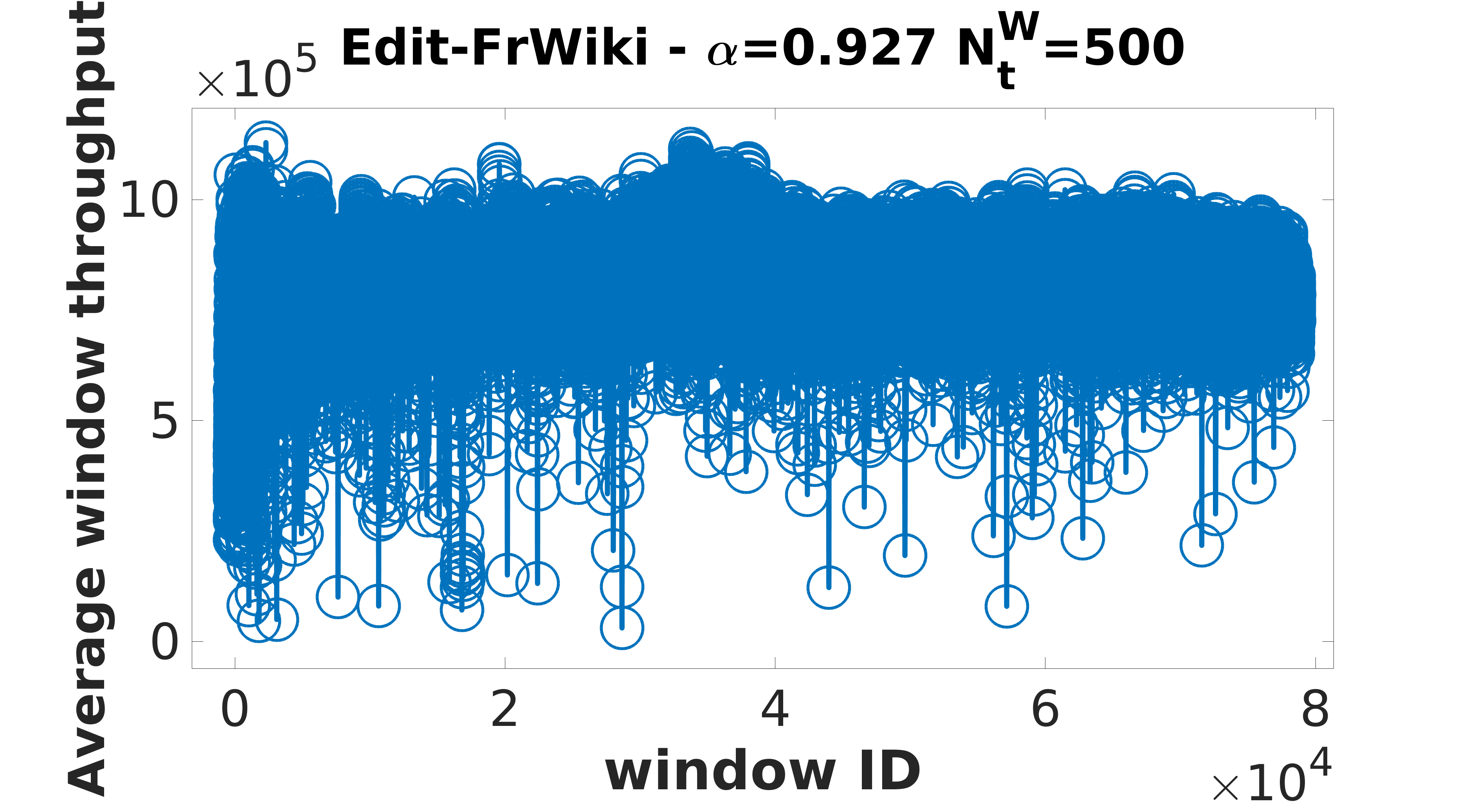}}
    \caption{Average window throughput (edge/s) of sGrapp-100 at the end of each window.}
   \label{fig:wthroughput100}
\end{figure*}

\section{Conclusion}
We studied the fundamental problem of dense bi-clique counting in streaming graphs. We introduced an effective and efficient framework for approximate butterfly counting, sGrapp. Following a data driven approach, we conducted extensive graph analysis to unveil the organizing principles of temporal butterflies in streaming graphs (the butterfly densification power law). These insights shed light on developing sGrapp algorithm. sGrapp utilizes a new exact counting core and a time-based windowing technique which adapts to the temporal distribution of the graph stream with no assumptions on the order and rate of stream, making it applicable to any real stream. 
sGrapp displays $MAPE<0.05$ in graph streams with almost uniform temporal distribution. The optimized version, called sGrapp-x, handles graph streams with non-uniform temporal distribution with MAPE below $0.14$. sGrapp-x lowers the minimum and maximum MAPE of sGrapp and also increases the probability of approximation error below $0.15$ and $0.2$, most notably in the densest graph streams. sGrapp variants perform much better than existing algorithms.

\clearpage

\bibliographystyle{plain}
\bibliography{main}

\end{document}